\pdfoutput=1 % to ensure arxiv compiles as PDF latex

\documentclass{ucdthesis} % based on Memoir class: https://mirror.mwt.me/ctan/macros/latex/contrib/memoir/memman.pdf
\usepackage{booktabs} % for nice tables
\usepackage{graphicx}
\usepackage{caption} %%% optional but probably necessary
\usepackage{subcaption} %%% optional but probably necessary
\usepackage{amsmath, amsthm, amssymb, amsfonts}
\usepackage{commath, mathtools, enumerate, dsfont,tikz-cd, hyperref, tikz, pgf, ytableau}
\usepackage{mathrsfs}
\usepackage{mathtools}
\usepackage{float}
\usepackage{makecell}
\usepackage[
   includehead,
   includefoot,
     left = 1in, 
      top = 1in, 
    right = 1in,
   bottom = 1in
]{geometry}
\usepackage{indentfirst}
\usepackage{setspace}
\usepackage{calc}
\numberwithin{figure}{chapter} 
\numberwithin{table}{chapter}
\numberwithin{equation}{chapter}
\numberwithin{section}{chapter}

\makeatletter
\newcommand*{\toccontents}{\@starttoc{toc}}
\makeatother

\usepackage{bm}

\usepackage{color}

\usepackage{enumitem}
\usepackage{mathtools}

\usepackage{hyperref}
\hypersetup{
    colorlinks=true,
    linkcolor=blue,
    filecolor=magenta,      
    urlcolor=teal,
    citecolor=cyan,
}

\newtheorem{theorem}{Theorem}[section]
\newtheorem{proposition}{Proposition}[section]
\newtheorem{corollary}{Corollary}[section]
\newtheorem{lemma}{Lemma}[section]
\newtheorem{remark}{Remark}[section]
\newtheorem{example}[theorem]{Example}
\newtheorem{definition}[theorem]{Definition}

\newtheorem{conjecture}[theorem]{Conjecture}

\usepackage{algorithm}
\usepackage[noend]{algpseudocode}
\algnewcommand\algorithmicinput{\textbf{Input:}}
\algnewcommand\algorithmicoutput{\textbf{Output:}}
\algnewcommand\Input{\item[\algorithmicinput]}%
\algnewcommand\Output{\item[\algorithmicoutput]}

\tikzstyle{tensor}=[rectangle,draw=blue!70,fill=blue!20,thick]
\tikzstyle{phys_sym}=[circle,draw=green!30,fill=green!20,thick]
\tikzstyle{bond_sym}=[circle,draw=purple!50,fill=purple!20,thick]

\DeclareMathOperator*{\slim}{s-lim}

\def\Ad{\mathrm{Ad}}
\def\ad{\mathrm{ad}}

\def\C{\mathbb{C}}

\def\g{\mathfrak{g}}
\def\h{\mathfrak{h}}

\def\u{\mathfrak{u}}

\def\N{\mathbb{N}}

\def\R{\mathbb{R}}

\def\sl{\mathfrak{sl}}
\def\gl{\mathfrak{gl}}
\def\so{\mathfrak{so}}
\def\su{\mathfrak{su}}

\def\Z{\mathbb{Z}}

\def\st{\, \vert \,}
\def\Aut{\mathrm{Aut}}

\def\ide{\texttt{id}}
\def\idty{\mathds{1}}

\def\Sym{\mathrm{Sym}}

\def\GL{\text{GL}}

\def\Hom{\mathrm{Hom}}
\def\im{\mathrm{im}\,}

\def\spec{\mathrm{spec} }

\def\eps{\varepsilon}

\def\SWAP{\texttt{SWAP}}

\def\calA{\mathcal{A}}
\def\calB{\mathcal{B}}
\def\calC{\mathcal{C}}

\def\calG{\mathcal{G}}
\def\calH{\mathcal{H}}

\def\calN{\mathcal{N}}

\def\calP{\mathcal{P}}
\def\calS{\mathcal{S}}
\def\calU{\mathcal{U}}

\def\bi{\textbf{i}}
\def\bj{\textbf{j}}
\def\bk{\textbf{k}}

\def\bbE{\mathbb{E}}
\def\bbF{\mathbb{F}}

\def\up{\uparrow}
\def\down{\downarrow}
\newcommand{\Exterior}{\mathchoice{{\textstyle\bigwedge}}%
	{{\bigwedge}}%
	{{\textstyle\wedge}}%
	{{\scriptstyle\wedge}}}

\newcommand{\wt}[1]{\widetilde{#1}}
\newcommand{\inprod}[1]{\left\langle #1 \right\rangle}
\newcommand{\bra}[1]{\left\langle #1 \right\vert}
\newcommand{\ket}[1]{\left\vert #1 \right\rangle}

\newcommand{\paran}[1]{\left( #1 \right)}
\newcommand{\brac}[1]{\left[ #1 \right]}

\newcommand{\Tr}{\mathrm{Tr}}
\newcommand{\sgn}{\mathrm{sgn}}

\let\emph\relax % hmm...need to fix this
\DeclareTextFontCommand{\emph}{\bfseries} 

\definecolor{gray}{rgb}{0.93,0.93,0.93}
\definecolor{light-gold}{rgb}{0.96,0.8,0}
\definecolor{light-red}{rgb}{1,0.4,0.4}
\definecolor{light-green}{rgb}{0.5,1,0.5}
\definecolor{light-blue}{rgb}{0.4,0.4,1}

\usepackage[
backend=biber,
style=alphabetic,
sorting=nyt
]{biblatex}
\usepackage{pdfpages}
\bibliography{refs}

\begin{document}
    %\includepdf[pages={1}]{Thesis_title_page.pdf} % Janky way of getting the desired first thing
   \frontmatter
   % Redefine plain page style so that the first pages of chapters have desired page style. Only needed for davis.
   %\pagestyle{plain}
    \pagestyle{headings} % use this for prettier arXiv

   \title{SO(n) AKLT Chains as Symmetry Protected Topological Quantum Ground States}
    \author{Michael Ragone}
    \degreesemester{Winter}
    \degreeyear{2024}
    \degree{Doctor of Philosophy}
    \chair{Bruno Nachtergaele}
    \othermembers{Greg Kuperberg \\ Martin Fraas }
    \numberofmembers{3}
    \field{Mathematics}
    \campus{Davis}
% Delete (or comment out) the \approvalpage line for the final version.

\begin{frontmatter}
\maketitle
% Delete (or comment out) the \approvalpage line for the final version.
% \approvalpage
% \copyrightpage

%\include{abstract}
\section*{Abstract}
This thesis studies a pair of symmetry protected topological (SPT) phases which arise when considering one-dimensional quantum spin systems possessing a natural orthogonal group symmetry. Particular attention is given to a family of exactly solvable models whose ground states admit a matrix product state description and generalize the AKLT chain. We call these models ``$SO(n)$ AKLT chains'' and the phase they occupy the ``$SO(n)$ Haldane phase''. We present new results describing their ground state structure and, when $n$ is even, their peculiar $O(n)$-to-$SO(n)$ symmetry breaking. We also prove that these states have arbitrarily large correlation and injectivity length by increasing $n$, but all have a 2-local parent Hamiltonian, in contrast to the natural expectation that the interaction range of a parent Hamiltonian should diverge as these quantities diverge.
We extend a definition of Ogata's~\cite{ogata2020classification} of an SPT index for a split state for a finite symmetry group $G$ to an SPT index for a compact Lie group $G$. We then compute this index, which takes values in the second Borel group cohomology $H^2(SO(n),U(1))$, at a single point in each of the SPT phases. The two points have different indices, confirming the two SPT phases are indeed distinct. Chapter~\ref{ch:Introduction} contains an introduction with a detailed overview of the contents of this thesis, which includes several chapters of background information before presenting new results in Chapter~\ref{ch:gapped_ground_states_phases} and Chapter~\ref{ch:SO(n)_Haldane_chains}.

\begin{acknowledgements}
I would first like to thank my wonderful wife Joyce and my family. Their love and support inspires me daily, and I could not ask for a better team. 

There are far too many people I would like to thank individually for their kindness. At risk of overfilling the section, I will stick to those who have served as long-term academic mentors.
My advisor Bruno Nachtergaele is a rare combination of wise, brilliant, kind, and generous. His support has been essential to my journey as a mathematician, and I am grateful for his insight and his time. My ``summer advisors'' Marco Cerezo, Carlos Ortiz Marrero, and Bojko Bakalov are all wonderful scientists and people alike, and I look forward to many years of collaboration and chats over excellent coffee with each of them. My thesis committee members Martin Fraas and Greg Kuperberg made this entire thesis possible, both thanks to their unique mathematical insights and their friendly openness to my many questions. I have been lucky enough to learn from many excellent mathematicians here at Davis: Roger Casals, John Hunter, Eugene Gorsky, Adam Jacob, Monica Vazirani, Michael Kapovich, Jerry Kaminker, and Andrew Waldron. Finally, I would like to thank a small handful of undergraduate professors from the University of Arizona who would lead me into a mathematical career: Klaus Lux, Jan Wehr, Jean-Marc Fellous, and Amanda Young, whose wonderful linear algebra course inevitably pointed me to the quantum world.  

Research supported in part by the National Science Foundation through DMS-1813149 and DMS-2108390. All new results are joint work with Bruno Nachtergaele.

\end{acknowledgements}

% You can delete the \clearpage lines if you don't want these to start on
% separate pages.
\tableofcontents
\end{frontmatter}

\pagestyle{headings} % use this for prettier arXiv submission
\pagenumbering{arabic}

\chapter{Introduction} \label{ch:Introduction}
In elementary school, we learn about the phases of matter: gas, liquid, solid. We then learn that the true story is vastly more complicated, but the elementary school story comes with some useful ideas which will serve as caricatures for themes in this thesis. Imagine a glass of water in a chamber where we may change the temperature and pressure. We start by gently changing these parameters, observing that most of the time, this only gently changes the properties of water. One sensible way to define phases in our experiment is by a sort of equivalence relation. We say that two states are in the same phase of matter if we may smoothly change the temperature and pressure to change one state into another. This ``smooth'' condition ensures that water, water vapor, and ice all fall into different phases--if we attempt to heat water into water vapor, we will encounter some form of derivative discontinuity in the energy of the system at the boiling point.

Now, we often think of phases of water not as equivalence classes of states, but in terms of qualitative and quantitative functions of these equivalence classes. For example, we might ask whether our water sample is rigid: if it is, we call it a solid, and if it isn't, we call it either a liquid or a gas. If we were nerdy mathematician types, we might say that rigidity is an invariant of the phase equivalence relation on states, meaning it does not change as we smoothly change the temperature and pressure. After all, ice is rigid both at $-10^\circ$C and at $-5^\circ$C, and liquid water is not rigid at $5^\circ$C and $10^\circ$C alike. By itself, this invariant is not quite rich enough to uniquely identify all three phases: it can distinguish solids from liquids and gases, but not liquids from gases. But by considering it alongside other invariants, like the ``does our sample have a well defined volume?'' invariant, we may obtain a complete classification of the three phases of water. This saves us from the arduous task of considering the enormous space of parameter paths which could feasibly connect two states of water, and instead converts the problem into checking these invariants at both points to determine whether they may occupy the same phase of matter.

This thesis is dedicated to the study of a particular pair of quantum phases of matter and their associated invariants. We work in the setting of quantum spin chains, where one has an infinite one-dimensional lattice $\Z$ with each site $x\in \Z$ supporting a quantum particle with finite degrees of freedom. These particles are frozen in place, and so the energy of the system is dictated by an interaction $\Phi$ which gives rise to a Hamiltonian $H$. When the interaction is sufficiently spatially localized, we may mathematically describe such a system as a $C^*$ dynamical system, which comes with an algebra of observables $\calA$ and a strongly continuous cocycle of automorphisms $\tau_{s,t}$ implementing dynamics. States are then positive linear functionals on this algebra, and we may study ground states $\omega$ which minimize the energy. We will restrict our attention to interactions $\Phi$ which possess a spectral gap above their ground state energy, and so do not have low lying energy states in the thermodynamic limit. 

In analogy with our water experiment, we may define a gapped ground state phases as an equivalence class of gapped interactions, where we think of these paths of interactions $\Phi(s)$ as depending on some ``experimental parameter'' $s$.
We have in mind one very particular phase diagram, which arises when we consider interactions which act only on nearest-neighbor particles and enjoy a certain symmetry under the orthogonal group $O(n)$.
\begin{figure}
    \centering
    \includegraphics{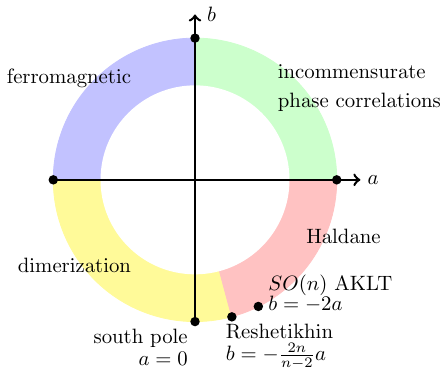}
    \captionof{figure}{The phase diagram for $O(n)$-invariant nearest-neighbor interactions (adapted from~\cite{bjornberg2021dimerization}). }
\end{figure} This phase diagram has been heavily studied over the past 30 years. In the $n=3$ case, this corresponds exactly to the bilinear-biquadratic chain phase diagram, where the south pole point is the $P^{(0)}$ chain~\cite{aizenman1994geometric}, the matrix product state point is the AKLT chain~\cite{affleck1988valence}, and the Reshetikhin point~\cite{reshetikin1983method} is the Takhtajan-Babujian model~\cite{babujian1982exact,takhtajan1982picture}. Our attention is particularly on the gapped ``dimerized'' and ``Haldane'' phases~\cite{tasaki2020physics}, separated by the gapless exactly solvable Reshetikhin point--as the colors suggest, it is physically expected that each of these constitutes a distinct gapped phase of quantum matter~\cite{tu2008class}. This thesis provides a piece of mathematical confirmation of this picture, proving that the south pole and the matrix product state points occupy distinct gapped ground state phases of matter (with an appropriate definition, that is). Much like the water experiment, the approach to demonstrating this is by finding invariants of gapped ground state phases and computing them at these points. In the odd $n$ case, a rather coarse invariant will suffice. In the even $n$ case, we will require a finer notion of equivalence which incorporates symmetry of our interaction--this gives rise to a definition of \textit{symmetry protected topological (SPT) phases} of quantum matter. It will be in this sense in which we can show that the dimerized and Haldane phases are distinct. In showing this, we provide a new generalization of Ogata's definition of an SPT index~\cite{ogata2020classification} which allows us to upgrade finite group symmetries to compact Lie group symmetries, fulfilling the predictions of physicists Chen, Gu, Liu, and Wen~\cite{chen2013symmetry}. 

We will spend most of our time trying to better understand the matrix product state point in the Haldane phase. Indeed, the lion's share of novel contributions of this thesis are results about the ground states at this point, which we have decided to call $SO(n)$ AKLT chains. These models, first studied by physicists Tu, Zhang, and Xiang in~\cite{tu2008class}, serve as a higher $n$ generalization of the AKLT chain, which is an exactly solvable model which historically allowed mathematicians and physicists to probe the behavior of the Haldane phase. Many of the results for these chains fundamentally rests upon their structure as representations of the Lie algebra $\so(n)$: indeed, the collection of matrices which define these matrix product states sit inside the Clifford algebra of rank-$n$, $\calC_n$, which is intimately related to the representation theory of the orthogonal groups.
We will give special attention to the even $n$ $SO(n)$ AKLT chains, which were only briefly discussed in the aforementioned work but possess a variety of unexpected structure.

Perhaps most notably, it was believed that the even $n$ models were dimerized models~\cite{tu2008class}. Dimerization refers to a type of symmetry breaking process where translation-invariance of an Hamiltonian is broken by a pair of 2-periodic ground states which are translates of one another. But it is more than this: the term ``dimerization'' suggests that this symmetry-breaking is essentially catalyzed by an interaction term which rewards spins for entangling exclusively with one neighbor (i.e. forming ``dimers''), whence the pair of ground states arises by preferentially choosing a neighbor with which to entangle. Indeed, this is the heuristic mechanism undergirding dimerization in the prototypical Majumdar-Ghosh model~\cite{majumdar1969next}, as well as the south pole model in the aptly named dimerized phase. One may characterize this behavior in two main ways: by observing an alternating strong-weak entanglement structure in the ground states which reflects the formation of ``dimers'', or by computing the expectation of a suitably chosen 2-local observable which may detect the presence or absence of a dimer bond. For the Majumdar-Ghosh chain, the latter is accomplished by choosing the projection onto the singlet and comparing the expectation in the ground states across e.g. sites $x=0$ and $x=1$. For the dimerized south pole point, there is an analogous natural choice of operator which distinguishes the two ground states~\cite{bjornberg2021dimerization}.

\begin{figure}
    \centering
    \includegraphics[scale=0.5]{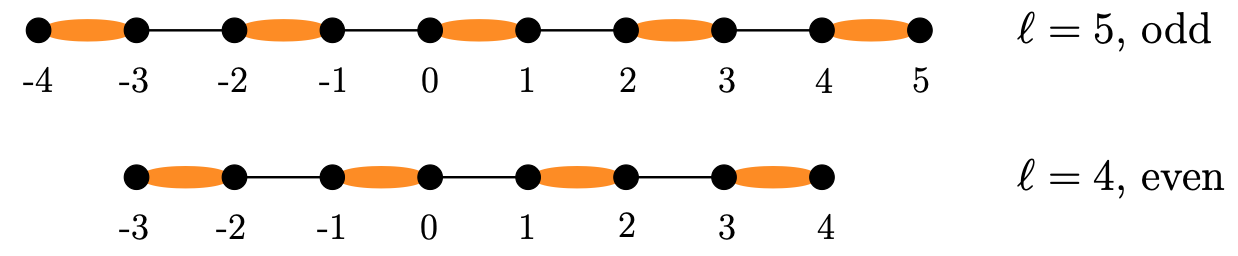}
    \captionof{figure}{Illustration of dimerization at the south pole point. Depending on whether $\ell$ is even or odd, the site $x=0$ is more entangled with either its left or its right neighbor. This gives rise to a pair of 2-periodic ground states in the thermodynamic limit (Figure from~\cite{bjornberg2021dimerization}). }
\end{figure}

The even $n$ $SO(n)$ AKLT chains are a different story, however: while they exhibit the same breaking of translation-invariance into period 2, the mechanism and result is fundamentally different from dimerization. The two states provably have identical entanglement structure, and when $n>4$, they are totally indistinguishably by any 2-local observable. To our knowledge, this is the first noted instance of this type of symmetry-breaking.

We mention a few results of independent interest. These even $n$ $SO(n)$ AKLT chains exhibit a different peculiar type of symmetry-breaking, which to our knowledge has not been observed in the literature: the $O(n)$ invariance of the interaction is broken into $SO(n)$ invariance of its ground states, which are related to one another by an on-site transformation of determinant $-1$. We prove that for all $n$, the $SO(n)$ AKLT chain matrix product states have a parent Hamiltonian given as a sum of 2-local interactions but an injectivity length and correlation length which may be made arbitrarily large by increasing $n$. This disproves an outstanding conjecture suggesting that the parent Hamiltonian of a family of matrix product states necessarily has interaction terms whose support scales with these quantities. Finally, we discover an alternative charge-parity-time reversal (CPT) symmetry breaking story for the even $n$ ground states.

\begin{section}{Outline and Novel Contributions}
    In Chapter~\ref{ch:quantum_spin_chains}, we review the mathematical framework which allows us to rigorously describe quantum spin systems. We begin by describing the formalism of finite quantum spin chains and quickly move to a collection of examples. The first few are quite trivial and serve mostly to clarify the notation and common computational tasks within the quantum spin system approach--these can be safely skipped by a more experienced practitioner. The last two however, the AKLT chain and Majumdar-Ghosh model, serve as important prototypes: the AKLT chain for $SO(n)$ chains, and the Majumdar-Ghosh model for the dimerization phenomenon. We then proceed to the full machinery of infinite quantum spin systems.

    In Chapter~\ref{ch:app-rep_theory}, we review some results from the representation theory of finite and compact Lie groups. Representation theory is a language of symmetry and will grant us a wide variety of powerful tools to better understand quantum systems and phases of quantum matter. Indeed, the formal definition of a quantum ``spin'' necessarily invokes the representation theory of $SU(2)$, the special unitary group of $2\times 2$ matrices. We have opted to spend a great deal of time in Section~\ref{sec:Spin and SU(2), by example} and Section~\ref{sec:Examples Revisited: Putting the ``Spin'' in Quantum Spin Systems} which together develop a representation-centric approach to quantum spin systems. While common language among the mathematical physics community, careful pedagogical development of these ideas can be difficult to find, so we thought it worthwhile to provide these sections and hope that it serves as a useful resource for new researchers. We then pivot to describing a bit of the representation theory of $SO(n)$--here, we take a more pragmatic approach, building only the theory we need to probe the behavior of the $SO(n)$ AKLT chains we will study in the final Chapter. Clifford algebras take a central role, both because they are used to later define $SO(n)$ AKLT matrix product ground states and because they are intimately tied to the representation theory of $SO(n)$. 

    In Chapter~\ref{ch:proj-reps}, we discuss the closely related notion of a projective representation of a finite or compact Lie group. We will later find that a certain cohomology theory associated to projective representations will provide a family of key invariants of symmetric gapped ground state phases.

    In Chapter~\ref{ch:AKLT_chain}, we discuss the Bilinear-Biquadratic phase diagram, which is a historical predecessor to and special $n=3$ case of the $O(n)$ phase diagram we showed earlier. Of essential importance is the AKLT chain, which is the $SO(3)$ AKLT chain. We probe three different equivalent perspectives to better understand this state: the valence bond state picture, the matrix product state picture, and the finitely correlated state picture. We also spend a little time with the much simpler Majumdar-Ghosh model, both to gain comfort with the machinery and to demonstrate what we mean by dimerization.

    In Chapter~\ref{ch:MPSs}, we briefly describe some tools from the theory of matrix product states. In particular, the $SO(n)$ AKLT chain ground states are matrix product states.

    In Chapter~\ref{ch:gapped_ground_states_phases}, we finally get to the main theme of the thesis. We provide a precise description of the phase diagram of $O(n)$-invariant nearest neighbor spin chains as well as formal definitions of gapped ground state phases and SPT phases given an on-site symmetry group $G$. This theory is well understood mathematically for the case of interactions with unique gapped ground states and a finite symmetry group $G$, thanks to the work of Ogata~\cite{ogata2020classification}. But in our case, we require a broader definition of SPT phase as an equivalence class of symmetric gapped interactions, originally appearing in~\cite{bachmann2014gapped}, and we allow our symmetries to be from compact Lie groups $G$. In Section~\ref{sec:index for split states}, we provide the new Theorem~\ref{thm:split state index} which proves that any 1D gapped ground state may be assigned a well-defined SPT invariant in $H^2(G,U(1))$ where $G$ is a compact Lie group. We also provide a computation of this index for MPSs in Section~\ref{sec:the H2 index for MPS}. Then, using this, Theorem~\ref{thm:SPT invariants of ground state space} guarantees that any gapped, symmetric interaction $\Phi$ may be assigned a collection of SPT invariants. We then compute for the first time this pair of invariants explicitly for the south pole interaction, finding that they are trivial SPT invariants as a Corollary of Theorem~\ref{thm:south pole ground state is trivial SPT}.

    In Chapter~\ref{ch:SO(n)_Haldane_chains}, we restrict our attention to the $SO(n)$ AKLT chains. Nearly all of the results proven in this chapter are new contributions.
    
    We prove in Section~\ref{sec:the ground states of the MPS point} that this interaction has a unique pure ground state for odd $n$ and two pure 2-periodic ground states for even $n$. The question of distinctness of these ground states in the even $n$ case is surprisingly subtle: Theorem~\ref{thm:omega_pm have identical entanglement} shows that these states have identical entanglement structure and Theorem~\ref{thm:w+ and w- are distinct dimerized states for l large enough} shows that these states have identical reduced density matrices for all chain lengths $\ell<n/2$, which in turn implies that no $\ell<n/2$-local observable may distinguish the two ground states. We further prove in Section~\ref{sec:symmetry breaking from O(n) to SO(n)} that these states exhibit a peculiar type of symmetry breaking, which to our knowledge has not been observed in the literature: the $O(n)$ invariance of the interaction is broken into $SO(n)$ invariance of its ground states, which are related to one another by an on-site transformation of determinant $-1$. This distinguishes them from the south pole ground states as $O(n)$ SPT phases. In the process of demonstrating $SO(n)$ invariance, we characterize the nature of the finite chain ground state spaces as representations of $SO(n)$, which is summarized in Corollary~\ref{cor:reps of odd ground state spaces} for odd $n$ and Corollary~\ref{cor:reps of even ground state spaces} for even $n$. This investigation ultimately allows us to explicitly calculate the SPT invariants associated to the model in Theorem~\ref{thm:SOn haldane chains have nontrivial index}, culminating in the central result of this thesis: Corollary~\ref{cor:south pole and MPS point are different SPT phases} guarantees that the south pole point and the $SO(n)$ AKLT point have for all $n$ different associated SPT invariants and so cannot occupy the same $SO(n)$ SPT phase, mathematically confirming the physical prediction that the Haldane and dimerized phases are distinct SPT phases. 
    
    In Section~\ref{sec:CPT symmetry} of Chapter~\ref{ch:SO(n)_Haldane_chains}, we further explore properties of these ground states. In Section~\ref{sec:CPT symmetry}, we observe an alternative symmetry breaking story for even $n$. For all even $n$, it is not hard to see that $\omega_\pm$ are both invariant under time-reversal symmetry. When $n\bmod 4 = 0$, the ground states $\omega_+$ and $\omega_-$ are invariant under charge conjugation and reflection parity symmetry, while when $n\bmod 4 = 2$, both symmetries are broken and these symmetries swap $\omega_+$ and $\omega_-$. This can be calculated directly, but it also admits a nice representation-theoretic description which reflects the behavior of $\so(n)$ irreps under the map which sends a representation to its dual. Thanks to the Lie algebra isomorphism $\su(4)\cong \so(6)$, we also obtain a connection in the $n=6$ case to a family of $SU(4)$-invariant quantum antiferromagnets with charge-conjugation symmetry breaking ground states studied by Affleck et. al.~\cite{affleck1991SU2n}.

    The final Section~\ref{sec:parent property (SO(n) chains)} generalizes the well-known result that for $n=3$, the Hamiltonian $H$ at the $SO(n)$ AKLT point is the parent Hamiltonian of the AKLT chain, meaning the finite chain AKLT MPSs are exactly the ground states of $H$ and there are no others. We upgrade this statement and prove in Theorem~\ref{thm:parent property (SO(n) chains)} that the Hamiltonian $H$ at the $SO(n)$ AKLT point, which is a sum of 2-local interactions, is the parent Hamiltonian of the family of $SO(n)$ AKLT MPSs. The proof fundamentally rests on the combination of the intersection property of frustration-free interactions combined with the structure of the MPSs as representations of $\so(n)$. This provides a counterexample to an outstanding conjecture via Corollary~\ref{cor:parent prop counterexample cor length}: in particular, $SO(n)$ AKLT MPSs are a family of MPSs whose correlation length and injectivity length may be made arbitrarily large by increasing $n$ but all have a parent Hamiltonian with nearest-neighbor interactions. 
\end{section}
\chapter{Mathematical Framework for Quantum Spin Systems}\label{ch:quantum_spin_chains}

In this chapter we briefly recall the mathematical formalism required to describe a quantum spin system as a $C^*$-dynamical system, a Heisenberg picture whose data consists of a $C^*$ algebra of observables $\calA$ and a strongly continuous one parameter group of automorphisms $\tau_t$ implementing time evolution. Quantum states are then describable as positive, normalized, linear functionals $\omega:\calA\to \C$.
This description should serve only as a refresher and is hardly comprehensive: we recommend the lecture notes~\cite{nachtergaele2016quantum} and the books~\cite{naaijkens2017quantum,tasaki2020physics}. For finite systems, one may find the closely related standard reference in quantum information~\cite{nielsen2010quantum} useful. 

\section{Finite Quantum Spin Chains}
Let $(\Lambda,d)$ denote a finite set equipped with a distance function. Our motivating example is the interval of integers of length $\ell$, $\Lambda = [1,\ell]\subseteq  \Z$ with the usual distance. To each site $x\in \Lambda$, we assign a copy of the finite dimensional Hilbert space $\calH_x := \C^d$ with a Hermitian inner product $\inprod{\cdot,\cdot}$, following the physicist's convention of skew-linearity in the first slot and linearity in the second slot.\footnote{It is common to allow for different on-site Hilbert spaces, but we will not need this.} Each such on-site Hilbert space $\calH_x$ corresponds to a \textit{qudit}, with the special case of $d=2$ earning the name \textit{qubit}. We may equivalently call a qudit a \textit{spin}-$s$ particle, where $d=2s+1$, relegating the discussion of spin and the representation theory of $SU(2)$ to Chapter \ref{ch:app-rep_theory}. The Hilbert space for the chain $\calH_\Lambda$ is then given by tensoring each qudit together:
\begin{equation}
    \calH_\Lambda := \bigotimes_{x\in \Lambda} \calH_x .
\end{equation}
Now, we can define an algebra of observables on each site $x$ by $\calA_x := M_d(\C)$, the $d\times d$ complex matrices. This forms a $C^*$ algebra when equipped with the adjoint operation $A\mapsto A^*\in \calA_x$ with respect to $\inprod{\cdot,\cdot}$ and the usual operator norm $\norm{A} = \sup_{v\in \C^d} \frac{\norm{Av}}{\norm{v}}$. The full algebra of observables is then given by
\begin{equation}
    \calA_\Lambda := \bigotimes_{x\in \Lambda} \calA_x . 
\end{equation} Given a subset $\Lambda_1\subseteq \Lambda$, we may identify the algebra $\calA_{\Lambda_1}$ of $\Lambda_1$ with the subalgebra $\calA_{\Lambda_1}\otimes \idty_{\Lambda / \Lambda_1}$.

\subsection{Dynamics}
The Hamiltonian of a quantum system is of great physical interest, encoding the energy of the system and serving as the generator of dynamics. We will define Hamiltonians via an \textbf{interaction}, a map $\Phi:\mathcal{P}(\Lambda)\to \calA_{\Lambda}$ which assigns subsets of $\Lambda$ to observables, with the restriction that $\Phi(X) = \Phi(X)^*\in \calA_X$ for every subset $X\in \mathcal{P}(\Lambda)$. For any volume $Z\subseteq \Lambda$, the Hamiltonian $H_Z$ described by the interaction $\Phi$ is
\begin{equation}
    H_Z = \sum_{X\in \mathcal{P}(Z)} \Phi(X). 
\end{equation} For finite volumes $\Lambda$, this gives a distinguished self-adjoint observable $H_\Lambda^* = H_\Lambda$, but the sum is generally not well-defined in the infinite case. Of special importance for this thesis are \textbf{nearest-neighbor} (and next-nearest neighbor) interactions, i.e. interactions for which $\Phi(X)=0$ whenever $X$ consists of more sites than immediate neighbors. In the examples at the end of this section, the Heisenberg and AKLT chains are given by nearest-neighbor interactions, and the Majumdar-Ghosh chain is given by a next-nearest-neighbor interaction. 

We will study the time evolution of a quantum system via Heisenberg dynamics, wherein the initial state $\omega$ is fixed while observables $A(t)\in \calA_{\Lambda}$ evolve over time $t$.\footnote{The alternative to Heisenberg picture is Schr\"odinger picture, where states evolve with fixed observables. The two are entirely equivalent.} Let $t\in \R$. The Heisenberg dynamics $\tau_t^\Lambda$ generated by a Hamiltonian $H_\Lambda$ (which is determined by the interaction $\Phi$) is an automorphism of the algebra $\calA_\Lambda$ given by
\begin{equation}
    \tau_t^{\Lambda}(A) = U^*_{\Lambda}(t) A U_\Lambda (t), \qquad \text{for all } A\in \calA_{\Lambda}, 
\end{equation} where $U_{\Lambda}(t)$ is the unitary operator 
\begin{equation}
    U_{\Lambda}(t) = e^{-itH_{\Lambda}}\in \calA_\Lambda . 
\end{equation} One can check by hand that the automorphisms $\tau_t^\Lambda$ solve the Heisenberg equation for observable time evolution:
\begin{equation}
    \frac{d}{dt} \tau_t^\Lambda(A) = i[H_{\Lambda}, \tau_t^{\Lambda}(A)]. 
\end{equation} The usual theorems from ODE suffice to show existence and uniqueness of the flow $\tau_t$. In the finite case, the family $U(t)$ forms a 1-parameter Lie group and $\tau_t^\Lambda$ forms a 1-parameter Lie group of $*$-automorphisms (compare to the Adjoint representation Example~\ref{ex:adjoint rep of su2}). When the Hamiltonian is a norm continuous function of time, $t\mapsto H_\Lambda(t)$, existence and uniqueness still grant dynamics described given a family of unitary operators $U(t,0)$ which satisfy the initial value problem
\begin{equation}
    \frac{d}{dt} U(t,0) = -iH(t) U(t,0), \qquad U(0,0) = \idty.
\end{equation} In this case, we may lose the group structure. However, $U(t,s)$ still satisfy the cocycle property $U(t,s)U(s,r) = U(t,r)$ for all $r,s,t\in \R$, so the Heisenberg dynamics $\tau_{t,s}$ form a cocycle of automorphisms defined by
\begin{equation}
    \tau_{t,s}(A) = U(t,s)^* A U(t,s),\qquad A\in \calA_\Lambda,
\end{equation} which satisfy the differential equation
\begin{equation}
    \frac{d}{dt} \tau_{t,0}(A) = i\tau_{t,0}([H(t), A]) = i[\tau_{t,0}(H(t)), \tau_{t,0}(A)].
\end{equation}

%--------------------------------
\subsection{States}
A \textbf{state} of a quantum spin system is a normalized, positive linear functional $\omega:\calA_{\Lambda}\to \C$, i.e. $\omega$ satisfies
\begin{equation}
    \omega(\idty) = 1, \text{ and } \omega(A^*A)\geq 0, \quad \text{ for all } A\in \calA_{\Lambda}.
\end{equation} When $\calH_{\Lambda}$ is finite dimensional, we may equip $\calA_{\Lambda}$ with the Hilbert-Schmidt inner product $\inprod{A,B} = \Tr A^*B$ and then the Riesz representation theorem guarantees that each state corresponds to a unique matrix $\rho$ such that for all $A\in \calA_\Lambda$
\begin{equation}
    \omega(A) = \Tr \rho A, \quad \text{ with } \Tr \rho = 1, \qquad \rho \geq 0 . 
\end{equation} We call $\rho$ the \textbf{density matrix} of the state $\omega$. 

One can easily see that the set of states forms a convex set in the space of linear functionals $C\subseteq \calA_\Lambda^*$. A \textbf{pure} state $\omega$ is an extremal point of $C$, i.e. if $\omega_1,\omega_2\in C$ and there is a $t\in (0,1)$ such that $\omega = t\omega_1 + (1-t)\omega_2$, then $\omega = \omega_1=\omega_2$. A \textbf{mixed} state is any state which is not pure. 

For finite dimensional $\calH_\Lambda$, pure states are exactly rank-1 orthogonal projections, i.e. there exists a vector $\ket{\psi}\in \calH_\Lambda$ such that $\omega = \ket{\psi}\bra{\psi}$. When this vector $\ket{\psi}$ is an eigenvector of a Hamiltonian $H_{\Lambda}$, i.e. $H_\Lambda \ket{\psi} = \lambda \ket{\psi}$, we call $\omega$ an eigenstate of $H_{\Lambda}$. Of particular importance are \textbf{ground states}, which correspond to eigenvectors of smallest eigenvalue of $H_{\Lambda}$:
\begin{equation}
    H_{\Lambda} \ket{\psi} = \min \text{spec}(H_{\Lambda}) \ket{\psi}. 
\end{equation} Of course, one may similarly allow ground states to be mixed, whence the condition becomes $\omega$ a ground state of $H_{\Lambda}$ if 
\begin{equation} \label{def:finite chain ground state}
    \omega(H_{\Lambda}) = \min \{ \omega(H_\Lambda) : \omega \text{ a state on } \calA_{\Lambda} \} .  
\end{equation}
We will write an appropriate generalization of this definition for infinite quantum spin systems later.

Finally, suppose our finite lattice is bipartite, splitting as $\Lambda = \Lambda_1 \cup \Lambda_2$. We call a state $\omega$ a \textbf{product state} if there exist states $\omega_1$ on $\calA_1$ and $\omega_2$ on $\calA_2$ such that $\omega(A_1\otimes A_2) = \omega_1(A_1)\omega_2(A_2)$ for all $A_1\in \calA_1$ and $A_2\in \calA_2$. We further call a state $\omega$ on $\calA_\Lambda$ separable if it is a convex combination of product states, and \textbf{entangled} if it is not separable. Quantifying entanglement is a business by itself, but let us give one such measure. Let $\omega$ be a pure state corresponding to $\psi\in \calH_{\Lambda_1}\otimes \calH_{\Lambda_2}$. Then the entanglement entropy $S_E(\psi)$ of $\omega$ is
\begin{equation}
    S_E(\omega) = -\Tr \rho_1 \log \rho_1, \qquad \text{where } \rho_1 = \Tr_{\calH_{\Lambda_2}} \ket{\psi}\bra{\psi} . 
\end{equation} One can show that $-\Tr \rho_1 \log \rho_1 = -\Tr \rho_2 \log \rho_2$, so this quantity is the same for both subsystems. One can also show that $\omega$ is a separable state iff $S_E(\psi) = 0$. Notice that given any pair of unitaries $U_1\in \calA_{1}$, $U_2\in \calA_2$, the action of $U_1\otimes U_2$ on $\calH_{\Lambda_1}\otimes \calH_{\Lambda_2}$ preserves the entanglement entropy, i.e.
\begin{equation} \label{eq:entanglement entropy invariant under local unitaries}
    S_E(U_1\otimes U_2 \, \psi) = S_E(\psi),
\end{equation} but it certainly is not invariant under a general $U\in \calA_1\otimes \calA_2$. Indeed, this feature of invariance under unitaries acting only on subsystems is a hallmark of any good measure of entanglement.

%--------------------------------
\subsection{Symmetry}\label{sec:symmetry}
Take a finite volume $\Lambda$ with dynamics $\tau_t^{\Lambda} = e^{itH_{\Lambda}}$. A \textbf{symmetry} $\alpha$ is an algebra automorphism $\alpha$ on $\calA_\Lambda$ which commutes with the dynamics:
\begin{equation} \label{def:sym commutes dynamics}
    \tau_t^{\Lambda} \circ \alpha = \alpha \circ \tau_t^{\Lambda} , \qquad \text{for all } t\in \R.
\end{equation} Differentiating with respect to $t$ and letting $\delta(\cdot) = [H_\Lambda, (\cdot)]$ denote a linear map on $\calA_\Lambda$, we see that this is equivalent to 
\begin{equation}
    \delta\circ \alpha = \alpha \circ \delta .
\end{equation} There are two important classes of symmetries. A \textbf{lattice symmetry} is a symmetry given by a representation of a symmetry of the underlying lattice $\Gamma$. Translation invariance is such an example, as is the ``lattice reflection symmetry''. For instance, let $\Lambda = [-\ell, \ell]$. Then the reflection group symmetry is implemented by the group $\Z_2$, where the nontrivial element acts as reflection about $0\in [-\ell,\ell]$. This then induces a map $\alpha:\Z_2\to \text{Aut}(\calA_{\Lambda})$ which sends $\alpha_1 = \ide$ and $\alpha_g(A_{-\ell}\otimes \dots \otimes A_{\ell}) = A_{\ell}\otimes \dots\otimes A_{-\ell}$.

Given a unitary representation of a Lie group $U:G\to \calU(\calH_x)$ on an on-site Hilbert space,\footnote{For the unfamiliar, we will discuss representations in Chapter~\ref{ch:app-rep_theory}.} an \textbf{on-site} or \textbf{local symmetry} is the representation $\alpha_{(\cdot)}:G\to \text{Aut}(\calA_{\Lambda})$ given by
\begin{equation} \label{def:unitary conjugation symmetry}
    \alpha_g(A) = \paran{\bigotimes_{x\in X} U_g} A \paran{\bigotimes_{x\in X} U_g^*} , \qquad A\in \calA_{X} .
\end{equation} Note that when we work on pure states of finite chains $\calH_{[1,\ell]}$, this corresponds to studying tensor representations $U^{\otimes \ell}$ of $G$. In practice, such symmetries are typically identified at the level of the interaction $\Phi:P(\Lambda)\to \calA_\Lambda$ defining a Hamiltonian by verifying the invariance condition 
\begin{equation}
    \alpha(\Phi(X)) = \Phi(X) , \qquad \text{ for all } X\in \calP(\Lambda),
\end{equation} which then implies (\ref{def:sym commutes dynamics}). We will see a concrete example demonstrating that the AKLT interaction is $SU(2)$-invariant at the end of the chapter in Example \ref{ex:checking for symmetry}.

Let us describe a special case of broken and unbroken symmetry before the later full definition in Section~\ref{sec:infinite quantum spin chains setup}.
Let $\Lambda = [1,\ell]$. Let $\alpha$ be the symmetry defined by (\ref{def:unitary conjugation symmetry}), and suppose we have a pure eigenstate $\omega$ corresponding to $\psi\in \calH_{[1,\ell]}$ which is an eigenvector of eigenvalue $\lambda$ of an $\alpha$-symmetric Hamiltonian $H_\Lambda$, i.e. $\alpha(H_\Lambda) = H_\Lambda$. Then
\begin{equation}\label{eq:symmetry eigenspaces}
    H_{\Lambda} U_g^{\otimes \ell} \ket{\psi} = U_g^{\otimes \ell} H_{\Lambda} \ket{\psi} = \lambda U_g^{\otimes \ell}\ket{\psi} , 
\end{equation} i.e. $U_g^{\otimes \ell}\ket{\psi}$ is also a $\lambda$ eigenstate of $H_\Lambda$, although not necessarily equal to $\ket{\psi}$ itself. In other words, the $\lambda$-eigenspaces $V_\lambda$ of $H_{\Lambda}$ are invariant subspaces of $U_g^{\otimes \ell}$. In particular, this is true for the ground state space. If $U_g^{\otimes \ell}\ket{\psi} \in \C \ket{\psi}$ for every $\ket{\psi}$ corresponding to a pure ground state and all $g\in G$, we say the symmetry $\alpha$ is unbroken. If instead there exists a vector $\ket{\psi}$ corresponding to a pure ground state and a group element $g\in G$ such that $U_g^{\otimes \ell}\ket{\psi} \not\in \C \ket{\psi}$, we say the symmetry $\alpha$ is (spontaneously) broken. 

\section{Examples}
We now describe some simple examples of quantum spin chains. The first example is a single particle problem, and the second consists of multiple independent single-particle problems. 

Before we begin, now is a good time to define the Pauli spin matrices, which are the following $2\times 2$ complex matrices:\footnote{In the quantum computing literature, most authors prefer to call the Pauli matrix $X$ instead of $\sigma^X$, and likewise for $Y,Z$.}
\begin{equation}
    \sigma^X = \begin{pmatrix} 0 & 1 \\ 1 & 0 \end{pmatrix}, \quad \sigma^Y = \begin{pmatrix} 0 & -i \\ i & 0 \end{pmatrix}, \quad \sigma^Z = \begin{pmatrix} 1 & 0 \\ 0 & -1 \end{pmatrix} . 
\end{equation} These self-adjoint operators enjoy the following commutation relations:
\begin{equation} \label{eq:pauli commutation(1st chap)}
    [\sigma^X, \sigma^Y] = 2i\sigma^Z , \quad [\sigma^Y, \sigma^Z] = 2i\sigma^X, \quad [\sigma^Z, \sigma^X] = 2i\sigma^Y . 
\end{equation} The Lie algebra $\su(2)$ may be defined as the (real) span of the skew-adjoint matrices $\{i\sigma^X,i\sigma^Y,i\sigma^Z\}$. In Chapter~\ref{ch:app-rep_theory}, we will mathematically define ``spin'' in terms of $\su(2)$.

\begin{example}{A single spin-1/2 particle}

Let $\calH = \C^2$, a single spin-1/2 particle (or qubit), and let $\ket{0},\ket{1}$ be an orthonormal basis of $\calH$. We will take as an algebra of observables $\calA = M_2(\C)$, the $2\times 2$ complex matrices. Consider the state $\omega:\calA\to \C$ corresponding to the vector $\ket{0}$, i.e. we have 
\[
    \omega(A) = \bra{0} A \ket{0} = \Tr \ket{0}\bra{0} A ,\qquad \text{for all } A\in \calA.
\]  The Pauli matrices $\sigma^X,\sigma^Y,\sigma^Z$ are self-adjoint, so we may take any of them to be our Hamiltonian $H$. Let us pick $H = \sigma^X$. Here, when we say that $H$ is the infinitesimal generator of dynamics, we mean that exponentiation yields a continuous 1-parameter group of unitaries $\{U(t): t\in \R\}$ given by
\begin{equation}
    U(t) = e^{-it H} = \cos(t)\idty - i\sin (t) H . 
\end{equation} The 1-parameter family of automorphisms $\tau_t$ which describe dynamics is then given by $\tau_t(A) := U^*(t) A U(t)$ where $A\in \calA$. We may then describe the expectation of an observable $A$ in the state $\omega$ at time $t$ by writing $\omega(\tau_t(A))$. Let us pick two interesting observables to measure: $A_0 := \ket{0}\bra{0}$ and $A_1:= \ket{1}\bra{1}$. Computing, we see that
\begin{align*}
\omega(\tau_t(A_0)) &= \bra{0} e^{itH} \ket{0}\bra{0} e^{-itH} \ket{0} = (\cos t)(\cos t) = \cos^2 t . \\
\omega(\tau_t(A_1)) &= \bra{1} e^{itH} \ket{0}\bra{0} e^{-itH} \ket{1} = (i\sin t)(-i\sin t) = \sin^2 t.
\end{align*}

Before leaving this example, we consider a different state $\omega_{-}$ corresponding to the vector $\ket{-}:= \frac{1}{2}\paran{\ket{0}-\ket{1}}$, i.e.
\[
    \omega_{-}(A) = \Tr \ket{-}\bra{-} A, \qquad \text{for all } A\in\calA.
\] The Hamiltonian $H= \sigma^X$ has two eigenvalues, $\pm 1$. This state is in fact an eigenstate whose corresponding energy is $-1$, since $H\ket{-} = -\ket{-}$. So $\omega_-$ is a ground state for this Hamiltonian. Eigenstates are stationary states since e.g. $U(t)\ket{-} = e^{-iHt}\ket{-} = e^{it}\ket{-}$ and so
\[
    \omega(\tau_t(A)) = \omega(A), \qquad \text{for all } A\in\calA .    
\] In other words, the state $\omega_-$ is invariant under the action of the 1-parameter symmetry group of time translation automorphisms $\tau_t$.
\end{example}

We can now construct a quantum system of $n$ spin 1/2 particles rotating independently. This is in essence exactly as exciting as the previous example, but serves to introduce common notation.

\begin{example}{Just transverse fields}\label{ex:just transverse fields}

    Let $\Gamma = \Z$ and let $\Lambda = [a,b]\subseteq \Gamma$ be any finite interval. Let the on-site Hilbert space for every $x\in \Gamma$ be a spin-1/2 particle $\calH_x = \C^2$, so the chain $\Lambda$ has
    \begin{equation}
        \calH_{\Lambda} = \bigotimes_{x=a}^b \C^2, \qquad \calA_{\Lambda} = \bigotimes_{x=a}^b M_2(\C) . 
    \end{equation}
    Let us write a Hamiltonian (and implicitly define our interaction $\Phi$). Letting  $i\in\{X,Y,Z\}$ and site $x\in \Lambda$ we identify as before with the algebras:
    \[
        \sigma_x^i = \idty\otimes\dots\otimes \idty \otimes \underbrace{\sigma^i}_{site \; x} \otimes \idty \otimes \dots \otimes \idty , 
    \]i.e. the only nontrivial operator acts on the $x$th factor of $\calA_{\Lambda}$. Then we can write a ``transverse field'' Hamiltonian
    \begin{align}
        H_{tf} = \sigma_a^Z + \sum_{x=a+1}^{b} \sigma_{x}^X ,
    \end{align} where we have made the boundary term $x=a$ a different Pauli for color. The terms of this sum of self-adjoint operators commute because they act on different tensor factors, and so we may simultaneously diagonalize them to diagonalize $H_{tf}$. The ground state of this model corresponds to the lowest energy eigenvector of $H_{tf}$, which can be checked to be
    \[
    \ket{1}\otimes \ket{-}^{\otimes b-a-1} \in \calH . 
    \]
\end{example}

So, since the interaction terms all commute with one another, diagonalizing the Hamiltonian (and so computing the ground state) was just as easy as diagonalizing each term (and computing $n$ separate ground states). 

The following examples are significantly richer. We will encounter them repeatedly through this thesis, and we will need tools to glean insight into their behavior and ground state structure. The representation theory of $SU(2)$ will play a key role, allowing us to rigorously work with the notion of spin to understand these interactions. But to provide some preliminary motivation, briefly recall the classical Heisenberg model from statistical mechanics. In the isotropic 1D case with no transverse field, one considers the finite lattice $\Lambda = [a,b]\subseteq \Z$ with each site $x$ storing a unit vector $\vec{s}_x \in \R^3$, a ``spin''. Given a configuration $\vec{s}$ which assigns a spin $\vec{s}_x$ to each site $x$, the Hamiltonian $H_{Is}$ in this context is the real-valued function of the configuration $\vec{s}$ given by
\begin{equation} \label{def:classical heisenberg}
    H_{Is}(\vec{s}) = - J \sum_{x = a}^b \vec{s}_x \cdot \vec{s}_{x+1}.
\end{equation} where $J\in \R$ is a parameter. When $J>0$ the interaction is called \textbf{ferromagnetic}, and in this case configurations $\vec{s}$ with many neighboring aligned spins are rewarded with lower energy $H(\vec{s})$. When $J<0$ the interaction is called \textbf{anti-ferromagnetic}, wherein neighboring anti-aligned spins are rewarded with lower energy. One can think of this model as a version of the Ising model with continuous local spins instead of discrete local spins. With this in mind we return to quantum spin chains to introduce our first nontrivial example.

\begin{example}{The Heisenberg Chain} \label{ex:heisenberg chain}

As in Example \ref{ex:just transverse fields}, let $\Gamma = \Z$ and let $\Lambda = [a,b]\subseteq \Gamma$ be any finite interval. Let the on-site Hilbert space for every $x\in \Gamma$ be a spin-1/2 particle $\calH_x = \C^2$, so the chain $\Lambda$ has
\begin{equation}
    \calH_{\Lambda} = \bigotimes_{x=a}^b \C^2, \qquad \calA_{\Lambda} = \bigotimes_{x=a}^b M_2(\C) . 
\end{equation}

Now, letting $J\in \R$ a real parameter, the quantum Heisenberg Hamiltonian on volume $\Lambda$ is given by 

\begin{equation}\begin{split}
    H_{\Lambda} &= -J \sum_{x=a}^{b-1} \vec{\sigma}_x \cdot \vec{\sigma}_{x+1} \\
    &= -J \sum_{x=a}^{b-1} \sigma_x^X \sigma_{x+1}^X + \sigma_x^Y \sigma_{x+1}^Y + \sigma_x^Z \sigma_{x+1}^Z .
\end{split} \end{equation} When $J>0$, this is called the \textit{ferromagnetic} Heisenberg chain, and when $J<0$, this is called the \textit{anti-ferromagnetic} Heisenberg chain. 
\end{example}

This nearest-neighbor interaction is already much more interesting than Example \ref{ex:just transverse fields}, both since the terms in the Hamiltonian do not commute with one another and since distinct qubits now interact with their neighbors.

The next example, first studied by physicists Majumdar and Ghosh in the seminal work~\cite{majumdar1969next}, will eventually serve as a prototype for ground state dimerization: we will later see that this translation-invariant Hamiltonian will have two 2-periodic ground states which look like translates of one another, and this is arguably the simplest example of this behavior.

\begin{example}{The Majumdar-Ghosh Chain}\label{ex:Majumdar-Ghosh chain}

    Let $\Lambda = [a,b]$ denote a finite chain of spin-1/2 particles, i.e. the identical $\calH_\Lambda$ and $\calA_{\Lambda}$ to the previous Examples \ref{ex:just transverse fields} and \ref{ex:heisenberg chain}. Now, define the Hamiltonian $H_{\Lambda}$ via a next-nearest-neighbor interaction $h^{(MG)}_{x,x+1,x+2}$ and positive real parameter $J>0$:
    \begin{equation} \begin{split}
        H_{\Lambda} &= \sum_{x=a}^{b-2} h^{(MG)}_{x,x+1,x+2} \\
        &= J \sum_{x=a}^{b-2} 2\vec{\sigma}_x\cdot \vec{\sigma}_{x+1} + \vec{\sigma}_x\cdot \vec{\sigma}_{x+2} . 
    \end{split} \end{equation} One may interpret this as a sort of sum of two quantum anti-ferromagnetic Heisenberg interactions: the stronger $2\vec{\sigma}_x \cdot \vec{\sigma}_{x+1}$ rewards neighboring anti-alignment, and the weaker $\vec{\sigma}_x \cdot \vec{\sigma}_{x+2}$ rewards anti-alignment at a small distance. 
\end{example}

Our final example is perhaps the most important for this thesis. The AKLT chain was studied by Affleck, Kennedy, Lieb, and Tasaki in~\cite{affleck1988valence}. It is a key exactly solvable model which both revolutionized the study of gapped quantum systems by establishing the conjectured existence of the Haldane phase and led to the invention of matrix product states, one of the most important classes of states in quantum science. It is also the $n=3$ case of a larger family of $O(n)$-invariant spin chains, the main phase diagram we intend to study throughout this thesis.

Before proceeding, we define the ``spin-1 Pauli matrices'' acting on $\C^3$:
\begin{equation}
    S^X = \frac{1}{\sqrt{2}}\begin{pmatrix}
    0 & 1 & 0 \\
    1 & 0 & 1 \\
    0 & 1 & 0 
    \end{pmatrix}, \quad S^Y = \frac{1}{\sqrt{2}}\begin{pmatrix}
    0 & -i & 0 \\
    i& 0 & -i \\
    0 & i & 0 
    \end{pmatrix}, \quad S^Z = \begin{pmatrix}
    1 & 0 & 0 \\
    0 & 0 & 0 \\
    0 & 0 & -1 
    \end{pmatrix} . 
\end{equation} It is worth mentioning that the explicit matrices really are not so important. What \textit{is} important is that they enjoy the same commutation relations (\ref{eq:pauli commutation(1st chap)}) as the (rescaled) Paulis $\frac{1}{2}\sigma^X, \frac{1}{2}\sigma^Y,\frac{1}{2}\sigma^Z$:
\begin{equation}
    [S^X, S^Y] = iS^Z , \quad [S^Y, S^Z] = iS^X, \quad [S^Z, S^X] = iS^Y . 
\end{equation} Mathematically, we are saying that there is a Lie algebra homomorphism $\frac{1}{2}\sigma^i \mapsto S^i$, $i=X,Y,Z$, i.e. a spin-1 representation of $\su(2)$. Let us define the AKLT chain.

\begin{example}{The AKLT Chain}\label{ex:AKLT chain}

Let $\Gamma = \Z$ and let $\Lambda = [a,b]\subseteq \Gamma$ be any finite interval. Let each on-site Hilbert space be a spin-1 particle $\calH_x = \C^3$, and so
\begin{equation}
    \calH_{\Lambda} = \bigotimes_{x=a}^b \C^3, \qquad \calA_{\Lambda} = \bigotimes_{x=a}^b M_3(\C) . 
\end{equation} Define the AKLT Hamiltonian by the interaction $h^{(AKLT)}_{x,x+1}\in M_3(\C)\otimes M_3(\C)$:
\begin{equation} \begin{split}
    H_{\Lambda} &= \sum_{x=a}^{b-1} h^{(AKLT)}_{x,x+1} \\
    &= \sum_{x=a}^{b-1} \frac{1}{3}\idty + \frac{1}{2} \vec{S}_x\cdot \vec{S}_{x+1} + \frac{1}{6} (\vec{S}_x\cdot \vec{S}_{x+1})^2 . 
\end{split} \end{equation} The initial appearance of this Hamiltonian is rather opaque, but these terms are carefully tuned so that $h_{x,x+1}^{(AKLT)}$ is exactly an orthogonal projector onto a special 5-dimensional subspace corresponding to a spin-2 particle (we will discuss spin in Chapter \ref{ch:app-rep_theory}). In this way, one obtains an exactly solvable model. The $\frac{1}{3}\idty$ term serves only to shift the spectrum of $H_{\Lambda}$ by a constant and so does nothing to affect ground state structure. Notice that after omitting the ``biquadratic'' term $(\vec{S}_x\cdot \vec{S}_{x+1})^2$, the remaining ``bilinear'' interaction $\vec{S}_x\cdot \vec{S}_{x+1}$ is a spin-1 version of the Heisenberg chain in Example \ref{ex:heisenberg chain}. 

\end{example}
Indeed, the primary motivation for constructing the AKLT Hamiltonian was to study the conjectured Haldane phase of matter containing the spin-$1$ antiferromagnetic Heisenberg chain ground state. By adding a small biquadratic term, one obtains a ``close'' exactly solvable model whose properties can be rigorously determined. Interestingly, while it has long been known that the AKLT chain confirms the existence of the Haldane phase--it has a unique gapped ground state with exponential decay of correlations, words which will have meaning in the next chapter--the case of the spin-1 Heisenberg antiferromagnet remains unproven, despite hefty numerical evidence and a variety of advancements in the mathematical toolbox for proving spectral gaps. 

%--------------------------------------------------               
\section{Infinite Quantum Spin Chains} \label{sec:infinite quantum spin chains setup}             
Let $(\Gamma,d)$ be a countable lattice with a metric $d$. Of particular importance for this thesis is the 1D spin chain $\Gamma = \Z$ (with the usual $d(x,y) = \abs{x-y}, x,y\in \Z$). 
Let each on-site Hilbert space $\calH_x$ be given by a spin-$s$ particle, $\calH_x = \C^{2s+1}$ and each on-site $C^*$ algebra given by matrices $\calA_x = M_{2s+1}(\C)$, equipped with the usual adjoint and operator norm. As before, for each finite volume $\Lambda\subseteq \Gamma$, the Hilbert space and algebra of observables are given by
\begin{equation}
    \calH_{\Lambda} = \bigotimes_{x\in \Lambda} \calH_x, \qquad \calA_{\Lambda} = \bigotimes_{x\in \Lambda} \calA_x . 
\end{equation} Again, given a subset $\Lambda_1\subseteq \Lambda$, we may identify the algebra $\calA_{\Lambda_1}$ as a subset of $\calA_{\Lambda}$. With this identification in mind, the \textbf{algebra of local observables} $\calA_{loc}$ on $\Gamma$ is given by
\begin{equation}
    \calA_{loc} := \bigcup_{\substack{ \Lambda \subseteq \Z, \\ \Lambda \text{ finite}}} \calA_{\Lambda} . 
\end{equation} Taking a norm closure of $\calA_{loc}$, we obtain the $C^*$ \textbf{algebra of quasi-local observables}
\begin{equation}
    \calA := \overline{\calA_{loc}}.
\end{equation} Occasionally, when it is important to explicitly write $\Gamma$, we may write $\calA = \calA_\Gamma$. This algebra is a $C^*$ algebra with the algebra of local observables as a dense subset, and we will be formulating our theory in terms of this algebra. Note that in the finite volume $\Lambda$ case, the closure is redundant and $\calA_{\Lambda} = \overline{\calA_{\Lambda}}$.

Now, given an interaction $\Phi: \calP_0(\Gamma)\to \calA$, we would like to define dynamics on an infinite system. To do so we require some locality assumptions on the interaction $\Phi$. So long as $\Phi$ is sufficiently local, given an exhaustive sequence of finite volumes $\Lambda \nearrow \Gamma$, the finite volume dynamics $\tau_t^\Lambda$ will converge to a strongly continuous one-parameter group of automorphisms $\tau_t: \calA\to \calA$ (or if $\Phi(t), t\in [0,1]$ a sufficiently regular path of interactions, a strongly continuous cocycle of automorphisms $\tau_{t,s}$). We may quantify this locality by way of an $F$-norm. Let $F:[0,\infty)\to (0,\infty)$ be called an \textbf{F-function} if it satisfies the following properties:
\begin{enumerate}
    \item (Non-increasing): for $0\leq r\leq s$, we have $F(r)\geq F(s)$.
    \item (Uniform integrability): 
    \begin{equation}
        \norm{F}:= \sup_{x\in \Gamma} \sum_{y\in \Gamma} F(d(x,y)) < \infty .
    \end{equation}
    \item (Convolution condition): There exists $C_F<\infty$ such that for any $x,y\in \Gamma$,
    \begin{equation}
        \sum_{z\in \Gamma} F(d(x,y))F(d(z,y)) \leq C_F F(d(x,y)) . 
    \end{equation}
\end{enumerate} Given such a function we may construct a Banach space of interactions $\Phi$ on $(\Gamma,d)$ with the $F$-norm
\begin{equation}
    \norm{\Phi}_F := \sup_{x,y\in \Gamma}\frac{1}{F(d(x,y))} \sum_{\substack{X\subseteq \Gamma: \\ x,y\in X}} \norm{\Phi(X)} . 
\end{equation}

Certainly when $\Phi$ is a finite-range interaction, i.e. there exists an $R>0$ such that if $\text{diam}(X)>R$ then $\Phi(X) = 0$, we have $\norm{\Phi}_F<\infty$. We also have that a variety of sufficiently rapidly decaying interactions are in this Banach space, since one such $F$-function for $\Gamma = \Z^\nu$ is given by
\begin{equation}
    F(r) = (1+r)^{-(\nu + \epsilon)} . 
\end{equation}

When one has an interaction $\Phi$ with bounded $F$-norm, we call $\Phi$ a \textbf{local interaction}, or colloquially we call its Hamiltonian $H$ a \textbf{local Hamiltonian}. When we have a local interaction $\Phi$, one may prove the existence of limiting dynamics. 

\begin{theorem}~\cite{nachtergaele2016quantum} 
Let $\Phi$ be an interaction on $(\Gamma,d)$ with $\norm{\Phi}_F<\infty$. Along any increasing, exhaustive sequence $\{\Lambda_n\}$ of finite subsets of $\Gamma$, the norm limit
\begin{equation}
    \tau_t(A) = \lim_{n\to \infty} \tau_t^{\Lambda_n}(A) 
\end{equation} exists for all $t\in \R$ and $A\in \calA_{loc}$. This convergence is uniform for $t$ in compact sets and independent of the choice of exhaustive sequence $\{\Lambda_n\}$. The family $\{\tau_t\}_{t\in \R}$, which we denote by the infinite volume dynamics corresponding to $\Phi$, defines a strongly continuous one-parameter group of $*$-automorphisms on $\calA$.
    
\end{theorem}
The key technical tool to prove existence of the dynamics is a Lieb-Robinson bound, a critical family of results which allow one to upgrade quasi-locality of an interaction to quasi-locality of the flow it generates--a sort of approximate version of a light cone argument. While the first such bound was found by Lieb and Robinson in 1972~\cite{lieb1972finite}, it was not until the mid 2000's that a flurry of results would make it clear that such bounds can be used to great effect in a variety of quantum applications~\cite{hastings2004lieb,hastings2006spectral,nachtergaele2006lieb}. For an introduction to these ideas, see~\cite{nachtergaele2016quantum} or the comprehensive review~\cite{nachtergaele2019quasi}.

The strongly continuous, one-parameter group of $*$-automorphisms $\{\tau_t\}_{t\in \R}$ then has a strong generator $\delta$, the closure of the operator $\wt{\delta}$ defined by $\text{Dom}(\wt{\delta}) = \calA_{loc}$ and 
\begin{equation}
    \wt{\delta} = \lim_{\Lambda\nearrow \Gamma} [H_\Lambda, A] , \qquad A\in \calA_{loc}. 
\end{equation} We often write $\tau_t = e^{it\delta}$.\footnote{Compare again to the Adjoint representation Example \ref{ex:adjoint rep of su2}.} We call the pair $(\calA,\tau_t)$ a \textbf{$C^*$-dynamical system}.

Now, we may define a \textbf{ground state} $\omega:\calA\to \C$ of the $C^*$-dynamical system $(\calA,\tau_t=e^{it\delta})$ as a positive, normalized, linear functional such that 
\begin{equation}
    \omega(A^* \delta(A)) \geq 0 \qquad \text{for all } A\in \text{Dom}(\delta).
\end{equation} Notice that for a finite volume $\Lambda$, $\delta(A) = [H_\Lambda,A]$, and one can show that the above definition is equivalent to that of (\ref{def:finite chain ground state}) when the density matrix $\rho$ of $\omega$ has $H\rho = \min_{\lambda\in \text{spec}(H)}\lambda \rho$.

At the moment, unlike the case of a finite quantum system, the data of an infinite quantum system has no associated Hilbert space: just a $C^*$ algebra $\calA$, dynamics $\tau_t$, and perhaps a collection of states $\omega$. The Gelfand-Naimark-Segal (GNS) construction produces an appropriate Hilbert space from this data, and indeed will allow us to rigorously describe a variety of physical properties like gappedness. Recall that given a representation $\pi:\calA\to \calB(\calH)$ of a unital $C^*$ algebra into the bounded linear operators of a Hilbert space $\calH$, we say a vector $\Omega\in \calH$ is \textbf{cyclic} if 
\begin{equation}
    \{\pi(A)\Omega: A\in\calA \}\subseteq \calH
\end{equation} is a dense subspace of $\calH$.

\begin{theorem} (GNS construction~\cite{nachtergaele2016quantum}) \label{thm:GNS construction}

Let $\omega$ be a state on a $C^*$ algebra $\calA$. Then there exists a Hilbert space $H_\omega$, a representation $\pi_\omega$ of $\calA$ on $\calH_\omega$, and a vector $\Omega_\omega\in \calH_\omega$, which is cyclic for $\pi_\omega$, such that
\begin{equation}
\omega(A) = \inprod{\Omega_\omega , \pi_\omega(A) \Omega_\omega} \qquad\text{ for all } A\in \calA . 
\end{equation} Moreover, the triple $(\calH_\omega,\pi_\omega,\Omega_\omega)$ is uniquely determined by $\omega$ up to unitary equivalence, i.e. given two cyclic representations $(\calH_1,\pi_1,\Omega_1)$ and $(\calH_2,\pi_2,\Omega_2)$ of the same state $\omega$, there exists a unique unitary map $U:\calH_1\to \calH_2$ satisfying
\begin{equation}
    \Omega_2 = U \Omega_1 \quad \text{ and } \quad \pi_2(A) = U \pi_1(A) U^* \qquad \text{ for all } A\in \calA. 
\end{equation}
\end{theorem}
The uniqueness condition gives rise to a useful corollary, 
\begin{corollary} \label{cor:GNS unitary implements symmetries}
Let $\calA$ be a $C^*$ algebra, $\omega$ a state on $\calA$, and $\alpha$ an automorphism of $\calA$ which leaves $\omega$ invariant, i.e. $\omega \circ \alpha = \omega$. Then, $\alpha$ is unitarily implementable in the GNS representation of $\omega$, i.e. there exists a unique unitary $U\in \calU(\calH_\omega)$ such that for all $A\in \calA$ we have
\begin{equation}
    \pi_\omega (\alpha(A)) = U \pi_\omega(A) U^*, \qquad U\Omega = \Omega. 
\end{equation}
\end{corollary}

Since ground states $\omega$ may be shown to be invariant under time translations, $\omega\circ\tau_t = \omega$ for all $t\in \R$, a brief argument assures the existence of a strongly continuous group of unitaries $U_t\in \calU(\calH_\omega)$ which implement time evolution. Stone's theorem (see Chapter 10.2 in ~\cite{hall2013quantum}) then guarantees there is a densely defined self-adjoint operator $H_\omega$ acting on a dense subset of $\calH_\omega$ for which
\begin{equation}
    U_t = e^{-itH_\omega} . 
\end{equation} When applied to the $C^*$-dynamical system $(\calA,\tau_t)$ defining a quantum spin system, we call $H_\omega$ the \textbf{GNS Hamiltonian}. One can show that $H_\omega \geq 0$ and that $H_\omega \Omega = 0$. 

In this way we may reach a key definition. The \textbf{spectral gap} above a ground state $\omega$ of a quantum spin system $(\calA,\tau_t)$ is 
\begin{equation}
    \gamma:= \sup\{ \eps > 0 : (0,\eps)\cap \text{spec}(H_\omega) = \emptyset\} , 
\end{equation} where we write $\gamma = 0$ if the RHS is empty. If $\gamma>0$, we call the ground state \textbf{gapped}, and if $\gamma=0$, we call the ground state \textbf{gapless}.

Let us extend the special case Equation (\ref{eq:symmetry eigenspaces}) to a full definition of broken and unbroken symmetry. We have in mind the same examples: lattice symmetries, like translations on $\Gamma=\Z$ given by the group $G=\Z$ acting via $\alpha_x(A_j) = A_{j+x}$ for any $A_j\in \calA_j$; and local symmetries given by a Lie group $G$ acting on finite volumes by the expression (\ref{def:unitary conjugation symmetry}). Suppose that $\alpha:G\to \text{Aut}(\calA)$ is a symmetry, so $\alpha_g\circ \tau_t = \tau_t \circ \alpha_g$ for all $g\in G$. One can show that the space $\mathcal{S} := \{\omega \text{ a ground state of }(\calA,\tau_t) \}$ is invariant under precomposition with $\alpha_g$, i.e.
\begin{equation}
    \omega \circ \alpha_g \in \mathcal{S} ,\qquad \text{ for all } g\in G.
\end{equation} If for all $\omega\in\mathcal{S}$ and all $g\in G$ we have $\omega\circ \alpha_g = \omega$, we say the symmetry is \textbf{unbroken}. If there exists an $\omega\in\mathcal{S}$ and a $g\in G$ such that $\omega\circ \alpha_g \neq \omega$, we say the symmetry is \textbf{broken}. Notice that if $\omega$ is the unique ground state of an interaction $\Phi$, one immediately obtains that the symmetry is unbroken.

The following example may give an idea of how to check for a Hamiltonian's symmetry. We note that symmetry is often expressed in the language of representation theory, and indeed we will use definitions and constructions from the following Chapter~\ref{ch:app-rep_theory}. We will do many more detailed calculations in Section \ref{sec:Spin and SU(2), by example}, some of which will make the result quite obvious. However, this suggests a rather general approach to checking whether a given representation is a symmetry and contextualizes our usage of tools from representation theory.
\begin{example}\label{ex:checking for symmetry}
    Consider the term $\vec{S}_x\cdot \vec{S}_{x+1}$ of the AKLT Hamiltonian (\ref{ex:AKLT chain}). We implicitly defined the spin-1 Lie algebra representation $\pi^{(1)}:\su(2)\to \gl(\calH_x)$, where $\calH_x\cong \C^3$, by $\frac{1}{2}\sigma^i \mapsto S^i, i=X,Y,Z$. The tensor representation $\wt{\pi}:\su(2)\to \gl(\calH_{[a,b]})$ is then given by
    \begin{equation}
        \frac{1}{2}\sigma^i \xmapsto{\wt{\pi}} \sum_{j=a}^b S_j^i = S^i\otimes \underbrace{\idty\otimes \dots \otimes \idty}_{b-a} + \idty\otimes S^i\otimes \underbrace{\idty\otimes \dots \otimes \idty}_{b-a-1} + \dots + \underbrace{\idty\otimes \dots \otimes \idty}_{b-a}\otimes S^i . 
    \end{equation} Now, observe that the interaction term commutes with the representative of the Pauli matrix $\sigma^X$ (and likewise for $\sigma^Y,\sigma^Z$):
    \begin{equation}\begin{split}
        [\wt{\pi}(\frac{1}{2}\sigma^X), \vec{S}_x\cdot \vec{S}_{x+1}] &= \sum_{j=a}^b [S_j^X, \vec{S}_x\cdot \vec{S}_{x+1}]\\
        &= [S_x^X + S_{x+1}^X, S_x^X S_{x+1}^X + S_x^Y S_{x+1}^Y + S_x^Z S_{x+1}^Z] \\
        &= [S_x^X, S_x^Y] S_{x+1}^Y + [S_x^X, S_x^Z] S_{x+1}^Z  + S_x^Y[S_{x+1}^X, S_{x+1}^Y] + S_x^Z[S_{x+1}^X, S_{x+1}^Z] \\
        &= 2i(S_{x}^Z S_{x+1}^Y - S_{x}^Y S_{x+1}^Z + S_x^Y S_{x+1}^Z - S_x^Z S_{x+1}^Y) \\
        &= 0 , 
    \end{split}\end{equation} where we have used the Pauli commutation relations (\ref{eq:pauli commutation}) and that $\wt{\pi}$ is a Lie algebra representation. It is straightforward to see that the same holds for $(\vec{S}_x\cdot \vec{S}_{x+1})^2$.

    Now, since $SU(2)$ is simply connected, Theorem \ref{thm:alg reps lift to group reps when G simply connected} guarantees that we may ``exponentiate'' to obtain a unique representation $\wt{\Pi}:SU(2)\to \GL(\calH_{[a,b]})$ such that\footnote{Compare to the adjoint representation, where this theorem says we have $e^X(\cdot)e^{-X} = \Ad_{e^X}(\cdot) = e^{\ad_X(\cdot)} = e^{[X,\cdot]}$ for all $X\in \su(2)$. }
    \begin{equation}
        \wt{\Pi}(g) (\vec{S}_x \cdot \vec{S}_{x+1}) \wt{\Pi}(g)^{-1} = \vec{S}_x\cdot \vec{S}_{x+1} , \qquad \text{for all } g\in SU(2) . 
    \end{equation} Indeed, $\wt{\Pi}(g)$ is just the exponential of the tensor representation $\wt{\pi}$, so 
    \begin{equation}
        \wt{\Pi}(g) = (\Pi^{(1)}(g))^{\otimes b-a + 1} = \underbrace{\Pi^{(1)}(g)\otimes \dots \otimes \Pi^{(1)}(g)}_{b-a+1} , 
    \end{equation} where $\Pi^{(1)}:SU(2)\to \GL(\calH_x)$ denotes the spin-1 representation. In other words, if $\alpha_g = \Ad_{\wt{\Pi}(g)}$, then 
    \begin{equation}
        \alpha_g(\vec{S}_x\cdot \vec{S}_{x+1}) = \vec{S}_x\cdot \vec{S}_{x+1} . 
    \end{equation}
\end{example} Thus, the symmetry $\alpha_g(\cdot)$ may be thought of as applying two consecutive operations--the adjoint rep $\Ad$ and the tensor rep $\wt{\Pi}$--to the basic spin-1 rep $\Pi^{(1)}$.
\chapter{Some Representation Theory}\label{ch:app-rep_theory}

Representation theory is an enormous subject with many veins of work across several domains of mathematics. It is the mathematical language of symmetry and so pervades physics as well. As such a large field, many wonderful texts have been written about representation theory. The goal of this chapter is to briefly introduce some key definitions, examples, and theorems so that we can work with representations more practically, both to define symmetry rigorously and to simplify a variety of problems in quantum spin systems.

In Section \ref{sec:notation and background}, we rapidly recall some key constructions, definitions and theorems which may be present in any standard course on representation theory. Much of the structure here is guided by the similarly brief~\cite{nachtergaele2016quantum}, but for details and further exposition, we direct the interested reader to~\cite{simon1996representations, hall2015lie, fulton1991representation}. 

In Section \ref{sec:Spin and SU(2), by example}, we get our hands dirty with $SU(2)$ and its representation theory. Unlike the preceding section, this chapter is highly computational and explicit, serving as a largely self-contained introduction. For readers who may have only gently brushed against Lie representation theory, it may be wise to work through this section before Section \ref{sec:notation and background}.

In Section \ref{sec:Examples Revisited: Putting the ``Spin'' in Quantum Spin Systems}, we use our newfound understanding of $SU(2)$ representations to study our guiding quantum spin system example models.

In Section \ref{sec:highest weight reps}, we recall the notion of highest weight representations of a complex semi-simple Lie algebra $\g$. We have in mind the orthogonal Lie algebras $\g = \so(n)$, whose representation theory will govern the behavior of the states we will study in the final Chapter \ref{ch:SO(n)_Haldane_chains}. Especially important for us are the exterior power representations, which we will construct explicitly and discuss. We follow the conventions of~\cite{fulton1991representation}, with a few references to~\cite{simon1996representations}.

In Section \ref{sec:Clifford algebras}, we recall some standard results on Clifford algebras, following~\cite{simon1996representations}. A handful of results will be proven for later use.

In Section \ref{sec:the spin representations Pi of SO(n)}, we use Clifford algebras to construct a family of representations of $\so(n)$ which will play an important role in describing the ground state representations in the final chapter \ref{ch:SO(n)_Haldane_chains}. We will also in the process form a meaty class of examples of projective representations of $SO(n)$, which is the subject of the following chapter.

%-------------------------------------------
\section{Notation and Background} \label{sec:notation and background}
A \emph{Lie Group} $G$ is a group with the compatible structure of a smooth manifold, where by compatible we mean that the product $(g,h)\mapsto gh$ and inversion mappings $g\mapsto g^{-1}$ are smooth. We call $G$ compact if $G$ is compact as a manifold, and we call $G$ connected if any two points can be connected by a continuous curve in $G$. 
A \emph{Lie algebra} $\g$ is a vector space $\g$ over a field $\bbF\in \{\R,\C\}$ equipped with a Lie bracket, a bilinear mapping $\g\times \g \to \g$ which is skew-symmetric $[X,Y]=-[Y,X]$ for all $X,Y\in \g$ and satisfies the Jacobi identity
\begin{equation}
    [[X,Y],Z] + [[Y,Z],X] + [[Z,X],Y] = 0 ,\qquad \text{ for all } X,Y,Z\in \g . 
\end{equation}
Especially important are the matrix Lie groups, which are closed subgroups of the general linear group on $V = \R^n, \C^n$, 
\begin{equation}
    \GL(V) := \{g \text{ an } n\times n \text{ invertible matrix on } V\} . 
\end{equation} 
Every Lie group $G$ has a corresponding Lie algebra $\g$, which is the tangent space $T_{\idty}(G)$ at the identity $\idty\in G$ equipped with the Lie bracket given by the vector field commutator of the vectors $x,y$ pushed forward by left-multiplication. But for the special case of matrix Lie groups $G\subseteq \GL(V)$, an equivalent characterization is that the corresponding Lie algebra $\g$ is the set of matrices whose exponential is in $G$, equipped with the usual commutator:
\begin{equation}
    \g = \{ X \text{ an } n\times n \text{ matrix on } V \st e^{tX}\in G, t\in \R\} .
\end{equation} Here we mean the usual matrix exponential given by its convergent power series.\footnote{This follows the mathematicians convention for a Lie algebra: the set of operators $X\in \g$ such that $e^{X}\in G$. Physicists bake in the imaginary number so that $X\in \g$ if $e^{-iX}\in G$. Generally, when working with complex representations, one can be lazy about this distinction, but we will remain careful for the sake of exposition.} 
In agreement with the tangent space at identity characterization, a basis for $\g$ may be found by computing directional derivatives on $G$ and evaluating at the identity, e.g. $X = \frac{d}{dt} e^{tX} \vert_{t=0}$ . 

It can then be readily shown that the corresponding Lie algebra of $\GL(V)$ is 
\begin{equation}
    \gl(V) := \{X \text{ an } n\times n \text{ matrix on } V\} , 
\end{equation} where $\gl(V) = M_n(\R)$ (or $M_n(\C)$) equipped with matrix commutator $[A,B] = AB-BA$ as its Lie bracket.
This surprisingly rich class of Lie groups will be our primary focus, thanks in no small part to their privileged role in describing on-site symmetries in quantum spin systems. 

Our key examples going forward will be a collection of compact Lie groups. First, the unitary and special unitary groups
\begin{equation}\begin{split}
    U(n) &= \{U\in \GL(\C^n): U^* U = \idty\} ,\\
    SU(n) &= \{U\in U(n) \text{ and } \det(U)=1\},
\end{split}\end{equation} and then the orthogonal and special orthogonal groups 
\begin{equation}\begin{split}
    O(n) &= \{R \in \GL(\R^n): R^T R = \idty\}, \\
    SO(n) &= \{R \in O(n) \text{ and } \det(R)=1\}.
\end{split}\end{equation} As a comment, we will commonly also use the notation $\calU(\calH)$ to express the group of unitaries on a Hilbert space $\calH$. When $n = \dim(\calH)$ is finite, this is exactly $\calU(\calH) = U(n)$. 

We may recover the corresponding matrix Lie algebras in a straightforward manner, starting with $U(n)$. Let $U(t) = e^{tX}\in U(n)$. Then differentiating the equation $U(t)^*U(t) = \idty$ at the identity $t=0$, we have
\begin{equation} \label{eq:unitary and skew hermitian}
    0 = \frac{d}{dt} \paran{(e^{tX})^* e^{tX}} \big\vert_{t=0} = \paran{X^* e^{tX} + (e^{tX})^* X }\big\vert_{t=0} = X^* + X . 
\end{equation} This means that the Lie algebra $\g = \u(n)$ of $G = U(n)$ is given by 
\begin{equation}
    \u(n) = \{ X\in \gl(\C^n) : X^* = - X \} . 
\end{equation} Similarly, to obtain a description of $\su(n)$, we observe that 
\begin{equation}
    1 = \det U(t) = \det e^{tX} = e^{t \Tr X} \qquad \text{ for all } t\in \R , 
\end{equation} and so 
\begin{equation}
    \su(n) = \{ X\in \u(n) : \Tr X = 0 \} . 
\end{equation} By a similar pair of calculations, we see that
\begin{equation}
    \so(n) = \{X\in \gl(\R^n) : X^T = - X \} , 
\end{equation} noting that $\so(n) = \mathfrak{o}(n)$ because $\Tr X = \Tr X^T = -\Tr X$ and so $\Tr X = 0$. 

We will spend a great deal of time with the Lie groups $SU(2)$ and $SO(3)$ (and their corresponding Lie algebras) in the following Section \ref{sec:Spin and SU(2), by example}. 
In the remainder of this section, we will briskly go through some of the theorems and definitions which will appear in this thesis--it may be wise to first work through Section \ref{sec:Spin and SU(2), by example} before returning. Note that the primary class of models we investigate in the final chapter of this thesis are a family of quantum ground states invariant under a certain representation of $SO(n)$ and the $SO(3)$ prototype corresponds exactly to the AKLT chain.

\subsection{Lie group-Lie algebra correspondence}
\begin{definition}
    Let $G$ be a Lie group. A \emph{representation} is a finite dimensional vector space $V$ with a Lie group homomorphism $\Pi:G\to \GL(V)$, i.e. a map 
    \begin{equation}
        \Pi(g_1g_2) = \Pi(g_1)\Pi(g_2) \qquad \text{ for all } g_1, g_2 \in G . 
    \end{equation}
\end{definition} 
Note that it is common to call either the vector space $V$ or the homomorphism $\Pi$ a representation: really, it is the shared data $(\Pi,V)$ that makes a representation. It is also common to refer to a representation $V$ as a $G$-module, which is the same data with a slightly different perspective.

\begin{definition}
    Let $\g$ be a Lie algebra. A \emph{representation} is a finite dimensional vector space $V$ with a Lie algebra homomorphism $\pi:\g\to \gl(V)$, i.e. a map 
    \begin{equation}
        \pi([X,Y]) = [\pi(X),\pi(Y)] \qquad \text{ for all } X,Y\in \g .
    \end{equation}
\end{definition} Again, it is common to call either $V$ or $\pi$ the representation, and it is common to call $V$ a $\g$-module. In essence, Lie representation theory can be thought of as studying Lie groups/algebras by their homomorphisms into matrix Lie groups/algebras. Let's write a few examples, noting that Section \ref{sec:Spin and SU(2), by example} is chock full of much more thoroughly worked examples for $SU(2)$, starting with Example \ref{ex:trivial rep of su2}.

\begin{enumerate}
    \item The trivial representation of a Lie group $G$ is $\Pi:G\to \GL(\C)$ which maps $\Pi_g = \idty$, the 1 $\times$ 1 identity matrix.
    \item Every matrix Lie group $G\subseteq \GL(V)$ has a \textit{defining} representation $\Pi_g = g$, i.e. each matrix represents itself.
    \item The adjoint representation of $G$ is $\Ad: G\to \GL(\g)$ given by 
    \begin{equation}
        \Ad_g(Y) = g Y g^{-1}, \qquad \text{ for all } Y\in \g . 
    \end{equation}
\end{enumerate}

We may build up our class of examples for representations by extending our usual linear algebra constructions.
\begin{definition}
    Let $\Pi:G\to \GL(V)$, $\Pi_1:G\to \GL(V_1)$, and $\Pi_2:G\to\GL(V_2)$ be representations of a Lie group $G$. 

    The \textit{tensor representation} $\Pi_1\otimes \Pi_2:G\to \GL(V_1\otimes  V_2)$ is defined by
    \begin{equation}
        (\Pi_1\otimes \Pi_2)(g) = \Pi_1(g)\otimes \Pi_2(g), \qquad \text{ for all } g\in G. 
    \end{equation}

    The \textit{direct sum representation} $\Pi_1\oplus \Pi_2: G\to \GL(V_1\oplus V_2)$ is defined by 
    \begin{equation}
        (\Pi_1\oplus \Pi_2)(g) = \Pi_1(g)\oplus \Pi_2(g), \qquad \text{ for all } g\in G. 
    \end{equation}

    The \textit{dual representation} 
    $\Pi^*: G\to \GL(V^*)$ is defined for all $g\in G$ by
    \begin{equation}
        \Pi^*(g) = (\Pi(g^{-1}))^T, 
    \end{equation} where the $^T$ is transpose with respect to the dual pairing $(\cdot,\cdot)$ of $V^*$ and $V$. 
\end{definition}

We should have a notion of equivalence for representations.
\begin{definition}\label{def:equivalent reps}
        Two finite dimensional representations $(\Pi_1,V_1),(\Pi_2,V_2)$ of a Lie group or Lie algebra $G$ are \textit{equivalent} if there exists an invertible linear transformation $T: V_1\to V_2$ such that
        \[
            \Pi_2(g)T = T\,  \Pi_1(g) \qquad \text{for all }g\in G,
        \] We call $T$ an \textit{intertwiner} or \textit{isomorphism}. 
\end{definition} 
It is common to think of this definition as ``two reps are equivalent if they differ only by a change of basis''. Note that it is common to work with intertwiners $T:V\to W$ which are only invertible when the codomain is restricted to its image $T(V)$. We will run into such intertwiners when we consider a certain isometric embedding, i.e. a map such that $T^*T=\idty_V$.

Much of Lie theory's power is driven by Lie group-Lie algebra correspondence: the group structure of $G$ and the linear structure of $\g$ make Lie theory highly amenable to a wide family of applications. Since $\g$ is the tangent space $T_{\idty} G$, we have seen that we may pass from $G$ to $\g$ by taking derivatives of one-parameter subgroups $\{U(t) = e^{tX} : t\in \R, X\in \g \} \subseteq G$. This also holds for representations. 

\begin{theorem} (Lie group reps pass to Lie algebra reps) \label{thm:lie gp passes to lie alg}

    Let $G$ be a matrix Lie group. If we have a representation $\Pi:G\to \GL(V)$, then there is a representation of the corresponding Lie algebra $\pi:\g\to \gl(V)$ defined by
    \begin{equation}
        \pi(X) = \frac{d}{dt}\paran{\Pi(e^{tX})} \Big\vert_{t=0} \qquad \text{ for all } X\in\g . 
    \end{equation}
\end{theorem} 
We may start with a representation of a Lie group $G$ and use this to compute the representation it induces on its Lie algebra $\g$. Again, more explicit examples abound in \ref{sec:Spin and SU(2), by example}. 
\begin{enumerate}
    \item The trivial representation of $\g$ is $\pi:\g\to \gl(\C)$ by $\pi(X) =0$.
    \item The defining representation of a matrix Lie algebra $\g$ is $\pi:\g\to \gl(V)$ by $\pi(X) =X$.
    \item The adjoint representation of $\g$ is $\ad: \g\to \gl(\g)$ given by 
    \begin{equation}
        \ad_X(Y) = [X,Y], \qquad \text{ for all } Y\in \g . 
    \end{equation}
    \item The tensor representation of $\pi_1,\pi_2$ on $V_1\otimes V_2$ is given by
    \begin{equation}
        \pi_1(X)\otimes \idty_{V_2} + \idty_{V_1} \otimes \pi_2(X), \qquad X\in \g . 
    \end{equation}
    \item The direct sum representation of $\pi_1,\pi_2$ on $V_1\oplus V_2$ is 
    \begin{equation}
        (\pi_1 \oplus \pi_2)(X) = \pi_1(X) \oplus \pi_2(X), \qquad X\in \g . 
    \end{equation}
    \item The dual representation $\pi^*:\g\to \gl(V^*)$ is 
    \begin{equation}
        \pi^*(X) = - \pi(X)^T, \qquad X\in\g. 
    \end{equation}
\end{enumerate}

But what about the converse of this theorem? Just as in the study of ordinary differential equations, we may exponentiate to find solutions. Locally, i.e. in open neighborhoods around $0\in \g$ and $\idty\in G$, the exponential map is a diffeomorphism. Notice that we require $G$ to be connected--a quick counterexample is given by $O(n)$, which since the determinant $\det: O(n)\to \{-1,1\}$ is a continuous map, consists of two disconnected pieces. Since the exponential map is continuous and $\so(n)$ is connected, it may only map into the identity component $SO(n)$ of $O(n)$.
\begin{proposition}
    Let $G$ be a connected Lie group. The exponential map $\exp(\cdot):\g\to G$ is a local diffeomorphism, i.e. there exists an open neighborhood $N_0\subseteq \g$ containing $0\in \g$ and an open neighborhood $N_\idty\subseteq G$ containing $\idty\in G$ such that $\exp(\cdot): N_0\to N_\idty$ is a diffeomorphism.
\end{proposition}

\begin{theorem} (Lie algebra reps locally lift to Lie group reps) \label{thm:lie alg lifts to lie gp}

    Let $G$ be a connected Lie group with Lie algebra $\g$. Let $\pi:\g\to \gl(V)$ a representation of $\g$. Then on the open subsets for which $\exp: N_0 \to N_\idty$ is a diffeomorphism, we may define a representation $\Pi:N_0\to \GL(V)$ by
    \begin{equation}
        \Pi(e^X) = e^{\pi(X)} \qquad \text{ for all } X\in N_0 . 
    \end{equation}
\end{theorem}

Care must be taken if we wish to extend beyond this open neighborhood. Global features of the Lie group matter: it may be that the representation $(\Pi,V)$ induced by exponentiation is not well-defined on the entire group $G$. The topological obstruction is exactly the fundamental group $\pi_1(G)$, and we will encounter a key instance of this when we encounter the spin representations of $SO(n)$ in Section \ref{sec:the spin representations Pi of SO(n)} since $\pi_1(SO(n)) = \Z_2$. When $G$ is simply connected, i.e. $\pi_1(G) = 0$, there is no such issue. 

\begin{theorem} \label{thm:alg reps lift to group reps when G simply connected} (Lie algebra reps lift to global Lie group reps when $G$ simply connected) 

    Let $G$ be a simply connected Lie group with Lie algebra $\g$. Let $\pi:\g\to \gl(V)$ a representation of $\g$. Then there is a unique representation $\Pi:G\to \GL(V)$ such that
    \begin{equation}
        \Pi(e^X) = e^{\pi(X)} \text{ for all } X\in \g . 
    \end{equation}
\end{theorem} This means that every representation of a simply connected $G$ can be obtained from a representation of its Lie algebra. But what can we say of other Lie groups? Recall the universal cover construction from topology.

\begin{theorem} (Universal Covers~\cite{simon1996representations})
 \label{thm:universal cover of a Lie group}  Let $G$ be an arbitrary connected Lie group. Then there exists a simply connected Lie group $\wt{G}$, called the universal covering group of $G$, and a discrete normal subgroup $H\cong \pi_1(G)$ such that $G\cong \wt{G}/H$ and the Lie algebras of $G$ and $\wt{G}$ are isomorphic. Up to canonical isomorphism $\wt{G}$ is unique. 
\end{theorem}
\begin{remark}\label{rem:discrete normal subgroups are in center}
    A discrete normal subgroup $H$ of a connected Lie group is in necessarily in the center $H\subseteq Z(G)$. To prove this, take a neighborhood of the identity $\calN_\idty\subseteq G$ where the exponential map is bijective. Let $g=e^{tX}\in \calN_{\idty}$ and $h\in H$ and consider $e^{tX} h e^{-tX}$. Since $H$ is normal, $e^{tX} h e^{-tX}\in H$. At $t=0$, this is just $h$, and since $H$ is discrete and this path is continuous, it must be that $e^{tX} h e^{-tX} = h$, or $gh = hg$. Then the result follows by noting that $\calN_\idty$ generates the connected component of $G$, which is all of $G$.
\end{remark}
It should be mentioned that universal cover of a matrix Lie group need not be a matrix Lie group, as is the case for $\sl(2,\R)$ (see ~\cite{hall2015lie}), and the universal cover of a compact Lie group need not be compact, as is the case for $\R$, the universal cover of $U(1)$. But this is not the case for our most important examples for this thesis: the special orthogonal groups $SO(n)$ are double covered by another (compact) matrix Lie group called $Spin(n)$, i.e. $Spin(n)/\{\pm \idty\} \cong SO(n)$. The key example, which we will spend a great deal of time working through in the next section, is $Spin(3) = SU(2)$. Moreover, $\su(2)\cong \so(3)$ as Lie algebras. This means that any representation of $\so(3)$ lifts to a unique representation of $SU(2)$. But what can we say about $SO(3)$? 

\begin{corollary} (Universal covers control reps~\cite{hall2015lie}) \label{prop:universal covers}
    Let $G$ be a connected Lie group with universal cover $\wt{G}$ and shared Lie algebra $\g$. If $H$ is a Lie group with Lie algebra $\h$ and $\pi:\g\to \h$ a Lie algebra homomorphism, there exists a unique homomorphism $\wt{\Pi}:\wt{G}\to H$ such that $\wt{\Pi}(e^X) = e^{\pi(X)}$ for all $X\in \g$.
\end{corollary} We had earlier in Theorem \ref{thm:lie gp passes to lie alg} that every Lie group representation $\Pi:G\to \GL(V)$ (which is just a special case of a Lie group homomorphism) induces a Lie algebra representation $\pi:\g\to \gl(V)$. This corollary then tells us that there is a Lie group representation $\wt{\Pi}:\wt{G}\to \GL(V)$, which by uniqueness must pass through the canonical quotient projection $p: \wt{G}\to \wt{G}/\pi_1(G)$, i.e. $\Pi = p\circ \wt{\Pi}$. However, it may be that there are representations of a Lie algebra $\g$ which do not lift to the original Lie group $G$ when $G$ is not simply connected.

Indeed, this is the case for $SO(3)$. We later describe in Example \ref{ex:the spin s irrep} a class of $\su(2)$ representations labeled by a half integer $s= 0,1/2,1,3/2,\dots$ (the spin-$s$ irrep). Since $SU(2)$ is simply connected, each rep of $\su(2)$ lifts to a representation of $SU(2)$. However, only the representations corresponding to integer $s$ pass through to representations of $SO(3)$. The half-integer $s$ representations $\wt{\Pi}:SU(2)\to \GL(\C^{2s+1})$ fail to pass through the projection $S:SU(2)\to SU(2)/\{\pm \idty\} \cong SO(3)$ since $\Pi(\idty)\neq \Pi(-\idty)$. We will compute these things in the following chapter, but for now, we can summarize this observation in a remark.
\begin{remark}
    Every Lie group representation $\Pi:G\to \GL(V)$ lifts to a representation of its universal cover $\wt{G}:G\to \GL(V)$. But if $\pi_1(G)\neq 0$, there may be representations of its Lie algebra $\pi:\g\to \gl(V)$ which lift to a representation of the universal cover $\wt{\Pi}:\wt{G}\to \GL(V)$ but do not pass to a representation of $G$. 
\end{remark}

This idea will resurface when we discuss \textit{projective} representations.

\subsection{Other basic representation theoretic notions}

\begin{definition}
    Let $(\Pi,V)$ be a representation of a Lie group $G$ (or a Lie algebra $\g$). We call a vector subspace $W\subseteq V$ an \textit{invariant subspace} if for all $g\in G$ (or analogously all $X\in \g$)
    \begin{equation}
        \Pi(g) w \subseteq W \qquad \text{ for all } w\in W.
    \end{equation} 
\end{definition} To mention some equivalent terminology, we may also call $W$ a subrepresentation or a $G$-submodule.

The following proposition is exceptionally important and ensures that when working on representations, we may still pass between a Lie group and its Lie algebra. 
\begin{proposition} ~\cite{nachtergaele2016quantum}
    Let $G$ be a Lie group with Lie algebra $\g$. 
    \begin{enumerate}
        \item If $(\Pi,V)$ a representation of $G$ and $W$ is a $\Pi$-invariant subspace, then $W$ is also a $\pi$-invariant subspace where $\pi$ is the induced representation of $\g$.
        \item If $(\pi,V)$ a representation of $\g$ and $W$ is a $\pi$-invariant subspace, then $W$ is also a $\Pi$ invariant subspace where $\Pi$ is the (locally) induced representation of $G$ acting on $V$.
    \end{enumerate}
\end{proposition}

Now for a key definition. Note that we will commonly shorten ``representation'' to ``rep'' and ``irreducible representation'' to ``irrep''.
\begin{definition}
    Let $(\Pi,V)$ be a representation of a Lie group $G$ (or a Lie algebra $\g$). If there exists a invariant subspace $W\subseteq V$ such that $W$ is non-trivial ($W\neq \{0\}$ and $W\neq V$), then we call $(\Pi,V)$ \emph{reducible}. If the only invariant subspaces are the trivial ones, we call $(\Pi,V)$ \emph{irreducible}.
\end{definition}

Central to any application of linear algebra or functional analysis is the concept of diagonalization of an operator $T:V\to V$: that is, decomposition of $V = \bigoplus V_{\mu}$ into subspaces $V_{\mu}$ invariant under the action of $T$ (so we can restrict $T:V_{\mu}\to V_{\mu}$).

\begin{theorem}\label{thm:complete reducibility}
    Let $(\Pi,V)$ be a representation of a Lie group $G$ (or a Lie algebra $\g$). $(\Pi,V)$ is called \textit{completely reducible} if there exists a direct sum decomposition of $V$ into subspaces $\{V_j\}$, i.e.
    \begin{equation}
        V = \bigoplus_{j} V_j ,
    \end{equation} with each $V_j$ a $\Pi$-invariant subspace for which $(\Pi,V_j)$ is irreducible. 
\end{theorem}
We have been intentionally vague with the indexing of this direct sum decomposition, as we will later refer to completely reducible representations on infinite dimensional separable Hilbert spaces.
Complete reducibility of a representation $V$ may be thought of as ``simultaneous block diagonalization'' of the representative matrices $\{\Pi(g): g\in G\}\subseteq \GL(V)$. Namely, if $V$ is completely reducible and finite dimensional, then $V = V_1\oplus \dots \oplus V_k$ then there exists a basis of $V$ for which the matrices $\Pi(g)$ act block diagonally:
\begin{equation}
    \Pi(g) = \begin{pmatrix}
        \Pi_1(g) & 0 & 0 & 0 \\
        0 & \Pi_2(g) & 0 & 0 \\
        0 & 0 & \ddots & 0 \\
        0 & 0 & \dots & \Pi_k(g)
    \end{pmatrix},  
\end{equation} where each $\Pi_j(g)\in \GL(V_j)$. 

Just as not all operators are diagonalizable, not all representations are completely reducible. A quick counterexample can be seen by defining $\Pi:\R\to \GL(\C^2)$ given by $\Pi(x) = \begin{pmatrix}
    1 & x \\ 0 & 1 
\end{pmatrix}$. We now introduce a particularly important class of representations.
\begin{definition}
Let $V$ be a Hilbert space with inner product $\inprod{\cdot,\cdot}$. We call the Lie group representation $\Pi:G\to \GL(V)$ \emph{unitary} if $\Pi(g)\Pi(g)^* = \idty$ for all $g\in G$. Likewise, we call the Lie algebra representation $\pi:\g\to \GL(V)$ \emph{skew-Hermitian} if $\pi(X)^* = -\pi(X)$ for all $X\in \g$.
\end{definition} It is straightforward to see by taking derivatives just as we did in \ref{eq:unitary and skew hermitian} that if $(\Pi,V)$ is a unitary representation of a Lie group $G$, the induced representation $(\pi,V)$ of its Lie algebra $\g$ is skew-Hermitian, and if $(\pi,V)$ is a skew-Hermitian representation of a Lie algebra $\g$, the (locally) induced representation $(\Pi,V)$ of the corresponding Lie group $G$ is unitary. 

This class of representations is even broader than it appears.
\begin{proposition} (Weyl's unitary trick~\cite{hall2015lie})
\label{prop:Weyl unitary trick}
   Let $G$ be a finite or compact Lie group and let $\Pi:G\to \GL(V)$ be a representation of $G$. Then $(\Pi,V)$ is equivalent to a unitary representation $U:G\to \GL(V)$. 
\end{proposition}
\begin{proof}
    Choose an inner product $\inprod{\cdot,\cdot}$ on $V$, then define a map $\inprod{\cdot,\cdot}_G:V\times V\to \C$ by 
    \begin{equation}
        \inprod{v,w}_G = \int_G \inprod{\Pi(g)v,\Pi(g)w} \, d\mu(g).
    \end{equation} where $\int_G \, d\mu(g)$ denotes integration with respect to the normalized Haar measure on $G$. It may be readily verified that this defines an inner product on $V$. Observe that for any $h\in G$, we have
    \begin{align*}
        \inprod{\Pi(h)v,\Pi(h)w}_G &= \int_G \inprod{\Pi(g)\Pi(h)v, \Pi(g)\Pi(h)w} \, d\mu(g) \\
        &= \int_G \inprod{\Pi(gh)v, \Pi(gh)w} \, d\mu(g) \\
        &= \int_G \inprod{\Pi(g)v, \Pi(g)w} \, d\mu(g) \\
        &= \inprod{v,w}_G , 
    \end{align*} where we have used right-invariance of the Haar measure in the third line. In particular, $(\Pi,V)$ is a unitary representation with respect to the inner product $\inprod{\cdot,\cdot}_G$, completing the proof.
\end{proof} 
\begin{remark}
    It is essential that $G$ is compact or finite so that the Haar measure of the entire group $\int_G \, d\mu(g)<\infty$, and so we may normalize. There are immediate counterexamples if one considers noncompact groups, like $SL(\C^2)$.
\end{remark}

The following theorem ensures that all unitary representations may be ``block diagonalized''. The theorem may be readily upgraded to an infinite dimensional Hilbert space $\calH$ version by requiring that a unitary representation be strongly continuous $\Pi:G\to \mathcal{U}(\calH)$.
\begin{theorem} \label{thm:unitaries are completely reducible}
    Every finite dimensional unitary representation is completely reducible.
\end{theorem}
\begin{proof}
    Let $(\Pi,V)$ be a unitary representation of a Lie group $G$. If $(\Pi,V)$ irreducible, we are done. Otherwise, there exists a nontrivial invariant subspace $W\subseteq V$. Write $V = W\oplus W^\perp$, and note that if $x\in W^\perp, y\in W$, and $\g\in G$, then
    \begin{equation}
        \inprod{\Pi(g)(x), y} = \inprod{x,\Pi(g)^* y} = \inprod{x,(\Pi(g))^{-1} y} = \inprod{x,\Pi(g^{-1}) y} = 0 , 
    \end{equation} where for the last equality we have used that $W$ is an invariant subspace. This proves $W^\perp$ is also an invariant subspace. Iterating this argument for the restrictions to the finite dimensional subspaces $W,W^\perp$ yields the result.
\end{proof} 
One may think of this corollary as a representation-theoretic version of ``eigenspaces of distinct eigenvalues are orthogonal''. Of course, care must be taken when two invariant subspaces are isomorphic.
\begin{corollary} 
    From the proof of Theorem \ref{thm:unitaries are completely reducible}, any pair of invariant subspaces $V_\mu, V_\nu\subseteq V$ such that $V_\mu$ and $V_\nu$ are distinct irreps must be orthogonal.
\end{corollary} This theorem may be upgraded by considering strongly continuous unitary representations $\Pi:G\to \mathcal{U}(\calH)$.

The following result is perhaps the most important in all of representation theory.
\begin{lemma}(Schur's lemma) \label{lem:schur} 
Let $V_1,V_2$ be complex vector spaces and let $(\Pi_1, V_1)$ and $(\Pi_2,V_2)$ be two irreducible representations of $G$. Suppose that $T:V_1\to V_2$ a linear map such that
\begin{equation}
    T \, \Pi_1(g) = \Pi_2(g) T \qquad \text{ for all }g\in G . 
\end{equation} Then there are two possibilities. 
\begin{enumerate}
    \item If the two representations are not equivalent $V_1\not\cong V_2$, then $T = 0$.
    \item If the two representations are equivalent $V_1=V_2$, then $T = \lambda \idty$ for $\lambda\in \C$. 
\end{enumerate}
\end{lemma} \begin{proof}
    To prove (1), observe that $\ker(T)\subseteq V_1$ is an invariant subspace and $\text{ran}(T)\subseteq V_2$ is an invariant subspace and use irreducibility. To prove (2), observe that since $T$ is a linear operator over a finite dimensional complex vector space, it has at least one eigenvalue $\lambda\in \C$. Then the non-empty eigenspace of $\lambda$ is an invariant subspace and so must be all of $V$.
\end{proof}

\subsubsection{An Aside on Complexification}
For the entirety of this chapter, we have considered representations over a complex vector space $V$. It can be shown by a universal property that the complex representations of a real Lie algebra $\g$ extend uniquely to representations of the so-called complexification $\g_\C:= \g\otimes_\R \C$ of $\g$. We call $\g$ the real form of $\g_\C$. In practice, this is a rather easy task: for instance, the complexification of $\su(2)$ is given by taking the real Pauli basis $\{\sigma_x,\sigma_y,\sigma_z\}$ and taking their complex span: $\su(2)_\C = \text{span}_\C \{\sigma_x,\sigma_y,\sigma_z\} = \sl(2,\C)$.
    Indeed, this streamlines a number of \textit{a priori} clunky expressions: notice that in several examples, like the 2 qubit tensor rep Example~\ref{ex:2 qubit tensor rep of su(2)} and the adjoint rep Example~\ref{ex:adjoint rep of su2}, we have several copies of $i$ floating around for every expression. If we consider $\pi$ as the complex representation of $\g_\C$ on $V$ instead of $\g$ on $V$, then $\pi$ is complex linear and so
    \[
        -i\pi \paran{\frac{i}{2}\sigma_Z} = \pi\paran{\frac{1}{2}\sigma_Z},
    \] and our expressions are much cleaner. When working with complex representations from the Lie algebra perspective, one almost always jumps immediately from the real form $\g$ to the complexification $\g_\C$. At the Lie algebra level, it is also straightforward to go from a complex Lie algebra to its real form, as we often do to return from $\sl(2,\C)$ to $\su(2)$. At the level of the Lie groups, the real form of a semisimple Lie algebra corresponds to a compact Lie group. Here, the exponential of $\su(2)$ is $SU(2)$, a compact real 3-manifold diffeomorphic to the 3-sphere. Note however that if we exponentiate the complexified Lie algebra $\sl(2,\C)$, we end up with $SL(2,\C)$, which is not compact. This subtlety will not matter much for our work here and we can proceed without worry: we will only be studying representations which come from real compact Lie groups. See~\cite{kirillov2008introduction} for more information.

%-------------------------------------------
\section{Spin and \texorpdfstring{$SU(2)$}{SU(2)}, by Example} \label{sec:Spin and SU(2), by example}
In this section we will study representations of the Lie group $SU(2)$ and its associated Lie algebra $\su(2)$. The spin $s\in \frac{1}{2}\N$ of a quantum particle is a label for the irreducible representation $V_{s} = \C^{2s+1}$ of $SU(2)$. In addition to their foundational role in quantum spin systems, they also serve as the prototypical examples of semi-simple Lie group representations. We approach this by way of example, leaving the full story and proofs to standard references~\cite{hall2015lie, fulton1991representation,simon1996representations}. By the end of this, we will be able to start doing routine spin calculations with a clearer mathematical picture, allowing us to study the ground states of our key examples, the Majumdar-Ghosh chain Example \ref{ex:Majumdar-Ghosh chain} and the AKLT chain Example \ref{ex:AKLT chain}. The story will begin with studying the structure of $SU(2)$ and $SO(3)$.

\subsection{$SU(2)$ and $SO(3)$}

The term ``spin'' is often described as a form of intrinsic angular momentum. What does this have to do with $SU(2)$? Shouldn't the relevant group be the group of rotations on $\R^3$? Indeed, these two groups are closely connected: they share the same Lie algebra $\so(3)\cong \su(2)$, and as Lie groups, $SU(2)/\Z_2 \cong SO(3)$. These two facts will be demonstrated in the remainder of this section, and we will get rather comfortable with these examples here.

Recall from earlier the Pauli matrices 
\begin{equation}
    \sigma^X = \begin{pmatrix} 0 & 1 \\ 1 & 0 \end{pmatrix}, \quad \sigma^Y = \begin{pmatrix} 0 & -i \\ i & 0 \end{pmatrix}, \quad \sigma^Z = \begin{pmatrix} 1 & 0 \\ 0 & -1 \end{pmatrix} . 
\end{equation} which for future reference enjoy commutation relations
\begin{equation} \label{eq:pauli commutation}
    [\sigma^X, \sigma^Y] = 2i\sigma^Z , \quad [\sigma^Y, \sigma^Z] = 2i\sigma^X, \quad [\sigma^Z, \sigma^X] = 2i\sigma^Y . 
\end{equation}

After some algebra, one realizes that
\begin{equation} 
    SU(2) = \left\{ A = \begin{pmatrix}z & -w \\ \overline{w} & \overline{z} \end{pmatrix}: (z,w)\in \C^2 \text{ and } \abs{z}^2 + \abs{w}^2 = 1 \right\} ,
\end{equation} which leads us to see that any $A\in SU(2)$ can be written as
\begin{equation} \label{eq:SU(2) is 3-sphere}
    A = x_0 \idty + i (x_1 \sigma^X + x_2 \sigma^Y + x_3 \sigma^Z), \qquad \text{ where } x= (x_0, x_1, x_2, x_3) \in \R^4 \text{ has } \abs{x} = 1 . 
\end{equation} This defines an invertible smooth map from the 3-sphere $S^3 = \{x\in \R^4: \abs{x} = 1\}$ to $SU(2)$, and so they are diffeomorphic as smooth manifolds. In particular, this means $SU(2)$ is compact, connected, and simply connected real 3-manifold. Now, the Lie algebra $\su(2)$ is the tangent space $T_\idty SU(2)$, where we see $\idty$ corresponds to $(1,0,0,0)\in S^3$ through the this map. The tangent space to the sphere at this point is generated by $(0,1,0,0), (0,0,1,0), (0,0,0,1)$, and so 
\begin{equation}
    \su(2) = \text{span}_{\R} \{i \sigma^X, i\sigma^Y, i\sigma^Z\} . 
\end{equation} One may equivalently arrive at this conclusion by considering three families of 1-parameter groups: the ``rotation paths'' parameterized by $\theta\in \R$
\begin{equation}
    U_X(\theta) = e^{i\theta \sigma^X}, \quad U_Y(\theta) = e^{i\theta \sigma^Y} , \quad U_Z(\theta) = e^{i\theta \sigma^Z},
\end{equation} which have derivatives at the identity $\theta = 0$
\begin{equation}
    U_X'(\theta)\Big\vert_{\theta = 0} = i\sigma^X, \quad U_Y'(\theta)\Big\vert_{\theta = 0} = i\sigma^Y, \quad U_Z'(\theta)\Big\vert_{\theta = 0} = i\sigma^Z. 
\end{equation}
Let us shift our attention to $SO(3)$. The group of rotations $SO(3)$ is generated by three one-parameter groups of rotations
\[
    R_X(\theta) = \begin{pmatrix}
        1 & 0 & 0 \\
        0 & \cos\theta & -\sin\theta \\ 
        0 & \sin\theta & \cos\theta
        \end{pmatrix},  R_Y(\theta) = \begin{pmatrix}
            \cos \theta & 0 &  \sin \theta\\
            0 & 1 &  0 \\ 
            -\sin \theta & 0 & \cos\theta
        \end{pmatrix},  R_Z(\theta) = \begin{pmatrix}
            \cos \theta & -\sin \theta & 0 \\
            \sin\theta & \cos\theta & 0 \\ 
            0 & 0 & 1
        \end{pmatrix}
\] with $\theta\in \R$. These correspond to rotations about the $X, Y, $ and $Z$ axes, respectively. That these rotations generate all of $SO(3)$ can be proven in a variety of ways, including using classical Euler Rotation Theorem, by Lie theoretic techniques (see Corollary 3.47 in~\cite{hall2015lie}), or by simply picking up a ball and using only products of rotations about two axes, say $X, Y$, to construct any other rotation. Taking derivatives at the identity $\theta = 0 $ of these 1-parameter groups, we get 
\begin{equation}
    R_X'(0) = \begin{pmatrix}
        0 & 0 & 0 \\
        0 & 0 & -1 \\ 
        0 & 1 & 0
        \end{pmatrix}, \;R_Y'(0)\Big\vert_{\theta = 0} = \begin{pmatrix}
        0 & 0 & 1 \\
        0 & 0 & 0 \\ 
        -1 & 0 & 0
        \end{pmatrix}, \; R_Z'(0)\Big\vert_{\theta = 0} = \begin{pmatrix}
        0 & -1 & 0 \\
        1 & 0 & 0 \\ 
        0 & 0 & 0
        \end{pmatrix}, \;
\end{equation} This means that these operators enjoy the same commutation relations as the Paulis (\ref{eq:pauli commutation}) after rescaling by 2i:
\begin{equation}
    [R_X'(0), R_Y'(0)] = R_Z'(0), \qquad [R_Y'(0), R_Z'(0)] = R_X'(0), \qquad [R_Z'(0), R_X'(0)] = R_Y'(0) . 
\end{equation} In particular, they form a basis for the Lie algebra $\so(3)$, and the map $R_j'(0) \mapsto \frac{1}{2i}\sigma^j$ for $j\in \{X,Y,Z\}$ grants the Lie algebra isomorphism $\su(2)\cong \so(3)$. 

Now, we already said that $SU(2)$ is simply connected since $S^3$ is simply connected. There are several ways to see that $SU(2)/\Z_2 \cong SO(3)$, but one nice way is the following. 

\begin{proposition}\label{prop:SU(2) is double cover of SO(3)}
    $SO(3)$ is isomorphic to $SU(2)/\{\idty,-\idty\}$.
\end{proposition}
\begin{proof}
    Let $V$ be the set of traceless Hermitian matrices, i.e. $V = \{\vec{a}\cdot \vec{\sigma}: \vec{a}\in \R^3\}$. In particular $V\cong \R^3$. Then the Hilbert-Schmidt inner product
    \[
    \inprod{A,B} = \frac{1}{2} \Tr A^* B 
    \] is exactly the Euclidean inner product 
    \[
    \inprod{\vec{a}\cdot \vec{\sigma}, \vec{b}\cdot \vec{\sigma}} = \sum_{i=1}^3 a_i b_i . 
    \] Given $U\in SU(2)$, define $S(U) \in GL(V)$ by
    \[
    S(U) \cdot A = UAU^{-1}.
    \] Then $S(U)$ is orthogonal since $A\mapsto U A U^{-1}$ preserves the Hilbert-Schmidt inner product, i.e. $S:SU(2)\to O(3)$. Then, since $S(\idty) = \idty\in SO(3)$ and $SU(2)$ is connected, $S: SU(2)\to SO(3)$. It is also clear that $S$ is a group homomorphism. To see that $S$ is onto, one needs to demonstrate that each of the ``rotation about an axis'' $R_X(\theta),R_Y(\theta),R_Z(\theta)$ 1-parameter groups are in the image. We will show one such case, the $Z$-axis rotations. Observe that we may rewrite (\ref{eq:SU(2) is 3-sphere}) to say that any 1-parameter group $U(\theta) \in SU(2)$ can be written in terms of real vector $\vec{x}\in \R^3$:
    \begin{equation}
        U(\theta) = e^{i\theta \vec{x}\cdot \vec{\sigma}} = \cos \theta \idty + i \sin \theta (\vec{x}\cdot \vec{\sigma}) . 
    \end{equation} Consider the family $U_Z(\theta) = e^{i\theta \sigma^Z}$, and observe how it acts on the basis $\{\sigma^X,\sigma^Y,\sigma^Z\}$ of $V$:
    \begin{align*}
        U_Z(\theta)\sigma^X U_Z^{-1}(\theta) &= (\cos^2\theta - \sin^2 \theta )\sigma^X + ( - 2\sin\theta \cos \theta )\sigma^Y \\
        U_Z(\theta)\sigma^Y U_Z^{-1}(\theta) &= ( 2\sin\theta \cos \theta )\sigma^X + (\cos^2\theta - \sin^2 \theta )\sigma^Y  \\
        U_Z(\theta)\sigma^Z U_Z^{-1}(\theta) &= \sigma^Z . 
    \end{align*} where we have used $\sigma^Z\sigma^X = i\sigma^Y$, $\sigma^Y\sigma^Z = i\sigma^X$, and the easily shown anticommutation relations
    \begin{equation}
        \sigma^i \sigma^j = - \sigma^j \sigma^i \text{  when }i\neq j, \quad \sigma_i^2 = \idty . 
    \end{equation} So, recalling the double angle formulas, this means that on the Pauli basis of $V$, $S(U_Z(\theta))\in \GL(V)$ has the matrix
    \begin{equation} \label{eq:matrix double cover}
        S(U_Z(\theta)) = \begin{pmatrix}
            \cos 2\theta & -\sin 2\theta & 0 \\
            \sin 2\theta & \cos 2\theta & 0 \\
            0 & 0 & 1 
        \end{pmatrix}.
    \end{equation} Thus, $S$ is onto. The final thing to show is that $\ker S = \{\idty, -\idty\}$, then we are done by the first isomorphism theorem. It is clear that $S(\idty) = S(-\idty)$. It is also clear that if $U\neq \pm \idty$, then one can always find an $A\in SO(3)$ for which $U A \neq A U$, so we are done. 
\end{proof}
As a corollary, we see that as a manifold $SO(3)\cong \R \mathbb{P}^3$, real projective 3-space, and indeed, $SU(2)$ is the double cover of $SO(3)$--i.e., there is a 2-to-1 map $S:SU(2)\to SO(3)$ which is locally a homeomorphism. From facts about real projective space, we immediately have that
\begin{equation} \label{eq:fundamental group SO(3)}
    \pi_1(SO(3))=\Z_2 . 
\end{equation} The intuition for this double cover idea is greatly clarified by the matrix formula for rotations (\ref{eq:matrix double cover}). Here, we see the appearance of angles $2\theta$ instead of $\theta$. So when $\theta\in [0,2\pi)$, the double angle $2\theta$ has wrapped around $SO(3)$ twice while $\theta$ only wraps once.

Moreover, since $SU(2)$ is simply connected, it is the universal cover of $SO(3)$. As per Corollary \ref{prop:universal covers}, the universal cover in some sense controls the representation theory of its quotients: indeed, every representation of $SO(3)$ must be obtained from a representation of $SU(2)$. As we discussed, there are representations of $SU(2)$ which do not pass through the quotient and become representations of $SO(3)$: these will exactly be projective representations of $SO(3)$, and they play a key role in the SPT story. In some sense, one can think of the fundamental group (\ref{eq:fundamental group SO(3)}) as a topological obstruction to this passage of representations. But we are getting ahead of ourselves. It is time to discuss the irreducible representations of $SU(2)$ which are labeled by spin $s$.

\subsection{The spin-$s$ irreducible representations of $G=SU(2)$}

We will now build some representations of $\su(2)$ with bare hands, eventually demonstrating the structure of \textit{any} irrep of $\su(2)$. Recall that since $SU(2)$ is simply connected, Theorem\ref{thm:alg reps lift to group reps when G simply connected} guarantees these representations are in one-to-one correspondence with the representations of $SU(2)$.

Like any group or algebra, we have a notion of \textit{trivial representation}.

\begin{example}{The trivial representation of $\su(2)$, aka the spin-0 or singlet rep} \label{ex:trivial rep of su2}

        Let $V = \C$ and define $\Pi: SU(2) \to GL(\C)$ to be $\Pi(g) = \idty$ for any $g\in SU(2)$. Theorem \ref{thm:lie gp passes to lie alg} equips us to pass to a corresponding Lie algebra representation $\pi: \su(2) \to \gl(\C)$. More concretely, we pass by taking derivatives of paths in $SU(2)$, or just by writing for any $X\in \su(2)$
        \[
            \idty = \Pi(e^{tX}) = e^{t\pi(X)},
        \] which means that $\pi(X) = 0$. Notice in particular that $-i\pi(\frac{i}{2}\sigma^Z)=0$ is the (trivially) diagonal $1\times 1$ matrix with eigenvalue $0$. 
    \end{example}
    When we later describe every irrep of $\su(2)$ in Example \ref{ex:the spin s irrep}, it will become clear why this observation leads us to call the trivial representation of $\su(2)$ the spin $s=0$ representation.

    Matrix Lie groups also come with the natural \textit{defining representation}. This just means ``the representation where the matrices are their own representatives'', $g= \Pi(g)$.

\begin{example}{The defining representation of $\su(2)$, aka the spin-$1/2$ or doublet rep} \label{ex:def rep su2}

    Let $V=\C^2$ and define $\Pi:SU(2) \to GL(\C^2)$ by $\Pi(g) = g$. Of course, this also means that $\pi(X) = X$ for all $X\in \su(2)$. In the standard orthonormal basis of $\C^2$ given by $\ket{\up} := \begin{pmatrix} 1 \\ 0 \end{pmatrix},\ket{\down} = \begin{pmatrix} 0 \\ 1 \end{pmatrix}$, we have the diagonalization
    \[
        -i\pi\paran{\frac{i}{2} \sigma^Z} = \begin{pmatrix} \frac{1}{2} & 0 \\ 0 & -\frac{1}{2}
        \end{pmatrix}.
    \]
    \end{example}
    The following representation is especially important, as this gives powerful insight into how a Lie algebra acts on itself. We in fact have already encountered this representation: it was key to the proof of Proposition \ref{prop:SU(2) is double cover of SO(3)}.
    \begin{example}{The adjoint representation of $\su(2)$, aka the spin-1 or triplet rep}\label{ex:adjoint rep of su2}
    
        Let $V = \g_\C = \text{span}_\C \{\sigma^X,\sigma^Y,\sigma^Z\}$.
        Then $G=SU(2)$ acts on the vector space $\g_\C$ in a natural way: let $: G\to GL(\g_\C)$ be defined for all $g\in G$ by 
        \[
            \Ad_g(X) = gXg^{-1} \qquad \text{for all }X\in \g_\C .
        \] Again, taking $\frac{d}{dt}$ of the above where $g=e^{tX}$ with $X\in \g$ and evaluating at $t=0$ (or just using the exponential map), this induces the Lie algebra representation $\ad:\g \to \gl(\g_\C)$ which has that for any $X\in \g$,
        \[
            \ad_X(Y) = [X,Y] \qquad \text{for all }Y\in \g_C .
        \] Now, let us define operators $\sigma^{\pm}\in \g_\C$ by $\sigma^{\pm}:= \frac{1}{2}(\sigma^X \pm i\sigma^Y)$, or in matrices,
        $\sigma^+ = \begin{pmatrix}
                0 & 1 \\
                0 & 0
            \end{pmatrix}$ and $ \sigma^- = \begin{pmatrix}
                0 & 0 \\
                1 & 0
            \end{pmatrix}$.
        The commutation relations of the Pauli matrices (or direct computation) reveals that
        \begin{equation}\begin{split}
            -i[\frac{i}{2}\sigma^Z, \sigma^+] &= \sigma^+ ,\\
            -i[\frac{i}{2}\sigma^Z, \sigma^Z] &= 0 ,\\
            -i[\frac{i}{2}\sigma^Z, \sigma^-] &= -\sigma^- .
        \end{split}\end{equation} Or in other words, if we choose the ordered basis for $\g_\C$ to be $\sigma^+, \sigma^Z, \sigma^-$, then we have the diagonalization
        \[
            -i \ad_{\frac{i}{2}\sigma^Z} = \begin{pmatrix}
                1 & 0 & 0 \\
                0 & 0 & 0 \\
                0 & 0 & -1 
            \end{pmatrix}.
        \]
    \end{example}

    In physics, the operators $\sigma^\pm$ are often called ``raising and lowering'' or ``ladder'' operators. The interpretation being that the eigenvalues of (the representative of) the diagonal matrix $\sigma^Z$ are like rungs of a ladder, and the matrix $\sigma^+$ moves you up the ladder, while $\sigma^-$ moves you down the ladder. We list the relevant commutation relations here, where recall that $\sigma^\pm := \frac{1}{2}(\sigma^X \pm i \sigma^Y)$: 
    \begin{equation} \label{eq:ladder operators}
        [\frac{1}{2}\sigma^Z, \sigma^\pm] = \pm \sigma^\pm, \qquad [\sigma^+,\sigma^-] = \sigma^Z . 
    \end{equation} Note that across the literature, different conventions for the scaling of these operators are often taken.

    Let us unify these examples and describe any spin-$s$ irrep, keeping this ladder picture in mind. The commutation relations play an exceptionally important role.

    \begin{example}{The spin-$s$ irrep}\label{ex:the spin s irrep}
    
    Let $V = \C^{2s+1}$ and define a representation $\pi:\su(2)\to \gl(V)$ by 
        \begin{equation}\label{eq:diagonalized J3}
            -i\pi\paran{\frac{i}{2}\sigma^Z} = \begin{pmatrix}
                s & 0 & 0 & \dots & 0 & 0 \\
                0 & s-1 & 0 & \dots & 0 & 0 \\
                0 & 0 & s-2 & \dots & 0 & 0 \\
                \vdots & \vdots & \vdots & \ddots & \vdots & \vdots \\
                0 & 0 & 0 & \dots & -s+1 & 0 \\
                0 & 0 & 0 & \dots & 0 & -s
                 
            \end{pmatrix} 
        \end{equation} where we have defined $J_3:= -i \pi\paran{\frac{i}{2}\sigma^Z}$. Likewise, call the representatives $J_\pm = \frac{1}{\sqrt{2}} \pi(\sigma^\pm) = -\frac{i}{\sqrt{2}} \pi\paran{\frac{i}{2}(\sigma^X + \sigma^Y)}$. Label the eigenbasis of $V$ by $\ket{s},\ket{s-1},\dots , \ket{-s}$. We can recover matrix elements of the representatives of $\sigma^X, \sigma^Y$ by equivalently recovering the matrix elements of the representatives of $\sigma^\pm$. Using the commutation relations (\ref{eq:ladder operators}) and that Lie algebra homomorphisms have $\pi([X,Y]) = [\pi(X),\pi(Y)]$, one quickly demonstrates the commutation relations
        \begin{equation} \label{eq:comm relations raising lowering}
            [J_3, J_\pm] = \pm J_\pm, \qquad [J_+,J_-] = J_3 . 
        \end{equation} Observe that for any eigenvector $\ket{\lambda}$, we have the following fundamental equation
        \begin{equation}\begin{split} \label{eq:raising operator}
            J_3 J_\pm \ket{\lambda} &= (J_{\pm}J_3 \pm J_{\pm}) \ket{\lambda} \\
            &= (\lambda \pm 1) J_\pm \ket{\lambda}.
        \end{split}\end{equation} 
        In particular, $J_+$ either maps $\ket{\lambda}\mapsto (\lambda+1) \ket{\lambda + 1}$ or $\ket{\lambda}\mapsto 0$, and similarly $J_-$ either maps $\ket{\lambda} \mapsto (\lambda - 1 ) \ket{\lambda-1}$ or $\ket{\lambda} \mapsto 0$. 
    \end{example}
    Again, thinking of the eigenvalues of $J_3$ as the rungs of a ladder, we might think of $J_+$ as hopping up the rungs until falling off the top, and $J_-$ as hopping down until falling off the bottom:
    \begin{figure}
        \centering
        \includegraphics[width=1\textwidth]{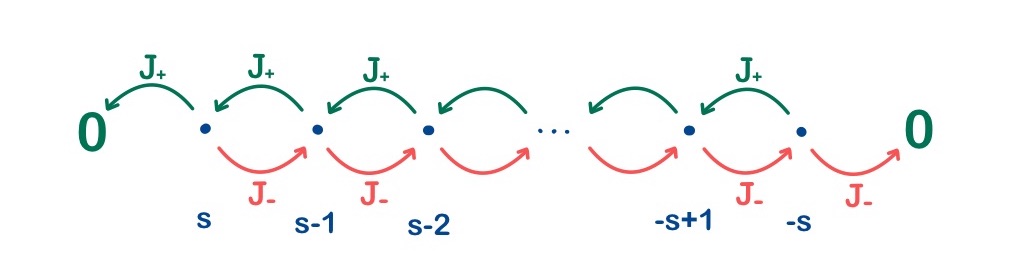}
        \caption{The ``ladder operators'' of the spin-$s$ irrep.}
        \label{fig:ladder_ops}
    \end{figure}

    It is not at all obvious from this construction that \textit{any} irrep of $\su(2)$ is equivalent to one of these. In fact, it is rather backwards from the typical presentation (see for instance Chapter 9 of~\cite{nachtergaele2016quantum}), but it is not hard to see the equivalence. Let us sketch this idea briefly. Let $V$ be irreducible. Since $V$ is a complex vector space, the representative of $\sigma^Z$, $J_3$, has at least one eigenvector $\phi\in V$. Then, since $V$ is finite dimensional, there exists a $k$ for which $J_+^k \phi \neq 0$ but $J_+^{k+1} \phi = 0$. Looking at Equation (\ref{eq:raising operator}), it is clear that $J_+^k \phi$ is the eigenvector of largest eigenvalue of $J_3$, call it $\ket{s}$. To see that $s$ must be an integer and not some arbitrary real number, observe that again by finite dimensionality and Equation (\ref{eq:raising operator}) there must exist a $k'$ such that $J_-^{k'}\ket{s} \neq 0$ but $J_-^{k'+1}\ket{s} = 0$. This can only happen when $s$ is an integer. One can show without too much effort that irreducibility implies that $\ket{s}$ is in fact the unique eigenvector of largest eigenvalue $s$ of $J_3$. This in turn forces $J_3$ to have the diagonalization (\ref{eq:diagonalized J3}). Once the diagonalization is determined, commutation relations decide the matrix elements of $J_\pm$. We thus realize that \textit{any} irreducible representation of $\su(2)$ is exactly determined by the largest ``eigenvalue'' $s$ of $J_3$. 
    
    We should be a careful saying ``eigenvalue'': after all, rescaling $J_3\mapsto cJ_3$ with $c\in \C$ yields an equivalent representation via Definition \ref{def:equivalent reps} but the eigenvalues are rescaled by $c$. This is easily remedied by thinking of $s$ not as an eigenvalue, but as a linear functional $s:\h \to \C$ where $\h := \text{span}_\C \{ \sigma^Z\}$ called a \textit{weight}.\footnote{This $\h$ is the first example of a \text{Cartan subalgebra}, which we will encounter later.} Letting $\ket{s},\ket{s-1}, \dots ,\ket{-s}$ be the eigenbasis of $V$ for $\pi(\sigma^Z)$, we define weights $m:\h\to \C$ by
    \begin{equation}
        \pi(\sigma^Z) \ket{m} = m(\sigma^Z) \ket{m} ,
    \end{equation} i.e. $m(\sigma^Z)$ is the $m^{th}$ diagonal entry of the representative $\pi(\sigma^Z)$. This has the benefit of storing eigendata while being invariant under rescaling. We will need a way of (partially ordering) these weights, but for now, the linear functional $s$ which corresponds to the largest eigenvalue of $\pi(\sigma^Z)$ will be called the highest weight. Then, per the argument above, irreps of $\su(2)$ are uniquely labeled by their highest weights $s$. This justifies our saying ``the spin-$s$ rep of $\su(2)$''. This is the $\su(2)$ case of the Theorem of Highest Weight, which we will present (without proof) in Section \ref{sec:highest weight reps}.

\subsection{Clebsch-Gordan and Tensor Decomposition}
    We can now identify any irrep of $\su(2)$ as a spin-$s$ irrep. Theorem~\ref{thm:complete reducibility} guarantees that any unitary representation (which is every finite dimensional representation by Weyl's unitary trick Proposition~\ref{prop:Weyl unitary trick}) is completely reducible and so can be decomposed into irreps. Tensor representations are of particular importance to quantum physicists and Lie theorists alike, and so understanding their decomposition into irreps is a subject of great study. Let $\mu,\nu\in \frac{1}{2}\N$ be spins labeling irreps $V_\mu, V_\nu$. We intend to study the irrep decomposition\footnote{In physics, this is often referred to as the ``Clebsch-Gordan decomposition''.}
    \begin{equation}\label{eq:tensor decomp su2}
        V_\mu \otimes V_\nu \cong \bigoplus_{\lambda\in \frac{1}{2}\N} V_{\lambda}^{\oplus c_{\mu,\nu}^\lambda} , 
    \end{equation} where $c_{\mu,\nu}^\lambda$ is an integer labeling the multiplicity with which the irrep $V_{\lambda}$ appears. Let us start off with a basic but illuminating example, wherein two identical spin-1/2 particles are rotated by an identical amount $U_g\in SU(2)$. Physically, one might think of this as an identical rotation of two dipoles, which preserves their alignment: ``rotated ferromagnets are still ferromagnets''. We will first need a lemma, an especially important tool which generalizes the commonly encountered lemma ``two diagonalizable operators are simultaneously diagonalizable if and only if they commute''.

    \begin{lemma}{(Simultaneous Block Diagonalization)}\label{lem:sim block diag}
    
    \noindent Let $U: G\to GL(V)$ be a representation of a group $G$ (or a Lie algebra $\g$), and let $H=H^*$ be a Hermitian operator such that $[U_g, H] = 0$ for all $g\in G$. Then, for any eigenvector $\ket{\psi}$ of $H$ with eigenvalue $\lambda$, $U_g \ket{\psi}$ is also an eigenvector of $H$ with eigenvalue $\lambda$.
    \end{lemma}
    \begin{proof}
        Observe that for any $g\in G$,
        \[
        H(U_g\ket{\psi}) = U_g H\ket{\psi} = U_g \lambda \ket{\psi} = \lambda (U_g \ket{\psi}) .
        \]
    \end{proof}
    \begin{remark}
        As a consequence, this means that for such a commuting pair $U,H$, any irreducible invariant subspace of $V$ must be contained in a $\lambda$-eigenspace of $H$.
    \end{remark}
    \begin{example}{The 2 qubit tensor rep $\C^2\otimes \C^2$ of $\g = \su(2)$} \label{ex:2 qubit tensor rep of su(2)}
    
    Recall from Example \ref{ex:def rep su2} that the defining representation $U_{(\cdot)}: SU(2) \to GL(\C^2)$ of $SU(2)$ is given by mapping each group element to its corresponding matrix $g = U_g$.
    We can tensor two of these together to construct another unitary representation. The tensor representation $\Pi := U^{\otimes 2}$, given by the map $\Pi: SU(2) \to GL(\C^2\otimes \C^2)$ defined by $g\mapsto \Pi(g) = U_g\otimes U_g$.
    Now, observe that this representation commutes with the Hermitian operator $\SWAP$, which is defined by $\SWAP \ket{u v} = \ket{vu}$:
    \[
        [U_g\otimes U_g, \SWAP] = 0 \qquad \text{for all } g\in SU(2). 
    \] $\SWAP$ can be readily diagonalized: it acts as $\idty$ on the symmetric subspace $\Sym^2(\C^2) = \text{span}\{ \ket{\up\up}, \ket{\up\down}+\ket{\down\up}, \ket{\down\down}\}$ and as $-\idty$ on the antisymmetric subspace $\Exterior^2(\C^2) = \text{span}\{\ket{\up\down}-\ket{\down\up} \}$.
    Thus, by the simultaneous block diagonalization Lemma \ref{lem:sim block diag}, every representative is block-diagonal in the ordered basis of $V = \Sym^2(\C^2)\oplus \Exterior^2(\C^2)$
    \begin{equation}\label{eq:su2-block-prod}
    U_g\otimes U_g =\begin{pmatrix} * & * & * & 0 \\
   * & * & * & 0 \\
    * & * & * & 0\\
    0 & 0 & 0 & * 
    \end{pmatrix}\, , \qquad  
    \SWAP = \begin{pmatrix}
        & & & 0 \\
        & \idty & & 0 \\
        & & & 0\\
        0 & 0 & 0 & -\idty 
    \end{pmatrix} \, .
\end{equation} This is exactly saying that $\Sym^2(\C^2)$ and $\Exterior^2(\C^2)$ are invariant subspaces of the representation $U^{\otimes 2}$.

Now, irreducibility of $\Exterior^2(\C^2)$ is immediate: it is a 1-dimensional representation and thus has no nontrivial invariant subspaces.
The irreps of $\su(2)$ are uniquely classified by spin, which is related to the dimension $d = 2s+1$. Since $d=1$, $s=0$, and so $\Exterior^2(\C^2)$ is the trivial representation from Example \ref{ex:trivial rep of su2}.

To discern the irreducibility of $\Sym^2(\C^2)$ takes only slightly more work. We proceed by diagonalizing the representative $\pi(\sigma^Z)$: if it matches (up to scalar multiplication) the diagonalization (\ref{eq:diagonalized J3}), then since each eigenvalue is distinct and shuffled about by the raising and lowering operators $\pi(\sigma^\pm)$ it must be irreducible, whence the Theorem of Highest Weight \ref{thm:highest weight reps} will uniquely identify it. But if any of the eigenspaces have dimension $>1$ (or as the physicists would say, if the eigenvalues have degeneracy), it must be that $\Sym^2(\C^2)$ is reducible from our classification of $\su(2)$ irreps.\footnote{A warning: other compact Lie groups can have irreps whose weight spaces--the generalization of this $\pi(\sigma^Z)$ eigenspace--are not necessarily 1-dimensional.}

The tensor representation $\Pi$ induces the tensor representation of the Lie algebra $\pi: \su(2) \to \gl(\C^2 \otimes \C^2)$ defined by $\pi(\sigma^Z) = \sigma^Z\otimes \idty + \idty \otimes \sigma^Z$. We will work on $\su(2)_\C$ so we can write $-i\pi(i\sigma^Z) = \pi(\sigma^Z)$. Now, notice that $\pi(\sigma^Z)$ is already diagonal in the earlier basis of $\Sym^2(\C^2)$:
\begin{equation}\begin{split} \label{eqn:Sym diagonalization}
    \pi(\sigma^Z) \ket{\up\up} &= 2 \ket{\up\up} \\
    \pi(\sigma^Z) \paran{\ket{\down\up} + \ket{\up\down}} &= 0 \paran{\ket{\down\up} + \ket{\up\down}} \\
    \pi(\sigma^Z) \ket{\down\down} &= -2 \ket{\down\down}.
\end{split}\end{equation}
So, rescaling, we see that $\pi(\frac{1}{2}\sigma_Z) = \begin{pmatrix}
    1 & 0 & 0 \\
    0 & 0 & 0 \\
    0 & 0 & -1
\end{pmatrix}$. Any two representations related by rescaling are equivalent, so this is exactly the spin-1 representation from Example \ref{ex:adjoint rep of su2}, aka the adjoint representation. 

While we are here, we may as well glance at the behavior of the raising and lowering operators from Example \ref{ex:the spin s irrep}. Recall that in the defining representation, we had $\sigma^+ = \begin{pmatrix} 0 & 1 \\ 0 & 0 \end{pmatrix}, \sigma^- = \begin{pmatrix} 0 & 0 \\ 1 & 0 \end{pmatrix}$, or 
    \[
    \ket{\down} \xmapsto{\sigma^+ \cdot } \ket{\up} \xmapsto{\sigma^+ \cdot } 0, \qquad \ket{\up} \xmapsto{\sigma^- \cdot } \ket{\down} \xmapsto{\sigma^- \cdot } 0.
\] Then in this tensor representation we have $\pi(\sigma^\pm) = \sigma^\pm \otimes \idty + \idty \otimes \sigma^\pm$, so 
    \begin{equation}\begin{split} \label{eqn:raising and lowering spin 1}
        &\ket{\down\down} \xmapsto{
        \pi(\sigma^+) \cdot } \paran{\ket{\up\down}+\ket{\down\up}} \xmapsto{\pi(\sigma^+) \cdot } 2\ket{\up\up} \xmapsto{\pi(\sigma^+) \cdot } 0  \\
        &\ket{\up\up} \xmapsto{\pi(\sigma)^- \cdot } \paran{\ket{\down\up}+\ket{\up\down}} \xmapsto{\pi(\sigma)^- \cdot } 2\ket{\down\down} \xmapsto{\pi(\sigma)^- \cdot } 0. 
    \end{split}\end{equation}
    As promised for any spin-$s$ irrep, the raising operator $\pi(\sigma^+)$ maps between the $-1\mapsto 0 \mapsto 1 $ eigenspaces of $\pi(\frac{1}{2}\sigma^Z)$ and annihilates the highest spin vector $\ket{\up\up}$ with $s=1$. 
    \end{example}

    Before fully answering the $\su(2)$ tensor irrep decomposition problem, i.e. computing the constants $c_{\mu,\nu}^\lambda$ for any spins $\mu,\nu$ from Equation (\ref{eq:tensor decomp su2}), let us think a bit more abstractly about properties of $\otimes, \oplus$. It is a standard fact, essentially linear algebra, that given reps $V_\mu,V_\nu,V_\kappa$, the distributive property holds:
    \begin{align*}
        V_\mu \otimes (V_\nu \oplus V_\kappa) &\cong (V_\mu \otimes V_\nu) \oplus (V_\mu \otimes V_\kappa )\\
        (V_\mu \oplus V_\nu) \otimes V_\kappa &\cong (V_\mu \otimes V_\kappa ) \oplus (V_\nu \otimes V_\kappa) , 
    \end{align*} and likewise, these operations on representations are associative. Indeed, the set of linear finite dimensional representations $\mathrm{Rep}_G$ equipped with the ``product'' $\otimes$ and ``sum'' $\oplus$ operations forms a ring, called the \textit{representation ring} of $G$. It is in fact a ring with identity, where the identity is the trivial representation $V_0$, since one can show that for any representation $V_\mu$, $V_0\otimes V_\mu \cong V_\mu \otimes V_0 \cong V_\mu$. This observation, combined with the following proposition, will allow us to do essentially all of our routine spin calculations.
    
\begin{proposition}(Clebsch-Gordan Rule)\label{prop:clebsch-gordan}
    
    Let $\mu,\nu\in \frac{1}{2}\N$ be spins labelling $\su(2)$ irreps $V_\mu, V_\nu$. Then their tensor product decomposes into irreps as 
    \begin{equation} \label{eq:clebsch-gordan}
        V_\mu \otimes V_\nu \cong V_{\mu+\nu} \oplus V_{\mu+\nu - 1} \oplus \dots \oplus V_{\,\abs{\mu - \nu}}. 
    \end{equation}
    
    \end{proposition} \begin{proof}
        Let $\Pi_\mu: SU(2)\to \GL(V_\mu)$ be a representation. Consider the ``character'' function $\chi_\mu:SU(2)\to \C$ given by 
        \begin{equation}
            \chi_\mu (g) = \Tr_{V_{\mu}} \Pi_\mu(g) . 
        \end{equation} We stated that any $SU(2)$ spin-$s$ irrep is equivalent to the form presented in Example \ref{ex:the spin s irrep}. Since the trace is invariant under equivalence of representations, i.e. under change of basis $\Tr A = \Tr S A S^{-1}$, we have that
        \begin{equation} \label{eq:character of sigmaZ}
            \chi_\mu\paran{e^{\frac{i\theta}{2} \sigma^Z}} = q^{\mu} + q^{\mu-1} + \dots + q^{-\mu}.
        \end{equation} where $q=e^{i\theta}$ is a complex number. By inspection it is clear that any irrep of $SU(2)$ yields a distinct character. By properties of trace, we have that the character of the direct sum $V_\mu \oplus V_\nu$ is $\chi_\mu + \chi_\nu$, and the character of the tensor product $V_\mu\otimes V_\nu$ is $\chi_\mu \cdot \chi_\nu$. WLOG let $\mu \geq \nu$. Then the tensor product character is 
        \begin{align*}
            \chi_\mu\cdot \chi_\nu &= (q^{\mu} + q^{\mu-1} + \dots + q^{-\mu}) (q^{\nu} + q^{\nu-1} + \dots + q^{-\nu}) \\
            &= (q^{\mu + \nu} + q^{\mu + \nu-1} + \dots + q^{-\mu - \nu}) + (q^{\mu + \nu-1} + \dots + q^{-\mu - \nu +1}) \\
            &\quad + \dots + (q^{\mu - \nu} + q^{\mu - \nu-1 } + \dots + q^{-\mu + \nu+1 } + q^{-\mu + \nu}) \\
            &= \chi_{\mu+\nu} + \chi_{\mu+\nu - 1} + \dots + \chi_{\mu-\nu}, 
        \end{align*} which by uniqueness is the character of the right hand side of \ref{eq:clebsch-gordan}.
    \end{proof}
    Characters are indispensable tools in representation theory, and this quick proof is just a shadow of their full power. See~\cite{fulton1991representation} for an introduction. 

    \begin{example}{Some examples and special cases}
        \begin{itemize}
            \item $V_{\mu}\otimes V_0 \cong V_\mu$ (the trivial rep is the identity in the representation ring)
            \item $V_{1/2}\otimes V_{1/2} \cong V_{1} \oplus V_{0}$ (see 2 qubit Example \ref{ex:2 qubit tensor rep of su(2)}) 
            \item $V_{2}\otimes V_{1} \cong V_{3} \oplus V_2 \oplus V_{1}$
            \item $V_{\mu}\otimes V_{1/2} \cong V_{\mu +1/2}\oplus V_{\mu -1/2}$
            \item $V_{1/2}\otimes V_{1/2} \otimes V_{1/2}\cong V_{1/2}\otimes (V_{1}\oplus V_0) \cong V_{3/2}\oplus 2V_{1/2}$.
        \end{itemize}
    \end{example}

\subsection{Total Spin: the Quadratic Casimir}
    
    We will now see two commonly appearing examples of quadratic Casimir operators, which will suffice for our needs here and allow us to study the spin interactions from earlier. We should note that these essentially come from the same source: whenever one has a Lie algebra $\g$ equipped with a non-degenerate $\Ad$-invariant bilinear form, it possesses a unique (up to scalars) quadratic Casimir element in the center of its universal enveloping algebra $U(\g)$. To go into this is beyond our scope, and we direct the curious reader to~\cite{kirillov2008introduction,hall2015lie}. All we will say at the moment is that the bilinear form $(A,B) = \Tr AB$ on $\su(2)$ is non-degenerate and $\Ad$-invariant, and every representation of $\su(2)$ stores a representative of the Casimir element. The utility of this element is its centrality: by Schur's lemma \ref{lem:schur}, central elements act as 0 between inequivalent irreps, and as scalar multiples of the identity between equivalent irreps.
    
    Let $\calH = \C^{2s+1}$ be the unitary spin $s$ irrep of $\su(2)$ where $\pi:\su(2)\to \gl(\calH)$. Denote the representatives of the Pauli matrices $S^i = -i\pi(\frac{i}{2}\sigma^i)$, $i=X,Y,Z$, and the representatives of the ladder operators $S^\pm = S^X \pm i S^Y$.
    We have encountered several similar expressions to the \textit{total angular momentum operator}
    \begin{equation} \label{def:total angular momentum}
        \vec{S}^2 := \vec{S}\cdot \vec{S} = (S^X)^2 + (S^Y)^2 + (S^Z)^2. 
    \end{equation} One can quickly compute that $[\vec{S}^2, S^i]=0$ for each generator $i=X,Y,Z$, and so $\vec{S}^2$ commutes with every representative. Schur's lemma \ref{lem:schur} then assures us that $\vec{S}^2 = \lambda \idty$ for some constant $\lambda\in \C$. We may compute this constant exactly. Just as before in (\ref{eq:ladder operators}), commutation relations of the Paulis yield commutation relations of the ladder operators
    \begin{equation}
        [S^Z, S^\pm] = \pm S^\pm, \qquad [S^+,S^-] = 2S^Z . 
    \end{equation} Now, $S^+$ annihilates the highest weight vector $\ket{s} \in \calH$. So, using the above and $[S^X,S^Y]=iS^Z$, we have 
    \begin{equation}\begin{split}\label{eq:eigenvalues of Casimir}
        0 = S^- S^+ \ket{s} &= (S^+ S^- - 2S^Z)\ket{s} \\
        &= \paran{(S^X + iS^Y)(S^X - iS^Y) - 2S^Z} \ket{s} \\ 
        &= \paran{ \vec{S}^2 - (S^Z)^2 - S^Z }\ket{s} \\ 
        &= \lambda \ket{s} - s^2 \ket{s} - s\ket{s}, 
    \end{split}\end{equation} which means that $\lambda = s(s+1)$ and so on the spin-$s$ irrep $\calH$ we have 
    \begin{equation}
        \vec{S}^2 = s(s+1)\idty. 
    \end{equation} 
    
    Let us now consider the case of $\calH \otimes \calH$, the tensor representation of two spin-$s$ particles. This is the setting for all of our nearest-neighbor interactions. The representatives then become $S^i\otimes \idty + \idty \otimes S^i$, or $S^i_1 + S^i_2$,  $i=X,Y,Z$. Then the Casimir operator here, call it $C\in \gl(\calH\otimes \calH)$, is given by
    \begin{equation}\begin{split}
        C :&= (S_1^X + S_2^X)^2 + (S_1^Y + S_2^Y)^2 + (S_1^Z + S_2^Z)^2 \\
        &= \vec{S}_1^2 + \vec{S}_2^2 + 2\vec{S}_1\cdot \vec{S}_2 \\
        &= 2s(s+1)\idty +  2\vec{S}_1\cdot \vec{S}_2 ,
    \end{split}\end{equation} where the last line holds by our earlier calculation of $\vec{S}^2$ on a spin-$s$ irrep, and each tensor slot in $\calH\otimes \calH$ is a spin-$s$ irrep.
    Again, $C$ commutes with every representative.
    The Clebsch-Gordan decomposition (\ref{prop:clebsch-gordan}) tells us that $\calH\otimes \calH\cong V_{2s} \oplus V_{2s-1} \oplus \dots \oplus V_{0}$. Let $P^{(\mu)}$ denote the orthogonal projection onto the invariant spin-$\mu$ subspace $V_{\mu}$. Multiplicity-freeness of this decomposition along with Schur's lemma (\ref{lem:schur}) allows us to expand the Casimir as 
    \begin{equation}
        C = \lambda_{2s} P^{(2s)} + \lambda_{2s-1} P^{(2s-1)} + \dots + \lambda_{0} P^{(0)}, 
    \end{equation} where $\lambda_\mu = \mu(\mu+1)$ by calculation (\ref{eq:eigenvalues of Casimir}) and irreducibility of $V_{\mu}$. Rearranging these two expressions for $C$, we have that
    \begin{equation} \label{eq:spin interaction}
        \vec{S}_1 \cdot \vec{S}_2 = - s(s+1)\idty + \frac{1}{2}(\lambda_{2s} P^{(2s)} + \lambda_{2s-1} P^{(2s-1)} + \dots + 2P^{(1)}) .
    \end{equation}

%-------------------------------------------

%---------------------------------------------
\section{Examples Revisited: Putting the ``Spin'' in Quantum Spin Systems} \label{sec:Examples Revisited: Putting the ``Spin'' in Quantum Spin Systems}
Let us finally return to our example interactions from earlier: the Heisenberg Chain (\ref{ex:heisenberg chain}), the Majumdar-Ghosh Chain (\ref{ex:Majumdar-Ghosh chain}), and the AKLT chain (\ref{ex:AKLT chain}).

\begin{example}{The Heisenberg Interaction, revisited}\label{ex:heisenberg, revisited}

Let $\calH_x$ be a spin-1/2 particle for every $x\in [a,b]\subseteq \Z$. Then by Clebsch-Gordan (\ref{prop:clebsch-gordan}), $V_{1/2}\otimes V_{1/2}\cong V_{1} \oplus V_{0}$, so by Equation (\ref{eq:spin interaction}) and $\sigma^i = 2S^i$ for $i=X,Y,Z$, the interaction becomes
\begin{align*}
 H_{\Lambda} = - J \sum_{x=a}^b \vec{\sigma}_x \cdot \vec{\sigma}_{x+1} = J\sum_{x=a}^b -4 ( P_{x,x+1}^{(1)} - \frac{3}{4} \idty ) = J \sum_{x=a}^b -4 P_{x,x+1}^{(1)} + 3 \idty 
\end{align*} But our work in the 2 qubit tensor rep Example \ref{ex:2 qubit tensor rep of su(2)} shows that $P^{(1)}$ is the projection onto the +1 eigenspace $\Sym^2(\C^2)$ of the $\SWAP$ operator, i.e. $P^{(1)} = \frac{1}{2}(\idty + \SWAP)$. This means that up to a trivial shift in ground state energy, the Hamiltonian is expressible as
\begin{equation}
    H_{\Lambda} = 2J\sum_{x=a}^b - \SWAP_{x,x+1} . 
\end{equation} In particular, for the ferromagnetic $J>0$, this Hamiltonian rewards symmetric vectors with lower energy.
\end{example}

\begin{example}{The Majumdar-Ghosh Interaction, revisited}\label{ex:Majumdar-Ghosh, revisited}

Again letting $\Lambda = [a,b]$ denote a finite chain of spin-1/2 particles, the Majumdar-Ghosh interaction is closely related to the Heisenberg interaction:
\begin{equation} \begin{split}
    H_{\Lambda} &= \sum_{x=a}^{b-2} h^{(MG)}_{x,x+1,x+2} \\
    &= J \sum_{x=a}^{b-2} 2\vec{\sigma}_x\cdot \vec{\sigma}_{x+1} + \vec{\sigma}_x\cdot \vec{\sigma}_{x+2} . 
\end{split} \end{equation}
Expanding out terms, one realizes that an equivalent formulation is
\begin{equation}
    H_{\Lambda} =  \underbrace{ J \vec{\sigma}_a \cdot \vec{\sigma}_{a+1} - J \vec{\sigma}_{b-1}\cdot \vec{\sigma}_b  }_{\text{``boundary terms''}}+ J \sum_{x=a}^{b-2} \vec{\sigma}_x \cdot \vec{\sigma}_{x+1}+ \vec{\sigma}_x \cdot \vec{\sigma}_{x+2} + \vec{\sigma}_{x+1} \cdot \vec{\sigma}_{x+2} . 
\end{equation}

We again perform a Clebsch-Gordan decomposition: $V_{1/2}\otimes V_{1/2}\otimes V_{1/2}\cong V_{3/2} \oplus 2 V_{1/2}$. One can then perform a similar calculation to that of \ref{eq:spin interaction} by using the quadratic Casimir on $V_{1/2}^{\otimes 3}$ to see that
\begin{equation}
    \vec{\sigma}_{x} \cdot \vec{\sigma}_{x+1}+ \vec{\sigma}_x \cdot \vec{\sigma}_{x+2} + \vec{\sigma}_{x+1} \cdot \vec{\sigma}_{x+2} = 6 P_{x,x+1,x+2}^{(3/2)} - 3\idty . 
\end{equation} So, up to a shift in the spectrum and absorbing positive scalars into $J$, we have
\begin{equation}
    H_{\Lambda} = (\text{boundary terms}) + J \sum_{x=a}^{b-2} P_{x,x+1,x+2}^{(3/2)} .
\end{equation} The kernel of each term consists of two spin-1/2 particles $2V_{1/2}$. We will later find that this Hamiltonian is frustration-free, which will allow us to compute exact ground states for the full chain. 
\end{example}

\begin{example}{The AKLT Interaction, revisited}\label{ex:AKLT, revisited}

Let $\Lambda = [a,b]\subseteq \Z$ be a chain of spin-1 particles and recall the AKLT interaction
\begin{equation} \begin{split}
    H_{\Lambda} &= \sum_{x=a}^{b-1} \frac{1}{3}\idty + \frac{1}{2} \vec{S}_x\cdot \vec{S}_{x+1} + \frac{1}{6} (\vec{S}_x\cdot \vec{S}_{x+1})^2 . 
\end{split} \end{equation} We again do a Clebsch-Gordan decomposition (\ref{prop:clebsch-gordan}) to see that 
\begin{equation}
    V_1\otimes V_1 \cong V_2 \oplus V_1 \oplus V_0 . 
\end{equation} Since the on-site spin is $s=1$, the calculation (\ref{eq:spin interaction}) tells us
\begin{equation}\begin{split}
    \vec{S}_{x}\cdot \vec{S}_{x+1} &= - 2\idty + \frac{1}{2}(2(2+1) P^{(2)} + 2P^{(1)}) \\
    & = -2\idty+3P^{(2)}+P^{(1)} \\
    (\vec{S}_{x}\cdot \vec{S}_{x+1})^2 &= 4 \idty - 3P^{(2)} - 3 P^{(1)} , 
\end{split}\end{equation} where for expanding the square we have used orthogonal projectors have $P^{(s)}P^{(s)}=P^{(s)}$ and that distinct irreducibles are orthogonal from one another, so $P^{(1)}P^{(2)} = P^{(2)}P^{(1)} = 0$. Altogether we have
\begin{equation}
    \frac{1}{3}\idty + \frac{1}{2} \vec{S}_x\cdot \vec{S}_{x+1} + \frac{1}{6} (\vec{S}_x\cdot \vec{S}_{x+1})^2 = \paran{\frac{1}{3} - 1 + \frac{2}{3} } \idty + \paran{\frac{3}{2} - \frac{1}{2} }P^{(2)} + P^{(1)}\paran{\frac{1}{2} - \frac{1}{2}} = P^{(2)} . 
\end{equation} So, the AKLT interaction is in fact an orthogonal projection onto the spin-2 subspace of two neighboring spin-1's:
\begin{equation}
    H_{\Lambda} = \sum_{x=a}^{b-1} P^{(2)}_{x,x+1}.
\end{equation} The ground state space of the 2-site interaction $P^{(2)}_{x,x+1}$ then contains a singlet and a triplet, $V_1\oplus V_0$. We will later learn that the exact ground state space of $H_{\Lambda}$ for finite chains of any length is still a singlet and a triplet. It is worth mentioning that we can think of these representations as arising in a different way. Take two spin-1/2 particles. Then by Clebsch-Gordan we have the following decomposition into irreps:
\begin{equation}
    V_{1/2}\otimes V_{1/2} \cong V_1 \oplus V_0 . 
\end{equation} Thus, the ground state space of the AKLT interaction can be described by two spin-1/2 particles. This observation will later form the basis for the valence bond state (VBS) description of the AKLT ground state, which is intimately tied to the matrix product state (MPS) description. 
\end{example}

%--------------------------------------------

%--------------------------------------------
\section{Highest Weight Representations} \label{sec:highest weight reps}

In the section following Example \ref{ex:the spin s irrep}, we sketched an argument that every finite dimensional irrep of $\su(2)$ is a spin-$s$ irrep for some $s\in \frac{1}{2}\mathbb{N}$. This is the $\su(2)$ case of the Theorem of Highest Weight, which entirely classifies finite dimensional irreducible representations of a semi-simple Lie algebra $\g$. We will cite the relevant definitions and the theorem. Our story will be highly biased towards two families of Lie algebras: the orthogonal Lie algebras $\so(n)$, which correspond to the series $B_{(n-1)/2}$ for odd $n$ and $D_{n/2}$ for even $n$ in the classification of simple Lie algebras. Each of these algebras are simple, except for $\so(4)\cong \so(3)\oplus \so(3)$ which is semi-simple. Their representation theory is particularly important for this Thesis, as they dictate the behavior of $SO(n)$ AKLT chains we are studying.

Throughout this chapter, we work with a complex semi-simple Lie algebra $\g$.
\begin{definition}
Let $\g$ a complex semi-simple Lie algebra. A \emph{Cartan subalgebra} $\h$ is a maximal abelian subalgebra in $\g$ such that for every $H\in \h$, the adjoint representative $\ad_H: \g \to \g$ is diagonalizable.
\end{definition} Let us find a nice choice of Cartan subalgebra for $SO(n)$ and then proceed.

First, let $R\in SO(3)$. $R$ commutes with $R^* = R^T$, and so we can apply the spectral theorem to diagonalize $R$ over $\C^3$. Since $\det(R) = 1$, the product of the eigenvalues is $1$. $R$ is real, so complex eigenvalues must appear in conjugate pairs, so $A$ must have eigenvalue 1. By $R$ unitary, the remaining pair of eigenvalues must have modulus 1 and so be of the form $e^{i\theta}, e^{-i\theta}$. Let $u,\overline{u}\in \C^3$ denote the eigenvectors of $R$, i.e.
    \[
        Ru = e^{i\theta} u, \qquad R \overline{u} = e^{-i\theta} \overline{u}.
    \] We can define $x := \frac{1}{2}(u - \overline{u})$ and $y:= \frac{-i}{2}(u+\overline{u})$ so that, using $\frac{1}{2}(e^{i\theta} + e^{-i\theta}) = \cos \theta$ and $\frac{-i}{2}(e^{i\theta} + e^{-i\theta}) = \sin \theta$, 
    \begin{equation}\begin{split}
        Rx &= (\cos \theta)x + (\sin \theta)y \\
        Ry &= -(\sin \theta)x + (\cos \theta)y 
    \end{split} \end{equation}
    So there exists a unitary change of basis $V$ such that
    \[
    R = V \begin{pmatrix}
       \cos \theta & \sin \theta  & 0\\
        -\sin \theta & \cos \theta & 0 \\
         0& 0 & 1 
    \end{pmatrix} V^{-1} . 
    \] This argument, used on conjugate pairs of eigenvalues, proves the following.
    \begin{proposition}~\cite{simon1996representations}
    When $n$ even, every $R\in SO(n)$ may be diagonalized to
    \begin{equation}
        R = V\begin{pmatrix}
        \cos \theta_1 & \sin \theta_1 & 0 & 0 &\dots & 0 & 0 \\
        -\sin \theta_1 & \cos \theta_1 & 0 & 0 &\dots & 0 & 0 \\ 
        0 & 0 & \cos \theta_2 & \sin \theta_2 &\dots & 0 & 0 \\
        0 & 0 & -\sin \theta_2 & \cos \theta_2 &\dots & 0 & 0 \\
        \vdots & \vdots & \vdots & \vdots &\ddots & \vdots & \vdots \\
        0 & 0 & 0 & 0 &\dots & \cos \theta_{n/2} & \sin \theta_{n/2}  \\
        0 & 0 &0 & 0 &\dots & -\sin \theta_{n/2} & \cos \theta_{n/2}  
        \end{pmatrix} V^{-1}
    \end{equation} When $n$ odd, the diagonalization is the same, except we add a $1$ in the $(n,n)^{th}$ entry of the above matrix.
\end{proposition}
Now, we may pass to the Lie algebra $\so(n)$ by taking derivatives in our coordinates $\theta_i$. This gives us a natural choice of Cartan subalgebra $\h$. Let $\ket{j}, j=1,\dots, n$ label an orthogonal basis of $\C^n$, and define the skew-symmetric matrix $L_{j,k} = \ket{j}\bra{k} - \ket{k}\bra{j}$. Then, when $n$ even, a basis for the Cartan subalgebra is given by $\{H_j = L_{2j-1,2j}: 1\leq j \leq n/2\}$. Or in block matrix form, where each copy of $L_{2j-1,2j}=\begin{pmatrix} 0 & 1 \\ -1 & 0 \end{pmatrix}$ may be thought of as specifying a $2\times 2$ matrix block,
    \begin{equation} \label{eq:cartan subalgebra so(n)}
        H_1 = \begin{pmatrix}
             L_{1,2} & & \\
             & \ddots & \\
             & & \ddots
        \end{pmatrix}, \; H_2 = \begin{pmatrix}
             \ddots  & & \\
             & L_{3,4} & \\
             & & \ddots 
        \end{pmatrix}, \dots , \; H_{n/2} = \begin{pmatrix}
              \ddots  & & \\
             & \ddots & \\
             & & L_{n-1,n}
        \end{pmatrix}
    \end{equation} When $n$ is odd, we can obtain a similar basis for $\h$ by taking $\{H_1,H_2, \dots, H_{(n-1)/2}\}$, noticing that the $n^{th}$ row and column of any matrix in $\mathfrak{h}$ is always 0 in this case.

\subsection{Weights of a representation}
    Let $\mathfrak{g}$ be a complex semi-simple Lie algebra and let $\h$ a fixed Cartan subalgebra of $\mathfrak{g}$. Let $(\pi,V)$ be a representation of $\mathfrak{g}$.
    \begin{definition}
    A linear functional $L\in \h^*$ is a \emph{weight} of $\pi$ if there exists a nonzero vector $v\in V$ such that
    \[
        \pi(H) v = L(H) v , \qquad \text{ for all }H\in \h .
    \] The \emph{weight space} corresponding to $L$ is the set of all $v\in V$ satisfying the above, and the \emph{multiplicity} of $L$ is the dimension of the corresponding weight space.\footnote{As a comment, there are many equivalent definitions for a weight, all highlighting different features. Indeed, as in~\cite{hall2015lie}, using e.g. the Hilbert-Schmidt inner product $\inprod{A,B} = \Tr A^* B$, one can identify $\lambda\in \h^*$ with $\inprod{\wt{\lambda}, \cdot}$ where $\wt{\lambda}\in \h$; it is this $\wt{h}$ that is called the weight. This has the advantage of being a bit more visually clear when working with root lattices and the Weyl group action, but we won't be doing that here.}
    \end{definition}
    
    Given a Cartan subalgebra $\h$, there are several commonly chosen bases for the space of linear functionals $\h^*$, each of which comes with its own merits.
    For the Cartan subalgebra (\ref{eq:cartan subalgebra so(n)}) of $\so(n)$, we will take the basis $\{L_j\}$ to be the dual basis to $\{H_j\}$ up to a slightly annoying but necessary constant $i$, i.e.
    \begin{equation}\label{def:weights of cartan subalgebra}
            L_j(H_k) = i\delta_{jk} ,\qquad 1\leq j,k\leq n/2.
    \end{equation}

    \begin{example}{The weights of the $\Sym^2(\C^2)$ representation} 
    Recall the two qubit tensor rep $\pi: \su(2)\to \gl(V)$ where $V = \C^2 \otimes \C^2$ from Example \ref{ex:2 qubit tensor rep of su(2)}. We did some calculations to identify the invariant subspace $\Sym^2(\C^2)$ with the spin-1 representation, but let us revisit them with our new definitions.
    In this case, the Cartan subalgebra is the one-dimensional subalgebra $\h = \text{span} \{H\}$, where $H= \frac{-i}{2}\sigma^Z$. Let us recast the key step, Equation (\ref{eqn:Sym diagonalization}), wherein we simultaneously diagonalized the Cartan subalgebra:
    \begin{equation}\begin{split}
        \pi(H) \ket{\up\up} &= 1\cdot L (H) \ket{\up\up} \\
        \pi(H) \paran{\ket{\up\down} + \ket{\down\up}} &= 0\cdot L (H) \paran{\ket{\up\down} + \ket{\down\up}} \\
        \pi(H) \ket{\down\down} &= -1\cdot L (H) \ket{\down\down},
    \end{split}\end{equation} where $L\in \h^*$ is the functional defined by $L(H) = 1$ (note that $\dim(\h) = \dim(\h^*) = 1$).
    So, each vector in this diagonalization has a corresponding weight: $\ket{\up\up}$ has weight $L$, $\ket{\up\down}+\ket{\down\up}$ has weight $0$, and $\ket{\down\down}$ has weight $-L$. 

    We can also take a moment to reflect on the action of the raising and lowering operators $\sigma^+ = \frac{1}{2}(\sigma^X + i\sigma^Y) = \begin{pmatrix}
        0 & 1 \\
        0 & 0
    \end{pmatrix}$ and $\sigma^- = \frac{1}{2}(\sigma^X - i\sigma^Y) = \begin{pmatrix}
        0 & 0 \\
        1 & 0
    \end{pmatrix}$. We saw that they ``raise and lower'' between eigenvalue rungs on a ladder in (\ref{fig:ladder_ops}). But now, it is clear that we could just as well take the rungs to be labeled as $L, 0, -L$ instead of $1, 0, -1$. Indeed, by doing so we remove the ambiguity caused by rescaling $\sigma^Z\mapsto c\sigma^Z$, which will change the eigenvalues of $\pi(c\sigma^Z)$ to $c,0,-c$, but will not affect the linear functionals. In this way, the representation yields a finite lattice of weight spaces corresponding to integer-size functionals $L$, where we may hop between sites on the lattice using $\sigma^\pm$. 

    Moreover, one may generate the entire representation $\C^3$ by starting with the ``highest weight'' vector, the $\ket{\up\up}$ with weight $L$, and repeatedly acting via $\pi(\sigma^-)$ to generate the lattice. One way to define the ``highest weight'' is say a vector is highest weight if it is annihilated by $\sigma^+$, i.e. $\pi(\sigma^+)\ket{\up\up} = 0$.
    \end{example}
    
    For $\su(2)$, the story is rather one-dimensional, because well, $\h$ is one-dimensional. But it does highlight some important features. The weights of finite dimensional irreps of other semi-simple Lie algebras will similarly form finite lattices. These weight lattices may be thought of as being generated by acting on a distinguished weight vector (the ``highest weight'' vector) with a generalization of raising and lowering operators (the ``roots'' of a Lie algebra). 
    This motivates the following definitions.

    \begin{definition} \label{def:root}
     A nonzero linear functional $\alpha\in \h^*$ is a \emph{root} if there exists a nonzero $X\in \g$ with 
    \begin{equation} 
        [H, X] = \alpha(H) X \qquad \text{for all } H\in\h .
    \end{equation} In other words, $\alpha$ is a nonzero weight of the adjoint representation $\text{ad} : \g \to L(\g)$, where $\text{ad}_H(\cdot) = [H,\cdot]$. We call the weight space corresponding to a root $\alpha$ the \emph{root space} of $\alpha$. 
    \end{definition}
    \begin{example}
        Consider $\su_\C(2)$ and let $\alpha\in \h^*$ be defined by $\alpha(\sigma^Z) = 1$. Then $\alpha$ is a root because the raising and lowering operators $\sigma^{\pm} = \frac{1}{2}(\sigma^X \pm i \sigma^Y)$ enjoy the commutation relations
        \begin{equation}
            [\sigma^Z, \sigma^\pm] = \pm \alpha(\sigma^Z) \sigma^\pm . 
        \end{equation}
    \end{example}
    
     To connect to the $\su(2)$ case, recall that the raising operator annihilates the highest spin vector and otherwise changes the spin $s$ of any other vector by the map $s\mapsto s+1/2$, while the lowering operator annihilates the lowest spin vector and the changes the others by $s\mapsto s-1/2$. Roots similarly allow us to ``hop'' between the weights of a representation in the following sense: if $(\pi,V)$ a representation, $\alpha$ a root with $X_\alpha\in \g$ in the associated root space (i.e. satisfying (\ref{def:root}), and $v\in V$ a vector of weight $\beta$, then
    \begin{equation}\begin{split} \label{eqn:lie theory fundamental computation}
        \pi(H) \pi(X_\alpha)v  &= \pi(X_\alpha) \pi(H)v  + [\pi(H),\pi(X_\alpha)]v \\
        &= \pi(X_\alpha)\beta(H)v + \alpha(H)\pi(X_\alpha) v \\
        &= (\beta(H) + \alpha(H))\pi(X_\alpha) v .
    \end{split}\end{equation} In particular, $v$ is a vector of weight $\beta$ and $\pi(X_\alpha) v$ is a vector of weight $\beta+\alpha$. If $\beta+\alpha$ leaves the finite weight lattice corresponding to the representation $(\pi,V)$, then $\pi(X_\alpha)v = 0$. This is exactly analogous to the $\su(2)$ case. In this case, the weight lattice of an irrep $V$ is a one dimensional finite lattice labeled by the spins of each $v\in V$. The raising operator then annihilates the highest spin vector and lowering operator annihilates the lowest spin vector. One could in principle choose to identify irreps with their highest spin or their lowest spin--a choice has to be made, and the convention is to describe irreps in terms of the highest spin vector, i.e. the vector which is annihilated by the raising operator. The set of positive simple roots serves to generalize this idea and allow us to consistently choose a vector of ``highest weight'' for a finite-dimensional irrep $V$.

    When $n$ is even (Lie type $D_{n/2}$), a set of positive simple roots $\Delta$ in terms of $L_i\in \h^*$ given by (\ref{def:weights of cartan subalgebra}) is 
    \begin{equation}
       \Delta := \{L_1-L_2, L_2-L_3, \dots, L_{n/2-1} - L_{n/2}, L_{n/2-1}+L_{n/2} \}.
    \end{equation} 
    When $n$ is odd (Lie type $B_{(n-1)/2}$),
    a set of positive simple roots $\Delta$ is 
    \begin{equation}
        \Delta := \{L_1-L_2, L_2-L_3, \dots, L_{(n-1)/2-1} - L_{(n-1)/2}, L_{(n-1)/2}\}.
    \end{equation} It is a fact that any root $\alpha$ can be expressed as an integer linear combination of the positive simple roots~\cite{hall2015lie}. Given a root $\alpha$, it is an exercise in linear algebra to find an explicit matrix $X_\alpha\in \g$ in the associated root space (see e.g. Chapter 7.7 of~\cite{hall2015lie}). 

    \begin{definition}
     Let $(\pi,V)$ a representation of $\g$. A \emph{highest weight vector} $v\in V$ is a vector annihilated by all of the positive simple roots, i.e. 
    \begin{equation}
        \pi(X_\alpha) v = 0 , \qquad \text{for all } \alpha\in \Delta .
    \end{equation} The \emph{highest weight} of a representation is then the weight of its highest weight vector.   
    \end{definition}
    In general, this vector may not be unique. But for an irrep $(\pi,V)$, there is exactly one highest weight vector $v$, and one can use the fundamental computation (\ref{eqn:lie theory fundamental computation}) to cyclically generate $V$:
    \begin{equation}
        V = \{ \pi(X_\alpha)v : \alpha \text{ a root}\}.
    \end{equation} We will provide several examples of this in the next section. But first, the big theorem.
    
    \begin{theorem}{(Theorem of highest weight~\cite{hall2015lie})} \label{thm:highest weight reps}
    
    Every irreducible, finite-dimensional representation of $\g$ has a highest weight, and two irreducible, finite-dimensional representations of $\g$ with the same highest weight are isomorphic.
    \end{theorem} This is actually a simplified version of the theorem of highest weight. The full theorem guarantees that such a highest weight must be dominant integral, and that any dominant integral weight produces a unique finite dimensional irreducible representation. We will not expand on these definitions, except to say that the ``dominant'' definition gives a convention to consistently choose a highest weight and the ``integral'' definition ensures that the irrep associated to a highest weight is finite dimensional. See~\cite{hall2015lie} for more information. Let us get to a meaty class of examples.
    
    \subsection{Exterior powers $\Exterior^k V$}\label{sec:exterior power reps}
    
    We begin with the defining representation $V=\C^n$, where all matrices $X\in \so(n)$ coincide with their representatives. We demonstrate the even $n$ case. The Cartan subalgebra $\h$ given by (\ref{eq:cartan subalgebra so(n)}) is then simultaneously diagonalized by the basis $\{f_1,f_2,\dots f_{n/2}, \wt{f}_1, \wt{f}_2,\dots \wt{f}_{n/2} \} \subseteq V$, which we express in terms of the earlier orthonormal basis $\ket{1},\ket{2},\dots \ket{n}$ of $V$:
    \begin{equation}
    f_k := \ket{2k-1} + i\ket{2k}, \quad \wt{f}_k := \ket{2k-1} - i\ket{2k}, \qquad 1\leq k \leq n/2 . 
    \end{equation} The computation below verifies that this basis simultaneously diagonalizes $\h$:
    \begin{equation} \label{eqn:sim diag cartan so(n)}
        H_j f_k = L_k(H_j) f_k, \qquad H_j \wt{f_k} = - L_k(H_j) \wt{f_k} ,
    \end{equation} where the linear functionals $L_k:\h\to \C$ are as in (\ref{def:weights of cartan subalgebra}). For example, in the $n=4$ case, this means that $f_1$ has weight $L_1$, $\wt{f_1}$ has weight $-L_1$, $f_2$ has weight $L_2$, and $\wt{f_2}$ has weight $-L_2$. In particular, one can draw the corresponding weight lattice, where the fundamental computation (\ref{eqn:lie theory fundamental computation}) shows how to use an element of the root space $X_\alpha\in \g$ where $\alpha\in \Delta$ a positive simple root to hop between the weight spaces:
\begin{figure} \centering
\includegraphics{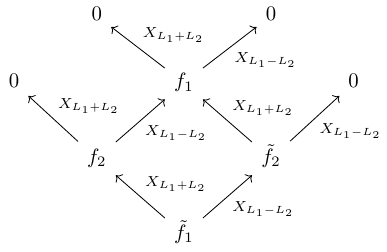}
\caption{The defining representation $V=\C^4$ of $\so(4)$. The Cartan subalgebra $\h$ from (\ref{eq:cartan subalgebra so(n)}) acts as diagonal matrices when we pick the basis $\{f_1,f_2,\wt{f}_1,\wt{f}_2\}$ of $V$. Given a root $\alpha$, an element of the root space $X_\alpha$ maps between the linear subspaces spanned by the written basis elements. For instance, $X_{L_1-L_2} f_2 \in \C f_1$. In this sense roots ``generalize'' the raising and lowering operators of $\su(2)$.}
\end{figure}

    Just as in the $n=4$ case, one can directly compute using explicit matrices $X_{\alpha}$ that $X_\alpha f_1 = 0$ for all $\alpha\in \Delta$, i.e. the highest weight of the defining representation is $L_1$.
    It is a nice exercise to see that this representation is irreducible.\footnote{Hint: each basis vector has a different weight.} By the theorem of highest weight, any irrep with highest weight $L_1$ must be isomorphic to the defining representation.

    We can now easily compute the highest weights of the representations $\Exterior^k V$ for $k<n/2$. The action of $\pi(X)$ is given by Leibniz rule, and after one verifies that $f_1\wedge f_2 \wedge \dots \wedge f_k$ is the highest weight vector, we have
    \begin{align*}
        \pi(H_j) f_1 \wedge \dots \wedge f_k &= \sum_{i=1}^k f_1 \wedge \dots \wedge H_j f_i \wedge \dots \wedge f_k \\
         &= \paran{\sum_{i=1}^k L_i(H_j)} (f_1  \wedge \dots \wedge f_k ),
    \end{align*} and so $f_1\wedge f_2 \wedge \dots \wedge f_k$ has weight $L_1+\dots + L_k$. Thus, the highest weight of the representation $\Exterior^k V$ is $L_1+\dots + L_k$. 

    In the case where $n$ even, $\Exterior^{n/2}V \cong U_+ \oplus U_-$. One of these irreps has weight $L_1+\dots + L_{n/2-1} + L_{n/2}$ with highest weight vector $f_1\wedge f_2 \wedge \dots \wedge f_{n/2-1} \wedge f_{n/2}$, while the other has weight $L_1+\dots + L_{n/2-1} - L_{n/2}$ with highest weight vector $f_1\wedge f_2 \wedge \dots \wedge f_{n/2-1} \wedge \wt{f}_{n/2}$.

    We now state a well-known proposition, which follows from the short discussion following Prop IX.10.4 in~\cite{simon1996representations}.

    \begin{proposition} \label{prop:ext k and ext n-k are isomorphic}
        As representations of $\so(n)$, $\Exterior^k V \cong \Exterior^{n-k} V$.
    \end{proposition}
    
    There are several approaches to prove this, none of which are particularly difficult. One approach is to use the $\so(n)$-equivariant wedge map $\Exterior^k V\otimes \Exterior^{n-k}V\to \Exterior^n V$ given by $(f,g)\mapsto f\wedge g$, and then check that $\Exterior^n V \cong \mathbf{1}$, the 1-dimensional trivial representation. This then identifies $\Exterior^{k}V \cong \overline{\Exterior^{n-k}}$, recalling that $\overline{(\cdot)}$ is the dual representation (or equivalently by unitarity, the conjugate representation). 
    Then, via e.g. Dynkin diagram automorphisms which implement the map sending a representation to its dual, one realizes that for the Lie types corresponding to $\so(n)$, $\Exterior^{k} V \cong \overline{\Exterior^k V}$ (see Theorem 3.E in~\cite{samelson1990notes} for details). This self-duality is special here: for instance, for $\su(3)$, the representation $\Exterior^1 V$ is not isomorphic to its dual $\Exterior^2 V$. Indeed, there are also representations of $\so(n)$ when $n=2\bmod{4}$ which are famously not self-dual: the ``half-spin'' representations $\Pi_+$ and $\Pi_-$ are not isomorphic but are dual to one another. We will encounter these later when we study Clifford algebras. 
    
    Another approach is to define the Hodge star map. Given an inner product $\inprod{\cdot,\cdot}$ on $V$ and a distinguished unit vector $\eta\in \Exterior^n V$,\footnote{A choice of orientation, if you wish.} this is the unique unitary\footnote{Note that in~\cite{simon1996representations}, Simon first defines the Hodge star as an antiunitary map. Our $\star$ consists of Simon's star composed with complex conjugation, which he discusses after the proof of the proposition.} map $\star: \Exterior^k V \to \Exterior^{n-k} V$ such that for all $u,v \in \Exterior^k V$,  
    \begin{equation} \label{def:original hodge star on exterior algebras}
        u \wedge \star v = \inprod{u,v} \eta .
    \end{equation} One can show that this map in fact intertwines the $SO(n)$ representations $\Exterior^k V\cong \Exterior^{n-k} V$. 

%---------------------------------------------

%----------------------------------------------
\section{The Rank-\texorpdfstring{$n$}{n} Clifford Algebra \texorpdfstring{$\calC_n$}{Cn}}\label{sec:Clifford algebras}
    We recall some standard results for Clifford algebras, see e.g.~\cite{simon1996representations,fulton1991representation}. The rank-$n$ Clifford algebra $\calC_n$ is generated by gamma operators $\gamma_1,\dots, \gamma_n$ subject to the relations
	\begin{equation} \label{def:gamma operators}
		\gamma_i \gamma_j + \gamma_j \gamma_i = 2\idty \delta_{ij} , \qquad \gamma_i^* = \gamma_i , \qquad i,j=1,\dots, n .
	\end{equation}
	We adopt the convention of defining the ``fifth gamma matrix'' $\gamma_0$ in the following way:
	\begin{equation} \label{def:gamma_0}
		\gamma_0 := \begin{cases}
			\gamma_1 \dots \gamma_n & n\equiv 0,1 \text{ mod}4 \\
			i\gamma_1 \dots \gamma_n & n\equiv 2,3 \text{ mod}4 \\
		\end{cases}
	\end{equation} The coefficient in front ensures that $\gamma_0^2 = \idty$ and $\gamma_0^*=\gamma_0$. Also note that if $n$ even, $\{\gamma_0,\gamma_i\} = 0$, while if $n$ odd, $[\gamma_0,\gamma_i]=0$.
	The element leads us towards a key definition: define a pair of mutually orthogonal commuting projectors $P_\pm\in\calC_n$ which sum to the identity $\idty\in \calC_n$ by
	\begin{equation} \label{def:P+ and P-}
		P_\pm := \frac{1}{2}(\idty \pm \gamma_0), \qquad \idty = P_+ + P_-, \qquad P_+P_- = P_-P_+ = 0 .
	\end{equation}
    When $n$ is even, we have that as associative algebras, $\calC_n\cong M_{2^{n/2}}(\C)$, the $2^{n/2}\times 2^{n/2}$ complex matrices. When $n$ is odd, since $[\gamma_0,\gamma_i]=0$, it can be checked on the basis (\ref{basis for gamma operators}) that $\gamma_0$ is a nontrivial element in the center, and so we have that as associative algebras, $\calC_n\cong P_+\calC_n \oplus P_-\calC_n$ where $P_+\calC_n\cong P_-\calC_n \cong M_{2^{(n-1)/2}}(\C)$.

    As a special case, when $n=3$, we have $\calC_3 \cong M_2(\C) \oplus M_2(\C)$. We can write $M_2(\C) = \text{span}_{\C}\{\idty, \sigma_x,\sigma_y,\sigma_z\}$, where the Pauli matrices enjoy the Clifford relations $\sigma_i\sigma_j +\sigma_j\sigma_i = 2\delta_{i,j}\idty$, i.e. $M_2(\C)$ is one of two irreps of $\calC_3$. Should one wish to work with Clifford algebras in a more explicit fashion, Chapter IV.3 of~\cite{simon1996representations} demonstrates how to realize their closely related Clifford groups in terms of Pauli strings.
	
	\subsection{Multi-index notation}
	We introduce multi-index notation to express elements of $\calC_n$ generated by products of $\gamma_i$: let 
	\begin{equation}
	\gamma_I =\gamma_{(i_1,\dots, i_k)} := \gamma_{i_1}\dots \gamma_{i_k} .
	\end{equation} where $1\leq i_2 < \dots < i_k \leq n$ and let $\abs{I} = $ \# of indices (e.g. $\abs{(2,3)} = 2$).
	Notice that any product of generators can be reorganized uniquely into $\pm\gamma_I$ by using the anticommutation relations (\ref{def:gamma operators}).
	We adopt the natural convention that if $I = ()$, then $\gamma_I = \idty$.
	The same notation can and will be applied to products $F_I$ described in the next section.
	
	Using this notation, one convenient basis for $\calC_n$ is given by 
	\begin{equation} \label{basis for gamma operators}
		\{\gamma_I : \abs{I}\leq n\} = \{\idty\} \cup \{\gamma_{i}\}_{1\leq i \leq n} \cup  \{\gamma_{i_1}\gamma_{i_2}\}_{1\leq i_1 < i_2 \leq n} \cup \dots \cup \{\gamma_{i_1}\dots \gamma_{i_{n-1}}\}_{1\leq i_1 < \dots < i_{n-1} \leq n} \cup \{\gamma_0\} .
	\end{equation} This basis reflects the vector space decomposition $\calC_n \cong \bigoplus_{k=1}^n \Exterior^k \C^n$. 

	Multi-index notation makes short work of traces of products of generators, as shown in the following lemma. 
	
	\begin{lemma} \label{lem:traces of products of gamma operators}
		Let $I=(i_1,\dots,i_a)$ and $J=(j_1,\dots,j_b)$. Then
		\[
		\gamma_I \gamma_J = \sgn(\pi) \gamma_K,
		\] where $K$ is the ordered symmetric difference of $I$ and $J$, so as sets, $K = (I\cup J) - (I\cap J)$; and $\pi$ is the permutation that orders the concatenation $(I\cup J)$, after omitting common indices in $(I\cap J)$, to the index $K$: i.e. $\pi$ maps $(i_1,\dots, \hat{i_p},\dots, i_a, j_1, \dots, \hat{j_q}, \dots, j_b) \mapsto K$, using $\hat{\cdot}$ to denote absent elements belonging to $I\cap J$.
	\end{lemma}
     As an example, let $I = (1,2)$ and $J=(1,3)$. Then $\gamma_I\gamma_J = \gamma_1\gamma_2\gamma_1\gamma_3 = -\gamma_1\gamma_1\gamma_2\gamma_3 = -\gamma_2\gamma_3$, so $K=(2,3)$ and $\sgn(\pi) = -1$.
	\begin{proof}
		If no elements are common to $I$ and $J$, then all terms in this product anticommute and as a set, $K=I\cup J$.
		If we have two elements that are common, say $x\in I, x\in J$, then we can anticommute until we have the term $\gamma_{x}\gamma_{x} = \idty$ by (\ref{def:gamma operators}).
	\end{proof}
	\begin{corollary}
		If we let $D := \Tr(\idty)$, then
		\[
		\Tr(\gamma_I \gamma_J) = \begin{cases}
			0 & \text{if } I\neq J \\
			D & \text{if } I = J\text{ and } \abs{I}=0,1 \mod4 \\
			-D & \text{if } I = J\text{ and } \abs{I}=2,3 \mod4  \\
		\end{cases}
		\] 
	\end{corollary}
	\begin{proof}
		First, observe that if $I\neq J$, then $K\neq ()$. Any nontrivial product $\gamma_K$ must have trace zero, using the relations (\ref{def:gamma operators}) and induction on $\abs{K}$: if $\abs{K}=1$, i.e. $\gamma_K = \gamma_k$, then picking $q\neq k$, we have
		\begin{align*}
			2 \, \Tr\gamma_k &= \Tr (\gamma_q\gamma_q + \gamma_q\gamma_q)\gamma_k \\
			 &=  \Tr \gamma_q\gamma_q\gamma_k + \gamma_q\gamma_k\gamma_q \\
			&=  \Tr \gamma_q(\gamma_q\gamma_k + \gamma_k\gamma_q) \\
			&= 0.
		\end{align*} where the second equality is cyclicity of trace and the third is by Clifford relations.
		
		If $\abs{K}>1$, then $\gamma_K = \gamma_{\wt{K}} \cdot \gamma_k$ where $\abs{\wt{K}} = \abs{K}-1$ and so taking a trace of anticommutation relations forces
		\[
		\Tr(\gamma_{\wt{K}} \cdot \gamma_k) + \Tr(\gamma_k \cdot \gamma_{\wt{K}}) = 0,
		\] whence cyclicity of trace gives $\Tr(\gamma_{\wt{K}}\cdot \gamma_k) = 0$. 
		
		Now, if $I=J$, then $K = ()$ and $\gamma_{K}=\idty$, so we need to determine the sign. In the case $\abs{I} = 1$, $\gamma_i \gamma_i = \idty$. In the case $2<\abs{I}\leq n$, writing $\gamma_I = \gamma_{\wt{I}}\cdot \gamma_i$
		\[
		\gamma_I \cdot \gamma_I = \gamma_{\wt{I}}\cdot \gamma_i\cdot \gamma_{\wt{I}}\cdot \gamma_i = (-1)^{\abs{\wt{I}}} \gamma_{\wt{I}} \cdot \gamma_{\wt{I}} \cdot \gamma_i\cdot \gamma_i = (-1)^{\abs{\wt{I}}} \gamma_{\wt{I}}\cdot \gamma_{\wt{I}} = (-1)^{\abs{I} - 1} \gamma_{\wt{I}}\cdot \gamma_{\wt{I}}.
		\] Induction on $\abs{I}$ gives the desired result.
		
	\end{proof}
 \begin{remark}
     Note that this labeling convention might be subtle when one passes to a representation of a Clifford algebra. Take for example the irrep $M_2(\C) = \text{span}_{\C} \{\idty, \sigma_x, \sigma_y,\sigma_z\}$ of $\calC_3$. Notice that in this representation, $\Tr \sigma_x\sigma_y \sigma_z = 2i$. But this does not contradict our statement: in this representation, $\sigma_x\sigma_y = i\sigma_z$. Indeed, this may be viewed as an artifact of the fact that this irrep is not faithful, since e.g. $\gamma_0 = i\sigma_x\sigma_y\sigma_z$ acts as $i\idty$. This is no longer an issue if we consider faithful irreps, or if we adjust this lemma to treat equivalence classes of strings up to the kernel of the representation. For the even $\calC_n$ we consider here, we are always taking a faithful representation. For the odd $\calC_n$, this requires us to work with the equivalence $\gamma_I \sim \gamma_0\gamma_I$. Then, one resolves the problem by restricting to either odd or even strings $\abs{I}$.
 \end{remark}

\subsection{The even subalgebra $\calC_n^{[ev]}$ and the odd subspace $\calC_n^{[odd]}$} \label{secapp:the even subalgebra}
    Let $n$ be even. Then we can use the commuting mutually orthogonal projectors $P_\pm$ (\ref{def:P+ and P-}) to break $\calC_n\cong M_{2^{n/2}}(\C)$ into four blocks:
    \begin{equation}\label{eqn:Cn splits into 4 blocks}
        \calC_n = P_+\calC_n P_+ + P_+\calC_n P_- + P_-\calC_n P_+ + P_-\calC_n P_-.
    \end{equation} Using that $\gamma_i\gamma_0=-\gamma_0\gamma_i$ for $i=1,\dots,n$ while $[\gamma_0,\gamma_I]=0$ for any even $\abs{I}$, we see that
    \begin{equation}\label{eq:even and odd subspace diagonalizations}
        \calC_n^{[ev]} = P_+\calC_n P_+ + P_-\calC_n P_-, \quad \calC_n^{[odd]} = P_+\calC_n P_- + P_-\calC_n P_+, 
    \end{equation} where $\calC_n^{[ev]}$ has a basis of even products $\gamma_I$ and odd $\calC_n^{[odd]}$ has a basis of odd products. Of course, only $\calC_n^{[ev]}$ is a subalgebra, since the product of two odd $\gamma_I,\gamma_J$ is even.
    
	The even subalgebra is (non-canonically) isomorphic to $\calC_{n-1}$, for instance by defining generators $\{i F_j\st j=2,\dots, n\}$ where the operators $F_j := \gamma_1\gamma_j$ enjoy the relations
	\begin{equation} \label{def:even Gamma operators}
		F_i F_j + F_j F_i = -2\delta_{i,j}\idty, \qquad F_i^* = - F_i , \qquad \text{for all } i,j=2,\dots, n.
	\end{equation} 
    As a consequence, we know that when $n-1$ even (and so $n$ odd), $\calC_n^{[ev]}\cong M_{2^{(n-1)/2}}(\C)$ is a square matrix algebra. When $n-1$ odd (and so $n$ even), the even subalgebra splits into two square matrix algebras $\calC_n^{[ev]}\cong M_{2^{n/2-1}}(\C) \oplus M_{2^{n/2-1}}(\C)$. Let us see this split in another way. Since $\gamma_0$ commutes with even products of gamma operators, it is a nontrivial element in the center $Z(\calC_n^{[ev]})$. So the block decomposition (\ref{eq:even and odd subspace diagonalizations}) is an algebra decomposition:
    \begin{equation} \label{even subalgebra splits into P+ and P-}
		\calC_n^{[ev]} = P_+ \calC_n^{[ev]} \oplus P_-\calC_n^{[ev]} , \qquad n \text{ even },
	\end{equation} where each of the summands is isomorphic to $M_{2^{n/2-1}}(\C)$, the $2^{n/2 - 1} \times 2^{n/2 - 1}$ matrices. As a concrete example, we write a basis for $\calC_4^{[ev]}$ which respects this decomposition in (\ref{concrete basis for Clifford algebra C4}).
		
    \subsubsection{Hodge Duals}\label{secapp:hodge duals}
	We define the ``Hodge star'' map on a basis of products $\{\gamma_I\}\subseteq \calC_n$ by writing
			\begin{equation} \begin{split}\label{def:Hodge dual}
				(\star\gamma_I)\gamma_I &= (-1)^{\abs{I}/2}\gamma_0, \qquad I \text{ even} \\ 
                (\star\gamma_I)\gamma_I &= (-1)^{(\, \abs{I}-1)/2}\gamma_0, \qquad I \text{ odd}. 
			\end{split} \end{equation} Up to a sign, this agrees with the earlier definition on exterior algebras (\ref{def:original hodge star on exterior algebras}), noting that $\gamma_0\in \Exterior^n \C^n \subseteq \calC_n$. Notice also that using anticommutation relations, we have that $\gamma_I^2 = (-1)^{\abs{I}/2}\idty$ for even $\abs{I}$ (or $\gamma_I^2 = (-1)^{(\, \abs{I}-1)/2}\idty$ for odd $\abs{I}$). In any case our sign convention is chosen so that in both the even and odd cases, one can multiply Equation (\ref{def:Hodge dual}) on the right by $\gamma_I$ to obtain
            \begin{equation} \label{eqn:hodge dual is gamma0 multiplication}
                \star \gamma_I = \gamma_0 \gamma_I.
            \end{equation}
            Recalling that $P_\pm = \frac{1}{2}(\idty \pm \gamma_0)$ from (\ref{def:P+ and P-}) are commuting orthogonal projectors, this will allow us to readily simultaneously diagonalize left multiplication $L_{P_\pm}:\calC_n\to \calC_n$, i.e. $L_{P_\pm}\gamma = P_\pm \gamma$. Elements in the image of $L_{P_\pm}$ are eigenvectors of eigenvalue 1 for the left multiplication map $L_{\gamma_0}:\calC_n\to \calC_n$, and elements in the image of $L_{P_-}$ are eigenvectors of eigenvalue $-1$ for $L_{\gamma_0}$. Observe that since $\gamma_0^2=\idty$,
			\begin{align*}
			 L_{\gamma_0} (\gamma_I + \star \gamma_I) &= \gamma_0 \gamma_I + \gamma_0 \star \gamma_I = \star \gamma_I + \gamma_I = \gamma_I + \star \gamma_I \\
            L_{\gamma_0} (\gamma_I - \star \gamma_I) &= \gamma_0 \gamma_I - \gamma_0 \star \gamma_I = \star \gamma_I - \gamma_I = -( \gamma_I - \star \gamma_I) .
			\end{align*} Thus, by starting with the basis $\{\gamma_I: \; \abs{I} \leq n \}$, the above leads to the following diagonalization of $L_{P_{\pm}}$, which respects the decomposition $\calC_n = P_+\calC_n \oplus P_-\calC_n$ from (\ref{eqn:Cn splits into 4 blocks}):
			\begin{equation} \label{eqn:basis gamma_I + star gamma_I} 
			\text{span}\{\gamma_I + \star \gamma_I : \; \abs{I}\leq n/2 \} \cup \{\gamma_I - \star \gamma_I : \; \abs{I}\leq n/2 \} = P_+\calC_n \oplus P_-\calC_n = \calC_n .
			\end{equation}
            When $n$ odd, this is in fact a decomposition into subalgebras.
            When $n$ even, it is only a decomposition of vector spaces like (\ref{eqn:Cn splits into 4 blocks}). But since $\star$ sends even products to even products and odd products to odd products, the trick similarly works after restricting to $\calC_n^{[ev]}$ and $\calC_n^{[odd]}$, obtaining bases respecting their decompositions (\ref{eq:even and odd subspace diagonalizations}). This diagonalization then becomes quite crucial to understanding the spin representation $\Pi$ of $Spin(n)$ from the following Section \ref{sec:the spin representations Pi of SO(n)}. The representatives of $\Pi$ are given by left multiplication by invertible elements $\gamma_I\in \calC_n^{[ev]}$ and so all commute with $\gamma_0$, and thus these $\gamma_I\pm \star \gamma_I$ are bases for the invariant subspaces of $\Pi_+\oplus \Pi_- \cong \Pi$. 
            \begin{example}
                When $n =4$, the above diagonalization of left multiplication by $\gamma_0$ yields the following bases:
                \begin{equation}\begin{split}\label{concrete basis for Clifford algebra C4}
				 &\{\idty + \gamma_0, \; \gamma_1\gamma_2 - \gamma_3\gamma_4, \; \gamma_1\gamma_3 + \gamma_2\gamma_4, \; \gamma_1\gamma_4 - \gamma_2\gamma_3\} \subseteq \im(P_+ \cdot )|_{\calC_n^{[ev]}} \\
				 &\{\idty - \gamma_0, \; \gamma_1\gamma_2 + \gamma_3\gamma_4, \; \gamma_1\gamma_3 - \gamma_2\gamma_4, \; \gamma_1\gamma_4 + \gamma_2\gamma_3\} \subseteq \im(P_- \cdot )|_{\calC_n^{[ev]}}.
			\end{split} \end{equation}
            \end{example}

%---------------------------------------------

%--------------------------------------------
\section{The Spin Representations \texorpdfstring{$\Pi$}{Pi} of \texorpdfstring{$SO(n)$}{SO(n)}}\label{sec:the spin representations Pi of SO(n)}
    In the following discussion, we cite the construction of spin representations in \cite{simon1996representations} to define a (projective) representation $\Pi:SO(n) \to \exp(\calC_n)$.\footnote{By placing any norm on the finite dimensional matrix algebra $\calC_n$, one sees that the exponential power series $e^A$ is well defined and invertible for any $A\in \calC_n$.} We have not yet defined the notion of projective representation, but the construction begins with a true representation, so we will only define it once it is needed. We will spend some time in this section developing theory for both $\Pi$ and for the ``Adjoint'' representation $\Pi(\cdot)\Pi^{-1}$.
    
	First, choose generators $L_{ij}$ for the Lie algebra $\so(n)$ by picking an orthonormal basis $\{\ket{i}\}$ of $\C^n$ and defining the matrix $L_{ij} = \ket{i}\bra{j} - \ket{j}\bra{i}$ for all $1\leq i<j\leq n$.
	Then, thinking of $\calC_n$ as a Lie algebra with the commutator as its Lie bracket $[\cdot, \cdot]$, observe that the map $\pi: \so(n)\to \calC_n$ which sends $L_{ij}\mapsto \frac{1}{2} \gamma_i\gamma_j$ defines a skew-Hermitian Lie algebra representation since it preserves the commutation relations of $\so(n)$:
		\begin{equation}
			\frac{1}{4}[\gamma_i\gamma_j,\gamma_r\gamma_s] = \delta_{jr}\gamma_i\gamma_s - \delta_{ir}\gamma_j\gamma_s + \delta_{is}\gamma_j\gamma_r - \delta_{js}\gamma_i\gamma_r .
		\end{equation}
		Using Theorem (\ref{thm:alg reps lift to group reps when G simply connected}), this Lie algebra representation $\pi$ then exponentiates to a unitary group representation $\wt{\Pi}$ of $Spin(n)$, the simply connected double cover $Spin(n)/\Z_2 \cong SO(n)$.
		The image of $\pi$ is contained in the even subalgebra $\calC_n^{[ev]}$, so after exponentiating, we can view $\wt{\Pi}$ as left multipiclation in $\calC_n$ via invertible elements of $\calC_n^{[ev]}$. For instance, the operator $\wt{\Pi}(g) = e^{\theta \gamma_1\gamma_2} \in \im(\Pi)$ for all $\theta\in \R$, and for any $\Gamma\in\calC_n$, $\wt{\Pi}(g)\cdot \Gamma = e^{\theta \gamma_1\gamma_2} \Gamma$.
        We will exponentiate the adjoint representation $[\pi,\cdot]$ to $\wt{\Pi}(\cdot)\wt{\Pi}^{-1}$ acting on the vector space $\calC_n$, where for example $\wt{\Pi}(g) \Gamma\, \wt{\Pi}(g)^{-1} = e^{\theta \gamma_1\gamma_2}\Gamma e^{-\theta \gamma_1\gamma_2}$.

        In both even and odd $n$ cases, $\wt{\Pi}$ is a true representation of $Spin(n)$, but it induces a (nontrivial) \textit{projective} representation $\Pi$ of $SO(n)$, meaning that if $w,v\in SO(n)$, then $\Pi(w)\Pi(v) = \alpha(w,v)\Pi(wv)$ where $\alpha(w,v)\in U(1)$ is a phase. As a special case, consider the spin-1/2 representation $\pi$ of $\so(3)$, which is given by the map $\pi:\so(3)\to P_+\calC_3\cong \gl(\C^2)$ sending $L_{ij}\mapsto \frac{1}{2}\sigma_i\sigma_j$ where $1\leq i<j\leq 3$ and $\{\sigma_i\}$ are the Pauli matrices. After exponentiating, we see that
        \begin{equation}\begin{split}
            \Pi(e^{\theta L_{12}})\Pi(e^{\theta L_{13}})\Big|_{\theta=\pi} = e^{\frac{\theta}{2}\sigma_1\sigma_2} e^{\frac{\theta}{2}\sigma_1\sigma_3}\Big|_{\theta=\pi} &= (\cos \frac{\theta}{2}\idty + \sin \frac{\theta}{2} \sigma_1\sigma_2 ) (\cos \frac{\theta}{2} \idty + \sin \frac{\theta}{2} \sigma_1\sigma_3)\Big|_{\theta=\pi} \\
             &= \sigma_1\sigma_2\sigma_1\sigma_3 \\
            &= -\sigma_2\sigma_3
        \end{split}\end{equation} and so for $\theta= \pi\in \R$, we have $\Pi(e^{\theta L_{12}})\Pi(e^{\theta L_{13}})\Pi(e^{\theta L_{12}})^{-1}\Pi(e^{\theta L_{13}})^{-1} = \sigma_2\sigma_3 \sigma_2\sigma_3 = -\idty$. This means that the representation $\Pi:SO(3)\to \gl(\C^2)$ is projective, since $e^{\pi L_{12}}e^{\pi L_{13}} = e^{\pi L_{13}}e^{\pi L_{12}}$ and so $e^{\pi L_{12}}e^{\pi L_{13}} e^{-\pi L_{12}}e^{-\pi L_{13}} = \idty$. It is clear that this computation can be easily repeated by replacing $\sigma_i$ with $\gamma_i$ to see that the spin representation $\Pi$ of $SO(n)$ is indeed projective for all $n\geq 3$. 
        
        One may then ask whether this projective representation may be de-projectivized--i.e. is there a sense in which this projective representation is  ``equivalent'' to a true representation? This question is a bit more subtle and it will be a guiding theme of Chapter \ref{ch:proj-reps}. We will define a natural notion of equivalence, and we will eventually find that this representation is nontrivial projective and may not be de-projectivized. But for now, let us summarize this calculation in a remark.
        \begin{remark} \label{rem:spin reps are projective}
            Let $n$ even and consider the spin representation $\wt{\Pi}:Spin(n)\to \calU(\C^{2^{n/2}})$, where we have identified $\calC_n\cong M_{2^{n/2}}(\C)$ so that $\calU(\C^{2^{n/2}})\subseteq \exp(\calC_n)$. The spin representation maps $\wt{\Pi}(e^{\theta L_{ij}}) = e^{\pi(\theta L_{ij})} =  e^{\frac{\theta}{2} \gamma_i\gamma_j}$ and satisfies $\wt{\Pi}(\idty) = \idty$ and $\wt{\Pi}(-\idty) = -\idty$. It thus descends to a projective representation $\Pi:SO(n)\to P\calU(\C^{2^{n/2}})$. Further, this representation is still projective when we restrict to the subgroup $\Z_2\times \Z_2  = \{\idty, e^{\pi L_{12}}, e^{\pi L_{23}}, e^{\pi L_{12}}e^{\pi L_{23}}\}\subseteq SO(n)$. 

            The same situation occurs when $n$ odd, after recalling that $\calC_n\cong M_{2^{(n-1)/2}}(\C)$.

            As a special case, $Spin(3)\cong SU(2)$, and $\wt{\Pi}:SU(2)\to \calU(\C^2)$ is exactly the defining spin-1/2 representation $\wt{\Pi}(g) = g$.
        \end{remark}

        We have singled out a certain finite subgroup $\Z_2\times \Z_2\subseteq SO(n)$ generated by $\pi$-rotations about two orthogonal axes $e^{\pi L_{12}}, e^{\pi L_{13}}$. This is of significant historical significance: it is this famous projective representation which endows the AKLT chain with a nontrivial $H^2(\Z_2\times \Z_2,U(1))$ SPT index. We will later show the extraction of this index in Chapter \ref{ch:gapped_ground_states_phases}, and indeed the story for $SO(n)$ AKLT chains in Chapter \ref{ch:SO(n)_Haldane_chains} will parallel this calculation.

    \subsection{Weights and reducibility of the spin representations}
        The image of the Cartan subalgebra of $\so(n)$ given by (\ref{eq:cartan subalgebra so(n)}) is given by
    \begin{equation} \label{eq:cartan subalg cliffords}
        \pi(\h) = \text{span}\{\gamma_1\gamma_2,\,  \gamma_3\gamma_4,\dots , \gamma_{n-1}\gamma_n\}.
    \end{equation} When $n$ is odd, $\Pi$ is an irreducible representation of $Spin(n)$ on $\C^{2^{(n-1)/2}}$ (recall that $P_+\calC_n\cong M_{2^{(n-1)/2}}(\C)$). Its highest weight may be computed directly to be $\frac{1}{2}(L_1 + \dots + L_{(n-1)/2})$ (see e.g. Proposition 20.20 in~\cite{fulton1991representation}). When $n$ is even, $\Pi$ is reducible and splits into so-called ``half-spin'' representations $\Pi\cong \Pi_+\oplus \Pi_-$, since $\gamma_0\in Z(\calC_n^{[ev]})$ and so the representation space splits  $\C^{2^{n/2}} \cong \im(P_+) \oplus \im(P_-)$ (recall that $\calC_n \cong M_{2^{n/2}}(\C)$). The respective highest weights of $\Pi_+$ and $\Pi_-$ are $\frac{1}{2}(L_1 + \dots + L_{n/2-1} + L_{n/2})$ and $\frac{1}{2}(L_1 + \dots + L_{n/2-1} - L_{n/2})$.

    This fundamental fact will later have physical consequences for $SO(n)$ AKLT chains: in the odd $n$ case, we obtain a unique ground state $\omega$, while in the even $n$ case, we obtain a pair of dimerized ground states $\omega_\pm$.

    \subsection{The ``Adjoint'' of the Spin Representation $\Pi (\cdot) \Pi^{-1}$} \label{sec:The Adjoint of the Spin Representation Pi (-) Pi}

    Now, while the spin representation $\Pi$ is a projective representation of $SO(n)$, the ``Adjoint'' representation $\Pi (\cdot) \Pi^{-1}$ acting on $\calC_n$ is a true representation of $SO(n)$ since the phases cancel.\footnote{Note that we may see this as the Adjoint representation by restricting the action of $\Pi(\cdot)\Pi^{-1}$ to the faithfully embedded Lie algebra $\Pi(\so(n))\subseteq \calC_n$} 
    It can be shown~\cite{simon1996representations} that this representation has the following action on generators $\gamma_i \in \calC_n$:
		\begin{equation} \label{eq:representation of SO(n) on clifford algebra}
			\Pi(w) \gamma_i \Pi(w)^{-1} = \sum_{j} w_{ji} \gamma_j, \qquad w\in SO(n).
		\end{equation} It is quickly clear how $\Pi(\cdot)\Pi^{-1}$ acts on the entire $\calC_n$, since on basis elements $\gamma_I \in \calC_n$,
        \[
            \Pi(w)(\gamma_{i_1}\gamma_{i_2} \dots \gamma_{i_k} ) \Pi(w)^{-1} = \Pi(w) \gamma_{i_1} \Pi(w)^{-1} \Pi(w) \gamma_{i_2} \Pi(w)^{-1} \dots \Pi(w) \gamma_{i_k} \Pi(w)^{-1}.
        \]
        This representation in fact decomposes into some familiar representations. We will prove this in a very hands-on way which produces nice bases, and do a more representation-theoretic proof in Section \ref{sec:another take on lemma pi reps}.

    \begin{lemma} \label{lem:rep of Pi( )Pi^-1 on bond algebras}
	Let $V=\C^n$ denote the defining representation of $SO(n)$. 
	Then the representation $\Pi(\cdot)\Pi^{-1}$ of $SO(n)$ on the algebra $\calB = P_+\calC_n$ and $\calB = \calC_n$ have the following decomposition into irreducible representations:
	\begin{itemize}
            \item When $n$ odd, $\calB = P_+\calC_n$ and 
            \begin{align*}
                \calB &\cong \Exterior^0 V \oplus \Exterior^2 V \oplus \dots \oplus \Exterior^{n-1} V \\
                &\cong \Exterior^1 V \oplus \Exterior^3 V \oplus \dots \oplus \Exterior^{n} V,
            \end{align*} where we can pick explicit bases for these irreps by
            \[
                \Exterior^k V = \text{span} \{\gamma_I + \star \gamma_I: \abs{I} = k\}.
            \]
            \item When $n$ even, $\calB = \calC_n$ and 
            \[
                \calB \cong \Exterior^0 V \oplus \Exterior^1 V \oplus \dots \Exterior^{n/2 - 1} V \oplus U_+ \oplus U_- \oplus \Exterior^{n/2 + 1} V \oplus \dots \oplus \Exterior^{n-1} V\oplus \Exterior^n V ,
            \] where $U_+\oplus U_- \cong \Exterior^{n/2} V$, with $\dim(U_+) = \dim(U_-)$ and $U_+\not\cong U_-$. We can pick bases for these irreps by 
            \[
            \Exterior^k V = \begin{cases} \text{span}\{\gamma_I + \star \gamma_I: \abs{I} = k\} & k<n/2 \\
            \text{span}\{\gamma_I - \star \gamma_I: \abs{I} = k\} & k>n/2
                \end{cases}
            \] and 
        \[
            U_+ = \text{span}\{\gamma_I + \star \gamma_I : \abs{I} = n/2\}, \quad U_- = \text{span}\{\gamma_I - \star \gamma_I : \abs{I} = n/2\}.
        \]
		\end{itemize}
        
        \end{lemma}
		\begin{proof}
			We always have the canonical linear isomorphism generated by $\gamma_i\mapsto \ket{i}\in \C^n$:
			\[
				\calC_n \cong \Exterior^0 V \oplus \Exterior^1 V \oplus \dots \oplus \Exterior^{n} V,
			\] where a convenient basis for the subspace of $\calC_n$ corresponding to $\Exterior^k V$ is given by $\{\gamma_K: \; \abs{K} = k\}$, recalling that the empty string $\abs{K} = 0$ corresponds to $\idty$.
            To see that these subspaces are invariant under $\Pi(\cdot)\Pi^{-1}$, we pass to the Lie algebra representation given by $[\pi(X),\cdot)]$ for all $X\in \so(n)$. It suffices to work on generators $\pi(X) = \gamma_a\gamma_b, a<b$. Fix a basis element $\gamma_K$ with $\abs{K} = K$ and observe that
            \begin{align*}
                [\pi(X), \gamma_K] &= [\gamma_a\gamma_b, \gamma_K] \\
                &= [\gamma_a,\gamma_K]\gamma_b + \gamma_a[\gamma_b,\gamma_K].
            \end{align*} There are three cases to consider. If $a\not\in K$ and $b\not\in K$, then since $\abs{K}$ is even, $\gamma_a \gamma_K = \gamma_K \gamma_a$ and so $[\gamma_a,\gamma_K]=0$, and likewise for $\gamma_b$. If $a\in K$ and $b\not\in K$ then
            \[
                (\gamma_a\gamma_I - \gamma_I\gamma_a)\gamma_b = \pm \gamma_{K-a} \gamma_b,
            \] which has length $(K-1)+1 = K$. 
            Finally, if both $a\in K$ and $b\in K$, then since $K = (i_1, \dots, a, \dots b \dots, i_k)$ is ordered, we can bring $a,b$ to the front to write $K = \pm(a,b, i_1, \dots, i_k)$ and then 
            \[
            [\gamma_a \gamma_b, \gamma_K] = \pm[\gamma_a \gamma_b, \gamma_a\gamma_b\gamma_{i_1}\dots \gamma_{i_k}] = \pm(\gamma_a\gamma_b \gamma_a\gamma_b \gamma_{i_1}\dots \gamma_{i_k} - \gamma_a\gamma_b \gamma_{i_1}\dots \gamma_{i_k}\gamma_a\gamma_b) = 0 ,
            \] since $a,b\not\in (K-ab)$. In all cases, $\Exterior^k V$ is indeed an invariant subspace.
    
			Now, as per the discussion in Section \ref{sec:exterior power reps}, each $\Exterior^k V$ is irreducible for odd $n$, and each $\Exterior^k V$ except $k=n/2$ is irreducible for $n$ even. In the $k=n/2$ case, $\Exterior^{n/2} V$ splits into two non-isomorphic irreps $U_\pm$ of equal dimension. 
   
            It remains to show the basis claim. Recall the basis $\{\gamma_I \pm \star \gamma_I : \abs{I}\leq n/2\}$ which diagonalizes left multiplication $L_{P_\pm} : \calC_n\to \calC_n$ \ref{eqn:basis gamma_I + star gamma_I}. From the above argument, it is clear that the action of $[\pi,\cdot]$ preserves the grading of any fixed $\gamma_I$. It remains to show that it commutes with $P_\pm$ and so maps $\gamma_I+\star\gamma_I$ to $\gamma_{\wt{I}} + \star \gamma_{\wt{I}}$ and $\gamma_I-\star\gamma_I$ to $\gamma_{\wt{I}} - \star \gamma_{\wt{I}}$. Since any representative $\pi(X)\in \calC_n^{[ev]}$, we have $\pi(X)\gamma_0 = \gamma_0 \pi(X)$, and so using $\star \gamma_I = \gamma_0\gamma_I$ we see 
            \begin{align*}
                [\pi(X),\gamma_I + \star \gamma_I] &= [\pi(X),\gamma_I] + [\pi(X),\gamma_0\gamma_I] \\
                &= [\pi(X),\gamma_I] + \gamma_0[\pi(X),\gamma_I] \\
                &= [\pi(X),\gamma_I] + \star [\pi(X),\gamma_I] ,
            \end{align*} and by the same computation, $[\pi(X),\gamma_I - \star \gamma_I] = [\pi(X),\gamma_I] - \star [\pi(X),\gamma_I]$. In the even $n$ case, this tells us that $U_\pm$ have the promised bases. In the odd $n$ case, this means that we can restrict $\Pi(\cdot)\Pi^{-1}$ to $P_+\calC_n$, leaving only $\gamma_I + \star \gamma_I$. The freedom in the description of these representations comes from the fact that for both even and odd $n$, we have $\Exterior^k V \cong \Exterior^{n-k}V$ from Proposition \ref{prop:ext k and ext n-k are isomorphic}.
			\end{proof}

	\begin{remark} \label{rem:transitivity of adjoint action [pi(X), ]} It isn't too hard to see irreducibility for the invariant subspaces $\Exterior^k V$ and $U_\pm$ from raw computation, although it is guaranteed by e.g. the Weyl Character Formula. 

    Here we actually have more than irreducibility: the representation acts transitively on each family $\{\gamma_I:\abs{I} = k\}$ (and since $\pi(X)$ commutes with $\star$ for all $X\in\so(n)$, this similarly holds for $\{\gamma_I+\star\gamma_I: \abs{I}=k\}$). I.e. given any pair $\gamma_I, \gamma_J$ with $\abs{I}=\abs{J}$, we can find an $X\in\so(n)$ such that $[\pi(X),\gamma_I]=\gamma_J$, and conversely, if there is an $X\in\so(n)$ satisfying this equality, then $\abs{I}= \abs{J}$. 
    The proof is straightforward, and it is perhaps clearer to demonstrate this by example. Let $\gamma_I = \gamma_1\gamma_2$ and let $\gamma_J= \gamma_2\gamma_3$. Then picking $\pi(X) = \gamma_1\gamma_3$ suffices, since $[\gamma_1\gamma_3,\gamma_1\gamma_2] = 2\gamma_2\gamma_3 = 2\gamma_J$. So, we may algorithmically convert $I$ into $J$ by using successive applications of these $[\pi(X_k),\cdot]$, and since this is a Lie algebra representation, these nested commutators may be replaced by a single $[\pi(X),\cdot]$. We will use this observation frequently in later proofs.
	\end{remark}
        \begin{remark} We need concrete bases to prove future theorems. However, this decomposition is in fact a standard result, e.g. Ch. 19 of \cite{fulton1991representation}, which writes that the spin representation $\Pi$ acting on $\calC_n$ decomposes into\footnote{We are abusing notation here and writing $\Pi$ as the representation space.}
        \[
             \Pi\otimes \Pi \cong \Exterior^0 V \oplus \Exterior^1 V \oplus \dots \oplus \Exterior^{n-1} V \oplus \Exterior^n V .
        \]
        But $\Pi(\cdot)\Pi^{-1}$ is in fact the representation $\Pi\otimes \overline{\Pi}$, where $\overline{\Pi}$ denotes the dual representation of $\Pi$, or equivalently the conjugate representation of $\Pi$ since $\Pi$ unitary. 
        When $n$ is odd, the only spin representation is self-dual, $\Pi\cong \overline{\Pi}$.
        When $n\bmod{4} = 0$, both half-spin representations are self-dual, so $\Pi_+ \cong \overline{\Pi}_+$ and $\Pi_- \cong \overline{\Pi}_-$; when $n\bmod{4} = 2$, they are duals of each other $\Pi_+ \cong \overline{\Pi}_-$. In any case,
        \[
            \overline{\Pi} = \overline{\Pi_+ \oplus \Pi_-} = \overline{\Pi}_+ \oplus \overline{\Pi}_- \cong \Pi,
        \] whence we obtain the desired decomposition, after carefully noting that $\Exterior^{n/2}V$ splits into two irreps when $n$ even, and that $\calC_n\cong P_+\calC_n \oplus P_-\calC_n$ when $n$ odd where as $SO(n)$ representations, $P_+\calC_n\cong P_-\calC_n$.
        \end{remark}
        As a sanity check, we can check dimensions, recalling that $\dim(\Exterior^k V) = \binom{n}{k}$. When $n$ odd, 
        \[
            \dim(P_+\calC_n) = \dim(M_{2^{(n-1)/2}}(\C)) = 2^{(n-1)/2} \cdot 2^{(n-1)/2} = \frac{1}{2} 2^{n} = \frac{1}{2} \sum_{k \text{ even}}^n \binom{n}{k} = \frac{1}{2} \sum_{k \text{ odd}}^n \binom{n}{k} , 
        \] and when $n$ even,
        \[
            \dim(\calC_n) = \dim(M_{2^{n/2}}(\C)) = 2^{n/2} \cdot 2^{n/2} = 2^{n} = \sum_{k}^n \binom{n}{k} .
        \]

\subsection{Another take on Lemma \ref{lem:rep of Pi( )Pi^-1 on bond algebras}} \label{sec:another take on lemma pi reps}
    
    We may get the result of Lemma \ref{lem:rep of Pi( )Pi^-1 on bond algebras} in a different way. The following approach is a bit more natural and elegant, but the bases we get out of Lemma \ref{lem:rep of Pi( )Pi^-1 on bond algebras} are much easier to work with in later chapters.
    
    Recall that the image of the Cartan subalgebra $\pi(\h)$ is given by Equation (\ref{eq:cartan subalg cliffords}). We wish to diagonalize the representatives $[\pi(H),\cdot]$ of the Cartan subalgebra $\h$ in a way that mimics (\ref{eqn:sim diag cartan so(n)}), i.e.
    \[
        [\gamma_1\gamma_2, \gamma_1 + i \gamma_2] =  i (\gamma_1 + i\gamma_2) , \qquad [\gamma_1\gamma_2, \gamma_1 - i \gamma_2] =  -i (\gamma_1 - i\gamma_2)
    \] So, to define our change of basis $S:\calC_n\to \calC_n$ from the basis $\{\gamma_I: \abs{I}\leq n\}$ to a new basis $\{f_K: \abs{K} \leq n\}$, we define it on individual operators and extend to products by $S(\gamma_{i_1} \gamma_{i_2}\dots \gamma_{i_\ell}) = S(\gamma_{i_1}) S(\gamma_{i_2}) \dots S(\gamma_{i_\ell})$. We first write $\{f_1,f_2,\dots, f_{n/2}, \wt{f_1},\wt{f_2}, \dots ,\wt{f}_{n/2}\}$ where
    \[
        f_k := \gamma_{2k-1} + i\gamma_{2k}, \quad \wt{f_k} := \gamma_{2k-1} - i\gamma_{2k},
    \] whence we obtain
    \[
        [\pi(H_j), f_k] = L_k(H_j) f_k, \qquad [\pi(H_j), \wt{f_k}] = - L_k(H_j) \wt{f_k}.
    \] Then, to work on arbitrary products $f_K$, we use Leibniz yet again. One first sees that
    \begin{align*}
        [\pi(H_j), f_{k_1}f_{k_2}] = f_{k_1}[\pi(H_j),f_{k_2}] + [\pi(H_j),f_{k_1}]f_{k_2} = (L_{k_1} + L_{k_2})(H_j) f_{k_1}f_{k_2} ,
    \end{align*} and by induction, it is clear that when $\abs{K}< n/2$, $f_K = f_1 f_2 \dots f_k$ is the highest weight vector with weight
    \[
        [\pi(H_j), f_K] = \paran{L_1 + L_2 + \dots + L_k}(H_j) f_K .
    \] By the theorem of highest weight, this means that $\text{span}\{f_K : \abs{K} = k\} = \text{span}\{\gamma_K : \abs{K} = k\}\cong \Exterior^k V$ whenever $k< n/2$. 

    Now, when $k>n/2$, we observe that since $\pi(H_j)\in \calC_n^{[ev]}$, it commutes with $\gamma_0$, and so from (\ref{def:Hodge dual}) we see that
    \[
        [\pi(H_j), \star f_K] = \star [\pi(H_j), f_K] = \paran{L_1 + L_2 + \dots + L_k}(H_j) \star f_K.
    \] Similarly, by seeing that $[\pi(X),\star f_K] = \star [\pi(X), f_K]$ for any representative $\pi(X)$ of $\so(n)$, it follows that $f_K$ a highest weight vector implies $\star f_K$ a highest weight vector, since they are both annihilated by the set of positive simple roots. Thus, $\star(f_1 \dots f_{n-k})$ is the highest weight vector of $\Exterior^k V$ with weight $L_1 + L_2 + \dots L_{n-k}$.

    Now, when $n$ even, we had that $\Exterior^{n/2} V \cong U_+\oplus U_-$. We can directly find the vectors of highest weight for each of these representations, just as before: $f_1 f_2 \dots f_{n/2-1} f_{n/2}$ has weight $L_1 + L_2 + \dots + L_{n/2-1} + L_{n/2}$, while $f_1 f_2 \dots f_{n/2-1} \wt{f}_{n/2}$ has weight $L_1 + L_2 + \dots + L_{n/2-1} - L_{n/2}$. This is a different (but equivalent) perspective from earlier, where we leveraged that $U_+\oplus U_-$ corresponded to the images of left-multiplication by $P_\pm$. To check which of these highest weight vectors respectively correspond to $U_+$ or $U_-$, one can simply multiply each on the left by $P_+$. An explicit instance of this is demonstrated for $n=4$ in Equation (\ref{eq:n=4 example, f1f2 and f1wtf2 expanded}).

\subsection{Comment for the proof of the parent property} \label{secapp:comment for proof of parent property}
    In the final Chapter \ref{ch:SO(n)_Haldane_chains}, we will need explicit highest weight vectors for the proof of Theorem \ref{thm:parent property (SO(n) chains)}. It is thus convenient to list a couple of summarizing statements from the previous section.
    \begin{itemize}
        \item $f_1 f_2 \dots f_k$ is the highest weight vector of $\Exterior^k V$ for $k< n/2$ with weight $L_1+ L_2 + \dots + L_{k}$.
        \item $f_1 f_2 \dots f_{n/2-1} f_{n/2}$ and $f_1 f_2 \dots f_{n/2-1} \wt{f}_{n/2}$ are the highest weight vectors of $U_+\oplus U_- \cong \Exterior^{n/2}V$ with respective weights $L_1+L_2 + \dots + L_{n/2-1} + L_{n/2}$ and $L_1+L_2 + \dots + L_{n/2-1} - L_{n/2}$.
        \item $\star(f_1 \dots f_{n-k})$ is the highest weight vector of $\Exterior^k V$ for $k> n/2$ with weight $L_1+ L_2 + \dots + L_{n-k}$.
    \end{itemize} Note that since $\Exterior^k V \cong \Exterior^{n-k} V$, we have degeneracy of the irreps appearing in $\calC_n$ and any linear combination of their highest weight vectors $a f_1 \dots f_k + b \star(f_1 \dots f_{n-k})$, $a,b\in \C$, will be a vector of highest weight. This means that the above choices are not unique, although more convenient for our purposes.

    \begin{example} \label{ex:n=4 highest weight vectors}
        Let $n=4$. Then, with respect to the representation $[\pi,(\cdot)]$ on $\calC_4$, $f_1$ has weight $L_1$, $f_2$ has weight $L_2$, $\wt{f}_1$ has weight $-L_1$, and $\wt{f}_2$ has weight $-L_2$. Then the corresponding highest weights of the irreps decomposition of $\calC_4$ are
    \begin{alignat*}{2} 
        \Exterior^0 V &\text{ has highest weight vector } \idty &&\text{ with weight } 0  \\
        \Exterior^1 V &\text{ has highest weight vector } f_1 &&\text{ with weight } L_1 \\ 
        U_+ &\text{ has highest weight vector } f_1f_2 &&\text{ with weight } L_1+L_2 \\ 
        U_- &\text{ has highest weight vector } f_1\wt{f}_2 &&\text{ with weight } L_1-L_2 \\
        \Exterior^3 V &\text{ has highest weight vector } \star f_1 &&\text{ with weight } L_1 \\
        \Exterior^4 V &\text{ has highest weight vector } \star \idty \quad &&\text{ with weight } 0.
    \end{alignat*} As a sanity check, observe that $\star f_1 = \star(\gamma_1 + i\gamma_2) = -\gamma_2\gamma_3 \gamma_4 + i\gamma_1\gamma_3\gamma_4$, which is certainly in $\Exterior^3 V$. We also note that 
    \begin{equation}\begin{split} \label{eq:n=4 example, f1f2 and f1wtf2 expanded}
        f_1f_2 = (\gamma_1+i\gamma_2)(\gamma_3+i\gamma_4) &= \gamma_1\gamma_3 - \gamma_2\gamma_4 + i(\gamma_1\gamma_4 + \gamma_2 \gamma_3) \\
        &= \gamma_1\gamma_3 + \star (\gamma_1\gamma_3) + i (\gamma_1\gamma_4 + \star (\gamma_1\gamma_4) )\\
        f_1 \wt{f}_2 = (\gamma_1+i\gamma_2)(\gamma_3-i\gamma_4) &= \gamma_1\gamma_3 + \gamma_2\gamma_4 + i(\gamma_1\gamma_4 - \gamma_2 \gamma_3) \\
        &= \gamma_1\gamma_3 - \star (\gamma_1\gamma_3) + i (\gamma_1\gamma_4 - \star (\gamma_1\gamma_4) ),
    \end{split}\end{equation} meaning $U_+ = P_+\calC_4$ and $U_-=P_-\calC_4$.
    \end{example}
    
\chapter{Projective Representations of Finite and Compact Lie Groups}\label{ch:proj-reps}

Projective representations are common in quantum mechanics and will play an important role in the story of SPT phases later in this thesis. Commonly, we think of these representations as arising from the following ambiguity: while we often treat pure states as vectors $\psi\in \calH$, we may only measure expectation values $\inprod{\psi,A\psi}$ of observables $A\in \calA$. But expectation values cannot distinguish phases $e^{i\theta}\in U(1)$, $\theta\in \R$, i.e. if $\phi= e^{i\theta}\psi$, then 
\begin{equation}
   \inprod{\phi, A \phi} = \inprod{e^{i\theta}\psi, Ae^{i\theta} \psi} = \inprod{\psi, e^{-i\theta} Ae^{i\theta} \psi} =   \inprod{\psi, A\psi} .
\end{equation}
To remove this ambiguity then, one may quotient by this phase, and indeed this is exactly what we will do to define projective representations.

In Section \ref{sec:Projective Representations and De-projectivization}, we present a first definition of a projective (unitary) representation on a Hilbert space and clarify what we mean by de-projectivization. Here, we restrict our attention to compact Lie groups (which includes finite groups), whose projective representations may always be chosen to be unitary by Weyl's unitary trick Proposition~\ref{prop:Weyl unitary trick}. We also present our motivating example of the spin representations, first encountered in Remark \ref{rem:spin reps are projective}, which will guide us through the entire chapter.

In Section \ref{sec:Central Extensions and Projective Representations} we transform the data of a projective representation into the equivalent data of a central extension of a group $G$ by the abelian group $U(1)$. This leads us to develop some theory of central extensions. We will eventually see that projective representations which may be de-projectivized are exactly those which yield trivial central extensions, i.e. central extensions which split and so admit a section.

In Section \ref{sec:Universal Covers and De-projectivization}, we show a ``motivating calculation'': using Bargmann's theorem and Schur's lemma, we show that for an irreducible projective representation of a semi-simple Lie group $G$, the only obstruction to de-projectivization comes from the fundamental group $\pi_1(G)$, which is finite. This calculation will be subsumed by Theorem~\ref{thm:second group cohomology of semi-simple} in the following section, but it provides valuable insight and intuition.

In Section \ref{sec:Central Extensions and Group Cohomology}, we relate the central extension picture with a group cohomology picture. It begins with another definition of a projective representation, which is shown to be equivalent to the earlier definition. The section concludes with Remark \ref{rem:equivalence of Borel group cohomology}, which consolidates equivalences between projective representations, $U(1)$-central extensions, and elements of the second cohomology group $H^2(G,U(1))$. We wrap up with Theorem \ref{thm:second group cohomology of semi-simple}, which gives a very explicit characterization of $H^2(G,U(1))$ for compact Lie groups $G$.

To write this chapter, we follow the book by Schottenloher~\cite{schottenloher2008mathematical}, the paper by Bagchi and Misra~\cite{bagchi2000note}, a note by Levi Poon\footnote{Somewhat embarassingly, I cannot find this wonderful note online anymore. If you find it, please let me know so I may add this citation.}, and the StackExchange postings of user @ACuriousMind. We have also found the work of Duivenvoorden and Quella~\cite{duivenvoorden2013topological} to be a useful exploration of the ideas we discuss here in the context of SPT phases, and their work has greatly shaped this writing. Finally, for those interested in a separate development of the character theory of projective representations paralleling that of standard representation theory, we point to the paper~\cite{cheng2015character}.

%---------------------------------------------------------
\section{Projective Representations and De-projectivization} \label{sec:Projective Representations and De-projectivization}
Let $G$ be a compact Lie group. Note that every finite group is a compact Lie group when we equip it with the discrete topology, noting that in this case the fundamental group $\pi_1(G)\cong G$. In this case, the celebrated Peter-Weyl theorem (and its projective analogue~\cite{cheng2015character}) guarantees that every irreducible representation of $G$ is finite-dimensional $\dim(\calH)<\infty$.\footnote{Much of the following story may be upgraded to analogous statements for strongly continuous representations $G\to \calU(\calH)$ where $\calU(\calH)$ denotes unitary operators on a separable Hilbert space $\calH$. This is a topological group~\cite{espinoza2016topological}, not a Lie group, so some care must be taken. See Schottenloher~\cite{schottenloher2008mathematical}.}

Let $\calH$ be a complex Hilbert space, and let $\calU(\calH)$ denote the unitary operators on $\calH$. We may embed $U(1)\to \calU(\calH)$ by the map $\lambda \mapsto \lambda \idty$. This is a normal subgroup of $\calU(\calH)$, so define the \textit{projective unitary group} $P\calU(\calH) := \calU(\calH)/U(1)$ and let $p:\calU(\calH)\to P\calU(\calH)$ be the projection. Then we have a short exact sequence of Lie groups\footnote{Even more, this is a principle $U(1)$-bundle. In general, given a Lie group $G$ and a Lie subgroup $H$, we may think of $G$ as a principle $H$-bundle over the coset space $G/H$.}
\begin{equation} \label{eq:projective unitary group SES}
    1 \longrightarrow U(1) \longrightarrow \calU(\calH) \xlongrightarrow{p} P\calU(\calH) \longrightarrow 1 .
\end{equation}

\begin{definition}
    A \textit{projective} unitary representation of $G$ is a Lie group homomorphism $\Pi:G\to P\calU(\calH)$.
\end{definition}
Note that for compact Lie groups, every projective representation may be made into a projective unitary representation by Weyl's Unitary trick, the proof technique for Proposition~\ref{prop:Weyl unitary trick}. We will henceforth shorten ``projective unitary representation'' to ``projective representation''.

Now, of course, any representation is automatically a projective representation. But what about the converse: can we de-projectivize every projective representation? I.e. given a projective representation $\Pi:G\to P\calU(\calH)$, when may we lift it to a true representation, i.e. a map $\Phi: G\to \calU(\calH)$ such that $p \circ \Phi = \Pi$? In diagram form,
\begin{equation} \label{tikzcd:lifting projective rep to rep}
\includegraphics{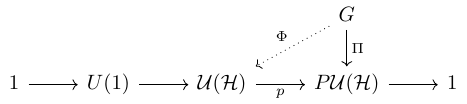}
\end{equation}
A key example to keep in mind throughout this section comes from Remark \ref{rem:spin reps are projective}. There we began with a representation $\pi:\so(3)\to \gl(\C^2)$ (equivalent to the spin-1/2 representation of $\su(2)\cong \so(3)$), exponentiated it to a representation $\wt{\Pi}:SU(2)\to \calU(\C^2)$, and then saw that this induces a projective representation $\Pi:SO(3)\to P\calU(\C^2)$. Is this equivalent to a true representation, i.e. can we find a representation $\Phi:SO(3)\to \calU(\C^2)$ agreeing with $\Pi$, $\Pi = p\circ \Phi$? The problem there was that since $\wt{\Pi}(-\idty) = -\idty$, it cannot pass through the projection $SU(2)\to SO(3)$ to yield a representation. But this argument does not really tell us that \textit{no} such $\Phi$ exists. It will turn out that indeed no such $\Phi$ exists, but to see this we will need to recast this question using some algebraic machinery (central extensions). Later, we will develop more computable topological invariants of this machinery (group cohomology) to answer similar problems.

\section{Central Extensions and Projective Representations} \label{sec:Central Extensions and Projective Representations}

Let us provide a small handful of definitions that will help us clarify the lifting problem. 
\begin{definition}
    Let $G$ and $A$ be Lie groups with $A$ abelian. A \emph{central extension} of $G$ by $A$ is a short exact sequence of Lie groups
    \begin{equation}
        1\longrightarrow A \longrightarrow E \longrightarrow G \longrightarrow 1 , 
    \end{equation} such that $A$ is embedded in the center $Z(E)$. 

    Two central extensions $E,E'$ of $G$ by $A$ are \emph{equivalent} if there exists an isomorphism $\Psi: E\to E'$ such that the following diagram commutes:
    \begin{equation}
        \includegraphics{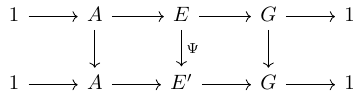}
    \end{equation}   
    
    We call the product group $A\times G$ the \emph{trivial extension}. 
\end{definition} 

\begin{example}
    Let $G=SO(3)$. Then a trivial central extension by $\Z_2$ is formed by $\Z_2 \times SO(3)\cong O(3)$:
    \begin{equation}
        1\longrightarrow \Z_2 \xlongrightarrow{\pm 1 \mapsto \pm \idty } O(3) \longrightarrow SO(3) \longrightarrow 1 ,
    \end{equation} and a nontrivial central extension is formed by $SU(2)$, since $SU(2)/\Z_2 \cong SO(3)$:
    \begin{equation}
        1\longrightarrow \Z_2 \xlongrightarrow{\pm 1 \mapsto \pm \idty } SU(2) \longrightarrow SO(3) \longrightarrow 1 .
    \end{equation} These are not equivalent extensions because $SU(2)\not\cong O(3)$, for instance since $SU(2)$ is connected but $O(3)$ is not. Note that in the first, the embedding map $\sigma:SO(3)\to O(3)$ is a Lie group homomorphism, but in the second, there is no Lie group homomorphism $\sigma:SO(3)\to SU(2)$. This may be seen by noting that such a continuous map would induce an embedding $\sigma_*:\pi_1(SO(3))\to \pi_1(SU(2))$, but $\pi_1(SO(3))=\Z_2$ and $\pi_1(SU(2)) = 0$.

    In the same way, for $G=SO(n)$, one can see that the following is a nontrivial central extension:
    \begin{equation}
        1\longrightarrow \Z_2 \longrightarrow Spin(n) \longrightarrow SO(n) \longrightarrow 1 . 
    \end{equation}
\end{example}

\begin{proposition} (Every projective representation yields a central extension~\cite{schottenloher2008mathematical}) \label{prop:proj rep yields central extension}

Let $G$ be a Lie group and $\Pi:G\to P\calU(\calH)$ a homomorphism. Then there is a central extension $E$ of $G$ by $U(1)$ and a homomorphism $p_U:E  \to \calU(\calH)$ so that the following diagram commutes:

\begin{equation}
\includegraphics{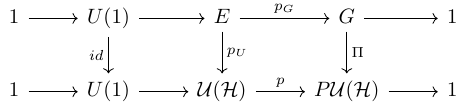}
\end{equation}
\end{proposition}
\begin{proof}
    Define $E$ to be the subgroup of the product group $\calU(\calH)\times G$ given by
    \begin{equation}\label{def:central extension of a proj rep}
        E := \{(U,g): p(U) = \Pi(g) \} . 
    \end{equation} $E$ is indeed a product group because $p, \Pi$ are both Lie group homomorphisms. Notice that the inclusion
    \begin{equation}\begin{split}
        U(1) &\longrightarrow E\subseteq \calU(\calH)\times G \\
        \lambda &\longmapsto (\lambda\idty_{\calH}, \idty) 
    \end{split}\end{equation} and the projection $p_G: E\to G$ given by $p_G(U,G) = G$ are homomorphisms such that the upper row is a central extension. Finally, the projection $p_U:E\to \calU(\calH)$ given by $p_U(U,G) = U$ is a homomorphism satisfying $\Pi \circ p_G = p_U \circ p$ . 
\end{proof}

Suppose that we had a projective representation which yields a central extension $E$. What does it mean if the central extension is trivial, i.e. equivalent to the product group extension $A\times G$? We always have the Lie homomorphism embedding $G$ into $A\times G$, given by $g\mapsto (1,g)$. So, when $E$ is trivial, we are guaranteed the existence of a section map $\sigma: G\to E$ with $p_G\circ \sigma = \ide_G$:
\begin{equation} 
\includegraphics{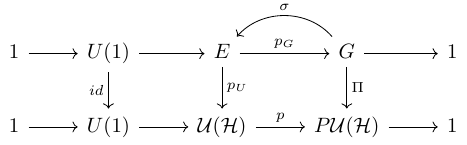}
\end{equation}
In this case, we may satisfy the question posed by (\ref{tikzcd:lifting projective rep to rep}): the desired representation is now given by $\Phi:G\to \calU(\calH)$ by $\Phi := \sigma\circ p_U$. Let us summarize this observation.

\begin{remark}
    If a projective representation $\Pi:G\to P\calU(\calH)$ yields a central extension $E$ (in the sense of Proposition \ref{prop:proj rep yields central extension}) which is equivalent to the trivial extension $U(1)\times G$, we may lift to a representation $\Phi:G\to \calU(\calH)$ where $p\circ \Phi = \Pi$. 
\end{remark}

The existence of this section map is actually a necessary and sufficient condition for saying a central extension is trivial.
\begin{lemma}\cite{schottenloher2008mathematical}
    A central extension $E$ is equivalent to the trivial extension $A\times G$ if and only if $E$ splits, i.e. if there exists a Lie homomorphism $\sigma: G\to E$ such that $p_G\circ \sigma = \ide_G$. 
\end{lemma} A key point is that $\sigma$ is a Lie homomorphism and not just some function (after all, since $p_G$ surjective, it is straightforward to find some choice of $\sigma$ such that $p_G\circ \sigma = \ide_G$.). This lemma then allows us to say the following.
\begin{corollary}
    If a projective representation $\Pi:G\to P\calU(\calH)$ yields a central extension $E$ which is not equivalent to the trivial extension, then there does not exist a representation $\Phi:G\to \calU(\calH)$ with $p\circ \Phi = \Pi$.
\end{corollary}
\begin{proof}
    Suppose towards contradiction that there existed such a $\Phi$. Then we could construct a Lie homomorphism $\sigma:G\to \calU(\calH)\times G$ by $g\mapsto (\Phi(g),g)$. Notice that $p \circ \Phi(g) = \Pi(g)$, so from the definition of $E$ (\ref{def:central extension of a proj rep}), $\sigma$ is actually a map $\sigma:G\to E$. But $E$ is a nontrivial extension so no such splitting map exists.
\end{proof}

So, given a projective representation, one may check to see whether it lifts to a true representation by checking whether the central extension it induces is trivial or not.

\begin{example}\label{ex:proj rep has SU(2) central extension}
    Recall the projective representation $\Pi:SO(3)\to P\calU(\C^2)$ from Remark (\ref{rem:spin reps are projective}). This is exactly the representation induced by the defining representation $\wt{\Pi}:SU(2)\to \calU(\C^2)$, which has $\wt{\Pi}(A) = A$, using that $SU(2)\cong e^{\so(3)}$. Note in particular that $\wt{\Pi}(\idty) = \idty$ and $\wt{\Pi}(-\idty) = -\idty$. Now, the central extension we get from Proposition \ref{prop:proj rep yields central extension} is the subgroup $E\subseteq \calU(\C^2)\times SO(3)$ 
    \begin{equation}
        E := \{ (U,g): p(U) = \Pi(g) \} .
    \end{equation} We claim that $E\cong SU(2)$. Recall that $S:SU(2)\to SO(3)$ implements the double cover projection. Consider the map 
    \begin{equation}\begin{split}
        SU(2) &\to E \subseteq \calU(\C^2)\times SO(3) \\
        A&\mapsto (\wt{\Pi}(A), S(A)) . 
    \end{split}\end{equation} This map is evidently a Lie group homomorphism. This map is invertible if its kernel is $\idty$. Observe then that $S$ is a double cover, so its kernel is $\pm \idty$. But $\wt{\Pi}(\idty)=\idty$ and $\wt{\Pi}(-\idty) = -\idty$, so this map is invertible and the claim is shown. 
    
    As stated earlier, there is no Lie group homomorphism $\sigma: SO(3)\to SU(2)$, so this central extension is nontrivial.
\end{example}

In a similar manner one can show that all half-integer $s$ representations of $SU(2)$ induce nontrivial projective representations of $SO(3)$, and likewise, that the spin representations from Remark \ref{rem:spin reps are projective} are all nontrivial projective. The key observation to do so is that for half-integer $s$ representations, $\wt{\Pi}(\idty) = \idty$ and $\wt{\Pi}(-\idty) = -\idty$. Compare this situation to that of integer $s$ representations, where $\wt{\Pi}(\idty) = \wt{\Pi}(-\idty)$. The following example is a rather overkill way of saying ``the defining representation of $SO(3)$ is a trivial projective representation'', but it serves to highlight the ideas of this section.

\begin{example} \label{ex:trivial projective rep, spin-1}
    The integer spin-1 representation $\wt{\Pi}:SU(2) \to \calU(\C^{3})$ naturally induces a projective representation $\Pi:SO(3)\to P\calU(\C^3)$. We can show that $\wt{\Pi}(\idty) = \wt{\Pi}(-\idty) = \idty$ by e.g. looking at the induced representation $\pi:\su(2)\to \gl(\C^3)$ which maps the Pauli $i\sigma^Z$ to the representative $H:= -i\pi(\frac{i}{2} \sigma^Z)= \begin{pmatrix}
        1 & 0 & 0 \\
        0 & 0 & 0 \\
        0 & 0 & -1
    \end{pmatrix}$. We have that $e^{\pi\sigma^Z} = -\idty$. So, 
        $\wt{\Pi}(-\idty) = e^{2\pi H} = \idty $. 
    In particular, since the double covering map $S:SU(2)\to SO(3)$ has $\ker S = \{\pm \idty\}$, we may meaningfully define the representation $\Phi:SO(3)\to \calU(\C^3)$ by $\Phi = \wt{\Pi}\circ S^{-1}$. Thus, this is the trivial extension $E = SO(3)\times \Z_2$.
\end{example} 

To summarize what has occurred to this point: a projective representation $\Pi:G\to P\calU(\calH)$ induces a central extension $E$ of $G$ by $U(1)$. When a central extension splits, it is equivalent to the trivial extension $U(1)\times G$ and we may find a true representation $\Phi:G\to \calU(\calH)$ which agrees with the projective representation $\Pi = p\circ \Phi$. When a central extension does not split (and so is not equivalent to the trivial extension), we can not find a true representation $\Phi:G\to \calU(\calH)$ agreeing with the projective representation. This shifts our problem to identifying when a central extension splits.\footnote{Bundle lovers may recall that a $U(1)$ principle bundle $E$ over a base $G$ is trivial if and only if there is a global section $\sigma:G\to E$.} Perhaps without realizing it, we have already been exposed to an important example of a central extension: the universal cover of a Lie group.

\section{Universal Covers and De-projectivization} \label{sec:Universal Covers and De-projectivization}

In this section, we will perform a brief analysis for the case of irreducible projective representations of a semi-simple Lie group $G$. While the result will be subsumed by Theorem~\ref{thm:second group cohomology of semi-simple}, the calculation provides valuable intuition for the full Theorem.

We start by citing a special case of an important theorem.\footnote{Typically, this theorem is stated as ``If the second Lie algebra cohomology $H^2(\g,\R)=0$, then every projective representation is equivalent to a true representation.'' Lie algebras are in one-to-one correspondence with their simply connected lie groups, and for semi-simple Lie algebras, $H^2(\g,\R)=0$.} As a warning, recall that by definition abelian Lie groups, like $U(1)$, are not semi-simple.
\begin{theorem} (Bargmann's Theorem~\cite{bargmann1954unitary})
    Let $G$ be a simply connected semi-simple Lie group. Then every projective representation $\Pi:G\to P\calU(\calH)$ may be lifted to a true representation $\Phi:G\to \calU(\calH)$ such that $p\circ \Phi = \Pi$. 
\end{theorem}

Already this is a huge step. For instance, this immediately tells us that every projective representation of $SU(2)$ may be de-projectivized, i.e. is equivalent to a true representation. So: how can we use this information to say something about other semi-simple Lie group representations? Recall Theorem (\ref{thm:universal cover of a Lie group}), which states that connected Lie groups $G$ have a universal cover $\wt{G}$, a simply connected Lie group with a discrete normal subgroup $H\cong \pi_1(G)$ such that $G\cong \wt{G}/H$ and the Lie algebras of $G$ and $\wt{G}$ are isomorphic. Remark~\ref{rem:discrete normal subgroups are in center} reveals that a discrete normal subgroup of a Lie group is contained in the center $Z(\wt{G})$, so we have the exact sequence 
\begin{equation}
    1\longrightarrow \pi_1(G) \longrightarrow \wt{G} \longrightarrow G \longrightarrow 1 . 
\end{equation} In fact, our situation is even better: for a compact semi-simple Lie group $G$, a result of Weyl's shows that the fundamental group $\pi_1(G)$ is a finite group~\cite{knapp1996lie}. This was the case for $SU(2)$ double-covering $SO(3)$ (and more generally $Spin(n)$ double-covering $SO(n)$). This seems to suggest that the equivalence classes of central extensions of $G$ by $U(1)$ are controlled by the fundamental group $\pi_1(G)$--let us see how this is so.

Let $\Pi:G\to P\calU(\calH)$ be irreducible. By passing to the Lie algebra $\g$ and then exponentiating to the universal cover $\wt{G}$, we get a projective unitary representation $\wt{\Pi}:\wt{G}\to P\calU(\calH)$. Bargmann's theorem guarantees that we may de-projectivize this representation, so we write $\wt{\Pi}:\wt{G}\to \calU(\calH)$. By Schur's lemma, the center of this group $\pi_1(G)\subseteq \wt{G}$ acts as scalar multiples of the identity $c\idty\in U(1)\subseteq \calU(\calH)$. In other words, we have a homomorphism $\lambda\in \Hom(\pi_1(G), U(1))$. When $\pi_1(G)$ acts trivially on $\calH$, i.e. when $\lambda(\gamma) = \idty$ for all $\gamma\in \pi_1(G)$, the representation $\wt{\Pi}:\wt{G}\to \calU(\calH)$ descends to a representation $\Pi:G\to \calU(\calH)$ which projective unitary representation $\Pi:G\to P\calU(\calH)$ we began with. In the language of central extensions, this is exactly the trivial central extension $\pi_1(G)\times G$. Compare this to Example \ref{ex:trivial projective rep, spin-1}, where we saw that the spin-1 representation of $SU(2)$ passes to a representation of $SO(3)$. To contrast, when $\pi_1(G)$ acts nontrivially on $\calH$, so $\lambda$ not constant, the representation $\wt{\Pi}:\wt{G}\to \calU(\calH)$ does not descend to a representation of $G$, and so we may not de-projectivize the representation $\Pi:G\to P\calU(\calH)$. Compare this to Example \ref{ex:proj rep has SU(2) central extension}.

In short: the argument above tells us that a nontrivial fundamental group $\pi_1(G)$ is the only obstruction to de-projectivizing a semi-simple Lie group. Let us summarize this in a remark. Of particular note is that the intimidating task of classifying equivalence classes of central extensions of semi-simple compact Lie groups by $U(1)$ may be transformed into a combinatorial question about finite abelian groups.

\begin{remark} \label{rem:upper bound irreducible proj reps}
    Each irreducible projective representation connected semi-simple Lie group gives rise to an equivalence class of central extensions by $U(1)$. The number of classes of central extensions is bounded above by the size of the finite group $\abs{\Hom(\pi_1(G),U(1))}$ and thus finite.
\end{remark}

This brief calculation will be a special case of the broader Theorem~\ref{thm:second group cohomology of semi-simple}, whose proof we omit. We will need to develop a group cohomology theory which will allow us to show that these equivalence classes of central extensions may be labeled by a certain cohomology group, $H^2(G,U(1))$. Each cohomology class will constructively label an inequivalent central extension, and then Theorem~\ref{thm:second group cohomology of semi-simple} will reveal that for any compact Lie group, $H^2(G,U(1))$ is finite and controlled by its fundamental group.

\section{Central Extensions and Group Cohomology} \label{sec:Central Extensions and Group Cohomology}
We now have a few equivalent situations. A projective representation may be de-projectivized when the central extension of $G$ by $U(1)$ it induces is trivial, or equivalently, when the $U(1)$-bundle it induces admits a global section. A short calculation for irreducible projective representations, combined with Bargmann's theorem, suggests that the number of inequivalent central extensions is controlled by $\Hom(\pi_1(G),U(1))$. The general story is resolved by constructing an appropriate invariant: group cohomology. The following discussion is an amalgam of results from~\cite{schottenloher2008mathematical,bagchi2000note,moore1964extensions}, although certainly many of these observations date back to earlier results. 

Let us start by considering an equivalent definition of a projective representation. Equivalence will be shown in Remark \ref{rem:equivalence of projective rep definitions}.
%Note that the regularity assumptions on these maps make this definition a subtle special case of the prior definition,\footnote{Every measurable group homomorphism between Lie groups is smooth, so the special case comes from a (lax) regularity assumption on $\alpha$.} so they are not strictly equivalent. However, they are equivalent for finite groups, and the equivalence classes yielded by the ensuing cohomology theory will actually attain our upper bound $\abs{\Hom(G,U(1))}$, so we will not worry too much about the distinction. 
In the ensuing discussion, we shorten ``projective unitary representation'' to ``projective representation''.

\begin{definition}
    A projective representation of $G$ is a (Borel) measurable map $\Pi:G\to \calU(\calH)$ together with a measurable map $\alpha: G\times G \to U(1)$ such that for all $g,h\in G$,
    \begin{equation} \label{def:proj rep with schur cocycle}
        \Pi(1) = \idty, \quad \Pi(g)\Pi(h) = \alpha(g,h)\Pi(gh) . 
    \end{equation}

    Two projective representations $\Pi_1:G\to \calU(\calH_1)$, $\Pi_2:G\to \calU(\calH_2)$ are  \emph{equivalent} if there exists a unitary $U:\calH_1\to \calH_2$ and a (Borel) function $\lambda:G\to U(1)$ such that
    \begin{equation}
        \Pi_2(g) U = \lambda(g) U \Pi_1(g) \qquad \text{ for all } g\in G. 
    \end{equation}
\end{definition}

We still have the close relationship with central extensions here. Let 
\begin{equation}
    1\longrightarrow U(1) \xlongrightarrow{\iota} E \xlongrightarrow{p_G} G \longrightarrow 1 
\end{equation} be a central extension and let $\tau:G\to E$ be a Borel function with $p_G \circ \tau = \ide_G$ and $\tau(\idty) = \idty$. We set $\tau_g:= \tau(g)$ for $g\in G$ and define
\begin{equation}\begin{split}
    \alpha: G\times G &\longrightarrow U(1)\cong \iota(U(1))\subseteq E, \\
    (g,h) &\mapsto \tau_g\tau_h \tau_{gh}^{-1} . 
\end{split}\end{equation} To see that this function is well-defined, observe that $\tau_g \tau_h \tau_{gh}^{-1} \in \ker p_G$, $\alpha(1,1) =1$, and use associativity of the group $G$ to see that $\alpha(1,1) = 1$ and 
\begin{equation}\label{eq:cocycle condition}
    \alpha(g,h)\alpha(gh,j) = \alpha(g,hj)\alpha(h,j) , \qquad g,h,j\in G . 
\end{equation}

\begin{remark} \label{rem:equivalence of projective rep definitions}
    We can now see why the two definitions of projective representation are equivalent. Every Borel map $\Pi:G\to \calU(\calH)$ obeying (\ref{def:proj rep with schur cocycle}) descends to a Lie group homomorphism $G\to P\calU(\calH)$, where continuity of the homomorphism comes from that every Borel group homomorphism between Lie groups is smooth. Now, let us see that every Lie group homomorphism $G\to P\calU(\calH)$ lifts to a Borel map $\Pi:G\to \calU(\calH)$ with a Borel map $\alpha:G\times G\to U(1)$ obeying (\ref{def:proj rep with schur cocycle}). Proposition \ref{prop:proj rep yields central extension} guarantees the existence of a principal $U(1)$-bundle $E$ over $G$. We may take a trivializing open cover $\{U_j\}$ of $G$. On each open set $U_j$ we may find a smooth local section $\tau_j:U_j\to E$. Define a function $\tau:G\to E$ such that $\tau(\idty) = \idty$, $\tau$ restricts to $\tau_j$ on $U_j$ whenever the intersection $U_j\cap U_k = \emptyset$, and $\tau$'s restriction to nonempty intersections is made by a coherent choice of $\tau_j$. Such a choice is guaranteed by the axiom of choice. This gives a Borel function $\tau:G\to E$.\footnote{$\tau$ need not be continuous. In fact, it often isn't--if it were, then $\tau$ would be a section of a principle bundle and thus $E$ would be the trivial bundle.} Then, the map $\alpha(g,h) = \tau_g\tau_h\tau_{gh}^{-1}$ is Borel, as is the map $\Pi:= p_U \circ \tau$ where $p_U:E\to \calU(\calH)$ as in Proposition \ref{prop:proj rep yields central extension}.
\end{remark}

\begin{definition}
    Any Borel function $\alpha: G\times G \to U(1)$ obeying (\ref{eq:cocycle condition}) is called a \emph{2-cocycle} (or just \emph{cocycle}) on $G$ with values in $U(1)$.

    The central extension of $G$ by $U(1)$ \emph{associated with a cocycle} $\alpha$ is given by the exact sequence
    \begin{equation}\begin{split}
        1\longrightarrow U(1) &\xlongrightarrow{\iota} U(1)\times_\alpha G \xrightarrow{p_G} G \longrightarrow 1 ,\\ 
        a&\longmapsto (a,\idty) ,
    \end{split}\end{equation} where $U(1)\times_\alpha G$ denotes the product $U(1)\times G$ endowed with the ``twisted'' multiplication defined by 
    \begin{equation}
        (a,g)(b,h) := (\alpha(g,h)ab,gh), \qquad (a,g),(b,h)\in U(1)\times G . 
    \end{equation}
\end{definition}
One must show that $U(1)\times_\alpha G$ is indeed a topological group. Associativity comes from the cocycle condition (\ref{eq:cocycle condition}). The regularity question is far more subtle and beyond the scope of these notes: we refer the interested reader to Bachgi and Misra~\cite{bagchi2000note} and Chapter 2.6 of Tao~\cite{tao2014hilbert}.\footnote{I would like to thank Jerry Kaminker for pointing me to Tao's book. A key point of the argument is showing that given $\alpha$, one may find a smooth cohomologous cocycle $\alpha'$ which allows one to define a smooth structure on $U(1)\times_{\alpha'} G$}

Thus, central extensions yield cocycles and cocycles yield central extensions. We would like a suitable notion of equivalence for cocycles to get a one-to-one correspondence with our (equivalence classes of) central extensions. So: when does a cocycle yield the trivial extension?

\begin{lemma} ~\cite{schottenloher2008mathematical} \label{lem:central extension splits for coboundaries}
    Let $\alpha:G\times G\to U(1)$ be a cocycle. Then the central extension $U(1)\times_\alpha G$ associated with $\alpha$ splits if and only if there is a function $\lambda:G\to U(1)$ with 
    \begin{equation}
        \lambda(gh) = \alpha(g,h) \lambda(g)\lambda(h) .
    \end{equation}
\end{lemma}\begin{proof}
    Consider the map $\sigma:G\to U(1)\times_\alpha G$ given by $\sigma(g) := (\lambda(g),g)$ for $g\in G$. Certainly we have $p_G\circ \sigma = \ide_G$. $\sigma$ is a homomorphism, and thus a section, if and only if
    \begin{equation}
        \sigma(gh) = \sigma(g)\sigma(h) \iff \lambda(gh) = \alpha(g,h)\lambda(g)\lambda(h) . 
    \end{equation} 
\end{proof} This gives an algebraic criterion to check whether our central extension is trivial, which in turn allows us to check whether we may lift a projective representation to a true representation. We can get our desired one-to-one correspondence then by using the lemma to take a suitable quotient.

\begin{definition} \label{def:second group cohomology}
    A cocycle $\alpha$ is called a \emph{coboundary} if there exists a Borel function $\lambda:G\to U(1)$ such that
    \begin{equation}
        \alpha(g,h) = \frac{\lambda(g)\lambda(h)}{\lambda(gh)}, \qquad g,h\in G.
    \end{equation} The set of cocycles $Z$ forms an abelian group with pointwise multiplication, and the set of coboundaries $B$ is a subgroup of $Z$. Define the \emph{second Borel group cohomology of $G$ with values in $U(1)$} by
    \begin{equation}
        H^2(G, U(1)) := Z/B . 
    \end{equation}
\end{definition} Note that one may indeed define a \textit{bona fide} cohomology theory which enjoys the usual cohomology features, like short exact sequence $\implies$ long exact sequence. See~\cite{moore1964extensions} for more information.

\begin{example} \label{ex:second group cohomology U(1)}
    We may immediately compute a rather boring group cohomology. If $G=U(1)$, then any projective representation $\Pi: U(1)\to \calU(\calH)$ with cocycle $\alpha$ lands in the center, i.e. $\im(\Pi)\subseteq U(1)\idty\subseteq \calU(\calH)$, and so may be treated as a map $\Pi: U(1)\to U(1)$. Picking $\lambda = \Pi$, the cocycle $\alpha$ must be a coboundary. Thus,
    \begin{equation}
        H^2(U(1),U(1)) = 0 . 
    \end{equation}
\end{example}

Let us wrap up this chapter with a summarizing remark and a theorem. In the original paper~\cite{moore1976group}, this result holds for $G$ and $A$ locally compact separable with $A$ abelian.

\begin{remark} (Theorem 10~\cite{moore1976group})\label{rem:equivalence of Borel group cohomology}
    From its definition and Lemma (\ref{lem:central extension splits for coboundaries}), we see that central extensions\footnote{Again, recall that these $U(1)$ central extensions of $G$ are the same as considering $G$ as a $U(1)$-principle bundle, which opens the door to other topological methods to classify these isomorphism classes: namely, the classifying space $BG$. See Proposition 4 in~\cite{schottenloher2018unitary} and use the long exact sequence in cohomology induced by the short exact sequence of trivial $G$-modules $1\to \Z\to \R \to U(1)\to 1$, using Appendix 4 of~\cite{chen2013symmetry} to complete the argument.}  of $G$ by $U(1)$ are in one-to-one correspondence with elements in $H^2(G,U(1))$, the Borel group cohomology of $G$ with values in $U(1)$. Altogether, we have the following one-to-one correspondences
    \begin{equation}
    H^2(G,U(1)) \leftrightarrow \left\{
        \parbox{4cm}{\centering 
           Projective representations $\Pi:G\to \calU(\calH)$}
    \right\}
     \leftrightarrow
     \left\{
        \parbox{4cm}{\centering 
           $U(1)$-central extensions of $G$}
     \right\}
    \end{equation}
\end{remark}

For finite groups $G$, the cohomology group $H^2(G,U(1))$ is finite and this is the same as the typical group cohomology. For connected compact Lie groups, we fulfill Remark~\ref{rem:upper bound irreducible proj reps} and present the following theorem of Moore. It is perhaps worth mentioning that the semi-simple case of this result was re-proven with different techniques by Bagchi and Misra in~\cite{bagchi2000note}: their goal was to classify projective representations, and their proof is in spirit a natural extension of the argument presented in Section~\ref{sec:Universal Covers and De-projectivization}.

Any compact Lie group is reductive and so can be decomposed as a product of an abelian Lie group and a semi-simple Lie group, $G=G_{ab}\times G_{ss}$~\cite{kirillov2008introduction}.

\begin{theorem}(Proposition 2.1 and 2.2~\cite{moore1964extensions})\footnote{It should be mentioned that Moore partially attributes these propositions to earlier results by Calabi and Shapiro.}\label{thm:second group cohomology of semi-simple}

Let $G$ be a compact Lie group. Then $H^2(G,U(1))$ is a finite group, and when $G$ is connected, we have the isomorphism
\begin{equation}
    H^2(G,U(1)) \cong \pi_1(G_{ss}),
\end{equation} where $G_{ss}$ is the semi-simple group appearing in the decomposition $G = G_{ab}\times G_{ss}$.

\end{theorem} 
Let us make a few comments, as the original statement by Moore was aesthetically different. Moore originally wrote that 
\begin{equation} \label{eq:original moore statement}
    H^2(G,U(1)) \cong \Hom(t(\pi_1(G)),U(1),
\end{equation} where $t(\pi_1(G))$ denotes the torsion subgroup of $\pi_1(G)$. By a standard result in algebraic topology (Proposition 1.12~\cite{hatcher2002algebraic}), $\pi_1(G)\cong \pi_1(G_{ab})\times \pi_1(G_{ss})$. Every compact connected abelian Lie group $G_{ab}$ is a product of $k$ tori~\cite{simon1996representations}, and so $\pi_1(G_{ab}) \cong \Z^k$. It is a well-known result that if $G_{ss}$ is semi-simple, then $\pi_1(G_{ss})$ is finite and abelian~\cite{knapp1996lie}. Thus, the torsion subgroup $t(\pi_1(G_{ab})\times \pi_1(G_{ss})) \cong \pi_1(G_{ss})$. Finally, since $\pi_1(G_{ss})$ is a finite abelian group, it is (non-canonically) isomorphic to its Pontryagin dual, $\pi_1(G_{ss})\cong \Hom(\pi_1(G_{ss}),U(1))$~\cite{terras1999fourier}. Thus, Moore's original statement Equation (\ref{eq:original moore statement}) is equivalent to Theorem~\ref{thm:second group cohomology of semi-simple}.

\begin{remark}
    For $G = SO(n)$, $H^2(SO(n),U(1)) = \Z_2$ since $\pi_1(SO(n)) = \Z_2$.
\end{remark} For $SO(3)$, the trivial projective representations are the integer spin-$s$ representations of $SU(2)$, and the nontrivial projective representations are the half-integer spin-$s$ representations of $SU(2)$. Note that in this case, when we restrict to the dihedral subgroup $\Z_2\times \Z_2\subseteq SO(3)$, it is likewise trivial projective for integer spin-$s$ and nontrivial projective for half-integer spin-$s$.

\chapter{The AKLT Chain and the Haldane Phase }\label{ch:AKLT_chain}

\section{The Bilinear-Biquadratic Spin-1 Phase Diagram}
Recall Example \ref{ex:heisenberg chain} (and the revisited Example \ref{ex:heisenberg, revisited}) of the Heisenberg antiferromagnetic ($J<0$) spin chain on $\Lambda = [a,b]\subseteq \Z$. We can generalize this slightly to consider the case where the on-site Hilbert space is the spin-$s$ irrep of $\su(2)$, $\calH_x = V_s =\C^{2s+1}$: 
\begin{equation}\label{ex:heisenberg (AKLT chapter)}
    H_\Lambda = -J \sum_{x=a}^{b-1} \vec{S}_{x} \cdot \vec{S}_{x+1} ,
\end{equation} where we have defined the spin matrices $S^X,S^Y,S^Z$ along with the total angular momentum operator in Equation (\ref{def:total angular momentum}). We saw earlier that this is a natural generalization of the classical Heisenberg antiferromagnet, wherein each site $x\in \Lambda$ hosts ``classical'' spins $\vec{r}_x\in \R^{2s+1}$ with $\abs{\vec{r}_x} = 1$,  and the Hamiltonian rewards anti-aligned neighboring spins. When $s=1/2$, the quantum Heisenberg antiferromagnet is quite well understood: exact eigenvalues and eigenvectors are known via Bethe ansatz, first performed in 1931~\cite{bethe1931theorie}; the ground state is known to be unique by Marshall-Lieb-Mattis; the spectral gap is known to vanish in the thermodynamic limit by Lieb-Schultz-Mattis theorem; and the ground state correlations have been computed (but not proven) to have power law decay~\cite{tasaki2020physics}.

It is natural to suspect that changing from $s=1/2$ to $s=1$ should not seriously disrupt these qualitative properties--after all, by classical analogy, $\R^2$ and $\R^3$ are not so different. So it came as a great surprise when Haldane predicted the following features of the $s=1$ antiferromagnetic quantum Heisenberg chain:
\begin{itemize}
    \item The ground state is unique.
    \item The ground state correlation function decays exponentially.
    \item There is a spectral gap above the ground state energy.
\end{itemize} Moreover, Haldane predicted these features to be the hallmarks of a full quantum phase of matter. Turns out, the $s=1$ Heisenberg antiferromagnet is a tricky beast. Despite abundant numerical and theoretical evidence suggesting these properties are true, proving the existence of a gap for this model remains an open problem even today. But in 1987, Affleck, Kennedy, Lieb, and Tasaki presented the AKLT chain, an exactly solvable model whose features matched Haldane's prediction. The AKLT chain can be thought of as a perturbation of the Heisenberg chain (\ref{ex:heisenberg (AKLT chapter)}), found by normalizing $J=-1$ and adding a ``biquadratic term'' $(\vec{S}_x \cdot \vec{S}_{x+1})^2$ to the already present bilinear term $\vec{S}_x\cdot \vec{S}_{x+1}$:
\begin{equation}
    \wt{H}_\Lambda = \sum_{x=a}^{b-1} \frac{1}{2} \vec{S}_x \cdot \vec{S}_{x+1} + \frac{1}{6} (\vec{S}_x \cdot \vec{S}_{x+1})^2 ,
\end{equation} which, up to a shift in the ground state energy accomplished by the terms $\frac{1}{3}\idty$, is the same as the AKLT chain Hamiltonian (\ref{ex:AKLT, revisited}).

But why stop at a small perturbation? Indeed, one can consider a whole family of models: the spin-1 bilinear biquadratic chains
\begin{equation} \label{def:bilinear-biquadratic}
    H_{\Lambda}(u,v) = \sum_{x=a}^{b-1} \cos \theta (\vec{S}_{x} \cdot \vec{S}_{x+1} ) + \sin \theta (\vec{S}_{x} \cdot \vec{S}_{x+1})^2 .
\end{equation} Up to some physically inconsequential rescaling so that the models are parameterized on a circle, this is the most general phase diagram for spin-1 $SU(2)$-invariant nearest-neighbor Hamiltonians, which can be seen in the following way. Write the homomorphism for the tensor representation of spin-$1$'s as $\pi: \su(2)\to \gl(V_1\otimes V_1)$.
One can calculate that both $\vec{S}_x\cdot \vec{S}_{x+1}$ and $(\vec{S}_x\cdot \vec{S}_{x+1})^2$ are in the commutant $\calC(V_1\otimes V_1)$ of the representation $V_1\otimes V_1$, i.e. 
\begin{equation}
    [\pi(X), \vec{S}_x\cdot \vec{S}_{x+1}] =[\pi(X), (\vec{S}_x\cdot \vec{S}_{x+1})^2] = 0, \qquad \text{ for all } X\in \su(2) .
\end{equation} Then, the Clebsch-Gordan decomposition (\ref{prop:clebsch-gordan}) of two neighboring spin-1 sites is $V_{1}\otimes V_{1}\cong V_2 \oplus V_1 \oplus V_0$. 
These irreducibles are all distinct, so by Schur's lemma (\ref{lem:schur}), the commutant $\calC(V_1\otimes V_1)$ is 3-dimensional and spanned by the orthogonal projections $P^{(2)}, P^{(1)}, P^{(0)}$. Notice that $\idty, \vec{S}_x\cdot \vec{S}_{x+1}, (\vec{S}_x\cdot \vec{S}_{x+1})^2$ are linearly independent and so form a basis for $\calC(V_1\otimes V_1)$. Any $SU(2)$-invariant interaction $h_{x,x+1}\in \gl(V_1\otimes V_1)$ by definition has $h_{x,x+1}\in \calC(V_1\otimes V_1)$, and so can be written as 
\[
    h_{x,x+1} = c_0 \idty + c_1 (\vec{S}_x\cdot \vec{S}_{x+1}) + c_2 (\vec{S}_x\cdot \vec{S}_{x+1})^2 , \qquad c_i\in \R . 
\] where we have used self-adjointness $h_{x,x+1} = h_{x,x+1}^*$ to see that $c_i\in \R$. But $c_0\idty$ inconsequentially shifts the spectrum of $h_{x,x+1}$ by $c_0$, and so does not contribute to the ground state structure. This completes the argument. Let us now proceed to look at the phase diagram of the Hamiltonian (\ref{def:bilinear-biquadratic}), which has been thoroughly studied throughout the past 40 years:
\begin{figure}
    \centering
    \includegraphics[width=0.5\textwidth]{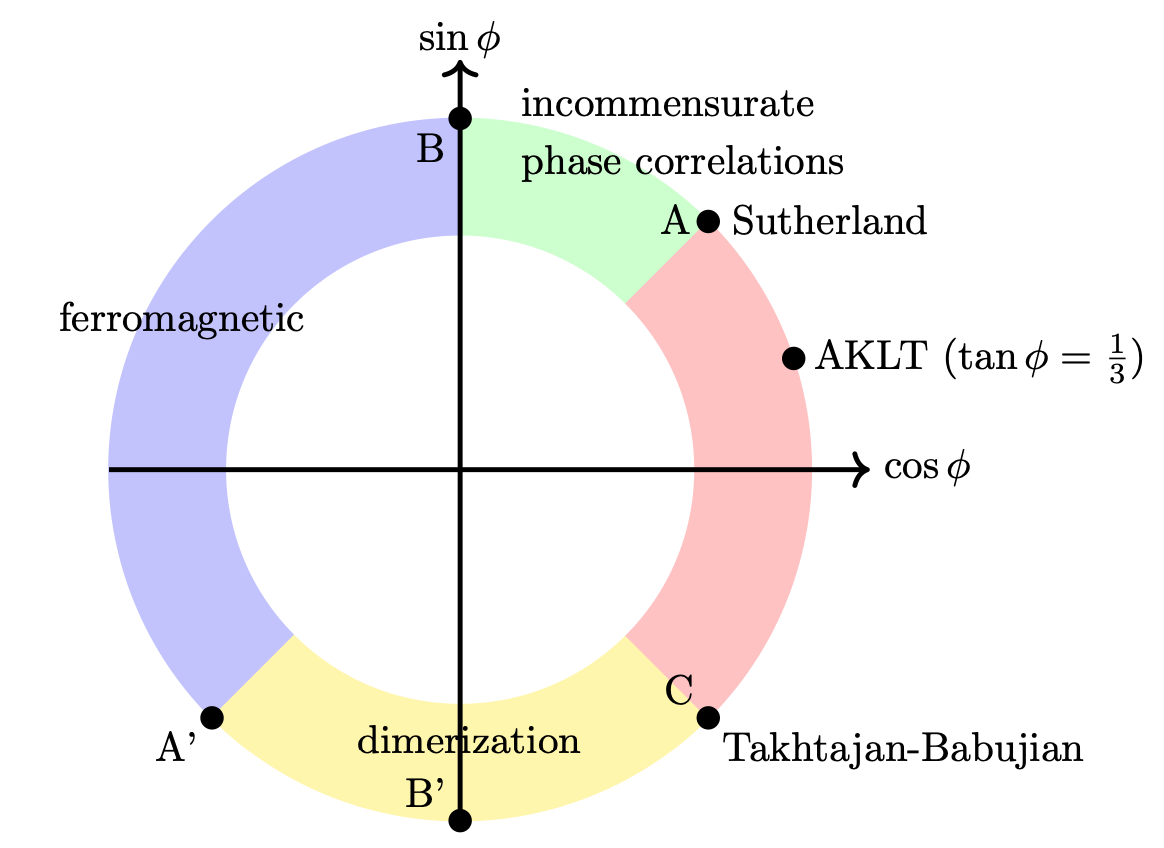}
    \caption{The spin-1 bilinear biquadratic phase diagram (Figure 2 in ~\cite{bjornberg2021dimerization}). The red phase corresponds to the Haldane phase, inhabited by the AKLT chain.}
    \label{fig:bilin_biquad}
\end{figure}

The Heisenberg antiferromagnet lands at $\phi = 0$, whence the picture of the AKLT model as a ``perturbation'' of this model seems a bit more realistic. A great deal of numerical evidence suggests that the red phase is indeed a stable gapped phase between the gapless models studied by Sutherland and Takhtajan-Babujian~\cite{bjornberg2021dimerization}. This figure was taken from~\cite{bjornberg2021dimerization}, wherein properties of the bottom yellow ``dimerization'' region were proven.\footnote{This same reference also features a helpful discussion on the regions in this phase diagram.} We will have more to say comparing the red and yellow regions in later chapters, but for now, the takeaway is that the AKLT chain is an exactly solvable model believed to qualitatively model the behavior of the red Haldane phase. 

\subsection{Overview of the chapter}

In this chapter, we will study the AKLT ground state from three different perspectives: the valence bond solid (VBS) picture, the matrix product state (MPS) picture, and the finitely correlated state (FCS) picture.\footnote{These notes have been largely shaped by Tasaki's book~\cite{tasaki2020physics} and Nachteraele and Sims' lecture notes~\cite{nachtergaele2016quantum}. Most of the theorems/propositions/lemmas are taken from one of these sources, with most of the VBS picture from Tasaki and most of the MPS/FCS picture from Nachtergaele and Sims.} This state would end up having impact beyond its unique role in probing the Haldane phase: it would lead to the invention of MPSs in the seminal pair of papers~\cite{white1992density,fannes1992finitely} and then the more general projected entangled pair states (PEPS), a massive class of useful low-entanglement quantum states commonly used to approximate other low-entanglement quantum states both analytically and numerically, as well as the rich study of symmetry protected topological (SPT) order states.

In Section \ref{sec:The AKLT Ground State as a Valence Bond State (VBS)}, we warm-up to the VBS picture by way of the Majumdar-Ghosh model, a predecessor to the AKLT chain whose VBS ground states provide a simple first view of the ``spins and bonds'' calculations common to $SU(2)$-invariant models. It also serves as a caricature of dimerization, the breaking of a Hamiltonian's translation invariance into a pair of 2-periodic ground states; here, the two ground states consist only of singlet bonds between nearest neighbors, and the translation-invariance is intuitively broken by the ground states ``choosing'' to either bond the spin at a site $x$ with its left $x-1$ or right $x+1$ neighbor. We then graduate to the VBS picture of the AKLT ground state. The process goes by constructing a ``virtual'' space of spin-1/2 particles and then constructing the ``physical space'' of spin-1s by way of an orthogonal projection $P_{sym}:V_{1/2}\otimes V_{1/2}\to V_1$. One constructs a line of entangled singlet bonds in the virtual space and then projects down to the physical space. Then, since $P_{sym}$ is an intertwiner of $SU(2)$ reps, we may commute the interaction terms $P_{x,x+1}^{(2)}$ of the AKLT chain past it, whence orthogonality of distinct irreps ensures that $P^{(2)}$ annihilates a singlet bond. In this picture, it is particularly clear that the finite-chain ground state space forms the $SU(2)$ representation $V_1\oplus V_0$. 

In Section \ref{sec:The AKLT Ground State as a Matrix Product State (MPS)} we proceed to the MPS picture. The fundamental object here is the tensor $T:\C^2\to \C^3\otimes \C^2$, which we introduce as an intertwining isometry. One then builds a tensor network by stringing several copies of this tensor together, where the free indices of this tensor network correspond to physical degrees of freedom. This gives us a very explicit expression of these states which grant access to a variety of calculations. This picture is in fact the same as the previous picture, since tensor network contractions may be equivalently described as projections of entangled pairs (hence PEPS). We spend some time in the end of this section detailing this correspondence. 

In Section \ref{sec:The AKLT Ground State as a Finitely Correlated State (FCS)} we reach our final picture, FCS. In essence this picture is the same as that of (translation invariant and $k$-periodic) MPSs, but described with the thermodynamic limit in mind. Here the tensor $T$ is repackaged as a CP map $\bbE:\calA_{x}\otimes \calB \to \calB$ given by $\bbE_A(B) = T^*(A\otimes B) T$, and the finite chain boundary data of the MPS are stored in a positive linear functional (``density matrix'') $\rho\in \calB^*$ and positive element $e\in \calB$. The advantage of this approach is that expectations in the thermodynamic limiting states can be readily accessed by diagonalizing the transfer operator $\bbE_\idty$. It is here that we will verify two of the defining features of the Haldane phase: that the thermodynamic limiting ground state is unique, and that the ground state exhibits exponential decay of correlations. The final property, that the Hamiltonian is gapped, will not be proven, but references will be provided.

\subsection{Setup}

Take a finite chain $\Lambda = [1,\ell]$ of spin-1's $V_1\cong \C^3$, letting as usual $\calH = \bigotimes_{x=1}^\ell \C^3$. As per Example (\ref{ex:AKLT, revisited}), the Hamiltonian can be written as
\begin{equation} \label{def:AKLT Hamiltonian (AKLT chap)}
    H_{\Lambda} = \sum_{x=a}^{b-1} P^{(2)}_{x,x+1} . 
\end{equation} Recall that by Clebsch-Gordan (\ref{prop:clebsch-gordan}) two neighboring spin-1's decompose as $V_1\otimes V_1 \cong V_2 \oplus V_1\oplus V_0$. Crucially, this Hamiltonian will turn out to be frustration-free.

Much of the key to this model's exact solvability is the frustration-free property enjoyed by its Hamiltonian.

\begin{definition}\label{def:frustration free}
    An interaction $\Phi:\calP_0(\Gamma) \to \calA_{loc}$ is \textbf{frustration-free} if for any finite volume $\Lambda\subseteq \Gamma$ the finite volume Hamiltonian $H_{\Lambda} = \sum_{X\subseteq \Lambda} \Phi(X)$ and each of the terms $\Phi(X)$ appearing in it have a common eigenvector belonging to their respective smallest eigenvalues.
\end{definition}
Let us clarify this definition by considering a special case, which corresponds to ``shifting'' each interaction term so that they are all $\Phi(X)\geq 0$.

\begin{lemma} \label{lem:frustration freeness for nonnegative operators} Let $A_j \geq 0$, $j=1,\dots, k$ be non-negative operators on a complex Hilbert space $\calH$. Then
\begin{equation}
    \ker \paran{\sum_{j=1}^k A_j} = \bigcap_{j=1}^k \ker (A_j) .
\end{equation}
\end{lemma} 
\begin{proof}
    Let us demonstrate this with two operators $A,B\geq 0$, the extension will be clear.

    First, we show $\ker(A+B)\supseteq \ker A \cap \ker B$. Let $\psi\in \ker(A)\cap \ker(B)$. Then $(A+B)\psi = A\psi + B\psi = 0$.
    
    Now, we show $\ker(A+B)\subseteq \ker A \cap \ker B$. Let $\psi \in \ker(A+B)$, so $(A+B)\psi = 0$. Since $A,B\geq 0$, we have $A + B \geq 0$. We can thus take a square root $(A+B)^{1/2}$, and by looking at the eigenvalue zero and doing functional calculus with the square root function (zero has a unique square root), we see that $\ker(A+B) = \ker(A+B)^{1/2}$. Thus,
    \begin{equation} \begin{split}
        0 = \norm{(A+B)\psi}^2 = \norm{(A+B)^{1/2}\psi}^2 &= \inprod{(A+B)^{1/2}\psi, (A+B)^{1/2}\psi} \\
        &= \inprod{\psi, (A+B)\psi} \\
        &\geq 0
    \end{split} \end{equation} where the second to last line follows by $A+B\geq 0 $. And thus, since this is a complex Hilbert space, $\inprod{\psi,A\psi}=\inprod{\psi,B\psi} = 0$ implies $\psi\in \ker (A)\cap \ker (B)$.
\end{proof}

\begin{remark} \label{rem:frustration free}
    If every interaction term $\Phi(X)\geq 0$, then the Hamiltonian $H_{\Lambda}$ is frustration-free if the intersection is nonempty $\bigcap_{X\subseteq \Lambda} \ker(\Phi(X))\neq \{0\}$.
\end{remark}

Note that this includes as a special case any commuting projector model, like the toric code. Certainly though it contains many models beyond this class, including the star of this chapter the AKLT chain.
The AKLT Hamiltonian will turn out to be frustration-free, once we demonstrate that it has a non-trivial $\ker H_{[1,\ell]}$.

\section{The AKLT Ground State as a Valence Bond State (VBS)} \label{sec:The AKLT Ground State as a Valence Bond State (VBS)}

\subsection{Warm-up: the Majumdar-Ghosh Chain Ground States} \label{sec:Warm-up: the Majumdar-Ghosh Chain Ground States}
Recall from Example \ref{ex:Majumdar-Ghosh, revisited} the Majumdar-Ghosh model, a spin-1/2 chain $\Lambda = [a,b] \subseteq \Z$ whose Hamiltonian can be expressed as 
\begin{equation} \begin{split}\label{ex:MG chain (AKLT chapter)}
    H_{\Lambda} &= \sum_{x=a}^{b-2} 2\vec{\sigma}_x\cdot \vec{\sigma}_{x+1} + \vec{\sigma}_x\cdot \vec{\sigma}_{x+2} \\
    &= (\textit{boundary terms}) +  \sum_{x=a}^{b-2} P_{x,x+1,x+2}^{(3/2)} , 
\end{split}\end{equation} where $P_{x,x+1,x+2}^{(3/2)}$ is the orthogonal projection onto the spin-3/2 subspace of three neighboring spin-1/2 particles $V_{1/2}\otimes V_{1/2}\otimes V_{1/2} \cong V_{3/2}\oplus 2V_{1/2}$. The original paper by Majumdar and Ghosh~\cite{majumdar1969next}, along with the work by Caspers, Emmett, and Magnus~\cite{caspers1984majumdar}, shows that the only finite chain ground states of this model (with open boundary conditions) are given by products of singlets. Recall that a singlet refers to the trivial spin-0 representation $V_0\cong \C$ of $SU(2)$: here, just as in Example \ref{ex:2 qubit tensor rep of su(2)}, the singlet appears as the one dimensional invariant subspace of the two-qubit tensor rep $V_{1/2}\otimes V_{1/2}$, and when we choose the basis $\ket{\up},\ket{\down}$ of $V_{1/2}\cong \C^2$, the singlet rep has 
\begin{equation}
    V_0 = \text{span}\{\ket{\up\down}-\ket{\down\up}\} \subseteq \C^2\otimes \C^2. 
\end{equation} Let us describe these ground states more explicitly, and let $\phi := \ket{\up\down}-\ket{\down\up}$ denote a copy of the singlet. When the chain length is even, say $[a,b] = [1,2\ell]$, we have a 5-dimensional ground state space given by
\begin{equation}\label{eq:MG even l ground states}
    \ker H_{[1,2\ell]} = \paran{\C \underbrace{\phi\otimes \dots \otimes \phi}_{\ell \, \text{copies} } }\cup \paran{\C^2\otimes \underbrace{\phi \otimes \dots \otimes \phi}_{\ell-2 \, \text{copies}} \otimes \C^2} , 
\end{equation} Recall from earlier that the singlet $V_0\cong \C \phi$ acts as the identity in the representation ring of $SU(2)$, which is just a fancy way of saying $W \otimes V_0 \cong W$ for any rep $W$. In particular, the left expression is a singlet $V_0\otimes \dots \otimes V_0 \cong V_0$, and the right expression is a product of two spin-1/2's $V_{1/2}\otimes V_0\otimes \dots \otimes  V_0 \otimes V_{1/2} \cong V_{1/2}\otimes V_{1/2}$, so the ground state space (\ref{eq:MG even l ground states}) is the $SU(2)$ rep given by
\begin{equation} 
    \ker H_{[1,2\ell]} = V_0 \oplus \paran{V_{1/2}\otimes V_{1/2}} . 
\end{equation} We can pictorially represent these states, using dots to denote ``free'' spin-1/2 particles $V_{1/2} = \C^2$ and lines to denote ``singlet bonds'' $V_0 = \C \phi$. Let us write them for $2\ell = 4$. Firstly, the singlet state $V_0$: 
\begin{figure}\centering
\includegraphics{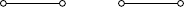}
\end{figure}

Now, the two free spin-1/2 particles $V_{1/2}\otimes V_{1/2}$: 
\begin{figure}\centering
\includegraphics{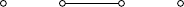}
\end{figure}

Let us also look at the odd length case $\Lambda = [1,2\ell+1]$. In this case we have a 4-dimensional ground state space 
\begin{equation}\label{eq:MG odd l ground states}
    \ker H_{[1,2\ell+1]} = \paran{\C^2\otimes \underbrace{\phi\otimes \dots \otimes \phi}_{\ell \, \text{copies} } }\cup \paran{ \underbrace{\phi \otimes \dots \otimes \phi}_{\ell \, \text{copies}} \otimes \C^2} , 
\end{equation} or in terms of $SU(2)$ reps,
\begin{equation}
    \ker H_{[1,2\ell+1]} = V_{1/2} \oplus V_{1/2} . 
\end{equation} Pictorially, we have for $2\ell+1=3$ sites 

\begin{figure}\centering
    \includegraphics{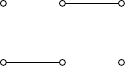}
\end{figure}

While we will not write a full argument proving that these are the exact kernels of the Hamiltonian (\ref{ex:MG chain (AKLT chapter)}), one can reconstruct a proof using two key observations. The first is that the interaction term $P^{(3/2)}$ will punish any vectors in the spin-3/2 representation of 3 neighboring spin-1/2s, leaving only representations of the form $V_0\otimes V_{1/2}$ and $V_{1/2}\otimes V_0$ in the kernel. In words, to be a ground state one must have a singlet supported on three consecutive sites. The second is that this Hamiltonian is frustration-free (see Definition \ref{def:frustration free}), which guarantees that any ground state of $H_\Lambda$ must be a ground state of every term $P_{x,x+1,x+2}^{(3/2)}$. Here, this has the powerful consequence that every consecutive three sites must contain a singlet representation--since we are working with spin-1/2 particles, this almost immediately forces the ground states to have the desired structure.

Notice before proceeding that the dimension of the finite-chain Hamiltonian's ground state space is somewhat misleading. Simply by inspection (and indeed, this is provably the case) it seems that there are only two pure ground states $\omega_\pm$ in the thermodynamic limit given by infinite products of singlet bonds where one is the translate of the other. In essence, the finite chain ground state space contains ``spurious'' ground states which differ only by a choice of boundary condition. Indeed, given any finitely supported observable $A\in \calA_{[1,2\ell+1]}$, one may pick a sufficiently large chain (here, $[-1,2\ell+3]$ will suffice) so that these boundary conditions are inconsequential, e.g. for any $\ket{\psi_+},\ket{\wt{\psi}_+}\in \paran{\C^2\otimes \phi^{\otimes \ell}}$, the expectation values agree 
\[
    \bra{\psi_+} A \ket{\psi_+} = \bra{\wt{\psi}_+} A \ket{\wt{\psi}_+} . 
\] Or in other words, the reduced density matrices of $\ket{\psi_+}$ and $\ket{\wt{\psi}_+}$ are identical. This ambiguity is resolved by taking the thermodynamic limit: in this case, we see only the two ``true'' pure ground states $\omega_\pm$, which are uniquely determined by the property that 
\[
    \omega_\pm (P_{x,x+1,x+2}^{(3/2)}) = 0 , \qquad \text{for all } x\in \Z.
\] The AKLT ground state will exhibit a similar behavior, as we will see in the next section.

\subsection{AKLT Ground State}
We will now present the valence bond state (VBS) construction of the ground state of the AKLT Hamiltonian (\ref{def:AKLT Hamiltonian (AKLT chap)}). Let us begin with an observation. 
Clebsch-Gordan yields the two spin-1/2 tensor decomposition $V_{1/2}\otimes V_{1/2}\cong V_{1}\oplus V_0$. Let $P_{sym}:V_{1/2}\otimes V_{1/2} \to V_1$ denote the orthogonal projection onto the spin-1 subspace. In more words, $P_{sym}$ acts as $\idty$ on spin-1 and $0$ elsewhere. By construction $P_{sym}$ is an intertwiner (\ref{def:equivalent reps}), which means that
\begin{equation}\label{eq:P_sym an intertwiner}
    P_{sym} \Pi^{(1/2)}_g \otimes \Pi^{(1/2)}_g  = \Pi^{(1)}_g P_{sym}, \qquad \text{for all } g\in SU(2) , 
\end{equation} where $\Pi^{(s)}: SU(2)\to GL(\C^{2s+1})$ gives the spin-$s$ representation. This then quickly implies that any tensor power $P_{sym}^{(\otimes \ell)}$ is an intertwiner in the same way.

Now we construct the ``virtual'' Hilbert space. Take a chain of spin-1/2 $V_{1/2}\cong \C^2$ particles of length $2\ell$, labeled as
\begin{equation}
    \calH_{virt} := V_{1/2}^{(L,1)} \otimes V_{1/2}^{(R,1)} \otimes V_{1/2}^{(L,2)} \otimes V_{1/2}^{(R,2)} \otimes \dots \otimes V_{1/2}^{(L,\ell)} \otimes V_{1/2}^{(R,\ell)}. 
\end{equation} Then, similarly to the Majumdar-Ghosh model in Equation (\ref{eq:MG even l ground states}), define the 4 dimensional space of vector states $W$ as singlets $\phi = \ket{\up\down}-\ket{\down\up}$ with ``dangling free spins'' at the boundary:
\begin{equation} 
    W_{virt} := \C^2\otimes \underbrace{\phi \otimes \dots \otimes \phi}_{\ell-2 \, \text{copies}} \otimes \C^2. 
\end{equation} Now, define the subspace $W_{gs}\subseteq \calH$ by projecting:
\begin{equation} \label{def:VBS ground states}
    W_{gs} := P^{\otimes \ell}_{sym} W_{virt} . 
\end{equation} Note that this projection does \textit{not} preserve the simple structure of ``separable product of singlets''. Indeed, it creates a significant amount of entanglement. We claim that this is the ground state space of the AKLT Hamiltonian. The first thing to check is that any vector $\ket{\psi} = P^{\otimes \ell}_{sym} \ket{\psi_{virt}} \in W_{gs}$ is in fact a ground state vector. This is a straightforward calculation using the intertwining relation (\ref{eq:P_sym an intertwiner}). One first sees that the projection onto the spin-2 subspace has $P^{(2)}_{x,x+1} P_{sym}^{(\otimes 2)} = P_{sym}^{(\otimes 2)} P^{(2)}_{(L,x),(R,x),(L,x+1),(R,x+1)}$, the orthogonal projection onto the spin-2 subspace of $V_{1/2}^{\otimes 4}$, then we have
\begin{equation}
    P^{(2)}_{x,x+1} P^{\otimes \ell}_{sym} \ket{\psi_{virt}} = P^{\otimes \ell}_{sym} P^{(2)}_{(L,x),(R,x),(L,x+1),(R,x+1)} \ket{\psi_{virt}}
    = 0 ,
\end{equation} where the second equation follows because the representation $W_{virt}$ decomposes as $W_{virt} \cong V_{1}\oplus V_0$ and contains no spin-2 irreps. So long as $W_{gs}\neq \{0\}$, $H_{\Lambda}$ is frustration-free and so $W_{gs}$ consists of ground states of $H_{\Lambda}$. We will now show that $W_{gs} \cong V_{1}\oplus V_0$, i.e. it is a 4-dimensional rep of $SU(2)$. By the intertwining relation (\ref{eq:P_sym an intertwiner}), $W_{gs}$ is an invariant subspace of $\calH$ since 
\[
    (\Pi_g^{(1)})^{\otimes \ell} W_{gs} = (\Pi_g^{(1)})^{\otimes \ell} P^{\otimes \ell}_{sym} W_{virt} = P^{\otimes \ell} (\Pi_g^{(1/2)}\otimes \Pi_g^{(1/2)})^{\otimes \ell} W_{virt} =  P^{\otimes \ell} W_{virt} = W_{gs}. 
\] Just as in the warm-up in Equation (\ref{eq:MG even l ground states}), $W_{virt}\cong V_0 \oplus V_1$. This is an irrep decomposition, so it suffices to demonstrate that these irreps are not annihilated by the projection, i.e. that there exists $v_0\in V_0$ with $P_{sym}^{\otimes \ell} v_0 \neq 0$ and that there exists $v_1\in V_1$ with $P_{sym}^{\otimes \ell} v_1 \neq 0$. It is clear from our 2 qubit Example (\ref{ex:2 qubit tensor rep of su(2)}) that the following two vectors respect this decomposition (we are essentially creating a singlet bond for the first and a triplet bond for the second):
\begin{equation}\begin{split}
    v_0 &:= \ket{\up} \otimes \phi^{\otimes \ell-2} \otimes\ket{\down} - \ket{\down} \otimes \phi^{\otimes \ell-2} \otimes\ket{\up} \\
    v_1 &:= \ket{\up} \otimes\phi^{\otimes \ell-2}\otimes \ket{\up}. 
\end{split}\end{equation}
We can compute the matrix elements of $P_{sym}$ in the tensor basis of $\ket{\up},\ket{\down}$ by
\begin{equation}
    P_{sym} = \ket{\up\up}\bra{\up\up} + \frac{1}{2}\paran{\ket{\up\down}+\ket{\down\up}}\paran{\bra{\down\up}+\bra{\up\down}} + \ket{\down\down}\bra{\down\down} . 
\end{equation} Observe that on overlapping singlets
\begin{equation} \begin{split}
    \idty \otimes P_{sym} \otimes \idty (\phi\otimes \phi ) &= \idty \otimes P_{sym} \otimes \idty \paran{\ket{\up\down\up\down} - \ket{\down\up\up\down} - \ket{\up\down\down\up} + \ket{\down\up\down\up}} \\
    &= -\ket{\down\up\up\down} - \ket{\up\down\down\up} + \frac{1}{2}\paran{\ket{\up\up\down\down} + \ket{\up\down\up\down} + \ket{\down\down\up\up} + \ket{\down\up\down\up}} \\
    &\neq 0 .
\end{split}\end{equation} One can perform similar calculations to explicitly (albeit tediously) verify that $P_{sym}^{\otimes \ell} v_0 \neq 0$ and $P_{sym}^{\otimes \ell} v_1 \neq 0 $. This completes the construction and proves that as a rep, $W_{gs}\cong V_0\oplus V_1$. 

The VBS construction is often portrayed pictorially, much as we did with the Majumdar-Ghosh ground states. Again, free dots represent spin-1/2 particles $V_{1/2}$, lines represent singlet $V_0$ bonds, and the new red circles represent spin-1 projections. 

\begin{figure}\label{fig:VBS picture}
\centering \includegraphics{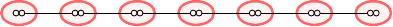}
\captionof{figure}{VBS picture of AKLT chain ground states}
\end{figure}
To recap, there are $2\ell$ ``virtual'' spin-1/2 particles $V_{1/2}$ and $\ell$ ``physical'' spin-1 $V_1$ particles, where each physical particle arises as the image of the projection map $P_{sym}$ applied to a pair of spin-1/2s $V_1^{(x)} = P_{sym} (V_{1/2}^{(L,x)}\otimes V_{1/2}^{(R,x)})$. At the end of it all, there are two free spin-1/2s at either end of the chain, meaning this rep is $V_{1/2}\otimes V_{1/2}\cong V_1\oplus V_0$. In the thermodynamic limit this four dimensional space $W_{gs}$ will coalesce into a single ground state $\omega$, uniquely specified by the property that
\begin{equation}
    \omega(P_{x,x+1}^{(2)}) = 0 \qquad \text{ for all } x\in \Z . 
\end{equation}

Proving uniqueness of a ground state is in general tricky. Given a candidate state space $W_{gs}$, it is generally a straightforward\footnote{Of course, even ``simple'' linear algebra becomes more involved when working with high dimensional vector spaces.} verification to see that $W_{gs} \subseteq \ker H_{[1,\ell]}$ for all chain lengths $\ell$. But to show that this containment is saturated may be very difficult in general. In this case however, the combination of the frustration-free property (\ref{def:frustration free}) and the structure of $W_{gs}$ as an $SU(2)$ rep makes it possible. It is totally possible using the VBS construction to prove uniqueness of the AKLT ground state (indeed, this is the approach in Tasaki's book~\cite{tasaki2020physics}), but the proof is so similar to that corresponding to the MPS construction in the following section that we will delay this. This type of uniqueness result means that the Hamiltonian $H_{\Lambda}$ is a \textit{parent Hamiltonian} for the family of ground states $W_{gs}$, a notion we will spend more time with in Chapter \ref{ch:MPSs}.

\begin{remark}
The VBS construction presented here consists of constructing a virtual Hilbert space of spins, entangling pairs of these spins, and then projecting to the physical Hilbert space. Initially it may not be obvious, but this is the same data as the following matrix product state (MPS) construction. Both are special cases of the more general projected entangled-pair state (PEPS) construction, which builds a wide variety of states in varying spatial dimensions using exactly this procedure. Note also that due to the $SU(2)$ symmetry so integrated in these models, we have repeatedly focused on the rep theory aspects of these states. General VBS, MPS, or PEPS need not have any symmetry whatsoever, only requiring a collection of entangled virtual particles and a projection to the physical Hilbert space. 
\end{remark}

\section{The AKLT Ground State as a Matrix Product State (MPS)} \label{sec:The AKLT Ground State as a Matrix Product State (MPS)}
We will now present the MPS construction of the ground state space for the AKLT chain (\ref{def:AKLT Hamiltonian (AKLT chap)}). We begin again with a Clebsch-Gordan (\ref{prop:clebsch-gordan}) decomposition $V_{1}\otimes V_{1/2} \cong V_{3/2}\oplus V_{1/2}$, which furnishes a unique-up-to-phase isometry $T: V_{1/2}\to V_{1}\otimes V_{1/2}$ which intertwines (see Definition \ref{def:equivalent reps}) these representations, i.e. 
\begin{equation} \label{eq:intertwiner MPS}
    (\Pi_g^{(1)}\otimes \Pi_g^{(1/2)}) T = T \, \Pi_g^{(1/2)} , \qquad \text{for all } g\in SU(2) . 
\end{equation} Just as a reminder, since $T:\C^2\to \C^3\otimes \C^2$, $T$ an isometry means that $T^*T =\idty_2$, where we have made the dimension a bit more explicit. One may use tensor network diagrams to write this relation, where each free leg corresponds to an index in the tensor and each bond between legs corresponds to an index contraction:
\begin{figure}
    \centering
    \includegraphics[width=0.5\textwidth]{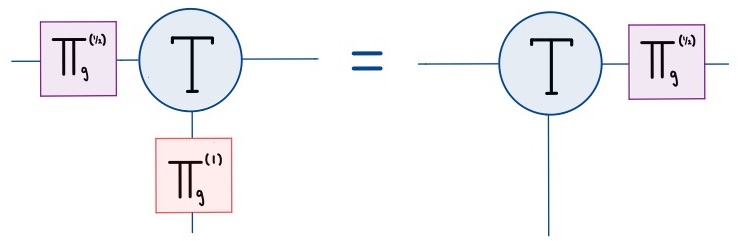}
    \caption{Tensor network diagram of the intertwining relation \ref{eq:intertwiner MPS}.}
    \label{fig:intertwiner}
\end{figure}

Let $\alpha,\beta\in V_{1/2} = \C^2$ and $n\geq 2$ and define $\psi_{\alpha\beta}^{(\ell)}\in \calH_{[1,\ell]}$ by\footnote{MPSs demand piles of linear algebra, which all at once may be overwhelming (as it was to the author). We recommend writing small chains, like a chain of length 2, to understand these sorts of expressions.}
\begin{equation} \label{def:MPS AKLT (isometry)}
    \psi_{\alpha\beta}^{(\ell)} = (\idty_3^{\otimes \ell}\otimes \bra{\beta}) (\underbrace{\idty_3\otimes \dots \otimes \idty_3}_{\ell-1}\otimes T) \dots (\idty\otimes T) T \ket{\alpha} . 
\end{equation} In tensor network form, we may write this state by
\begin{figure}
    \centering
    \includegraphics[width=0.5\textwidth]{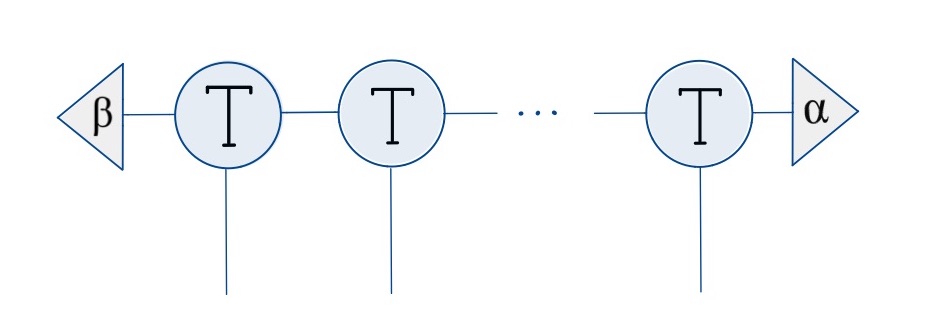}
    \caption{Tensor network diagram of $\psi_{\alpha\beta}^{(\ell)}$ (\ref{def:MPS AKLT (isometry)}).}
    \label{fig:MPS_tensor}
\end{figure}
By using the intertwining property (\ref{eq:intertwiner MPS}) tensored $\ell$ times, one sees that for any $g\in SU(2)$ 
\begin{equation}\begin{split} \label{eq:AKLT SU(2) rep}
    (\Pi_g^{(1)})^{\otimes \ell}\psi_{\alpha\beta}^{(\ell)} &= (\idty_3^{\otimes \ell}\otimes \bra{\Pi_g^{(1/2)}\beta}) (\underbrace{\Pi_g^{(1)}\otimes \dots \otimes \Pi_g^{(1)}}_{\ell} \otimes \Pi_g^{(1/2)} T) \dots  T \ket{\alpha} \\
    &= (\idty_3^{\otimes \ell}\otimes \bra{\Pi_g^{(1/2)}\beta}) (\underbrace{\idty_3\otimes \dots \otimes \idty_3}_{\ell-1} \otimes T) \dots  (\idty_3\otimes T)T \ket{\Pi_g^{(1/2)}\alpha}
\end{split}\end{equation} where in the first line we have used that $\bra{\beta}(\Pi_g^{(1/2)})^* = \bra{\Pi_g^{(1/2)} \beta}$. In tensor network form, using the diagram Figure \ref{fig:intertwiner} corresponding to the intertwining relation (\ref{eq:intertwiner MPS}), we have 
\begin{figure}
    \centering
    \includegraphics[width=1\textwidth]{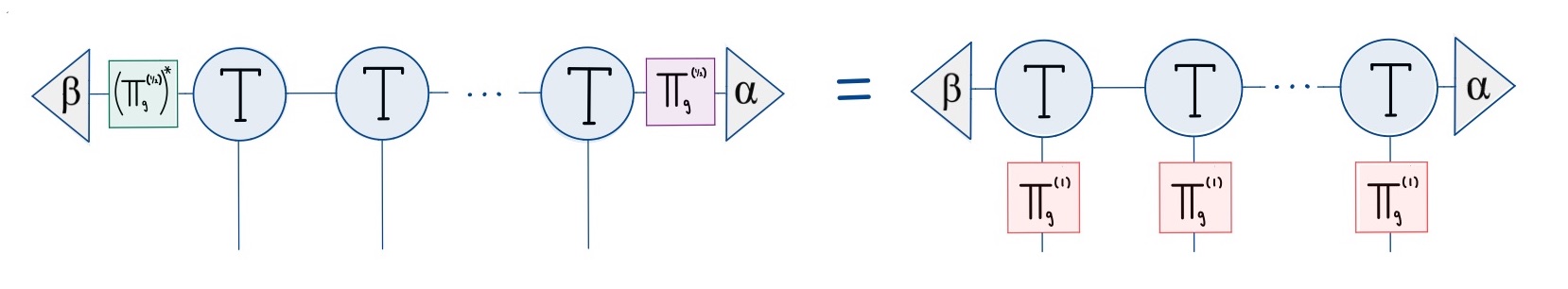}
    \caption{Tensor network version of Equation (\ref{eq:AKLT SU(2) rep}).}
    \label{fig:MPS_bulk_boundary}
\end{figure}

This means that as a rep of $SU(2)$, the vector space $\{\psi_{\alpha\beta}^{(\ell)} : \alpha,\beta\in \C^2\}$ is the rep $V_{1/2}^* \otimes V_{1/2}$. But for $SU(2)$, since dimension $2s+1$ is a unique label for irreps\footnote{This is \textit{not} true in general. E.g. for $SU(n)$ with $n\geq 3$, the dual of the defining rep is not isomorphic to the defining rep.} and since a rep $V$ is irreducible iff the dual rep $V^*$ is irreducible, $V_{1/2}^*\cong V_{1/2}$ and so $V_{1/2}^*\otimes V_{1/2} \cong V_{1}\oplus V_0$. This contains no spin-2 irreps $V_2$, and so by orthogonality of distinct irreps we have for $\ell=2$
\begin{equation}\label{eq:2 site MPS is ground state}
    P^{(2)} \psi_{\alpha \beta}^{(2)} = 0 ,\qquad \text{ for all } \alpha,\beta\in \C^2.
\end{equation} For the same reason we see that if the map $\psi^{(\ell)}: \C^2\otimes \C^2 \to \calH_{[1,\ell]}$ which maps $\psi^{(\ell)}\paran{\ket{\alpha}\bra{\beta}} = \psi^{(\ell)}_{\alpha\beta}$ is injective, then we may conclude that for any $\ell$ and any site $x\in [1,\ell-1]$, $P^{(2)}_{x,x+1} \psi^{(\ell)}_{\alpha\beta} = 0$ and thus show that $\{\psi^{(\ell)}_{\alpha\beta}\} = \ker H_{[1,\ell]}$.

This family of maps $\psi^{(\ell)}$ will in fact be our matrix product states, and in the following argument we will be able to directly show that 
\begin{equation}
    \text{span} \{ \psi_{\alpha\beta}^{(\ell)} \st \alpha,\beta\in \C^2 \} = \ker H_{[1,\ell]} . 
\end{equation} While we already know that $\{\psi_{\alpha\beta}^{(\ell)}\} \subseteq \ker H_{[1,\ell]}$, it will be instructive to demonstrate this again more directly.

Recall from Chapter (\ref{ch:app-rep_theory}) that $\su(2)\cong \so(3)$ and so $SU(2)\cong Spin(3) = e^{\so(3)}$. By our classification of spin-$s$ irreps (\ref{ex:the spin s irrep}), we see that $\Pi^{(1)}$ corresponds to the 3D defining rep of $Spin(3)$, i.e. $\Pi^{(1)}_w = w$ for all $w\in Spin(3)$, and $\Pi^{(1/2)}$ is the 2D spin representation of $Spin(3)$.
\begin{lemma} 
    Pick the standard (orthonormal) basis $\ket{1},\ket{2},\ket{3}$ of $\C^3$ and define $2\times 2$ matrices $t_i$, $i=1,2,3$ as the rescaled Pauli matrices
\begin{equation}
    t_1 := \frac{1}{\sqrt{3}} \sigma^X, \quad t_2 := \frac{1}{\sqrt{3}} \sigma^Y, \quad t_3 := \frac{1}{\sqrt{3}} \sigma^Z. 
\end{equation} 
    Then the isometry $T: \C^2\to \C^3\otimes \C^2$ defined by Equation (\ref{def:MPS AKLT (isometry)}) is exactly given by the matrix\footnote{As always in linear algebra, the matrix depends on the choice of basis. This is especially important to remember when working with tensor states.}
    \begin{equation}
        T = \sum_{i=1}^3 \ket{i} \otimes t_i . 
    \end{equation}
\end{lemma}
    \begin{proof}
        First, we verify the isometry condition $T^*T = \idty$. Using $t_i^*=t_i$ and the Clifford anticommutation relation of the Paulis $\sigma^i \sigma^j + \sigma^j \sigma^i = 2\delta_{ij}$, we see that
        \[
        T^*T = \frac{1}{3} \sum_{i,j=1}^3 \inprod{i,j}\otimes t_i^* t_j = \frac{1}{3} \sum_{i=1}^3 t_i t_i = \idty . 
        \]

        Now, to see that $T$ satisfies the intertwining relation (\ref{def:MPS AKLT (isometry)}), we recall Equation (\ref{eq:intertwiner MPS}) from the discussion of spin representations of $SO(3)$:
        \begin{equation} 
	       T \Pi^{(1/2)}_w = (\Pi^{(1)}_w\otimes \Pi^{(1/2)}_w)T, \qquad w\in Spin(3).
        \end{equation} But since $SU(2)\cong Spin(3)$, this is exactly the same as the intertwining relation from earlier.
\end{proof}

With an explicit matrix in hand, we can now compute matrix elements of $\psi_{\alpha\beta}^{(\ell)}$ in the tensor basis of $\calH_{[1,\ell]}$\footnote{This readily provable trick is almost like ``cyclicity of trace''. To compute these matrix elements, we again recommend working out the $\ell=1,2$ examples.}:
\begin{equation}
    \inprod{i_1,\dots,i_\ell \st \psi_{\alpha\beta}^{(\ell)}} =  \bra{\beta} t_{i_{\ell}} \dots t_{i_1} \ket{\alpha} = \Tr \ket{\alpha}\bra{\beta} t_{i_\ell} \dots t_{i_1} . 
\end{equation} This means that
\begin{equation}
    \psi_{\alpha\beta}^{(\ell)} = \sum_{i_1,\dots,i_\ell} \Tr \paran{\ket{\alpha}\bra{\beta} t_{i_{\ell}} \dots t_{i_1}} \ket{i_1,\dots , i_n}.
\end{equation} By linear extension we get the promised map $\psi^{(\ell)}: M_2(\C) \to \calH_{[1,\ell]}$ : 
\begin{equation} \label{def:MPS AKLT (full matrix form)}
    \psi^{(\ell)}(B) = \sum_{i_1,\dots,i_\ell} \Tr \paran{B t_{i_{\ell}} \dots t_{i_1}} \ket{i_1,\dots , i_n}, \qquad B\in M_2(\C).
\end{equation}
We now define the state space $\calG_\ell\subseteq \calH_{[1,\ell]}$ as 
\begin{equation}
    \calG_{\ell} := \{ \psi^{(\ell)}(B) \st B\in M_2(\C) \}, 
\end{equation} which is of course equal to the earlier space, $\calG_\ell = \text{span}\{ \psi_{\alpha\beta}^{(\ell)} \st \alpha , \beta \in \C^2 \}$.

The next lemma we have already shown in Equation (\ref{eq:2 site MPS is ground state}), but we show this proof to demonstrate an illuminating direct calculation.
\begin{lemma}
    The matrix product states are ground states of $H_{[1,\ell]}$, i.e.
    \begin{equation}
        \calG_{\ell} \subseteq \ker H_{[1,\ell]}.
    \end{equation}
\end{lemma} 
\begin{proof}
    Let $B\in M_2(\C)$ and let $x=1,\dots, \ell-1$ and compute the expectation of $P_{x,x+1}^{(2)}$ in the state $\psi^{(\ell)}(B)$:
    \begin{equation}\begin{split}
    &\inprod{\psi^{(\ell)}(B), P^{(2)}_{x,x+1} \psi^{(\ell)}(B)}\\
    &= \sum_{\substack{i_1,\dots, i_\ell\\j_1, \dots, j_\ell}} \overline{\Tr[Bt_{i_\ell} \dots t_{i_1}]}  \Tr[Bt_{j_\ell}\dots t_{j_1}] \paran{ \bra{i_1,\dots, i_\ell} \idty_{[1,x-1]} \otimes P^{(2)} \otimes \idty_{[x+2,\ell]} \ket{j_1,\dots,j_\ell}} \\
    &= \sum_{\substack{ i_1,\dots,i_{x-1} \\ i_{x+2}, \dots, i_\ell }} \sum_{\substack{i_x,i_{x+1}\\j_x,j_{x+1}}}  \overline{\Tr[Bt_{i_\ell} \dots t_{i_1}]} \Tr[Bt_{j_\ell} \dots t_{j_1}] \bra{i_{x},i_{x+1}} P_{x,x+1}^{(2)} \ket{j_x,j_{x+1}} \\
    &= \sum_{\substack{ i_1,\dots,i_{x-1} \\ i_{x+2}, \dots, i_\ell }} \sum_{\substack{i_x,i_{x+1}\\j_x,j_{x+1}}}  \Bigg( \overline{\Tr[(t_{i_{x-1}}\dots t_{i_1} B t_{i_\ell} \dots t_{i_{x+2}}) t_{i_{x+1}} t_{i_x} }] \Tr[(t_{i_{x-1}}\dots t_{1} B t_{i_\ell} \dots t_{i_{x+2}}) t_{j_{x+1}} t_{j_x} ] \\
    &\qquad\qquad \bra{i_{x},i_{x+1}} P_{x,x+1}^{(2)} \ket{j_x,j_{x+1}} \Bigg) \\
    &= \sum_{\substack{ i_1,\dots,i_{x-1} \\ i_{x+2}, \dots, i_\ell }} \inprod{\psi^{(2)}(t_{i_{x-1}}\dots t_{1} B t_{i_\ell} \dots t_{i_{x+2}}), P^{(2)} \psi^{(2)}(t_{i_{x-1}}\dots t_{1} B t_{i_\ell} \dots t_{i_{x+2}})} . 
\end{split}\end{equation} But we have already computed in (\ref{eq:2 site MPS is ground state}) that for any two site MPS $\psi^{(2)}(C)$, the expectation of $P^{(2)}$ vanishes. Therefore the final line is a sum of many $0$'s, and so 
\[
    \inprod{\psi^{(\ell)}(B), P^{(2)}_{x,x+1} \psi^{(\ell)}(B)} = 0 .
\] frustration-freeness of $H_{\Lambda}$ then completes the proof. 
\end{proof} 
\begin{remark}
    We will later prove the same result for a class of $SO(n)$-invariant MPSs which generalize the AKLT case in Lemma \ref{lem:SO(n) MPS are ground states}.
\end{remark} Note that this entire calculation goes through thanks to the vanishing energy of the length-2 chain. Many matrix product state calculations proceed in a similar manner, where the 2 is replaced by the \textit{injectivity length}, the smallest chain length $\ell$ such that the map $\psi^{(\ell)}(\cdot)$ is injective. Here, $\ell=2$, which we will show soon.

We now proceed to show that this family of MPSs in fact saturates the ground state space, i.e. every ground state of $H_{[1,\ell]}$ is an MPS. This means exactly that $H_{[1,\ell]}$ is a \textit{parent Hamiltonian} of the family of matrix product states $\calG_{\ell}$.

\begin{theorem} \label{thm:AKLT parent Hamiltonian}
    The ground state space is exactly the MPS space, i.e.
    \begin{equation} \label{eq:AKLT parent Hamiltonian}
    \calG_{\ell} = \ker H_{[1,\ell]} . 
    \end{equation}
\end{theorem}
\begin{remark}
    When $\calG_{\ell}$ a family of MPSs satisfying (\ref{eq:AKLT parent Hamiltonian}), we say that $H_{[1,\ell]}$ is a parent Hamiltonian for the MPS space $\calG_{\ell}$.
\end{remark}
\begin{proof}
    We begin by showing that the maps $\psi^{(\ell)}: M_2(\C) \to \calH_{[1,\ell]}$ are injective by induction. 

\textbf{Injectivity for chains of length $\ell=2$.}
    
    First, let $\ell = 2$. Recall the tensor action of the spin-1 representation $\Pi^{(1)}: SU(2) \to \GL(\C^3)$ on the MPS space Equation (\ref{eq:AKLT SU(2) rep}). For $\ell=2$ this means that for all $B\in M_2(\C)$, 
    \begin{equation} \label{eq:2 site MPS intertwining rep}
        (\Pi_g^{(1)} \otimes \Pi_g^{(1)}) \psi^{(2)}(B) = \psi^{(2)}(\Pi_g^{(1/2)} B (\Pi_g^{(1/2)})^* ), \qquad g\in SU(2).
    \end{equation} Either these vectors vanish, or $\psi^{(2)}(\idty)$ is a spin-0 rep and the span of $\psi^{(2)}(\sigma^i)$, $i=X,Y,Z$ is a spin-1 rep. A direct calculation reveals that none of these vectors vanish. Using the anticommutation relations of the Paulis, one has that $\Tr \sigma^i \sigma^j = 2\delta_{ij}$, and so
    \begin{equation}
        \psi^{(2)}(\idty) = \sum_{i_1,i_2 = 1}^3 \Tr \, t_{i_2} t_{i_1} \ket{i_1, i_2} = \frac{2}{3}(\ket{1,1} + \ket{2,2} + \ket{3,3}).
    \end{equation} Similarly for the spin-1 rep, recognizing that $\sigma^Y\sigma^Z = i\sigma^X$, 
    \begin{equation}\begin{split}
        \psi^{(2)}(\sigma^X) = \sum_{i_1,i_2 = 1}^3 \Tr \, \sigma^X t_{i_2} t_{i_1} \ket{i_1, i_2} &= \frac{2}{3}(\Tr \sigma^X\sigma^Y\sigma^Z\ket{3,2} + \Tr\sigma^X\sigma^Z\sigma^Y \ket{2,3}) \\
        &= \frac{2i}{3}(\ket{3,2}-\ket{2,3}) . 
    \end{split}\end{equation} and $\psi^{(2)}(\sigma^Y) = \frac{2i}{3}(\ket{1,3} - \ket{3,1})$ and $\psi^{(2)}(\sigma^Z) = \frac{2i}{3}(\ket{2,1} - \ket{1,2})$. This demonstrates that $\psi^{(2)}$ is injective.

\textbf{Injectivity for arbitrary chains $\ell$ by induction.}

    Assume now that $\psi^{(\ell)}$ injective for $\ell\geq 2$. We wish to show $\psi^{(\ell+1)}$ injective. As an exercise, one can show that the map $\psi^{(\ell)}$ between finite dimensional vector spaces\footnote{If the vector spaces were infinite dimensional, the existence of such a positive constant is stronger than injectivity, i.e. there exist injective maps where no such $c_\ell>0$ exists. For a counterexample, take a basis $\{e_n: n \in \N\}$ of $L^2(\R)$ and define the diagonal linear operator $T e_n = \frac{1}{n} e_n$. } is injective if and only if there exists a constant $c_\ell> 0 $ such that 
    \begin{align}
        \norm{\psi^{(\ell)}(B)}^2 \geq c_\ell \Tr B^* B \qquad \text{for all } B\in M_n(C) . 
    \end{align} Let us now estimate $\norm{\psi^{(\ell+1)}(B)}$:

    \begin{equation}\begin{split}
        \norm{\psi^{(\ell+1)}(B)}^2 &= \norm{ \sum_{i_1,\dots, i_\ell} \Tr[Bt_{i_{\ell+1}}t_{i_\ell}\dots t_{i_1} ]\ket{i_1, \dots, i_\ell, i_{\ell+1}} }^2 \\
        &= \sum_{i_{\ell+1}}\sum_{i_1,\dots, i_\ell} \abs{\Tr B t_{i_{\ell+1}} t_{i_\ell} \dots t_{i_1} }^2 \\
        &= \sum_{i_{\ell+1}} \norm{\psi^{(\ell)}(Bt_{i_{\ell+1}})}^2 \\
        &\geq c_\ell \sum_{i_{\ell+1}} \Tr(B t_{i_{\ell+1}})^* Bt_{i_{\ell+1}} \\
        &= c_\ell \sum_{i_{\ell+1}} \Tr(t_{i_{\ell+1}}t_{i_{\ell+1}}^* B^* B) \\
        &= c_\ell \Tr \left[ \sum_{i_{\ell+1}} t_{i_{\ell+1}} t_{i_{\ell+1}}^* \right] B^*B \\
        &= c_\ell \Tr ([\idty] B^*B) \\
        &= c_\ell \Tr B^*B
    \end{split}\end{equation} where the second to last line follows by computing the bracketed sum with the explicit form of the matrices $t_i$. Since $c_\ell > 0$ by induction hypothesis, we are done and $\psi^{(\ell+1)}$ is injective.

\textbf{A short chain intersection property and $\ker H_{[1,3]} = \calG_3$.}

    We have thus far shown that $\dim(\calG_{\ell}) = 4$ for all $\ell$. Noting that we already have $\calG_3\subseteq \ker H_{[1,3]}$, we will now show an intersection property for short chains
    \begin{equation} \label{eq:short chain intersection property}
        \ker H_{[1,3]} = (\calG_2\otimes \C^3) \cap (\C^3\otimes \calG_2) = \calG_3 . 
    \end{equation} The first equality holds by observing that $H_{[1,2]}, H_{[2,3]}\geq 0$ (this is just $P^{(2)}_{x,x+1}$, the orthogonal projection onto the spin-2 subspace $V_2\subseteq \C^3\otimes \C^3$) and by recalling Lemma (\ref{lem:frustration freeness for nonnegative operators}) which characterizes the frustration-freeness condition for nonnegative operators as non-emptiness of the intersected kernel. The second equality requires us to demonstrate that $\dim \ker H_{[1,3]} \leq 4$, which we can accomplish by doing some representation theory. Clebsch-Gordan decomposition demands that as $\su(2)$ reps, 
    \begin{equation}
        \calG_{2}\otimes \C^3 = (V_0 \oplus V_1) \otimes V_1 \cong V_2 \oplus 2 V_1 \oplus V_0 .
    \end{equation} We wish to show that the kernel contains at most $V_1\oplus V_0$, i.e. we will use orthogonality of distinct irreducibles to show that $V_2\perp \ker H_{[2,3]}$ and one copy of $V_1 \perp \ker H_{[2,3]}$. Let us start with $V_1$. We know automatically that the highest weight vector of $V_0\subseteq \calG_2$ is $\psi^{(2)}(\idty)$. Since the highest weight vector of the standard rep $\C^3$ is $\ket{1}$, we may extract a highest weight vector\footnote{Since there are two copies of $V_1$, there are many highest weight vectors and no uniqueness is guaranteed.} $\xi_1:= \psi^{(2)}(\idty)\otimes \ket{1}$ and check whether it is in the kernel of $P^{(2)}_{2,3}$:
    \begin{equation}
        (\idty \otimes P^{(2)})\xi_1 = (\idty \otimes P^{(2)}) \frac{1}{3}\paran{\ket{1,1,1} + \ket{2,2,1} + \ket{3,3,1}}. 
    \end{equation} To compute this, we need a different form of $P^{(2)}$. Notice that since we have the decomposition $\C^3\otimes \C^3 \cong V_2\oplus V_1 \oplus V_0$ where $V_0 = \C\ket{\phi}$ with  $\ket{\phi}=\frac{1}{\sqrt{3}} \paran{\ket{1,1} + \ket{2,2} + \ket{3,3}}$, $V_2$ consists of symmetric vectors, and $V_1$ consists of antisymmetric vectors, we realize that $P^{(2)}$ may be similarly expressed as $P^{(2)} = \frac{1}{2}\paran{\SWAP - 2\ket{\phi}\bra{\phi}+ \idty }$. We can then directly compute to see that $(\idty\otimes P^{(2)}) \xi_1 \neq 0$, which applying orthogonality is sufficient to show that at least one $V_1$ is orthogonal to $\ker H_{[2,3]}$.
    
    Now for $V_2$. From the right hand side of Equation (\ref{eq:2 site MPS intertwining rep}), we see that $V_1\subseteq \calG_2$ is the adjoint rep (recall Example \ref{ex:adjoint rep of su2}) passed through the injective intertwiner $\psi^{(2)}$. So, the highest weight vector is $\psi^{(2)}(\sigma^X + i\sigma^Y)$, the image of $2\sigma^+$. Since the highest weight vector of the standard rep $\C^3$ is $\ket{1}$, we may extract the highest weight vector $\xi_2:= \psi^{(2)}(\sigma^X + i\sigma^Y)\otimes \ket{1}$ of $V_2\subseteq \calG_{2} \otimes \C^3$ and check whether it is in the kernel of $P^{(2)}_{2,3}$.
    \begin{equation}
        (\idty\otimes P^{(2)}) \xi_2 = (\idty\otimes P^{(2)})\frac{2i}{3}\paran{ \ket{3,2,1} - \ket{2,3,1} + i\ket{1,3,1} - i\ket{3,1,1}}.
    \end{equation}  Thus, we can again directly compute to see that $(\idty\otimes P^{(2)}) \xi_2 \neq 0$, whence orthogonality of distinct irreps assures us that $V_2$ must be orthogonal to $\ker H_{[2,3]}$.

\textbf{The intersection property for arbitrary $\ell$ chains.}
    
    To wrap up this proof, we will extend the short chain intersection property (\ref{eq:short chain intersection property}) to a full intersection property. Split a chain of length $\ell$ into three subchains, $\ell = L+M+R$ where $L,R\geq 0 $ and $M\geq 2$. The intersection property we will show is 
    \begin{equation} \label{eq:intersection property}
        (\calG_{L+M} \otimes \calH_{[1,R]}) \cap (\calH_{[1,L]}\otimes \calG_{M+R}) = \calG_{L + M + R} . 
    \end{equation}
    
    It will be convenient to use multi-indices $\bi = (i_1,\dots, i_L)$,$\bj = (j_1,\dots,j_M)$, $\bk = (k_1,\dots,k_R)$, and to let $t_\bi$ denote the product $t_{i_L} \dots t_{i_1}$, etc. 
    Then, for $\phi\in (\calG_{L+M} \otimes \calH_{[1,R]}) \cap (\calH_{[1,L]} \otimes \calG_{M+R})$, we have $C_\bi, D_\bk \in M_2$ such that
    \begin{equation} \label{eq:expression for intersected MPS}
        \phi = \sum_{\bi} \ket{\bi} \otimes \psi^{(M+R)}(C_\bi) = \sum_{\bk} \psi^{(L + M)} (D_\bk) \otimes \ket{\bk} . 
    \end{equation} We can expand this using our earlier definition of matrix product states (\ref{def:MPS AKLT (full matrix form)}):
    \begin{equation} \begin{split}
        0 = \paran{\sum_{\bi} \ket{\bi}\otimes \paran{\sum_{\bj,\bk} \Tr [C_\bi t_{\bk}t_{\bj} ] \ket{\bj,\bk} }} - \paran{\sum_{\bk} \paran{\sum_{\bi,\bj} \Tr [D_\bk t_{\bj}t_{\bi} ] \ket{\bi,\bj}}  \otimes \ket{\bk}} . 
    \end{split} \end{equation} On the level of coefficients,
    \begin{equation}
        0 = \Tr[C_\bi t_\bk t_\bj] - \Tr[D_\bk t_\bj t_\bi] = \Tr[(C_{\bi} t_{\bk} - t_\bi D_{\bk})t_{\bj}].
    \end{equation} This means that for all $\bi,\bk$, $\psi^{(M)}(C_\bi t_\bk - t_\bi D_\bk) = 0$. By assumption, $M\geq 2$, and so $\psi^{(M)}$ is injective and hence
    \begin{equation}
        C_\bi t_\bk - t_\bi D_\bk = 0 , \quad \text{for all }\bi,\bk.
    \end{equation} Multiplying this relation on the left by $t_\bi^*$ and summing over $\bi$, we find
    \begin{equation}
        \paran{\sum_{\bi} t_\bi^* C_{\bi}} t_\bk = \sum_{\bi} (t_{\bi}^* t_{\bi}) D_\bk .
    \end{equation} The isometry property of $T$ means that $T^*T= \sum_{\bi} t_\bi^*t_\bi = \idty$, and therefore
    \begin{equation}
        D_\bk = B t_\bk, \quad \text{with } B = \sum_{\bi} t_{\bi}^* C_{\bi}.
    \end{equation} Inserting this expression for $D_\bk$ into (\ref{eq:expression for intersected MPS}) gives
    \begin{equation} \begin{split}
        \phi &= \sum_{\bk} \psi^{(L + M)}(D_\bk) \otimes \ket{\bk} \\
        &= \sum_{\bk} \paran{\sum_{\bi,\bj} \Tr[D_\bk t_{\bj}t_{\bi}]\ket{\bi,\bj} }\otimes \ket{\bk} \\
        &= \sum_{\bk} \paran{\sum_{\bi,\bj} \Tr[Bt_\bk t_{\bj}t_{\bi}]\ket{\bi,\bj} }\otimes \ket{\bk} \\
        &= \psi^{(L + M + R)}(B)
    \end{split}\end{equation} 
    which completes the proof of the intersection property (\ref{eq:intersection property}).
    
    We can finally finish the proof of the theorem by combining frustration-freeness al\'a Lemma \ref{lem:frustration freeness for nonnegative operators}, $\ker H_{[x,x+1]} = \calG_2$ from (\ref{eq:short chain intersection property}), and the intersection property (\ref{eq:intersection property}). Noting the conventions $\calH_\emptyset = \C$ and $[a,b] = \emptyset $ if $a>b$, 
    \begin{equation} \begin{split} \label{eq:intersection property AKLT}
        \ker H_{[1,\ell]} &= \bigcap_{x=1}^{\ell-1} \calH_{[1,x-1]} \otimes \ker H_{[x,x+1]} \otimes \calH_{[x+2, \ell]} \\
        &= \bigcap_{x=1}^{\ell-2} \calH_{[1,x-1]} \otimes (\calG_2 \otimes \C^3 \cap \C^3 \otimes \calG_2) \otimes \calH_{[x+3, \ell]} \\
        &= \bigcap_{x=1}^{\ell-2} \calH_{[1,x-1]} \otimes \calG_3 \otimes \calH_{[x+3, \ell]} \\
        &= \bigcap_{x=1}^{\ell-3} \calH_{[1,x-1]} \otimes \calG_4 \otimes \calH_{[x+3, \ell]} \\
        &\vdots \\
        &= \calG_\ell
    \end{split} \end{equation}
    
\end{proof}

\subsubsection{Connecting the VBS and MPS pictures}

As we said before, the VBS and MPS descriptions of the AKLT chain ground states are totally equivalent, and commonly the intersection property (\ref{eq:intersection property}) is ``visualized'' using the VBS picture. Let us say a bit more about the connection between these two pictures. We established in Equations (\ref{eq:intertwiner MPS}) and (\ref{eq:P_sym an intertwiner}) that both the isometry $T: V_{1/2}\to V_1\otimes V_{1/2}$ and the projector $P_{sym}: V_{1/2}\otimes V_{1/2}\to V_1$ are intertwiners of the corresponding $SU(2)$ representations, uniquely determined linear maps up to scalar multiplication by $c\in \C$ thanks to Schur's lemma. Recall that any intertwiner of $G$-representations $B\in \Hom_G(W,Y)$ by definition commutes with the action of $G$, so $gB = Bg$, and thus may be equally considered as an element of a one-dimensional invariant subspace $B\in W^*\otimes Y\cong \Hom_G(W,Y)$. In particular, using that any representation of $SU(2)$ is self-dual and so $W^*\cong W$, we may think of $T$ and $P_{sym}$ as elements in one-dimensional invariant subspaces of $V_{1/2}\otimes V_{1/2}\otimes V_1$. But Clebsch-Gordan tells us that $V_{1/2}\otimes V_{1/2}\otimes V_1 \cong V_{0} \oplus 2V_{1} \oplus V_{2}$. In particular, there is only one one-dimensional invariant subspace $V_0$, and so $P_{sym}$ and $T$ are exactly the same tensor, up to scalar multiplication by $c\in \C$.\footnote{While we have made heavy use of symmetry here, it is not essential to understand the correspondence between MPSs and one-dimensional PEPSs. The key is typically some form of ``tensor rearrangement'' as we are doing here. Tensor network diagrams often come in handy, making such rearrangements more visual.}

We can be a bit more explicit (and if one wanted to, we could even write explicit matrices). Let $\phi = \ket{\up\down} -\ket{\down\up}\in V_{1/2}\otimes V_{1/2}$ denote a singlet vector. Consider the map $\left[(P_{sym} \otimes \idty)((\cdot)\otimes \phi)\right]: V_{1/2}\to V_1\otimes V_{1/2}$. Observe that for all $g\in SU(2)$, 
\begin{equation}\begin{split}
    \left[ (P_{sym} \otimes \idty)((\cdot)\otimes \phi)\right] \Pi_g^{(1/2)} &= (P_{sym} \otimes \idty)((\Pi_g^{(1/2)}(\cdot))\otimes \phi) \\
    &= (P_{sym} \otimes \idty)((\Pi_g^{(1/2)}(\cdot))\otimes (\Pi_g^{(1/2)})^{\otimes 2} \phi) \\
    &= (\Pi_g^{(1)} \otimes \Pi_g^{(1/2)}) \left[ (P_{sym} \otimes \idty)((\cdot)\otimes \phi)\right].
\end{split}\end{equation} where in the second line, we used $(\Pi_g^{(1/2)})^{\otimes 2} \phi = \phi$ since $\phi$ is a singlet, and in the third line, we used the intertwining property (\ref{eq:P_sym an intertwiner}) of $P_{sym}$. In particular, the map $\left[(P_{sym} \otimes \idty)((\cdot)\otimes \phi)\right]$ obeys the same intertwining relation (\ref{eq:intertwiner MPS}) as $T$, and so by Schur's lemma must be a scalar multiple of $T$. 

In particular, starting from the very first MPS expression (\ref{def:MPS AKLT (isometry)}) of $\psi_{\alpha\beta}^{(\ell)}$, given $\alpha,\beta\in V_{1/2} = \C^2$, we see that 
\begin{equation}\begin{split}
    \psi_{\alpha\beta}^{(\ell)} &= (\idty_3^{\otimes \ell}\otimes \bra{\beta}) (\underbrace{\idty_3\otimes \dots \otimes \idty_3}_{\ell-1}\otimes T) \dots (\idty\otimes T) T \ket{\alpha} \\
    &= (\idty_3^{\otimes \ell}\otimes \bra{\beta}) (\underbrace{\idty_3\otimes \dots \otimes \idty_3}_{\ell-1}\otimes T) \dots (\idty\otimes T) (P_{sym}\otimes \idty) (\ket{\alpha}\otimes \phi) \\
    & \vdots \\
    &= \brac{ \underbrace{P_{sym}\otimes \dots \otimes P_{sym}}_\ell \otimes \bra{\beta} } \brac{\ket{\alpha} \otimes \underbrace{\phi \otimes \dots \otimes \phi}_\ell } \\
    &= \brac{\underbrace{P_{sym}\otimes \dots \otimes P_{sym}}_\ell } \brac{\ket{\alpha} \otimes \underbrace{\phi \otimes \dots \otimes \phi}_{\ell-1} \otimes \ket{\wt{\beta}} } \\
    &= VBS
\end{split}\end{equation} where $\wt{\beta}\in V_{1/2}$ satisfies $(\idty_2\otimes \bra{\beta} )\phi = \ket{\wt{\beta}}$ (or in matrices, if $\ket{\beta} = \beta_1 \ket{\up} + \beta_2 \ket{\down}$, then $\ket{\wt{\beta}} = \beta_2\ket{\up} - \beta_1 \ket{\down}$). This is exactly the description of the VBS space $W_{gs}$ we gave in Equation (\ref{def:VBS ground states}).

\section{The AKLT Ground State as a Finitely Correlated State (FCS)} \label{sec:The AKLT Ground State as a Finitely Correlated State (FCS)}

In the previous two sections, we have said little to nothing regarding ground state correlations and the thermodynamic limit. This is where the third picture of \textit{finitely correlated states} (FCSs), equivalent for the case of translation-invariant (or $k$-periodic) ground states to MPS and VBS, comes into play. Note that many wonderful introductions~\cite{cirac2021matrix,zeng2015quantum,pollmann2017symmetry} detailing MPS present the following material without ever mentioning the term FCS. 

Take a local observable $A_1\otimes \dots \otimes A_\ell \in \calA_{[1,\ell]}$, recalling that each onsite algebra is $\calA_{x}\cong M_3(\C)$. Let us compute the expectation of this observable in the MPS $\psi_{\alpha\beta}^{(\ell)}$ given by (\ref{def:MPS AKLT (full matrix form)}):
\begin{equation}\begin{split}
    \bra{\psi_{\alpha\beta}^{(\ell)}}& A_1\otimes \dots \otimes A_\ell \ket{\psi^{(\ell)}_{\alpha\beta}} \\
    &= \sum_{\substack{i_1,\dots,i_\ell\\j_1,\dots j_\ell}} \overline{\bra{\beta} t_{i_\ell} \dots t_{i_1} \ket{\alpha} } (A_1)_{i_1j_1}\dots (A_\ell)_{i_\ell j_\ell}  \bra{\beta} t_{j_\ell} \dots t_{j_1} \ket{\alpha} \\
    &= \sum_{\substack{i_1,\dots,i_\ell\\j_1,\dots j_\ell}}  (A_1)_{i_1j_1}\dots (A_\ell)_{i_\ell j_\ell}  \bra{\alpha} t_{i_1}^* \dots t_{i_\ell}^* \ket{\beta} \bra{\beta} t_{j_\ell} \dots t_{j_1} \ket{\alpha} \\
    &= \sum_{\substack{i_1,\dots,i_{\ell-1}\\j_1,\dots j_{\ell-1}}}  (A_1)_{i_1j_1}\dots (A_{\ell-1})_{i_{\ell-1}j_{\ell-1}}  \\
    &\qquad \times \bra{\alpha } t_{i_1}^* \dots t_{i_{\ell-1}}^* \brac{\sum_{i_\ell,j_\ell} (A_\ell)_{i_\ell j_\ell} t_{i_\ell}^* \ket{\beta}\bra{\beta} t_{j_\ell} } t_{j_{\ell-1}} \dots t_{j_1} \ket{\alpha}. 
\end{split}\end{equation}
Recognising that the bracketed piece of this expression looks like the Kraus representation of a completely positive map,\footnote{For more information on CP maps, see~\cite{nielsen2010quantum,wolf2012quantum}} we define for all $A\in \calA_x \cong M_3(\C)$ the map $\bbE_A : \calB\to \calB$ where $\calB:= M_2(\C)$ given by\footnote{If you aren't used to these sorts of maps, an exercise the author found helpful was to write this map more explicitly using block matrices and treating the tensor product as a Kronecker product.}
\begin{equation} \label{def:AKLT E map}
    \bbE_{A}(B) = \sum_{ij} A_{ij} t_i^* B t_j = T^*(A\otimes B) T ,
\end{equation} where as usual $A_{ij}$ refers to the $i,j^{th}$ entry of the $3\times 3$ matrix $A$.

This allows us to write the expectations of general local observables in a more compact form: 
\begin{equation}\begin{split}
    \bra{\psi_{\alpha\beta}^{(\ell)}}& A_1\otimes \dots \otimes A_\ell \ket{\psi^{(\ell)}_{\alpha\beta}} \\
    &= \sum_{\substack{i_1,\dots,i_{\ell-1}\\j_1,\dots j_{\ell-1}}} (A_1)_{i_1j_1} \dots (A_\ell)_{i_{\ell-1}j_{\ell-1}} \bra{\alpha} t_{i_1}^* \dots t_{i_{\ell-1}}^* [\bbE_{A_\ell}(\ket{\beta}\bra{\beta})] t_{j_{\ell-1}}\dots t_{j_1}\ket{\alpha} \\
    &= \bra{\alpha} \bbE_{A_1} \circ \dots \circ \bbE_{A_\ell} (\ket{\beta}\bra{\beta})\ket{\alpha}. ,
\end{split}\end{equation}

Let us pause for a second. Essentially, what we have done is realized that every expectation of a local observable (and by continuity quasi-local observable) in the MPS $\psi$ may be computed given the following data:
\begin{itemize}
    \item An algebra of on-site observables $\calA_x$ and a so-called bond algebra $\calB$. Here, $\calB = M_2(\C)$, so we say the \textit{bond dimension} of this MPS is 2.
    \item A completely positive map $\bbE:\calA_x\otimes \calB \to \calB$.\footnote{In some texts, this is called the ``double tensor'' or ``generalized transfer matrix''. }
    \item A positive element $\ket{\beta}\bra{\beta} \in \calB$ and a positive, normalized functional $\Tr\ket{\alpha}\bra{\alpha}(\cdot) \in \calB^*$ (a state on $\calB^*$ if you prefer).
\end{itemize} This is the data of a finitely correlated state~\cite{fannes1992finitely}, modulo some normalization and compatibility conditions that our current choices of $\ket{\beta}\bra{\beta}$ and $\Tr\ket{\alpha}\bra{\alpha}(\cdot)$ do not yet satisfy.

Returning to the task at hand, it is now straightforward to calculate expectations in the thermodynamic limiting states of our MPSs. 
Adding $L$ sites to the left and $R$ to the right of the interval $[1,\ell]$ gives the following expression for the expectation of $A = \idty^{\otimes L} \otimes A_1 \otimes \dots \otimes A_\ell \otimes \idty^{\otimes R} \in \calA_{[-L+1, \ell+R]}$ in the vector state $\psi_{\alpha\beta}^{(L+\ell+R)}$:

\begin{equation}\begin{split}
\bra{\psi_{\alpha\beta}^{(L+\ell+R)}} A \ket{\psi_{\alpha\beta}^{(L+\ell+R)}}
    &= \Tr \ket{\alpha}\bra{\alpha} \bbE_{\idty}^L \circ \bbE_{A_1} \circ \dots \circ \bbE_{A_n} \circ \bbE_{\idty}^R(\ket{\beta}\bra{\beta})  \\
    &= \Tr(\bbE_\idty^T)^L (\ket{\alpha}\bra{\alpha}) \circ \bbE_{A_1} \circ \dots \circ \bbE_{A_n} \circ \bbE_{\idty}^R(\ket{\beta}\bra{\beta}),
\end{split} \end{equation}
where $(\bbE_\idty^T)$ denotes the adjoint of $\bbE_\idty$ with respect to the Hilbert-Schmidt inner product $\inprod{B,C} = \Tr B^* C$ on $\calB$. The map $\bbE_\idty$ is called the \textit{transfer operator}, and its spectral properties control the long chain limits $\lim_{L\to \infty}, \lim_{R\to \infty}$. We thus proceed to diagonalizing $\bbE_{\idty}$ and $\bbE_{\idty}^T$. Our present job is made slightly easier by observing that this particular transfer operator is self-adjoint, $\bbE_{\idty} = \bbE_{\idty}^T$, since for all $B,C\in\calB$, we have 
\begin{equation}
    \inprod{B,\bbE_{\idty}(C)} = \sum_{i = 1}^3 \Tr B^* t_i^* C t_i = \sum_{i = 1}^3 \Tr (t_i^* B t_i)^*  C = \inprod{\bbE_{\idty}(B),C} ,
\end{equation} where for the final equality we have used that each rescaled Pauli is Hermitian $t_{i}^* = t_i$. Time to diagonalize.

\begin{lemma}
    The transfer operator $\bbE_{\idty}:\calB\to \calB$ admits the following diagonalization:
    \begin{equation}
        \bbE_{\idty}(\idty) = \idty,\qquad  \bbE_{\idty}(\sigma^i) = -\frac{1}{3} \sigma^i ,\quad  i=X,Y,Z.
    \end{equation} 
\end{lemma}\begin{proof}
    This calculation can be fully performed just by recalling the Pauli anticommutation relations $\sigma^i\sigma^j + \sigma^j\sigma^i = 2\delta_{ij}\idty$ and that $\sigma_i=\sigma_i^*$. Observe:
    \begin{equation}
        \bbE_{\idty}(\idty) = \sum_{i=1}^3 t_i^* t_i = \frac{1}{3} (\idty +\idty+\idty) = \idty,
    \end{equation} and 
    \begin{equation}
        \bbE_{\idty}(\sigma^X) = \frac{1}{3}( \sigma^X \sigma^X \sigma^X + \sigma^Y \sigma^X \sigma^Y + \sigma^Z \sigma^X \sigma^Z) = \frac{1}{3}\sigma^X (1-1-1) = -\frac{1}{3} \sigma^X , 
    \end{equation} and similarly for $\sigma^Y,\sigma^Z$. Since $\{\idty, \sigma^X,\sigma^Y,\sigma^Z\}$ forms a basis for $\calB \cong M_2(\C)$, we are done.
\end{proof} 
This CP map $\bbE$ is \textit{primitive}, which means that it has 1 as a simple eigenvalue and all other eigenvalues have $\abs{\lambda}<1$. In particular, $\lim_{p\to\infty} \bbE_{\idty}^p = P_{\idty}$, the rank-1 orthogonal projection onto $\idty$,\footnote{If this calculation is unfamiliar, it may help to write $\bbE_{\idty}$ as a diagonal matrix in the Pauli basis.} defined by $P_{\idty}:\calB\to \calB$ as 
\begin{equation} \label{def:rank 1 orth proj onto idty}
    P_\idty(B) = \frac{1}{2} (\Tr B) \idty , \qquad B\in \calB . 
\end{equation} Primitivity is the key condition to get a single well-defined thermodynamically limiting state. But let us not rush through this. 

Since the (rescaled) Pauli basis of $\calB$ is orthonormal, we may write for any $B\in \calB$
\begin{equation}
    B = \frac{1}{2} (\Tr B) \idty + \frac{1}{2}\sum_{i=1}^3 (\Tr \sigma^i B) \sigma^i , 
\end{equation} and so from our diagonalization we can write
\begin{equation}
    \bbE_{\idty}(B) = \frac{1}{2} (\Tr B) \idty - \frac{1}{3}\paran{B- \frac{1}{2}(\Tr B) \idty }, 
\end{equation} so that 
\begin{equation}\label{eq:powers of AKLT transfer operator}
    \bbE_{\idty}^p (\ket{\beta}\bra{\beta}) = \frac{1}{2}\norm{\beta}^2 \idty + \paran{-\frac{1}{3}}^p \paran{\ket{\beta}\bra{\beta} - \frac{1}{2} \norm{\beta}^2\idty} . 
\end{equation} This implies that
\begin{equation}\begin{split} \label{def:AKLT omega}
    \lim_{L\to \infty,R\to \infty} &\inprod{\psi_{\alpha\beta}^{(L+\ell+R)},  A \psi_{\alpha\beta}^{(L+\ell+R)}} \\ &= \lim_{L\to \infty,R\to \infty} \Tr(\bbE_\idty^T)^L (\ket{\alpha}\bra{\alpha}) \circ \bbE_{A_1} \circ \dots \circ \bbE_{A_\ell} \circ \bbE_{\idty}^R(\ket{\beta}\bra{\beta}) \\
    &= \frac{\norm{\alpha}^2 \norm{\beta}^2}{2} \Tr (\frac{1}{2}\idty) \bbE_{A_1}\circ \dots \circ \bbE_{A_\ell}(\idty) \\
    &=: \omega( A_1\otimes \dots \otimes A_\ell ), 
\end{split}\end{equation} where $\omega$ is a translation invariant pure state on $\calA_{\Z}$ uniquely determined by the above expression for simple tensor observables. We say $\omega$ is the finitely correlated state defined by the CP map $\bbE:\calA_x\otimes \calB \to \calB$, the positive linear functional $\rho(\cdot) := \Tr \frac{1}{2}\idty (\cdot)\in \calB^*$, and the positive element $e := \idty\in \calB$. These satisfy the normalization condition $\rho(e) = 1$ and the compatibility conditions $\bbE_{\idty}(e) = e$ and $\rho(\bbE_\idty(B)) = \rho(B)$ for all $B\in\calB$. 

To calculate 2-point correlation functions in this state, we observe that Equation (\ref{eq:powers of AKLT transfer operator}) implies that powers of the transfer operator converge exponentially fast\footnote{Indeed, this rate of convergence is controlled by the second largest eigenvalue $\lambda$ of $\bbE_{\idty}$.} to the rank-1 orthogonal projection onto $\idty$:
\begin{equation}
    \norm{\bbE_{\idty}^p - P_\idty} \leq \frac{2}{3^p} . 
\end{equation} Then, setting $A_2=\dots = A_{\ell-1} = \idty$, we have
\begin{equation}\begin{split}
    \abs{\omega(A_1\otimes \idty \otimes \dots \otimes \idty \otimes A_\ell) - \omega(A_1)\omega(A_\ell)} &= \abs{\frac{1}{2} \bbE_{A_1}\circ (\bbE_\idty^{(\ell-2)} - P_{\idty})\circ \bbE_{A_\ell}} \\
    &\leq \norm{A_1}\norm{A_\ell} \frac{C}{3^\ell} . 
\end{split}\end{equation} This confirms that the ground state $\omega$ has exponential decay of correlations. The final thing for us to prove is that $\omega$ is the unique zero-energy ground state of the AKLT chain. 

\begin{theorem}
    The finitely correlated state $\omega$ defined by (\ref{def:AKLT omega}) is the unique state on $\calA_\Z$ such that $\omega(P_{x,x+1}^{(2)}) = 0$ for all $x\in \Z$.
\end{theorem}
\begin{proof}
    Since local observables are dense in $\calA_\Z$, any state $\eta$ is uniquely determined by its restrictions to the subalgebras $\calA_{[a,b]}, a<b$.
    Let $\rho_{[a,b]}$ denote the density matrices of $\eta$ restricted to $\calA_{[a,b]}$. Notice that $\eta(P_{x,x+1}^{(2)}) = 0,$ for $x=a,\dots,b-1$. This tells us that $\text{ran}\, \rho_{[a,b]} \subseteq \calG_{b-a+1}$: to see why, recall that any density matrix $\rho$ is expressible as a convex combination of pure states $\ket{\psi_i}\bra{\psi_i}, \psi_i\in \calH$: i.e., $0\leq c_i\leq 1$ and 
    $\rho = \sum_{i} c_i \ket{\psi_i}\bra{\psi_i}$ with $\sum_i c_i = 1$.
    So every vector $\ket{\psi_i}$ in the range of our operator $\rho_{[a,b]}$ has the property that it has zero expectation for the operator $P^{(2)}$, which means it is in $\ker H_{[a,b]}$. By Theorem \ref{eq:AKLT parent Hamiltonian}), $\calG_{[a,b]} = \ker H_{[a,b]}$.

    We then have for all chains $t_1<b_1$ that
    \begin{align*}
        \eta(A_{t_1} &\otimes \dots \otimes A_{b_1})\\
        &= \lim_{\substack{a\to -\infty \\ b\to\infty}} \Tr \rho_{[a,b]} \idty_{a, t_1- 1} \otimes A_{t_1} \otimes \dots \otimes A_{b_1} \otimes \idty_{b_1 + 1, b} \\
        &= \omega(A_{t_1} \otimes \dots \otimes A_{b_1}),
    \end{align*} which is what we wanted to show. 
\end{proof}
\begin{remark}
    Warning! The dimension of the finite-chain ground state space $\calG_{\ell}$ is still $4$ for all $\ell\geq 2$. But there is only one thermodynamic limiting state (i.e. a single weak-$*$ accumulation point) $\omega$, hence ``uniqueness of the ground state''. This sort of subtlety is why we need to be careful with the thermodynamic limit. Indeed, there are also cases where one has a unique ground state for each finite lattice but the thermodynamic limit has multiple limiting states. This happens for instance in the dimerized models appearing at the B' point in the yellow phase of Figure (\ref{fig:bilin_biquad}): it was proven e.g. in ~\cite{bjornberg2021dimerization} that each finite chain ground state is unique, but there are at least two thermodynamic limiting states. The differentiating feature, visualized in figure 3 of this same paper, is that when the chain is of the form $[-\ell,\ell]$ and $\ell$ even, the site $x=0$ is more entangled with its $x=-1$ neighbor, while when $\ell$ odd, $x=0$ is more entangled with its $x=1$ neighbor (compare to the Majumdar-Ghosh ground states from earlier).
\end{remark}

We have now shown that the AKLT ground state $\omega$ possesses the first two out of the three promised properties of the Haldane phase: 
\begin{itemize}
    \item The ground state is unique.
    \item The ground state correlation function decays exponentially.
    \item There is a spectral gap above the ground state energy.
\end{itemize} The AKLT chain ground state indeed satisfies this final property, but we will not demonstrate this in this thesis. It is an extremely important result for several reasons, not the least of which being that it is the first descriptor in the term ``gapped ground state phase''. There is a wealth of proof strategies for showing spectral gaps, particularly for frustration-free models like the AKLT chain. We direct the interested reader to the lecture notes upon which this chapter is largely based ~\cite{nachtergaele2016quantum}, as well as the introduction ~\cite{young2023quantum} and the review ~\cite{nachtergaele2019quasi}.
\chapter{Matrix Product States}\label{ch:MPSs}

In the previous chapter we spent a great deal of time studying the ground state of the AKLT chain. Our primary interest in the AKLT chain thus far has been that it is an exactly solvable model in the Haldane phase. Exactly solvable models are useful for probing quantum phases in the same way that concrete examples are useful for probing mathematical theory: we may squeeze a great deal of information out of them, which in turn supplies fuel for conjectures, definitions, and proofs for the general case. The AKLT chain is the prototypical example of a nontrivial MPS. As it turns out, MPSs are wonderful quantum states for variational procedures: by varying the bond dimension $D$ of a MPS, one may construct quantum states with varying levels of cross-chain entanglement, allowing one to explore the Hilbert space of states on a chain while retaining a clear structure. A famous result of Hastings~\cite{hastings2007area} demonstrates that 1D gapped ground states may be well-approximated by matrix product states, and more precise estimates of the requisite bond dimension have since been calculated (e.g. Corollary 1 of ~\cite{huang2015computing}). MPSs would later pave the way to projected entangled pair states (PEPS) and other tensor network states, yielding a thriving subfield of research across quantum many-body physics and quantum computation. See the excellent review~\cite{cirac2021matrix} for more information on MPS, PEPS, and tensor network states. 

We will need only a few essentials from the theory of MPSs.
In what follows, we denote the bond dimension $D$, so the bond algebra is $\calB\cong M_D(\C)$, and the physical dimension $n$, so on-site Hilbert space at $x\in \Z$ is $\calH_x\cong \C^n$ and the on-site observable algebra is $\calA_x\cong M_n(\C)$. For the AKLT chain, $D=2$ and $n=3$. Let $T:\C^D\to \C^n\otimes \C^D$ be a tensor, and pick an orthonormal basis $\ket{i}, i=1,\dots , n$ of $\calH_x$ so we may write
\begin{equation}
    T = \sum_{i=1}^n \ket{i}\otimes t_i . 
\end{equation} Then the matrix product state $\psi_\ell : \calB\to (\C^n)^{\otimes \ell}$ defined by $T$ is given by
\begin{equation}
    \psi_\ell (B) = \sum_{i_1,\dots, i_\ell} \Tr (B t_{i_\ell} \dots t_{i_1} ) \ket{i_1,\dots , i_\ell}, \qquad B\in\calB . 
\end{equation}

For any matrix product state defined by a tensor $T:\C^D\to \C^n\otimes \C^D$, we have a naturally associated completely positive map $\bbE: \calA_x \otimes \calB \to \calB$, defined for $A\in \calA_x, B\in \calB$ by
\begin{equation}
    \bbE_A(B) := T^* (A\otimes B) T = \sum_{i,j=1}^n A_{ij} \, t_i^* B t_j.
\end{equation}

When we passed from MPS to FCS picture for the AKLT chain, we diagonalized the transfer operator $\bbE_\idty$, and its spectral properties controlled the thermodynamic limiting states of the MPS.

\begin{theorem}(Primitivity) ~\cite{wolf2012quantum,ogata2020classification} \label{thm:Primitivity}
For $T:\C^D\to \C^n\otimes \C^D$ and $\calB\cong M_D(\C)$, consider the transfer operator $\bbE_\idty:\calB \to \calB$, the completely positive map given by
\begin{equation}
    \bbE_\idty(B) = \sum_{i,j=1}^n t_i^* B t_j , \qquad B\in \calB . 
\end{equation} Assume the spectral radius of $\bbE_\idty$ is 1. Then the following properties are equivalent.
\begin{enumerate}
    \item There exists a unique faithful state $\rho:\calB\to \C$ and a strictly positive element $e\in \calB$ satisfying
    \begin{equation}
        \lim_{\ell\to \infty} \bbE_\idty^\ell (B) = \rho(B) e, \qquad B\in \calB . 
    \end{equation}
    \item The spectrum $\spec(\bbE_\idty)$ of $\bbE$ has 1 as a simple eigenvalue and all other eigenvalues $\lambda$ have $\abs{\lambda}<1$. There exists a faithful $\bbE_{\idty}$-invariant state $\rho$ and a strictly positive $\bbE_{\idty}$-invariant element $e\in \calB$.
    \item There exists an $\ell \in \N$ such that 
    \begin{equation}
        \calB = \text{span}\{ t_{i_1} t_{i_2}\dots t_{i_\ell} : 1\leq i_k\leq n \text{ for each } k=1,\dots,\ell \} . 
    \end{equation}
\end{enumerate}
\end{theorem}

This theorem essentially grants the condition which ensures that translation-invariant MPSs and FCS are the same, just as we saw with the AKLT chain.

\begin{definition} \label{def:primitive MPS}
    A transfer operator $\bbE_\idty$ satisfying the equivalent conditions in Theorem \ref{thm:Primitivity} is called \emph{primitive}. In this case, we similarly call the MPS associated to the tensor $T:\C^D\to \C^n\otimes \C^D$ \emph{primitive}.
    
    Given a primitive $T$, we may use the state $\rho:\calB\to \C$ and positive element $e\in \calB$ to construct a pure translation-invariant MPS (FCS) $\omega:\calA\to \C$ 
    \begin{equation}
        \omega( A_1 \otimes \dots \otimes A_\ell ) = \rho( \bbE_{A_1} \circ \dots \circ \bbE_{A_\ell} (e)).
    \end{equation} Purity of this state is ensured by Proposition 3.1 of~\cite{fannes1992finitely}.
\end{definition}

Note in particular that primitive MPSs correspond to pure FCSs.

The following theorem has various forms: we cite the review ~\cite{cirac2021matrix}, but we will need a small detour to this formulation's proof in ~\cite{cirac2017matrix}.

\begin{theorem} (Fundamental Theorem of MPS)~\cite{cirac2017matrix, cirac2021matrix} \label{thm:Fundamental Theorem of MPS} Suppose that two primitive MPSs defined by $T:\C^D\to \C^n\otimes \C^D$, $\wt{T}:\C^{\wt{D}}\to \C^n\otimes \C^{\wt{D}}$ give the same finitely correlated state $\omega_T = \omega_{\wt{T}}$. Then $D=\wt{D}$, and there is a phase $\lambda \in U(1)$ and a unitary matrix $\Pi\in M_{D}(\C)$ such that
\begin{equation}
    \wt{t}_i = \lambda \, \Pi \,  t_i \, \Pi^*, \qquad i=1,\dots,n ,
\end{equation} where we mean $T = (t_1,\dots, t_i, \dots, t_n)$ with each $t_i\in M_D(\C)$. 
    
\end{theorem} We will actually need a bit more than this. In particular, we will be considering the case where $\wt{T}$ is given as the image of $T$ under the action of a Lie group $G$. Namely, let $G\subseteq \GL(\C^n)$ be a matrix Lie group (we will later take it to be the image of a representation). We would like a result that shows that the phase $\lambda(g)$ and unitary matrix $\Pi(g)$, thought of as functions of $g\in G$, may be taken to continuously depend on $g$. To do this, we must prove a slightly stronger version of Lemma A.2 from ~\cite{cirac2017matrix}, 
\begin{remark} \label{rem:continuity of fundamental theorem of MPS}
    Let $G\subseteq \GL(\C^n)$ a matrix Lie group, and let $\wt{t}_i = \sum_{j=1}^n g_{ij} t_j$. We may choose the phase $\lambda:\calN\to U(1)$ and unitary matrix $\Pi:\calN \to M_D(\C)$ to be continuous functions in an open neighborhood of the identity $\calN\subseteq G$.
\end{remark}
\begin{proof}
    The proof proceeds essentially identically to that of Lemma A.2 in ~\cite{cirac2017matrix}, we will just prove that we can extract this corollary from it. Let $g\in G$ and consider the family of linear maps $\rho_g:M_D(\C)\to M_D(\C)$\footnote{This $\rho_g$ is a special case of what is often called the ``mixed transfer operator''.} given by 
    \begin{equation}
        \rho_g (\cdot):= \sum_{i=1}^n t_i^* (\cdot) \wt{t}_i = \sum_{i=1}^n t_i^* (\cdot) \paran{\sum_{j=1}^n g_{ij} t_j}. 
    \end{equation} After tracing the argument to equation A.5, we are guaranteed the existence of a family of unitaries $\Pi_g\in M_D(\C)$ which are eigenvectors for $\rho_g$, i.e.
    \begin{equation}
        \rho_g(\Pi_g) = \lambda_g\Pi_g . 
    \end{equation} Notice that $\rho_\idty = \bbE_{\idty}$, the transfer operator associated to the MPS tensor $T$. In particular, $\rho_\idty$ has simple eigenvalue 1 and all other eigenvalues have modulus $\abs{c}<1$ by the primitivity assumption. Then, since $\rho_g$ is a continuous function of the matrix elements of $g$, its eigenvector $\Pi_{(\cdot)}:\calN \to M_D(\C)$ may be treated as a continuous function of $g$ on some open neighborhood of the identity $\calN$. Then, proceeding through the proof in~\cite{cirac2017matrix}, we are guaranteed that $\abs{\lambda_g} = 1$ and $\lambda_g$ satisfies equation (A.6):
    \begin{equation}
        \frac{1}{\lambda_g} t_i = \Pi_g \sum_{j=1}^n g_{ij} t_j \Pi_g^*. 
    \end{equation} We may then use that $T$ is an isometry and so $\sum_{i=1}^n t_i^* t_i = \idty$, giving
    \begin{equation}
    \frac{1}{\lambda_g} = \frac{1}{\lambda_g} \sum_{i=1}^n t_i^* t_i = \sum_{i,j=1}^n g_{ij} t_i^* \Pi_g  t_j \Pi_g^* . 
    \end{equation} In particular, the right hand side consists of smooth functions of $g$, so $\lambda_{(\cdot)}:\calN \to U(1)$ is a smooth function.
\end{proof}

\subsection{Periodic and Ergodic Decompositions}
We often work with mixed states by finding their pure state decomposition. Given a FCS $\omega$ generated by $(\bbE,\rho, e)$, it is natural to wonder whether it decomposes into pure FCSs. Thanks to the work~\cite{fannes1992finitely}, the answer is yes, but the translation invariance condition must be gently relaxed to allow for ``periodic'' FCSs. This situation will occur in Chapter \ref{ch:SO(n)_Haldane_chains}, where we will have a translation-invariant state $\omega$ which is the equal weight superposition of two 2-periodic states $\omega = \frac{1}{2}(\omega_+ + \omega_-)$ which are translates of one another.
Note that we may think of any p-periodic state $\eta$ as a translation-invariant state by ``blocking'' pairs of neighboring spins: i.e. combining $p$ local spins on sites $[x,x+1,\dots, x+p-1]\subseteq \Z$ into one new on-site Hilbert space $\C^{n^p}\cong (\C^n)^{\otimes p}$. If a $p$-periodic state $\eta$ is a FCS after this blocking procedure, we call $\eta$ a \textbf{$p$-periodic FCS}. 

We say that $\omega$ is an \textbf{ergodic} FCS if it is extremal in the convex set of translation invariant states. Proposition 3.1 of~\cite{fannes1992finitely} ensures that this equivalent to saying $\omega$ is generated by $(\bbE,\rho,e)$ such that $e=\idty$ is the only eigenvector of $\bbE_\idty$ of eigenvalue one. The normalization condition $\bbE_\idty(e) = e$ combined with complete positivity forces $\norm{\bbE_\idty}\leq 1$. When there are no other eigenvalues $\abs{\lambda}=1$, then $\bbE_\idty$ is primitive and so $\omega$ is pure. Consider the alternative case where there are other eigenvalues $\abs{\lambda} = 1$. One may show that these eigenvalues must be $p^{th}$ roots of unity for some $p\in \N$, i.e. $\lambda^p=1$, and in this case $\omega = \frac{1}{p}(\omega_1+ \omega_2 +\dots + \omega_p)$ where each $\omega_i$ is a distinct pure $p$-periodic FCS. 

\begin{proposition} (Section 3~\cite{fannes1992finitely}) \label{prop:ergodic and periodic FCS decompositions}
\begin{enumerate}
    \item Every FCS $\omega$ generated by $(\bbE,\rho,e)$ has a unique convex decomposition into a finite sum of ergodic FCSs.
    \item Every ergodic FCS $\omega$ may be uniquely written as the average of $p$ $p$-periodic pure FCSs, which are translates of each other.
\end{enumerate}
\end{proposition}

If the reader desires a treatment of $p$-periodic MPSs which avoids blocking (as blocking partially forgets local entanglement structure), we direct them to~\cite{de2017irreducible}.

%-------------------------------------------
\section{The Parent Property}\label{sec:The Parent Property}

We encountered a frustration-free Hamiltonian back when we considered the AKLT chain Hamiltonian. Recall Definition \ref{def:frustration free} that an interaction $\Phi:\calP_0(\Gamma) \to \calA_{loc}$ is frustration-free if for any finite volume $\Lambda\subseteq \Gamma$ the finite volume Hamiltonian $H_{\Lambda} = \sum_{X\subseteq \Lambda} \Phi(X)$ and each of the terms $\Phi(X)$ appearing in it have a common eigenvector belonging to their respective smallest eigenvalues. When each of these terms $\Phi(X)\geq 0$, Remark \ref{rem:frustration free} gave the equivalent characterization that the interaction is frustration-free if for any finite $\Lambda\subseteq \Gamma$ we have
\begin{equation}
    \bigcap_{X\subseteq \Lambda} \ker (\Phi(X)) \neq \{0\} . 
\end{equation}

The ground state spaces $\calG_{\Lambda}$ of a frustration-free Hamiltonian $H_\Lambda$ then satisfy the \textbf{intersection property}. We say a sequence of subspaces $\{\calG_\ell\}_{\ell \in\N} \subseteq \calH_{[x,x+\ell-1]}$ for $x\in \N$, satisfies the intersection property if there is an $m\in \N$ such that
\begin{equation}\label{eq:intersection property (general case)}
    \calG_\ell = \bigcap_{x=1}^{\ell-m} \calH_{[1,x-1]} \otimes \calG_m \otimes \calH_{[x+m,\ell]} ,
\end{equation} holds for all $\ell \geq m$, with the convention that $\calH_\empty = \C$ and $[a,b]=\empty$ if $a>b$.

In the argument culminating in Equation (\ref{eq:intersection property AKLT}), we proved that the MPS spaces of the AKLT chain 
\begin{equation} \label{def:MPS space (general case)}
    \calG_{\ell} = \text{span}\{ \psi_\ell(B): B\in \calB\}
\end{equation} enjoyed the intersection property with $m=2$. Indeed, this same proof goes through for arbitrary primitive MPSs for a sufficiently large $m$. In particular, we need $m$ to be large enough so that the map $\psi_m$ is injective--appropriately, the smallest such $m$ is called the \textbf{injectivity length} of the MPS. One can show that a primitive MPS always has a finite length. We also proved in Theorem \ref{thm:AKLT parent Hamiltonian} that these are the only ground states of the AKLT interaction. 

\begin{definition}\label{def:parent hamiltonian}
Let $\Phi$ be a frustration-free interaction which is nonnegative, so $\Phi(X)\geq 0 $ for all $X\in P_0(\Z)$, and defines a gapped Hamiltonian. Let $\calG_\ell$ be given by MPSs $\psi_\ell$ as in (\ref{def:MPS space (general case)}). If for any finite volume $[1,\ell] \subseteq \Gamma$ we have
\begin{equation}
    \ker H_{[1,\ell]} = \calG_\ell , 
\end{equation} then we say $H_{[1,\ell]}$ is a \emph{parent Hamiltonian} for the MPS space $\calG_\ell$.
\end{definition}

Parent Hamiltonians are not unique. But we may readily construct them for primitive MPSs. For notational convenience, we work on the half-infinite chain $\Gamma= \Z_{\geq 1}$, but this similarly works for $\Z$. Let $m$ be the injectivity length and let the chain length $\ell\geq m$. Let $h_{[1,m]}$ be the orthogonal projection onto $\calG_m^\perp$ in $(\C^n)^{m}$. Consider the Hamiltonian given by
\begin{equation}
    H_{[1,\ell]} = \sum_{x=1}^{\ell-m} h_{[x,x+m-1]}.
\end{equation} Clearly the MPS spaces $\calG_\ell$ are ground states of $H_{[1,\ell]}$. But we have even more. 

\begin{theorem}~\cite{fannes1992finitely} \label{thm:fannes parent Hamiltonian}
    Let the chain length $\ell\in \N$ and let $\psi_\ell:\calB\to (\C^n)^{\otimes \ell}$ a family of primitive MPSs. Then there is a parent Hamiltonian $H$ for these MPSs which has a unique gapped ground state $\omega$ in the thermodynamic limit.
\end{theorem} Again, it may well be that the MPSs have many finite chain ground states, but there is only one thermodynamic limiting states $\omega$. Notice that the support of the terms in the Hamiltonian grows with the injectivity length $m$. One might conjecture that any parent Hamiltonian for MPS space $\calG_\ell$ need grow in this way, but the family of MPSs we study in Chapter \ref{ch:SO(n)_Haldane_chains} provides a set of counterexamples. 
\chapter{Gapped Ground State and Symmetry Protected Topological (SPT) Phases }\label{ch:gapped_ground_states_phases}

Recall the phase diagram for the spin-1 bilinear-biquadratic models, Figure \ref{fig:bilin_biquad}. The bilinear-biquadratic models are defined by nearest-neighbor Hamiltonians which are invariant under an on-site $O(3)$ symmetry. Letting the local Hilbert spaces instead be $\calH = \C^n$, we can ask a natural question: what phase diagram do we obtain when we consider nearest-neighbor Hamiltonians with an $O(n)$ symmetry? This phase diagram, which we will present in Section~\ref{sec:phase diagram o(n) spin chains}, is the beating heart of this thesis: essentially every result presented is to better understand the behavior of the models in the red ``Haldane'' phase of this diagram and their relation to those in the yellow ``dimerized'' phase investigated in~\cite{bjornberg2021dimerization}. The AKLT chain is the prototypical example of a model in the Haldane phase for $n=3$, and the $SO(n)$ AKLT chains studied in Chapter \ref{ch:SO(n)_Haldane_chains} are a natural generalization to other $n$.

In this chapter, we delve into the mathematical theory of gapped ground state phases with the goal of precisely describing what we mean when we say the red Haldane phase and yellow dimerized phase are distinct. Physically, two interactions should be in the same phase if we can ``smoothly deform'' one to the other. In Section~\ref{sec:gapped ground state phases}, this will lead to our first definition of a gapped ground state phase as an equivalence class of gapped interactions, which will be followed by Theorem~\ref{thm:bachmann gapped} detailing some invariants of this equivalence relation. Corollary~\ref{cor:odd n are different gapped phases} reveals that this is already enough to distinguish the Haldane phase from the dimerized phase for odd $n$, from the observation that models within these phases have different numbers of ground states.

The even $n$ case is more interesting, as the models within these phases are believed to have the same number of ground states. This inspires us to define in Section~\ref{sec:SPT Phases} a more refined notion of phase which incorporates a symmetry group: a symmetry protected topological (SPT) phase will be an equivalence class of gapped interactions which are related by paths of gapped interactions which respect a symmetry $G$. By enforcing this symmetry condition, single gapped phases may fracture into multiple SPT phases. An important example of this phenomenon comes from the AKLT chain, which served to demonstrate that the Haldane phase is a distinct SPT phase from the trivial phase which contains interactions with a product ground state. Pollmann-Berg-Turner-Oshikawa~\cite{pollmann2012symmetry} demonstrated that when one enforces a particular on-site $G=\Z_2\times \Z_2$ symmetry, the AKLT chain must occupy a nontrivial SPT phase. To argue this, they constructed a certain invariant of SPT phases and showed that it attains different values for the trivial model and for the AKLT chain.

The insights of Pollman-Berg-Turner-Oshikawa were converted into a full physical classification of $G$-SPT phases for matrix product states by Chen-Gu-Liu-Wen~\cite{chen2013symmetry}, who discovered a family of $H^2(G,U(1))$ topological indices which serve as a complete invariant. This was later converted into a mathematical proof for all 1D models with unique gapped ground states for the special case of finite $G$ by Ogata~\cite{ogata2020classification}. 

We begin to upgrade Ogata's result by considering the case where $G$ is a compact Lie group and allowing for ground state spaces to have multiple ground states. The appropriate index again takes values in $H^2(G,U(1))$. In Section~\ref{sec:the H2 index for MPS}, we show a computation of this index for MPSs. In Section~\ref{sec:index for split states}, we provide the new Theorem~\ref{thm:split state index} which proves that the index is well-defined for 1D gapped ground states. The theorem requires strong continuity of a particular projective representation of $G$ on an infinite dimensional Hilbert space, so much of this section is dedicated to describing results which we use in the proof. In Section~\ref{sec:excess spin operator}, we re-contextualize results from~\cite{bachmann2014gapped}, pointing out that they find a highly explicit and physical implementation of the otherwise abstract projective representation appearing in the previous section for precisely the ground states appearing at the south pole and the $SO(n)$ AKLT points. Finally, in Section~\ref{sec:the dimerized phase is trivial}, we compute the value of this index for the south pole ground states for all $n$. 

\section{The Phase Diagram for \texorpdfstring{$O(n)$}{O(n)}-Invariant Nearest Neighbor Spin Chains} \label{sec:phase diagram o(n) spin chains}
    Take a 1D quantum spin chain, i.e. the lattice $\Z$, with on-site Hilbert space $\calH_x=\C^n$ for every site $x\in \Z$. We will work on finite chains $[a,b]\subseteq \Z$, but for ease of notation we will often consider the chain $[1,\ell]$. The algebra of observables on a finite chain is $\calA_{[1,\ell]} = \bigotimes_{x=1}^{\ell} \calA_x$, where each site's observable algebra is given by $\calA_x = M_n(\C)$, the $n\times n$ matrices equipped with the operator norm. The algebra of quasilocal observables is given by the operator norm closure $\calA = \overline{\calA_{loc}}$ of local observables $\calA_{loc} = \bigcup_{[a,b]\subseteq \Z} \calA_{[a,b]}$.
    
    We wish to consider the phase diagram for a family of nearest-neighbor Hamiltonians invariant under a local $O(n)$-symmetry. This phase diagram admits a straightforward parameterization in the following way. 
    Let $H_\ell=H_{\ell}^*\in \calA_{[1,\ell]}$ be the nearest-neighbor Hamiltonian
    \begin{equation} \label{def:model Hamiltonian}
		H_\ell = \sum_{x=1}^{\ell-1} h_{x,x+1}  ,
	\end{equation} where $h_{x,x+1}$ denotes a copy of the operator $h=h^*\in M_n(\C)\otimes M_n(\C)$ acting on nearest-neighbor pairs.
    Now, the local $O(n)$ symmetry we have in mind is the tensor power of the defining representation $O(n)\to GL(\calH_\ell)$ defined by $g \mapsto g^{\otimes \ell}$, which gives rise to a representation on observables $g\mapsto g^{\otimes\ell}(\cdot)(g^{-1})^{\otimes \ell}$. To check that $H_\ell$ is invariant under this symmetry is to check that $[h,g\otimes g] = 0$ for all $g\in O(n)$. One can verify that the tensor representation $g\otimes g$ of $O(n)$ admits the following decomposition into irreducible representations:\footnote{\dots except when $n=4$. In this case, $\Exterior^2(\C^4)$ splits into two distinct irreps $U_+,U_-$.}
    \begin{equation} \label{eq:irrep decomposition of two site O(n) rep}
        \C^n \otimes \C^n \cong M_2 \oplus \Exterior^2(\C^n) \oplus \C \ket{\xi} , 
    \end{equation} where $\Exterior^2(\C^n)$ denotes the antisymmetric space, and $M_2 \oplus \C \ket{\xi}$ is an orthogonal decomposition of the symmetric subspace. Letting $\{\ket{i} \st i=1,\dots, n\}$ be an orthonormal basis of $\C^n$, the maximally entangled vector $\ket{\xi}$ may be expressed as
    \begin{equation}
        \ket{\xi}:= \frac{1}{\sqrt{n}}\sum_{i=1}^n \ket{ii}, 
    \end{equation} and its orthogonal complement in the symmetric subspace $M_2$ has as an orthonormal basis
    \begin{equation}
        M_2 = \text{span} \paran{\{ \ket{ij}+\ket{ji}: 1\leq i<j\leq n\} \cup \{\ket{11}-\ket{ii}: 1<i\leq n\} } .
    \end{equation} Now, by applying Schur's lemma~\ref{lem:schur} to each irrep, we have $[h,g\otimes g] =0 $ if and only if $h$ is a linear combination of orthogonal projections onto these three irreps $P_{M_2}, Q, P_{\Exterior^2\C^n}$, where $Q:=\ket{\xi}\bra{\xi}$. It is convenient to perform a change of basis: letting $\SWAP$ be the operator defined by $\SWAP(v\otimes w) = w\otimes v$ for all $v,w\in \C^n$, we may write $\SWAP = P_{M_2} + Q - P_{\Exterior^2 \C^n}$ and $\idty= P_{M_2} + Q + P_{\Exterior^2 \C^n}$. We can thus write $h = a\SWAP + bQ + c\idty$ for $a,b,c\in \R$, since $h=h^*$. Since shifting by $\idty$ only shifts the spectrum of $H_\ell$ by a constant $c$ and so only shifts the ground state energy, it suffices to consider $h$ of the form
    \begin{equation} \label{def:interaction aSWAP+bQ}
        h = a\SWAP + bQ, \qquad a,b\in \R.
    \end{equation}
    This parameterization gives rise to the following phase diagram. For a discussion of this phase diagram, we refer the reader to~\cite{bjornberg2021dimerization}. In the $n=3$ case, up to a change of basis and shift in ground state energy, this parameterized set of nearest-neighbor interactions is precisely the same as the bilinear-biquadratic interaction $\cos(\theta)\vec{S}_x\cdot \vec{S}_{x+1} + \sin(\theta)(\vec{S}_x\cdot \vec{S}_{x+1})^2$ we encountered in Figure \ref{fig:bilin_biquad}.

    % As always, the source code for this figure is in the ``Externalize Tikz graphics'' 
    \begin{figure}
    \centering
    \includegraphics{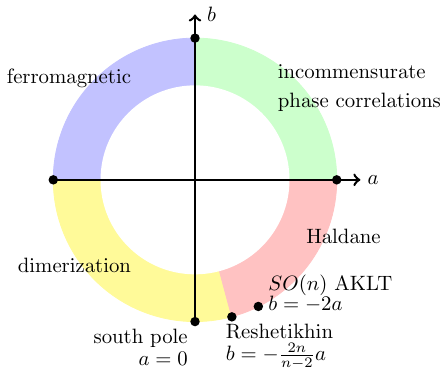}
    \captionof{figure}{The phase diagram for $SO(n)$-invariant interactions (\ref{def:interaction aSWAP+bQ}) (adapted from~\cite{bjornberg2021dimerization}). }
    \label{fig:SO(n) phase diagram}
    \end{figure}

    We will spend a great deal of time in the final Chapter~\ref{ch:SO(n)_Haldane_chains} studying the $SO(n)$ AKLT point $b=-2a$. These are precisely matrix product ground states of a frustration-free interaction, and they are well known to be gapped and robust to perturbation (for a review of proving stability of spectral gaps for frustration-free models, see~\cite{nachtergaele2022quasi}). In the $n=3$ case, this is the AKLT chain we studied in Chapter~\ref{ch:AKLT_chain}. When $n$ is odd, this interaction has a unique gapped ground state. When $n$ is even, it has a pair of translation-invariance symmetry breaking gapped ground states which are 2-translates of one another.

    The south pole point $a=0$ is known to possess for all $n$ a pair of gapped dimerized ground states which are 2-translates of another~\cite{nachtergaele2017direct,aizenman2020dimerization}. The $n=3$ case is exactly the well-studied $P^{(0)}$ chain~\cite{aizenman1994geometric}, where the interaction consists only of the orthogonal projection onto the singlet representation of $\su(2)\cong \so(3)$. More recently, it was proven in~\cite{bjornberg2021dimerization} that when $n\geq 3$, the models in an open region surrounding the south pole point similarly possess a pair of gapped dimerized ground states. Such models are describable as ``random loop models'', wherein one may obtain expectation values of certain physically relevant observables in the ground state by integrating over a space of loop configurations equipped with an appropriate measure. In Section~\ref{sec:the dimerized phase is trivial}, we will briefly recall some key features of this model so as to better understand its nature as a gapped ground state, and we conclude by computing its SPT index.

% ----------------------------------------
\section{Gapped Ground State Phases}\label{sec:gapped ground state phases}
We will now define gapped ground state phases following the approaches of~\cite{bachmann2014gapped,ogata2020classification} with the intent of studying the phase diagram Figure \ref{fig:SO(n) phase diagram}. Assume the setup from Section \ref{sec:infinite quantum spin chains setup}: $\Gamma$ is a lattice, $\calA$ is the algebra of quasi-local observables, $\Phi:\calP_0(\Gamma)\to \calA_{loc}$ is an interaction with bounded $F$-norm $\norm{\Phi}_F<\infty$ defining a local Hamiltonian on any finite $\Lambda\subseteq \Gamma$
\begin{equation}
    H_\Lambda = \sum_{X\subseteq \Lambda} \Phi(X) . 
\end{equation} We will often abuse notation and interchangeably refer to an interaction $\Phi$ as a Hamiltonian $H$. 

We require the following spectral gap condition on our Hamiltonians, following~\cite{bachmann2014gapped}. This condition, which asks for a uniform gap for finite volume Hamiltonians, is slightly stronger than our original definition requiring a spectral gap of the GNS Hamiltonian of a ground state, often called a ``bulk'' gap. It can be shown that the uniform finite volume gap implies a bulk gap (for any of its ground states), and due to the more recent work by Moon and Ogata~\cite{moon2020automorphic}, it is expected that requiring only the bulk gap condition for each of a finite number of ground states should produce highly analogous results to many of those presented in this Chapter.
\begin{definition}
    Let $\lambda_0(\Lambda) = \inf \spec(H_\Lambda)$. 
    The model is \textbf{gapped} if there exists a constant $\gamma>0$ and a family $0\leq \epsilon_\Lambda < \gamma$ such that $\lim_{\Gamma\to \Lambda} \epsilon_\Lambda = 0$ and 
    \begin{equation}
        \spec(H_\Lambda) \cap (\lambda_0(\Lambda) + \epsilon_\Lambda , \lambda_0(\Lambda) + \gamma) = \emptyset , \qquad \text{ for all } \Lambda. 
    \end{equation}
\end{definition} The models in open regions surrounding the $SO(n)$ AKLT point in the Haldane phase and the south pole point in the dimerized phase are known to satisfy this condition~\cite{nachtergaele1996spectral,bjornberg2021dimerization}.

We will define a phase as an equivalence class of interactions, where two interactions are in the same phase if there is a smooth gapped path $\Phi(s), s\in [0,1]$ of interactions connecting them. Whenever we choose to restrict to those Hamiltonians with unique gapped ground states we may equally well interpret this as an equivalence class of states; this is the approach of Ogata, whose work~\cite{ogata2020classification} over the past decade will play a central role in this chapter. With MPSs defined and the parent Hamiltonian construction, we now have a great deal of examples to probe phases in 1D where $\Gamma=\Z$.

We require that the gap $\gamma(s)$ be uniformly bounded below by a positive constant $0<\gamma<\gamma(s), s\in [0,1]$. Our notion of differentiability of $\Phi(s)$ here is made by defining $\partial \Phi(Z,s) = \abs{Z}\Phi'(Z,s)$ for all finite $Z\subseteq \Gamma$, where the prime is differentiation with respect to $s$, and requiring that $\norm{\partial \Phi}_F<\infty$ uniformly in $s$.

\begin{definition}~\cite{bachmann2014gapped} \label{def:gapped ground state phases}
Two gapped local Hamiltonians $H_0$ and $H_1$ are in the same gapped ground state phase if there exists a smooth family of gapped local Hamiltonians $[0,1]\ni s\mapsto H(s)$ such that $H_0 = H(0)$ and $H_1 = H(1)$. This forms an equivalence class, whence we write $H_0\sim H_1$.
\end{definition} 

Let $\calS(s)\subseteq \calA^*$ denote the set of ground states of $H(s)$ on the lattice $\Gamma$. As proven in~\cite{bachmann2012automorphic}, a consequence of this definition is the existence of a strongly continuous cocycle of automorphisms $\tau_{s,t}, 0\leq s,t\leq 1$ of the algebra of quasi-local observables $\calA$ which smoothly flows between the sets of ground states $\calS(0)$ and $\calS(1)$:
\begin{equation} \label{def:spectral flow autos}
    \calS(t) = \calS(s) \circ \tau_{s,t} . 
\end{equation} This justifies the name ``ground state phase''.
By definition the \textit{trivial phase} is the phase containing interactions which act only on-site, i.e. $\Phi(X)=0$ for any subset $X\subseteq \Gamma$ which has more than one element. Such interactions have product ground states, e.g. the ground state of the ``just transverse fields''\footnote{We can ignore the boundary term we wrote earlier, as it does not contribute to the thermodynamic limit.} Hamiltonian $H = \sum_{x=a}^b \sigma_x^X$ we encountered in Example (\ref{ex:just transverse fields}).

The physical idea undergirding these comes from the idea of adiabatic time evolution and the spectral flow. The existence of such automorphisms intuitively tells us that an initial state $\omega(0)\in \calS(0)$ and a final state $\omega(1):= \omega(0)\circ \tau_{0,1}\in \calS(1)$ have ``similar macro-scale entanglement properties'' in the thermodynamic limit. When an automorphism relating two states in this manner exists, we call the states \textit{automorphically equivalent}.

When we consider only those Hamiltonians in 1D with unique gapped ground states, there is only the trivial phase\footnote{This is a common physics definition of ``short range entangled''.}: physicists demonstrated~\cite{chen2013symmetry,zeng2015quantum} an argument that shows that all MPSs are in the trivial phase, and Ogata confirmed that there is only one phase in 1D using a closely related (but not identical) definition of gapped ground state phase~\cite{ogata2017class}. When symmetry enters the game however, the story gets more interesting.

% ----------------------------------------
\section{SPT Phases} \label{sec:SPT Phases}
Let $G$ be a group and suppose we have a group of automorphisms $\{\alpha_g: g\in G\}\subseteq \Aut(\calA_{\Gamma})$ representing this symmetry. We have in mind the on-site symmetry on spin chains $\Gamma=\Z$ from Section \ref{sec:symmetry}: namely, if $G$ a finite or compact group, then let $U:G\to\calU(\calH_x)$ be a unitary representation acting on site $x\in \Z$ and let the automorphism $\alpha:\calA_\Z \to\calA_\Z$ be defined on any finite volume\footnote{Such an automorphism is then uniquely defined on $\calA_\Z$ by the bounded linear transformation theorem.} $X\subseteq \Z$ by 
\begin{equation} \label{def:on-site symmetry}
    \alpha_g(A) = \paran{\bigotimes_{x\in X} U(g)} A \paran{\bigotimes_{x\in X} U(g)^*} , \qquad A\in \calA_{X} .
\end{equation} Compactness is not a strict prerequisite, but it is quite natural to assume since by Weyl's unitary trick, which we used to prove Proposition~\ref{prop:Weyl unitary trick}, every representation of a compact Lie group may be taken to be unitary. For the phase diagram of $O(n)$-invariant spin chains, we have in mind $G=O(n)$ and $G=SO(n)$, the connected component of $O(n)$.\footnote{Most attention in this chapter will be on the latter case, as we will often perform calculations at the Lie algebra level $\g=\so(n)$ and exponentiate, which has image only in the connected component of a Lie group containing the identity.}

We now ask for the families of Hamiltonians $H(s)$ to be invariant under these automorphisms, i.e the interaction has for all $g\in G$ that
\begin{equation}
    \alpha_g(\Phi(Z,s)) = \Phi(Z,s) 
\end{equation} for each finite volume $\Lambda \subseteq \Gamma$ and $s\in [0,1]$. In this case we write the equivalence relation $H_0\sim_G H_1$, and we call this set of equivalence classes of interactions \textit{symmetry protected topological (SPT) phases}.\footnote{There are several other common definitions of these phases, and the relation of these definitions is, to the author's knowledge, subtle and unclear mathematically.}  It was shown in~\cite{bachmann2014gapped} that the cocycle of automorphisms $\tau_{s,t}$ inherits the $G$-invariance of the interactions and so for all $g\in G$
\begin{equation}
    \tau_{s,t} \circ \alpha_g = \alpha_g\circ \tau_{s,t} \qquad \text{for all }s,t\in [0,1] . 
\end{equation} Indeed, this also means that the symmetric structure of the ground states $\calS(s)$ is preserved within a phase, which is the content of the following theorem.

\begin{theorem}~\cite{bachmann2014gapped} \label{thm:bachmann gapped}
Let $H_0$ and $H_1$ be in the same $G$-symmetry protected topological phase on a lattice $\Gamma$. Let $\overline{\calS}_0,\overline{\calS}_1$ be the vector spaces spanned by elements of $\calS_0,\calS_1$, respectively. Then
\begin{enumerate}[label=(\roman*)]
    \item $\dim(\overline{\calS}_0) = \dim(\overline{\calS}_1)$, in the sense that if one is finite-dimensional, the other is too and with the same dimension , 
    \item $\calS_0$ and $\calS_1$ both carry representations $\Theta_i: G\to \mathcal{L}(\overline{\calS_i})$, $i=0,1$.
    \item $\Theta_0$ and $\Theta_1$ are equivalent representations.
\end{enumerate}
\end{theorem} 
Using a result we will later prove, this theorem quickly tells us about the gapped ground state structure of the $O(n)$ spin chain phase diagram for odd $n$. This is in fact agnostic to the role of symmetry and is a consequence of earlier results on automorphic equivalence~\cite{bachmann2012automorphic}.
\begin{corollary} \label{cor:odd n are different gapped phases}
    When $n$ is odd, the yellow dimerized phase and red Haldane phase are distinct gapped ground state phases.
    Letting $G$ be e.g. the trivial symmetry, this follows because the south pole point in the dimerized phase has $\dim(\calS) \geq 2$ by~\cite{bjornberg2021dimerization}, while the $SO(n)$ AKLT point in the Haldane phase has $\dim(\calS) = 1$ by Corollary \ref{cor:parent prop means uniqueness}. 
\end{corollary}
The even $n$ case is more interesting: in this case, the south pole point has $\dim(\calS)\geq 2$ by~\cite{bjornberg2021dimerization} (it is believed to be that $\dim(\calS)=2$), and Corollary~\ref{cor:parent prop means uniqueness} reveals that the $SO(n)$ AKLT point has $\dim(\calS)=2$. It is not immediately clear that these points should occupy distinct $SO(n)$ SPT phases. We will be able to show this by constructing a finer invariant which incorporates the role of symmetry.

The 1D unique gapped ground state setting is rather well understood and will guide our search. We mentioned earlier that there is a great deal of evidence that when one ignores symmetry, there is only one gapped ground state phase for one dimensional models with unique ground states. A natural question arises: given a symmetry $G$, is there more than one SPT phase? In the seminal paper~\cite{pollmann2012symmetry}, Pollmann-Berg-Turner-Oshikawa demonstrated that when one enforces a $G = \Z_2\times \Z_2$ dihedral rotation on-site symmetry, there is a nontrivial SPT phase occupied by the AKLT chain.\footnote{They in fact showed that there are multiple symmetry groups which accomplish this.} In other words, there are at least two SPT phases for $G$-invariant interactions with unique ground states. They accomplished this by finding an invariant of the relation $\sim_G$ which they could separately compute on the trivial product state and on the AKLT state, and they found the two states had distinct associated invariants. 

This led to a full classification program: given $G$, classify the SPT phases which arise for $G$-symmetric interactions with unique ground states. The Pollmann-Berg-Turner-Oshikawa calculation suggested that finding invariants, often called topological invariants or topological indices, of these phases to differentiate them would be a fruitful approach. The classification then could be realized by demonstrating completeness of these invariants, i.e. that there is a one-to-one correspondence between SPT phases and values of an invariant. The work of Chen-Gu-Liu-Wen~\cite{chen2013symmetry} succeeded in classifying $G$-SPT phases for finite (and compact) $G$ for translation-invariant matrix product states. The topological indices labeling the SPT phases took values in the group cohomology of $G$ with values in $U(1)$, $H^2(G,U(1))$. Ogata would later rigorously extend their result for finite $G$ to handle all $G$-symmetric translation-invariant local Hamiltonians with a unique gapped ground state~\cite{ogata2020classification}. The proof is lengthy and requires three main steps:
\begin{enumerate}
    \item One needs a meaningful definition of the topological index associated to each state.
    \item The topological index must be constant along a phase, i.e. respect the equivalence relation $\sim_G$ dictated by paths of symmetric Hamiltonians $H(s)$.
    \item The topological invariant must assign a unique index to each phase. 
\end{enumerate} Our context is slightly different: we will allow for compact Lie group symmetries $G$; we will allow for the set of ground states $\calS$ to consist of more than a single ground state; and we will not pursue a completeness result. Again, our eventual goal in this thesis is to differentiate the Haldane and dimerized SPT phases for even $n$, which will not require completeness.

Analogously to Step 1, we need a definition of a topological index for a compact Lie group symmetry $G$. We will associate to each state $\omega_0\in\calS_0$ an index.
In Section \ref{sec:the H2 index for MPS}, we will describe how to compute this index for MPSs and compute it for the AKLT chain. In Section \ref{sec:index for split states}, we show that each gapped ground state may be assigned a well-defined index. Actually, we assign an index to each split state, a class of states which includes gapped ground states by a result of Matsui. This is a new theorem extending the finite $G$ case presented by Ogata to the compact $G$ case, and it will require some data on strongly continuous unitary representations.

Our analogue to step 2 almost immediately follows from the work of Ogata. Ogata shows~\cite{ogata2020classification} that for a split state $\omega_0$, the associated index is invariant under the cocycle of automorphisms $\tau_{s,t}$ (\ref{def:spectral flow autos}), and so $\omega_1 = \omega_0\circ \tau_{0,1}$ has the same index. This fundamentally rests upon the proof of Theorem \ref{thm:bachmann gapped} from~\cite{bachmann2014gapped}, which uses only the abstract group structure of $G$ and so holds just as well for compact $G$. Then, it is immediate from \ref{def:spectral flow autos} that the collection of indices associated to the ground states in $\calS(0)$ must be the same as the collection of indices associated to $\calS(1)$. In Section \ref{sec:the dimerized phase is trivial}, we wrap up by stating a theorem essentially proven in~\cite{bachmann2014gapped} which demonstrates that the yellow dimerized phase corresponds to the trivial phase. The red Haldane phase SPT index will have to wait until the final chapter of this thesis. 

%-------------------------------------------------
\section{The \texorpdfstring{$H^{2}(G,U(1))$}{H2(G,U(1))} Index for MPS} \label{sec:the H2 index for MPS}
Let us recall how to extract this index, following Ogata~\cite{ogata2020classification} but now carefully tracking the continuity of our Lie group representation. The key ingredient is the fundamental theorem of MPS, Theorem \ref{thm:Fundamental Theorem of MPS}. Let $G$ be a compact Lie group and let $\alpha:G\to \Aut(\calA_Z)$ be the representation of $G$ defined by Equation (\ref{def:on-site symmetry}). Let $\omega_T$ be an $\alpha$-invariant MPS given by a primitive $T:\C^D\to \C^n\otimes \C^D$. Fix $g\in G$. Because $\omega_T$ is $\alpha_g$ invariant, we have that the map $\wt{T}(g)$ defined by
\begin{equation}
    \wt{t}_i(g) := \sum_{j=1}^n U(g)_{ji} t_j , \qquad i=1,\dots , n 
\end{equation} gives the same state
\begin{equation}
    \omega_{\wt{T}(g)} = \omega_T . 
\end{equation} Then, the Fundamental Theorem of MPS \ref{thm:Fundamental Theorem of MPS} guarantees the existence of a phase $\lambda(g)\in U(1)$ and a unitary matrix $\Pi(g)\in M_D(\C)$ such that
\begin{equation}
    \lambda(g) \Pi(g) t_i \Pi(g)^* = \wt{t}_i(g) = \sum_{j=1}^{n} U(g)_{ji} t_j . 
\end{equation} Now, let us consider $\wt{t}_i(gh)$. From above, we have
\begin{equation}
    \wt{t}_i(gh) = \lambda(gh) \Pi(gh) t_i \Pi(gh)^* . 
\end{equation} On the other hand, from the definition of $\wt{t}_j(gh)$ we can compute
\begin{equation}\begin{split}
    \wt{t}_i(gh) = \sum_{j=1}^n U(gh)_{ji} t_j &= \sum_{j,k = 1}^n U(g)_{jk} U(h)_{ki} t_j \\
    &= \sum_{k=1}^n U(h)_{ki} \wt{t}_k(g) \\
    &= \lambda(g)\Pi(g)\sum_{k=1}^n U(h)_{ki} t_k \Pi(g)^*  \\
    &= \lambda(g)\lambda(h)\Pi(g)\Pi(h) t_i \Pi(h)^* \Pi(g)^* . 
\end{split}\end{equation} Altogether we have
\begin{equation}
    \lambda(gh) \Pi(gh) t_i \Pi(gh)^* = \lambda(g)\lambda(h) \Pi(g)\Pi(h) t_i \Pi(h)^* \Pi(g)^*,\qquad \text{ for all } i=1,\dots,n . 
\end{equation} By primitivity of the MPS (\ref{thm:Primitivity}), we have that for all $X\in M_D(\C)$
\begin{equation} \label{eq:the bond u form projective representations}
    \frac{\lambda(gh)}{\lambda(g)\lambda(h)} \Pi(gh) X \Pi(gh)^* = \Pi(g)\Pi(h) X \Pi(h)^* \Pi(g)^*,\qquad \text{ for all } i=1,\dots,n . 
\end{equation}
Finally, we recall Remark \ref{rem:continuity of fundamental theorem of MPS}, which guarantees that in a neighborhood of the identity $\calN\subseteq G$, the maps $\lambda:\calN\to U(1)$ and $\Pi:\calN \to \calU(\C^D)$ are continuous. Equation (\ref{eq:the bond u form projective representations}) exactly means that we have a projective representation $\Pi:\calN\to P\calU(\C^D)$. But since the map $\Pi:G\to P\calU(\C^D)$ is a homomorphism between Lie groups, we may upgrade continuity in $\calN$ to continuity on $G$, thus obtaining a projective representation $\Pi:G\to P\calU(\C^D)$. Remark \ref{rem:equivalence of projective rep definitions} then guarantees a unique corresponding cohomology class $\sigma\in H^2(G,U(1))$.

It is important to note that while this approach grants a well-defined projective representation $\Pi:G\to P\calU(\C^D)$, we still need to do some computation to identify this representation as a familiar representation (or find some other means of computing its associated index). For the following two examples, recall Theorem \ref{thm:second group cohomology of semi-simple} that $H^2(SO(3),U(1))\cong \Z_2=\{1,\sigma\}$.

\begin{example}(A boring example)
Note that by Clebsch-Gordan \ref{eq:clebsch-gordan} we have that $V_1\otimes V_1\cong V_2\oplus V_1 \oplus V_0$ as representations of $\su(2)\cong \so(3)$. So there is a unique up-to-phase intertwining isometry which embeds $V_1\hookrightarrow V_1\otimes V_1$ given by $T:\C^3\to \C^3\otimes \C^3$ which satisfies
\begin{equation}
    (g\otimes g) T = T g , \qquad g\in SO(3) . 
\end{equation} Let $\omega_T$ be the MPS generated by this tensor. Let $U:G\to \calU(\C^3)$ be the defining representation $U(g) = g$, and $\Pi^{(1)} = U$. Then we may rewrite the above as $(U(g)\otimes \Pi^{(1)}(g))T=T\Pi^{(1)}(g)$. 
Then the above equation gives that for $i=1,2,3$ 
\begin{equation} \label{eq:boring example intertwiner}
    \sum_{i,j} U_{ji}(g) t_j = \Pi^{(1)}(g)^{-1} t_i \Pi^{(1)}(g) . 
\end{equation} Evidently $\Pi^{(1)}:SO(3)\to \calU(\C^3)$ is a true representation of $SO(3)$, and so the cohomology class associated to $\omega_T$ is $1\in H^2(SO(3),U(1))$.
\end{example}

\begin{figure}
    \centering
    \includegraphics[width=0.75\textwidth]{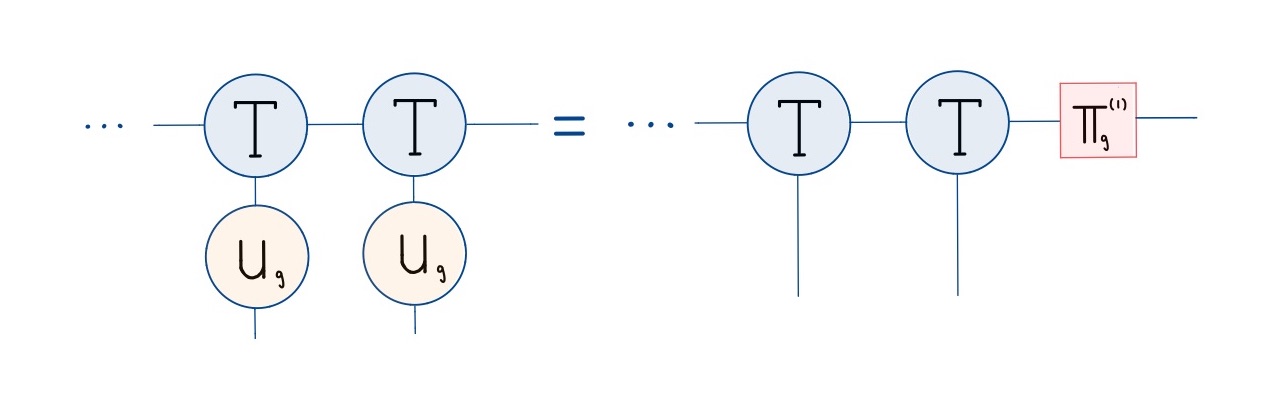}
    \caption{A boring MPS with a real representation $\Pi^{(1)}:SO(3)\to \calU(\C^3)$ and trivial cohomology index $1\in H^2(SO(3),U(1))$. A fruitful way to think of the index is to imagine the on-site symmetry acting on the MPS ground state $\omega_T$ restricted to the half-infinite chain $\Z_L = (-\infty,0]\cap \Z$. Using the intertwining relation (\ref{eq:boring example intertwiner}), we see that the symmetry action on physical legs of the pure state $\omega_T$ (the ``bulk'') is equivalent to a symmetry action on the bond leg (the ``boundary''). This is a description of excess spin acting on MPSs, which we will discuss in Section~\ref{sec:excess spin operator}.}
    \label{fig:excess_spin_real}
\end{figure} 

\begin{example}(AKLT Chain)

Let $U:SO(3)\to \calU(\C^3)$ act via the defining representation on $\C^3$, so $U(g)=g$. This is the same as the spin-1 representation $\Pi^{(1)}:SU(2)\to \calU(\C^3)$ and indeed descends to a representation of $SO(3)$ by the argument in Example \ref{ex:trivial projective rep, spin-1}. Let $\Pi^{(1/2)}:SO(3)\to P\calU(\C^2)$ denote the (nontrivial) projective spin-1/2 representation of $SO(3)$ from Example \ref{ex:proj rep has SU(2) central extension}. This is precisely the same as the intertwining relation (\ref{eq:intertwiner MPS}) which we used to define the AKLT chain, but after passing from $SU(2)$ to $SO(3)$:
\begin{equation}
    (U(g) \otimes \Pi^{(1/2)}(g)) T = T \,  \Pi^{(1/2)}(g), \qquad g\in SO(3) . 
\end{equation} 
Theorem \ref{thm:second group cohomology of semi-simple} tells us that $H^2(SO(3),U(1))\cong \Z_2\cong \{1,\sigma\}$, so since $\Pi^{(1/2)}$ is nontrivial projective, the associated cohomology class must be $\sigma$.
\end{example}

\begin{figure}
    \centering
    \includegraphics[width=0.7\textwidth]{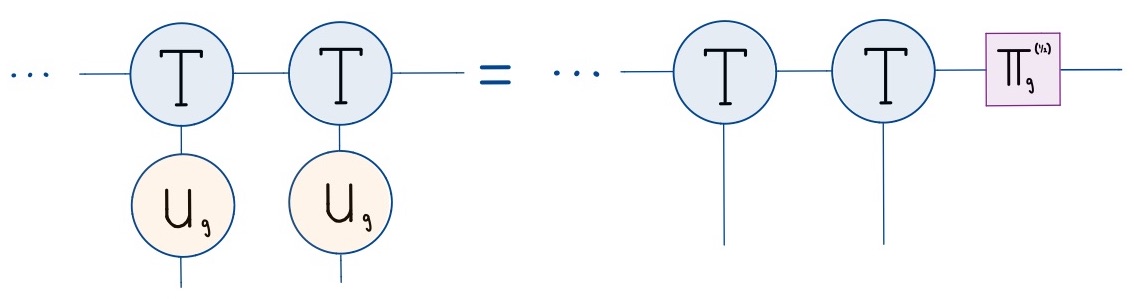}
    \caption{The AKLT chain's projective representation $\Pi^{(1/2)}:SO(3)\to \calU(\C^2)$ and nontrivial $\sigma\in H^2(SO(3),\Z_2)$}
    \label{fig:excess_spin_proj}
\end{figure}

%-------------------------------------------------
\section{The \texorpdfstring{$H^{2}(G,U(1))$}{H2(G,U(1))} Index for Split States}\label{sec:index for split states}

\subsection{Background: Strongly Continuous Projective Representations}
Up until now we have happily worked only with projective representations on finite dimensional Hilbert spaces. We need a bit more for the following section. We will not re-do Chapter \ref{ch:proj-reps}, but instead give relevant definitions and point out that much of this story works equally well for separable Hilbert spaces when our representations are strongly continuous.\footnote{Continuity in operator norm is too strong, as it does not admit a variety of ``natural'' representations. Chief among them is the regular representation $(\Pi, L^2(G))$ where $(\Pi(g) f)(x) = f(g^{-1} x) $ for all $x\in G$, which houses every irrep of a compact Lie group $G$ in one package.} We direct the interested reader to Schottenloher~\cite{schottenloher2008mathematical,schottenloher2018unitary} and to~\cite{spiegel2022continuous}, noting that Spiegel's doctoral thesis~\cite{spiegel2023c} provided some particularly lucid information to the author.

Given a separable Hilbert space $\calH$, recall that the strong operator topology on $\calB(\calH)$ is the coarsest topology such that for any fixed $\psi\in \calH$, the evaluation map $T\mapsto T\psi$ is continuous in $T$. One may think of this as the ``topology of pointwise convergence'': the sequence $T_n\to T$ strongly if for all $\psi\in \calH$, $T_n\psi \to T\psi$.

\begin{proposition}~\cite{schottenloher2018unitary, simms1970topological} The unitary group $\calU(\calH)$ is a topological group with respect to the strong topology.
\end{proposition}
\begin{corollary}
    Let us identify $U(1) = \{e^{i\theta}\idty: \theta\in \R\}\subseteq \calU(\calH)$, which is a normal subgroup of $\calU(\calH)$. Then the projective unitary group is the quotient
    \begin{equation}
        P\calU(\calH) = \calU(\calH)/U(1) ,
    \end{equation} and has the structure of a topological group. In particular, $\calU(\calH)$ is a principal $U(1)$-bundle over $P\calU(\calH)$.
\end{corollary}
As a sidenote, in ~\cite{schottenloher2018unitary} it is shown that this topology is metrizable and $U(\calH)$ is contractible.

Lie homomorphisms are continuous group homomorphisms between Lie groups. Strongly continuous unitary representations may be thought of as topological group representations into the unitary group with the strong topology.

\begin{definition} (Strongly continuous unitary representation)
Let $G$ be a topological group and $\Pi:G\to \calU(\calH)$. We call $\Pi$ a \textbf{strongly continuous unitary representation} if $\Pi(g)\Pi(h) = \Pi(gh)$ for all $g,h\in G$ and if $\Pi$ is strongly continuous, i.e. for each fixed $\psi\in \calH$, the map $g\mapsto \Pi(g) \psi$ is a norm continuous function.     
\end{definition}

We would now like to develop some facts about the projective unitary group $P\calU(\calH)$. We endowed this group with the quotient topology from the strong topology on $\calU(\calH)$. However, one may instead quotient the Hilbert space $\calH$ first and then equip the ``operators'' on this space with the strong topology.

\begin{definition}
Let $\calH$ be a separable Hilbert space. The \textbf{projective Hilbert space} is the set
\begin{equation}
    P\calH = \{\C\psi: \psi\in \calH \setminus \{0\} \},
\end{equation} where $\C\psi = \{\lambda\psi: \lambda\in\C\}$ is the one-dimensional subspace of $\calH$ generated by $\psi$. We equip $P\calH$ with the quotient topology from the canonical projection
\[
    p:\calH \setminus \{0\} \to P\calH, \quad p(\psi) = \C\psi ,
\] where $\calH \setminus \{0\}$ is given the subspace topology inherited from $\calH$. 
\end{definition}

\begin{proposition}~\cite{simms1970topological,spiegel2022continuous} \label{prop:projective hilbert space and rank-1 projections}
      We may equip $P\calH$ with the metric 
      \[d(\C\psi,\C\varphi) = \norm{ \ket{\psi}\bra{\psi} - \ket{\varphi}\bra{\varphi}},
      \] which gives $P\calH$ the same topology as the quotient topology. In particular, the map $\C\psi\mapsto \ket{\psi}\bra{\psi}$ is an isometric embedding $P\calH\to \calB(\calH)$ with image rank-1 orthogonal projections.   
\end{proposition}

Now, since any unitary $U\in \calU(\calH)$ is linear, it descends to a well-defined map $[U]:P\calH\to P\calH$. We may call this collection of maps $\calU(P\calH)$, and we may equip this collection with the strong operator topology, i.e. the coarsest topology such that for any fixed $\C\psi\in P\calH$, the evaluation map $[U]\mapsto [U]\C\psi = \C U\psi$ is continuous in $[U]$. The following theorem assures us that this procedure yields the same topological groups. 

\begin{theorem}~\cite{simms1970topological, spiegel2022continuous} \label{thm:equivalence of topological unitary groups}
    The topological group $\calU(P \calH)$ is the same as the quotient group $P\calU(\calH)$. 
\end{theorem} 

\begin{theorem} (Every $*$-automorphism of $\calB(\calH)$ is inner)\label{thm:every auto is inner}
    Let $\calH$ a separable Hilbert space and let $\gamma\in \Aut(\calB(\calH))$ be a $*$-automorphism. Then there is a unitary $U\in \calU(\calH)$ which implements the automorphism $\gamma = \Ad(U)$, i.e.
    \begin{equation}
        \gamma(A) = \Ad(U)(A) = U A U^* ,\qquad \text{for all } A\in \calB(\calH). 
    \end{equation}
\end{theorem}
\begin{corollary}
    There is a one-to-one correspondence between $*$-automorphisms $\Aut(\calB(\calH))$ and elements of the projective unitary group $P\calU(\calH)$.
\end{corollary}\begin{proof}(of Corollary) 
    Let $\gamma$ be implemented by two unitaries $U,V\in \calU(\calH)$. Then $\Ad(UV^{-1}) = \ide_{\calB(\calH)}$, which means
    $UV^{-1}\in \calB(\calH)'$, the commutant of $\calB(\calH)$. But by an analogue of Schur's lemma the commutant $\calB(\calH)' = \C\idty$, and so $UV^{-1}\in \C \idty$. This exactly means $[U]=[V]\in P\calU(\calH)$.  
\end{proof}

\begin{proposition}
Consider $\Aut(\calB(\calH))$ and $P\calU(\calH)$ with their strong operator topologies. Then the (bijective)
map $\Aut(\calB(\calH))\to P\calU(\calH)$ is continuous.
\end{proposition}
\begin{proof}
    We use Theorem~\ref{thm:equivalence of topological unitary groups} and instead show continuity $\Aut(\calB(\calH))\to \calU(P\calH)$. Fix an arbitrary $\C\psi\in P\calH$. We want to show that the map $\Ad(U)\mapsto \C U\psi$ is continuous. First, observe that the map $\Ad(U)\mapsto \Ad(U)(\ket{\psi}\bra{\psi})$ is continuous as a function $\Aut(\calB(\calH))\to \calB(\calH)$, since conjugation by a unitary preserves norm and every norm continuous function is strongly continuous. Now we write 
    \begin{equation}
        \Ad(U) \ket{\psi}\bra{\psi} = U\ket{\psi}\bra{\psi} U^* = \ket{U\psi}\bra{U{\psi}}. 
    \end{equation} Now, using Proposition \ref{prop:projective hilbert space and rank-1 projections}, the map $\ket{U\psi}\bra{U\psi}\mapsto \C U\psi$ is an isometry and thus continuous. 
\end{proof}
\begin{corollary} \label{cor:strong cont autos give strong cont proj unitaries}
    Suppose that $\gamma_{(\cdot)}:G\to \Aut(\calB(\calH))$ is a strongly continuous group representation. Then there is a strongly continuous group representation $G\to P\calU (\calH)$ which implements these automorphisms, i.e. a strongly continuous family of projective unitaries $[U_{(\cdot)}]:G\to P\calU(\calH)$ such that
    \begin{equation}
        \gamma_g(A) = \Ad(U_g)(A), \qquad A\in \calB(\calH).
    \end{equation}
\end{corollary}

We will not re-work our discussions of projective representations from Chapter \ref{ch:proj-reps}. Indeed, since we have the following short exact sequence of topological groups
\begin{equation}
    1 \longrightarrow U(1) \longrightarrow \calU(\calH) \longrightarrow P\calU(\calH) \longrightarrow 1 , 
\end{equation} the proof of Proposition \ref{prop:proj rep yields central extension} goes through once we replace ``Lie group homomorphism'' with ``topological group homomorphism''. Bargmann's theorem works just well, as do the arguments culminating in Remark \ref{rem:equivalence of Borel group cohomology}. In particular, there is still a 1-1 correspondence between isomorphism classes of projective representations and $H^2(G,U(1))$, the Borel group cohomology. See Schottenloher for details~\cite{schottenloher2008mathematical}.

%---------
\subsection{The index for split states}
The earlier MPS calculation is quite suggestive, especially Figure \ref{fig:excess_spin_real} and Figure \ref{fig:excess_spin_proj}. Ogata mimics this calculation and finds a way to assign a $H^2(G,U(1))$ index to each unique gapped ground state $\omega$ (which may well not be an MPS). The key definition is that of the \textit{split property} with respect to the decomposition $\Z=\Z_L\cup \Z_R$, where the left side of the infinite chain is $\Z_L = (-\infty, 0]\cap \Z$ and the right side is $\Z_R = [1,\infty)\cap \Z$.

The following definition is highly nonstandard but suffices for our purposes. This is a special case of the split property for pure states on a quantum spin chain $\Z = \Z_L\cup \Z_R$, a consequence of e.g. Proposition 1.18~\cite{ogata2020classification}.

\begin{definition}\label{def:ogata split thm}
    For a pure state $\omega$ on $\calA_{\Z_L\cup \Z_R}$, we say $\omega$ satisfies the \textbf{split property} with respect to the cut $\Z_L\cup \Z_R = \Z$ iff there are irreducible representations $(\calH_L,\pi_L), (\calH_R,\pi_R)$ of $\calA_{\Z_L}, \calA_{\Z_R}$ and a unit vector $\Omega\in \calH_L\otimes \calH_R$ such that $(\calH_L\otimes \calH_R, \pi_L\otimes \pi_R, \Omega)$ is a GNS-triple of $\omega$.
\end{definition} Physically, we may think of the split property as describing when a state $\omega$ does not have ``too much'' entanglement across a cut of the lattice $\Z$ into $\Z_L$ and $\Z_R$. It is not too involved to see that primitive MPSs are split, but much more is true. 

\begin{theorem}~\cite{matsui2013boundedness} \label{thm:matsui gapped ground states are split}
    Let $\omega$ be a gapped ground state of a local Hamiltonian on $\Z$. Then $\omega$ is split with respect to $\Z=\Z_L\cup \Z_R$.\footnote{It is not known whether every split state is a gapped ground state of a local Hamiltonian, although this is believed to be true.}
\end{theorem}

We are ready to demonstrate how to extract a projective representation, and thus a $H^2(G,U(1))$ index, from a split state. This will supply the desired topological index for a $G$-invariant gapped ground state $\omega$ of a sufficiently local $G$-symmetric interaction $\Phi$. This was originally done by Matsui~\cite{matsui2001split}, with a few improvements by Ogata~\cite{ogata2020classification}. Ogata's result applies only to a finite group, but in what follows, we will be able to extend the result to $G$ a compact Lie group. The bulk of the argument traces the argument in~\cite{ogata2020classification}, with adjustments present to ensure the continuity of group representations at play. Let us begin by establishing a standard result.

\begin{proposition}\label{prop:strong continuity on-site symmetry}
    Let $\alpha:G \to \Aut(\calA_\Z)$ be the representation of a Lie group $G$ given by on-site unitary conjugation from (\ref{def:on-site symmetry}). Then $\alpha$ is strongly continuous. 
\end{proposition}
\begin{proof}
    It suffices to show strong continuity at $\idty\in G$, since $G$ a Lie group. Note that $\alpha_\idty = \text{id}$, the identity operator. Let $\wt{A}\in \calA_\Z$ and fix $\epsilon>0$. Then there exists a local observable $A\in \calA_\Z^{loc}$ with $\norm{\wt{A}-A}<\epsilon$. Note that since $\alpha_g$ acts by conjugation by a unitary on local observables $A\in \calA_\Z^{loc}$, it is norm continuous and we may find a neighborhood of the identity $\calN_\idty\subseteq G$ for which $\norm{\alpha_g(A)-A}<\epsilon$. We then estimate for all $g\in \calN_\idty$:
    \begin{align*}
        \norm{\alpha_g(\wt{A}) - \wt{A}} &\leq \norm{\alpha_g(\wt{A}) - \alpha_g(A)}+ \norm{\alpha_g(A)- A} + \norm{A-\wt{A}} \\
        &\leq 2\norm{\wt{A} - A} + \norm{\alpha_g(A)- A} \\
        &= 3\epsilon,
    \end{align*} where we have used that automorphisms have operator norm $\norm{\alpha_g} = 1$. 
\end{proof}

\begin{theorem} \label{thm:split state index}
    Let $\omega$ be a split state with respect to $\Z = \Z_L \cup \Z_R$. Let $G$ be a Lie group acting by on-site unitary conjugation $\alpha:G\to\Aut(\calA_\Z)$ as in (\ref{def:on-site symmetry}), and assume $\omega$ is $G$-invariant, meaning $\omega = \omega\circ \alpha_g$ for all $g\in G$. Then we may assign a well-defined index $h_\omega \in H^2(G,U(1))$ to $\omega$.
\end{theorem}
\begin{remark}
    Because we do not require translation invariance of $\omega$, it may be that $h_\omega$ depends on the location of the cut $\Z = \Z_L\cup \Z_R$. 
\end{remark}
\begin{proof}
     By Proposition \ref{def:ogata split thm}, the GNS triple of $\omega$ is of the form $(\calH_L\otimes \calH_R, \pi_L\otimes \pi_R, \Omega)$ with irreducible representations $(\calH_L,\pi_L),(\calH_R,\pi_R)$ of $\calA_L,\calA_R$ and a unit vector $\Omega\in \calH_L\otimes \calH_R$. In this representation we have that the double commutants are
\[
    ((\pi_L\otimes \pi_R)(\calA_R))'' = \C \idty_{\calH_L} \otimes \calB(\calH_R) , \qquad ((\pi_L\otimes \pi_R)(\calA_L))'' = \calB(\calH_L)\otimes \C \idty_{\calH_R} . 
\] Now, $\alpha$-invariance of $\omega$ gives by Corollary \ref{cor:GNS unitary implements symmetries} that there is a unitary representation $V:G\to \calU(\calH_L\otimes \calH_R)$ such that
\begin{equation}
    \Ad(V_g)\circ (\pi_L\otimes \pi_R) = \pi_L\circ \alpha_g^L \otimes \pi_R \circ \alpha_g^R, \qquad V_g\Omega = \Omega. 
\end{equation} We may restrict this to $\calA_R$ to get
\[
    \Ad(V_g) (\idty_{\calH_L}\otimes \pi_R(A)) = \idty_{\calH_L} \otimes \pi_R\circ \alpha^R_g(A), \qquad A\in \calA_R.
\] Now, since $\pi_R$ is irreducible, this restriction defines a $*$-automorphism $\gamma_g$ of the entire algebra $\calB(\calH_R)$ such that
\begin{equation}\begin{split}
    \Ad(V_g) (\idty_{\calH_L}\otimes x) = \idty_{\calH_L} \otimes \gamma_g(x), \qquad x\in \calB(\calH_R), \\
    \gamma_g\circ \pi_R(A) = \pi_R \circ \alpha_g^R(A) , \qquad A\in \calA_R.
\end{split}\end{equation} Since $\alpha$ is strongly continuous by Proposition \ref{prop:strong continuity on-site symmetry}, the second equation above guarantees that $\gamma:G\to \Aut(\calB(\calH_R))$ is a strongly continuous group representation. Corollary \ref{cor:strong cont autos give strong cont proj unitaries} then guarantees the existence of a strongly continuous projective representation $[u_R]:G\to P\calU(\calH_R)$ implementing $\gamma$, i.e.
\begin{equation} \label{eq:uR implements right chain reps}
    \Ad(u_R(g))(x) = \gamma_g(x), \qquad x\in \calB(\calH_R) . 
\end{equation} Let $h_R\in H^2(G,U(1))$ be the second group cohomology class associated to the projective representation $[u_R]:G\to P\calU(\calH_R)$. 

To see that $h_R$ is independent of the choice of GNS representation $(\calH_R,\pi_R)$, suppose that we started with $(\wt{\calH}_R,\wt{\pi}_R)$ and obtained the index $\wt{h}_R$ associated to the projective representation $[\wt{u}_R]:G\to P\calU(\wt{\calH}_R)$. Then the GNS construction, Theorem~\ref{thm:GNS construction}, guarantees the existence of a unitary map $V:\wt{\calH}_R\to\calH_R$ such that the irreducible representations $\pi_R,\wt{\pi}_R$ are unitarily equivalent:
\[
    \pi_R(A) = V \wt{\pi}_R(A) V^*, \qquad A\in \calA_R . 
\] In particular, $[u_R] = V[\wt{u}_R]V^*$, which means they define equivalent projective representations and so have identical associated indices $h_R=\wt{h}_R$. Thus, we have a well-defined index $h_\omega :=h_R$ for each split state.
\end{proof}

To see that these are invariants of SPT phases, one must show that $h_\omega$ is invariant under any cocycle of automorphisms $\tau_{s,t}$ induced by a path of $G$-symmetric gapped Hamiltonians $H(s)$. Ogata proves this invariance in Theorem 2.10~\cite{ogata2020classification} by using the techniques of~\cite{bachmann2014gapped}, specifically Theorem \ref{thm:bachmann gapped} in this thesis. This theorem requires only the abstract group structure of $G$ and makes no use of topological data of $G$, and indeed Ogata's proof inherits this feature: thus her proof applies equally well to the Lie group $G$ case.
Then, in our setting, since we have a family of ground states $\calS(s)$ related by $\calS(t) = \calS(s) \circ\tau_{s,t}$ from (\ref{def:spectral flow autos}), we simply assign each family $\calS(s)$ a collection of invariants 
\begin{equation}
    h_{\calS}(s) := \{h_\omega(s)\in H^2(G,U(1)): \omega(s)\in \calS(s)\}
\end{equation} and apply Ogata's result to each path of states $\omega(s)$ separately.
We thus arrive at the following theorem. 
\begin{theorem} \label{thm:SPT invariants of ground state space}
    Let $H_0$ and $H_1$ be in the same $G$-symmetry protected topological phase on $\Z$, where $G$ a Lie group acting via on-site unitary conjugation (\ref{def:on-site symmetry}), with ground state spaces $\calS_0,\calS_1$. Then the collection of second cohomology classes of these ground state spaces is equal, 
    \begin{equation}
        h_{\calS}(0) = h_{\calS}(1) , 
    \end{equation} and so $h_{\calS}$ defines an SPT phase invariant.
\end{theorem} 

 For the phase diagram Figure~\ref{fig:SO(n) phase diagram}, we will take the symmetry to be $G=SO(n)$. In Section~\ref{sec:the dimerized phase is trivial}, we will compute this family of indices $h_\calS$ for the pair of ground states at the south pole point in the yellow dimerized region and find that it consists of two trivial indices. In Chapter~\ref{ch:SO(n)_Haldane_chains}, we will find that when $n$ even, the family of indices the pair of ground states $\omega_\pm$ at the $SO(n)$ AKLT point in the red Haldane phase consists of two nontrivial indices. Together, this proves the dimerized phase and Haldane phase are distinct SPT phases when $n$ even.

Before proceeding, we mention a couple of conjectures regarding these invariants.
Ogata proves~\cite{ogata2020classification} that when $G$ is a finite group $G$, this invariant is a complete invariant for unique gapped states $\omega$.\footnote{As we mentioned earlier, it is not known whether every split state is the unique gapped ground state of a local Hamiltonian, so completeness of the invariant is subtly open. It is however believed to still be true.} This spawns a natural conjecture.

\begin{conjecture}
    Let $G$ be a compact Lie group. Then $h_\Phi\in H^2(G,U(1))$ is a complete SPT phase invariant of the space of unique gapped ground state interactions, meaning given two unique gapped ground states $\omega_0,\omega_1$ of interactions $\Phi_0,\Phi_1$ with the same SPT invariant $h_{\omega_0}= h_{\omega_1}$, then the interactions are connected by a smooth uniformly gapped path of $G$-invariant interactions, i.e. $\Phi_0\sim_G \Phi_1$. 
\end{conjecture}

In our working examples, we have considered the Lie group $SO(n)$, and Theorem \ref{thm:second group cohomology of semi-simple} assures us that $H^2(SO(n),U(1)) \cong \Z_2$. But notice that we could have just as well asked to take the dihedral $\Z_2\times \Z_2$ subgroup of $SO(n)$ given by $\pi$-rotations about two axes, or more concretely $\{\idty, e^{\pi L_{12}}, e^{\pi L_{13}}, e^{\pi L_{12}}e^{\pi L_{13}}\}$. Here, $H^2(\Z_2\times \Z_2, U(1)) \cong \Z_2$. The SPT phases when we enforce the symmetry $SO(n)$ or the symmetry $\Z_2\times \Z_2$ are exactly the same! This example, as well as Theorem \ref{thm:second group cohomology of semi-simple} guaranteeing that $H^2(G,U(1))$ is a finite group for compact $G$, leads us to pose the following conjecture. 

\begin{conjecture}
    Let $G$ be a compact Lie group. Then there exists a finite subgroup $H\leq G$ for which the set of $G$-SPT phases is exactly the same as the set of $H$-SPT phases. 
\end{conjecture}

%-----------------------------------------------------------------
\section{The Excess Spin Operator and Half-Chain Rotations}\label{sec:excess spin operator}
From the proof of Theorem~\ref{thm:split state index}, the index arose by starting with a $G$-invariant split state $\omega$ and considering the action of half-chain rotations $\alpha_g^R$, which is just the action of the automorphism $\alpha:G\to \Aut(\calA)$ restricted to the right half-chain $\Z_R$.\footnote{Of course, this can all be similarly described for left half-chain rotations.} As a consequence of the split property, we saw that the GNS Hilbert space of $\omega$ may be expressed as $\calH_L\otimes \calH_R$, which allowed us to see that the half-chain rotations are unitarily implementable via a projective representation $u_R$ of $G$ on the GNS Hilbert space $\calH_R$ of the right half-chain. This is all well and good for defining the projective representation which grants us the index, but to compute this index explicitly for a given model one must find this projective representation more concretely. 

For several classes of models, including finitely correlated states and the random loop models describing the south pole point of the phase diagram Figure~\ref{fig:SO(n) phase diagram}, Bachmann and Nachtergaele proved that one may unitarily implement the automorphisms $\alpha_g^R$ using excess spin operators~\cite{bachmann2014gapped}. On finite chains $[-\ell,\ell]$, the right half-chain rotation operator refers to the on-site unitary symmetry restricted to the right half-chain, so the representatives $\bigotimes_{x=1}^\ell U_g$ of the group elements $g\in G$. Passing to the Lie algebra $\g$ representation, the excess spin operator is the infinitesimal generator of this symmetry given by $\sum_{x=1}^\ell L_j$, where $e^{itL_j} = U_g$ for some real $t$ (compare to Example \ref{ex:checking for symmetry}, where $L=S$ is exactly a spin matrix). By definition, the half-chain rotation automorphism $\alpha_g^R = \bigotimes_{x=1}^\ell \Ad(U_g)$ is an inner automorphism on $\calA_{[1,\ell]}$. For infinite chains $\Z$, one should not in general expect the half-chain rotation $\alpha^R$ to remain inner on the algebra of quasi-local observables $\calA$. But we proved already that when $\omega$ is a $G$-invariant split state, $\alpha^R$ is an inner automorphism on the GNS representation of $\omega$. The content of Theorem~\ref{thm:bachmann excess spin FCS} and Theorem~\ref{thm:excess spin random loop} is that the half-chain rotation operator, given as the exponential $U_g^R$ of the formal sum $L^R := \sum_{x=1}^\infty L_x$, converges in the GNS representation of $\omega$ and implements this automorphism, i.e. $\alpha_g^R = \Ad(U_g^R)$. For our context, this gives a strongly continuous representation of $Spin(n)$, which then passes to a projective representation of $SO(n)$: if this representation is trivial projective, then $\omega$ has a trivial SPT index, and if the representation is nontrivial projective, then $\omega$ is a nontrivial SPT. 

To state the first theorem, we slightly abuse notation and assume that there is an representative $L=L^*$ of the Lie algebra rep such that $U_g= \exp(ig L)$, where $g\in \R$. We will describe the (formal) operator $\sum_{x=1}^\infty L_x$ as a limit $\ell\to\infty$ of the local approximation
\begin{equation}
    L^R(\ell) = \sum_{x=1}^{\ell^2} f_\ell(x-1) L_x, 
\end{equation} where $f_\ell:\Z_{\geq 0} \to \R$ is given by 
\[
    f_\ell(m\ell + n) = 1-m/\ell ,\qquad \text{for $m,n\in [0,\ell-1]$, and } f_\ell(x) = 0 \qquad \text{ for $x\geq \ell^2$}. 
\] We shall then write 
\[
    U_g^R(\ell) = \exp(i g L^R(\ell)) . 
\]
\begin{theorem} (Theorem 3.1~\cite{bachmann2014gapped}) \label{thm:bachmann excess spin FCS}
    Let $\omega$ be a $G$-invariant pure $p$-periodic finitely correlated state and let $(\calH_L\otimes \calH_R,\pi_L\otimes \pi_R,\Omega_\omega)$ be its GNS representation. Then the strong limit
    \begin{equation}
        U^R_g = \slim_{\ell\to \infty} e^{ig\cdot \pi_R(L^R(\ell))}
    \end{equation} exists on $\calH_R$ for all $g\in G$ and defines a strongly continuous representation $U^R:G\to \calU(\calH_R)$. . 
\end{theorem} Note that every pure finitely correlated state is a gapped ground state of its parent Hamiltonian by Theorem~\ref{thm:fannes parent Hamiltonian}. The gap then implies by the Exponential Clustering Theorem~\cite{nachtergaele2006lieb} that there is exponential decay of spatial correlations in the ground state, so we omit this condition in our restatement of the theorem.
\begin{remark}
    The original statement of this theorem only concerns finitely correlated $\omega$. However, tracing through the arguments, one sees that it holds equally well for $p$-periodic finitely correlated states. In particular, it holds for the ground states of the $SO(n)$ AKLT point in the red Haldane phase of Figure~\ref{fig:SO(n) phase diagram}.
\end{remark}

The excess spin operator perspective adds some clarity to our MPS examples in Figure~\ref{fig:excess_spin_real} and Figure~\ref{fig:excess_spin_proj}. Recall that when $\omega$ is a pure finitely correlated state with minimal triple $(\bbE,\Tr\rho (\cdot), \idty)$, it is a primitive MPS. When $\omega$ is $G$-invariant, the Fundamental Theorem of MPS~\ref{thm:Fundamental Theorem of MPS} tells us that there is a projective unitary representation $\Pi:G\to \calU(\C^D)$ for which its defining isometry $T:\C^D\to \C^n\otimes \C^D$ enjoys the intertwining relation:
\begin{equation}
    (U_g\otimes \Pi_g) T = T \, \Pi_g \qquad g\in G. 
\end{equation} Of course, this is equivalent to the following intertwining relation enjoyed by the CP map $\bbE: M_n(\C)\otimes M_D(\C)\to M_D(\C)$ given by $\bbE_A(B) = T^*(A\otimes B)T$:
\begin{equation}
    \bbE_{U_g}(\Pi_g) = \Pi_g . 
\end{equation} So, $\Pi_g$ is an eigenvector of $\bbE_{U_g}$ of eigenvalue 1, and since $\bbE_\idty$ is primitive, it is clear that $\bbE_{U_g}$ is primitive for all $g$ in an open neighborhood of the identity $\calN_\idty \subseteq G$.

The representation content of $\Pi$ and the half-chain rotation representation are in fact the same.

\begin{theorem} (Theorem 4.3~\cite{bachmann2014gapped}) \label{thm:rep content excess spin MPS}
    Let $\omega$ be a pure finitely correlated state. The representation $U^R:G\to \calU(\calH_R)$ contains only the representation $\Pi:G\to \calU(\C^D)$, with infinite multiplicity. 
\end{theorem}
\begin{proof}
    Notice that for any $A,B\in \calA_L$, the representative $(\pi_L\otimes \pi_R)(A) = \pi_L(A) \otimes \idty$ necessarily commutes with the half-chain rotation operator $U_g^R$ for all $g\in G$. Then, one may extract the following relation (for details, see Equation 4.2 in~\cite{bachmann2014gapped}):
    \begin{equation} \label{eq:mimick eqn 4.2 in bachmann}
    \inprod{\pi_L(A)\Omega_\omega, U_g^R \pi_L(B) \Omega_\omega} = \inprod{\Omega_\omega, \pi_L(A^*B) U_g^R \Omega_\omega} = \Tr \paran{ \rho \, \bbE_{A^* B}(\Pi_g) }  .
    \end{equation} By the characterization of primitivity in Theorem~\ref{thm:Primitivity}, we have
    \begin{equation}
        \text{span} \{ \rho \circ \bbE_{A_{-\ell}} \circ \dots \circ \bbE_{A_0} : \ell\in \mathbb{N}, A_i\in M_n(\C) \} = M_D(\C) . 
    \end{equation} In particular, Equation (\ref{eq:mimick eqn 4.2 in bachmann}) gives a one-to-one correspondence between the matrix elements of the representatives $U_g^R$ and the matrix elements of $\Pi_g$. Thus, $\Pi_g$ determines the full set of representation content of $U_g^R$. Since $U^R:G\to \calU(\calH_R)$ is infinite dimensional, it must contain the finite dimensional $\Pi:G\to \calU(\C^D)$ with infinite degeneracy.
\end{proof}
This theorem justifies the tensor network picture from Figure~\ref{fig:excess_spin_real} and Figure~\ref{fig:excess_spin_proj}, in that one may cut an MPS according to $\Z = \Z_L\cup \Z_R$, perform a half-chain rotation, and extract the index from the projective representation induced on the bond indices.

\begin{remark} 
Adjusting this proof for $p$-periodic pure finitely correlated states is not particularly difficult by using Proposition 3.3 in~\cite{fannes1992finitely}, but some care must be taken: the representation content of $U^R$ will in this case depend upon the exact location of the cut $\Z = \Z_L\cup \Z_R$, since the states are not translation invariant. We omit the somewhat clumsy general $p$ statement of this theorem and instead later use the basic idea of this proof to extract the index of the MPS.
\end{remark}

%-------------------------------------------------
\section{The Dimerized Phase is a Trivial SPT Phase}\label{sec:the dimerized phase is trivial}
In this section, we wish to compute the SPT indices $h_\calS$ of the pair of gapped ground states at the south pole point of Figure~\ref{fig:SO(n) phase diagram}. As we mentioned earlier, Bachmann and Nachtergaele proved existence of the excess spin operator for a class of random loop models with $SU(2)$ invariance~\cite{bachmann2014gapped}. We will apply nearly identical methods to demonstrate existence of the excess spin operator for the south pole ground states, and this expression will be enough to compute the index. To discuss the random loop model, we will use the notation and definitions from~\cite{nachtergaele1994quasi}, mentioning the pedagogical lecture notes~\cite{ueltschi2014graphical}. For the south pole point, many of the results for $\su(2)$ case essentially carry over unchanged.

Let us say a little about the random loop model which describes the south pole point. The $n=3$ prototype is the same as $P^{(0)}$ chain, where the interaction is given only by projection onto the singlet representation of $\su(2)\cong \so(3)$ inside the product of two spin-$s$ particles, $V_0\subseteq V_{s}\otimes V_{s}$ (see Example 2 in~\cite{nachtergaele1994quasi} for a discussion of this model). For general $n$, the interaction (\ref{def:interaction aSWAP+bQ}) consists only of $-Q$, the orthogonal projection onto the maximally entangled vector $\ket{\xi} = \frac{1}{\sqrt{n}} \sum_{i=1}^n \ket{ii}$:
\begin{equation} \label{eq:south pole interaction}
    H_{[a,b]} = \sum_{x=a}^{b-1} - Q_{x,x+1} + \idty
\end{equation} where we have added an inconsequential $\idty$ for application of a theorem momentarily.
Bj\"ornberg et. al.~\cite{bjornberg2021dimerization} studied an open region of the phase diagram centered at this point and proved that in this region, there are at least two gapped ground states which are 2-translates of one another. One may extract many of the following expressions from the special $a=0$ case of their theorems, particularly Theorem 2.1. Physically, one might expect dimerization because each interaction term $-Q_{x,x+1}$ rewards high entanglement between sites $x$ and $x+1$ with low energy, but monogamy of entanglement prevents site $x$ from being simultaneously highly entangled with $x-1$ and $x+1$. Thus, it is energetically favorable for the site $x$ to ``choose'' to entangle more strongly with its left neighbor or its right neighbor, leading to translation-invariance symmetry breaking.

Let $L = \ket{\alpha}\bra{\alpha'} - \ket{\alpha'}\bra{\alpha}, 1\leq \alpha<\alpha'\leq n$ be an element of $\so(n)$. Just as one measures spin-spin correlators in $SU(2)$ invariant models, we will now measure correlators of $L$ observables (indeed, spin-spin is a special case of this via $\su(2)\cong \so(3)$). The key observation is that expectation values of functions of $L$ in a state $\omega_\beta$ at inverse temperature $\beta$ may be expressed as an integral over a space of loop configurations $\nu$ equipped with an appropriate probability measure $d\mu_\beta$. 

Recall from~\cite{nachtergaele1994quasi} that we have the following Poisson integral formula for the exponential of $H=H_{[a,b]}$ as given by (\ref{eq:south pole interaction}):
\begin{equation} \label{eq:poisson integral formula}
    e^{-\beta H} = \int d\rho_\beta(\nu) K(\nu),
\end{equation} where
\begin{itemize}
    \item $\nu$ is a loop configuration on space-time $[a,b]\times [0,\beta]$ (see Figure~\ref{fig:random_loop}), where the time direction $[0,\beta]$ has periodic boundary conditions. It is equivalently specified by a collection of ``double-bars'' at coordinates $(e_i,t_i)$ in space-time. 
    \item $d\rho_\beta$ is the probability measure of a Poisson process of intensity 1 running over the time interval $[0,\beta]$. An instance of this process is a configuration $\nu$, where this determines the location and number of double-bars $B = \{\{x,x+1\}: a\leq x\leq b-1\}$.
    \item $K(\nu)$ is the time-ordered product of operators $Q$. In particular, if the time-ordered locations of the double-bars in a configuration $\nu$ are $\{(b_1,t_1), (b_2,t_2),\dots,(b_k,t_k)\}\subseteq B\times [0,\beta]$ with $t_1<t_2<\dots <t_k$, then $K(\nu) = Q_{b_1} Q_{b_2} \dots Q_{b_k}$. 
\end{itemize} Now, we wish to eventually understand the expectation of observables which are given as functions of $L$, $f(L)\in \calA$ in the equilibrium state $\omega_\beta$, and then take $\beta\to \infty$ to recover the ground state. So we will first write
\begin{equation}
    \Tr \paran{ f(L) e^{-\beta H} } = \int d\rho_\beta(\nu) \Tr \paran{f(L) K(\nu)} . 
\end{equation} We may now apply a useful observation, quite analogous to that of the $P^{(0)}$ chain, Example 2 in~\cite{nachtergaele1994quasi}. For a fixed configuration $\nu$, this trace is not terribly difficult to compute. One realizes\footnote{This may be seen explicitly by taking an orthonormal basis of the chain $\calH_{[a,b]}$, say $\{\ket{\sigma}: \sigma = (\sigma_a,\sigma_{a+1}, \dots, \sigma_b), 1\leq \sigma_x \leq n \}$, and computing the trace by inserting a resolution of the identity $\sum_{\sigma} \inprod{\sigma,K(\nu)\sigma} = \sum_{\sigma_1,\dots,\sigma_k} \inprod{\sigma_k, Q_{b_1} \sigma_1}\inprod{\sigma_1,Q_{b_2}\sigma_2} \dots \inprod{\sigma_{k-1}, Q_{b_k} \sigma_k}$ .} that the trace of $K(\nu)$ factorizes to be a product of traces over loops $\gamma\in\nu$, and so we have $\Tr K(\nu) = \prod_{\gamma_\in \nu} \Tr \idty = n^{l(\nu)}$ where $l(\nu)$ is the number of loops in a configuration $\nu$. Now, to handle $\Tr f(L) K(\nu)$, we describe a couple of graphical rules. We may think of the action of $L_x, x\in \Z$ as landing on the loop passing through $x$. We may of course ``slide'' $L_x$ up and down without issue. If $L_x$ encounters a double-bar, then we may use the following observation, similar to that of the $P^{(0)}$ chain in Example 2~\cite{nachtergaele1994quasi}: since $\ket{\xi} = \frac{1}{\sqrt{n}}\sum_{i=1}^n \ket{ii}$ is the trivial representation of $\so(n)$ and $Q = \ket{\xi}\bra{\xi}$ we have
\begin{equation} \label{eq:antisymmetry of projection}
    (L\otimes \idty + \idty\otimes L)\ket{\xi} = 0 , \text{ which gives } (L\otimes \idty)Q = -(\idty\otimes L) Q . 
\end{equation} From this observation, we may freely move the $L_x$ along a loop, imposing antisymmetry whenever we hit a double-bar. We may track the sign of $L_x$ as we move it by assigning the even and odd sublattices a sign, $+$ or $-$ (see Figure~\ref{fig:random_loop}). When two $L_x$'s hit one another, since $L = \ket{\alpha}\bra{\alpha'} - \ket{\alpha'}\bra{\alpha}$, we may quickly calculate powers of these matrices:
\begin{equation}
    L^k = \begin{cases}
        (-1)^{(k-1)/2} L & \text{if } k \text{ is odd} \\
        (-1)^{k/2} \paran{ \ket{\alpha}\bra{\alpha} + \ket{\alpha'}\bra{\alpha'}} & \text{if } k \text{ is even}.
    \end{cases}
\end{equation} If a loop has no copies of $L$, it has trace $\Tr \idty = n$. If any loop has an odd number of $L$'s, it is traceless, and if it has an even number of $L$'s, it has $\Tr L^k = 2 (-1)^{k/2}$. We may take the trace of a loop once all $L$'s have been collected into a single location using this graphical calculus, and then the trace $\Tr f(L) K(\nu)$ still factorizes into a product of traces of loops. This allows us to calculate any $\Tr f(L) K(\nu)$. For instance, if $f(L) = L_x L_y$, then 
\begin{equation}
    \Tr L_x L_y K(\nu) = 2\, n^{l(\nu)-1} (-1)^{\abs{x-y}+1}  I(x\sim y)
\end{equation} where $I(x\sim y)$ is the indicator function which returns 1 when sites $x$ and $y$ are on the same loop. In general, it is clear that given any loop configuration $\nu$, we have a well-defined linear functional $E_\nu$ on the algebra of local observables which computes
\begin{equation}
    E_\nu(f(L)) = \frac{1}{n^{l(\nu)}} \Tr\paran{ f(L) e^{-\beta H} } = \frac{1}{n^{l(\nu)}} \prod_{\gamma\in \nu} F_f(\gamma) , 
\end{equation} where $F_f$ is defined on a loop $\gamma$ by using the graphical calculus which allows us to move $L$'s followed by taking the trace. For particular choices of observable, one may derive closed form expressions for this quantity in terms of e.g. sums over consistent spin configurations--see for example Equation 2.5 in~\cite{bjornberg2021dimerization}.
Now, using the Poisson integral formula (\ref{eq:poisson integral formula}), we may compute expectations of the observables $f(L)$ in the equilibrium state $\omega_\beta$ by normalizing 
\begin{equation}
    \omega_\beta(f(L)) = \frac{\Tr \paran{(f(L)) e^{-\beta H}}}{\Tr e^{-\beta H}} = \int d\mu_{\beta} E_\nu(f(L)),
\end{equation} where $d\mu_{\beta} = Z_\beta^{-1} n^{l(\nu)} d\rho_\beta$ with $Z_\beta = \Tr e^{-\beta H}$ being the partition function. Notice that for example, the correlator
\begin{equation} \label{eq:correlator and probability that x is connected to y}
    \omega_\beta(L_x L_y) \propto (-1)^{\abs{x-y}+1} \text{Prob}_{\mu_\beta}(x\sim y),
\end{equation} namely the probability that $(x,0)$ and $(y,0)$ belong to the same loop. 
The ground state expectations $\omega$ may then be computed by taking the limit $\beta\to \infty$.

\begin{figure}
\centering \includegraphics[scale=0.7]{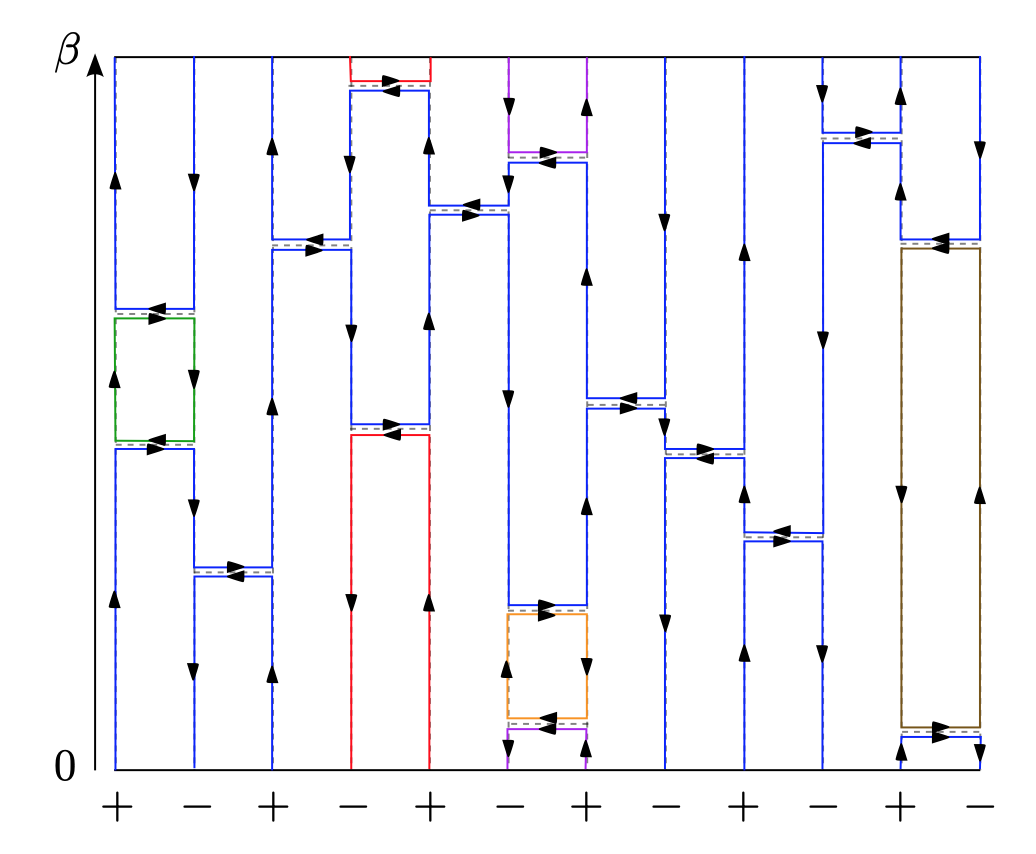}
\captionof{figure}{ (Figure from~\cite{bachmann2014gapped}) One particular loop configuration $\nu$ for the south pole model on the finite chain $[a,b]$. The distribution of ``double bars'' in the space-time $[a,b]\times [0,\beta]$ is given by the Poisson process $\rho_\beta$. At each double bar coordinate $(t_i,e_i)$, the projection $Q^{e_i}$ imposes antisymmetry by (\ref{eq:antisymmetry of projection}). If $f(L)$ acts on this configuration, this antisymmetry allows us to ``slide'' $k$ copies of $L$ on a given loop in a configuration to get a single copy of $\pm L^k$, whence we may compute the trace of $f(L) K(\nu)$ as a product of traces of loops.}
\label{fig:random_loop}
\end{figure}

Now, we wish to describe the excess spin $\sum_{x\geq 1} L_x$ which generates right half-chain rotations. The insight of Bachmann and Nachtergaele was that the exponential of this (formal) sum converges in the GNS representation of the ground state $\omega$. The key observation is that closed loops contribute no spin, which one may see by Equation (\ref{eq:antisymmetry of projection}). Thus, the only possible source of spin comes from loops which intersect both the left and right half-chain, which we may express as loops enclosing the point $(1/2,0)$. When $\omega$ is gapped, it has exponential decay of correlations, and so one can see by Equation (\ref{eq:correlator and probability that x is connected to y}) that the probability that $\nu$ is a configuration containing a loop connecting sites $x$ and $y$  vanishes exponentially in $\abs{x-y}$. We define the approximant $\hat{L}^R(\epsilon)\in \calA_{R}$ for $\epsilon>0$ by
    \begin{equation}
        \hat{L}^R(\epsilon) = \sum_{x\geq 1} e^{-\epsilon x} L_x . 
    \end{equation}

    One can carefully track the estimates in~\cite{bachmann2014gapped} to see that the following Lemma and Theorem may be readily generalized from the $\su(2)\cong \so(3)$ case to the $\so(n)$ case. Note that the correlator decay condition in the original statement of this theorem is automatically satisfied by both south pole ground states, since they are gapped~\cite{bjornberg2021dimerization} and thus have exponential decay of correlations~\cite{nachtergaele2006lieb,hastings2006spectral}.
    \begin{lemma} \label{lem:mimicking bachmann nachtergaele} (Lemma 3.3~\cite{bachmann2014gapped}) 
    Let $\omega$ be one of the south pole ground states. For any configuration $\nu$ of loops $\gamma$, let 
    \begin{equation}\label{eqn:excess spin acts on loops containing origin}
        \hat{L}_\nu^R := \sum_{x\geq 1 : N_{(1/2,0)}(\gamma(x))} L_x , 
    \end{equation} where $\gamma(x)$ denotes the loop to which $(x,0)$ belongs, and $N_{(1/2,0)}(\gamma)$ is the indicator function of $(1/2,0)\in \emph{int}(\gamma)$. then
    \begin{equation}
        \lim_{\eps\to 0} \int d\mu_\infty (\nu) E_\nu \paran{(\hat{L}^R(\eps) - \hat{L}_\nu^R)^2} = 0.
    \end{equation}
    \end{lemma}
    
    \begin{theorem} \label{thm:excess spin random loop} (Theorem 3.4~\cite{bachmann2014gapped}) Let $\omega$ be one of the south pole ground states, and let $(\calH_L\otimes \calH_R,\pi_L\otimes \pi_R,\Omega_\omega)$ denote its GNS representation. Then the strong limit
    \begin{equation}
        \hat{U}^R_g = \slim_{\epsilon\to 0} e^{i g \pi_R(\hat{L}^R(\epsilon))}
    \end{equation} exists and defines an infinite-dimensional, strongly continuous representation of $\hat{U}^R:Spin(n)\to \calU(\calH_R)$.
    \end{theorem}

    We now finally arrive at the main theorem for this section. 
    \begin{theorem} \label{thm:south pole ground state is trivial SPT}
        Let $\omega$ be a ground state of the south pole point. The representation $\hat{U}^R_g:Spin(n)\to \calU(\calH_R)$ descends to a genuine representation of $SO(n)$.
    \end{theorem}
    \begin{proof}
        A consequence of the proof of Theorem \ref{thm:excess spin random loop} in~\cite{bachmann2014gapped} is that since $\omega$ is gapped and so has exponential decay of correlations, the subset of configurations $\nu$ with infinitely many loops containing $(1/2,0)$ is a measure zero subset. Let $\#N_{(1/2,0)}(\nu)$ denote the number of loops in $\nu$ containing $(1/2,0)$.
        Then, letting $A,B\in \calA_{loc}$ be local observables, we can compute matrix elements of this representation in the GNS representation, which for brevity we will here write $(\calH_\omega, \pi_\omega, \Omega_\omega)$: 
        \begin{align*} \inprod{\pi_\omega (A), \hat{U}_g^R \pi_\omega(B)} &= \omega(A^* \exp(
        ig \hat{L}^R_\nu) B) \\
        &= \int d\mu(\nu) E_\nu(A^* \exp(
        ig \hat{L}^R_\nu) B)  \\
        &= \sum_{k\geq 0} \int_{\{\nu: \#N_{(1/2,0)}(\nu) = k\}} d\mu (\nu) E_\nu (A^* \exp(
        ig \hat{L}^R_\nu) B) ,
        \end{align*} where in the third line we have used that $\#N(1/2,0)$ is finite with probability 1. Consider a fixed configuration $\nu$ with $k=\# N_{(1/2,0)}(\nu)$. Unraveling the Definition (\ref{eqn:excess spin acts on loops containing origin}), we see that if $V\cong \C^n$ denotes the defining representation of $\so(n)$, then $\hat{L}_\nu^R$ acts as the tensor representation $V^{\otimes k}$. In particular, since the defining representation of $\so(n)$ $V$ lifts to a genuine representation of $SO(n)$, any tensor power of the defining representation of $\so(n)$ similarly lifts to a genuine representation of $SO(n)$. This means that the matrix elements above exactly correspond to those of a sum of genuine representations, and so $U_g^R$ itself is a genuine representation.
    \end{proof}
    \begin{corollary}\label{cor:south pole ground state is trivial SPT}
        Let $G=SO(n)$. The pair of second cohomology classes $h_\calS$ of the south pole ground states from Theorem~\ref{thm:SPT invariants of ground state space} is a pair of trivial cohomology classes.
    \end{corollary}

\chapter{\texorpdfstring{$SO(n)$}{SO(n)} AKLT Chains}\label{ch:SO(n)_Haldane_chains}

\section{The Frustration-Free \texorpdfstring{$SO(n)$}{SO(n)} AKLT point}
    Let us consider the $SO(n)$ AKLT point in the red phase of the phase diagram, Figure \ref{fig:SO(n) phase diagram}. When $n=3$, this phase is exactly the Haldane phase and the $SO(n)$ AKLT point corresponds to the AKLT chain.
    At this point, the nearest-neighbor interaction $h$ defined in (\ref{def:interaction aSWAP+bQ}) is $h=a(\SWAP - 2Q)$. We may set $a=1$ and add a copy of $\idty$, which does not affect our analysis of ground states except to give a more convenient $h\geq 0$:
    \begin{equation}\label{def:FF interaction h=1+SWAP-2Q}
        h = \idty + \SWAP - 2Q, \qquad H_{\ell} = \sum_{x=1}^{\ell-1} h_{x,x+1} . 
    \end{equation} We will see that this is a frustration-free point (recall Definition \ref{def:frustration free}) in the phase diagram by demonstrating that on finite chains $[1,\ell]$, the Hamiltonian $H_\ell = \sum_{x=1}^{\ell-1} h_{x,x+1}$ and the terms $h_{x,x+1}$ in it share a common eigenvector belonging to their respective ground state spaces. These ground states will be given by MPSs, hence the name. We can quickly compute the lowest energy eigenspace $\calG\subseteq \C^n\otimes \C^n$ of a single copy of $h$ using the irrep decomposition (\ref{eq:irrep decomposition of two site O(n) rep}) and $[\SWAP,Q]=0$: $h$ has eigenvalue $2$ on $M_2$ and eigenvalue $0$ on the ground state space $\calG$, which consists of antisymmetric vectors and the maximally entangled $\ket{\xi}$, 
    \begin{equation} \label{eq:two site ground states MPS}
        \calG = \Exterior^2(\C^n) \oplus \C \ket{\xi} .
    \end{equation}
    
    In Section \ref{sec:the ground states of the MPS point}, we revisit the construction in~\cite{tu2008class} of MPS ground states for the Hamiltonian $H_\ell$ and then use the equivalent finitely correlated states picture~\cite{fannes1992finitely} to see that in the limit of the infinite chain $\Z$, these MPS ground states correspond to exactly one pure ground state $\omega$ when $n$ is odd and a pair of pure 2-periodic ground states $\omega_\pm$ when $n$ is even. This distinction between even and odd behavior was first noted in~\cite{tu2008class}. The odd $n$ case analogizes the AKLT chain to higher dimensions and has been well studied, but the even $n$ case is interesting in its own right and demonstrates a variety of unusual behaviors which we investigate in Section \ref{sec:symmetry breaking from O(n) to SO(n)} and Section \ref{sec:Two Distinct Ground States on Long Chains}. 
    As a result of independent interest, we later prove in Section \ref{sec:parent property (SO(n) chains)} that these MPS ground states are in fact the only ground states for this interaction, making $H$ a parent Hamiltonian.
	
	In the next Section \ref{sec:the ground states of the MPS point}, we construct a class of matrix product ground states and their associated completely positive map $\bbE$. 
	These will recover local expectations of the translation-invariant thermodynamic limit ground state $\omega$.
	It will turn out that $\omega = \frac{1}{2}(\omega_+ + \omega_-)$, where $\omega_\pm$ are pure, 2-periodic ground states admitting matrix product ground states.

% % ---------------------------------------------------
\section{\texorpdfstring{$SO(n)$}{SO(n)} AKLT Chain Ground States}\label{sec:the ground states of the MPS point}
	
	We now construct the matrix product ground states for the finite-chain Hamiltonian $H_\ell$, which was first demonstrated in the language of bond products and spin representations of $\so(n)$ in~\cite{tu2008class}.\footnote{The bond product construction presented in the original~\cite{tu2008class} has much in common with the VBS picture for the AKLT chain we presented in Section \ref{sec:The AKLT Ground State as a Valence Bond State (VBS)}.}
	We start with the bond algebra given by the rank $n$ Clifford Algebra $\calC_n$ generated by operators $\gamma_1,\dots,\gamma_n$ subject to the below anticommutation relations:
    \begin{equation}\label{def:gamma operators (final chap)}
        \gamma_i \gamma_j + \gamma_j\gamma_i = 2\delta_{ij}\idty, \qquad \gamma_i^* = \gamma_i, \qquad 1\leq i,j \leq n . 
    \end{equation} A brief review of Clifford algebras containing several pertinent definitions can be found in Section \ref{sec:Clifford algebras}. When $n$ is even, we have that as associative algebras, $\calC_n\cong M_{2^{n/2}}(\C)$, the $2^{n/2}\times 2^{n/2}$ complex matrices. In this case let the bond algebra be $\calB := \calC_n$. When $n$ is odd, $\calC_n\cong P_+\calC_n \oplus P_-\calC_n$ where $P_+\calC_n\cong P_-\calC_n \cong M_{2^{(n-1)/2}}(\C)$ as associative algebras. In this case we restrict to a subalgebra for our bond algebra $\calB := P_+\calC_n$, noting that the redundant $P_-\calC_n$ generates an equivalent MPS and so can be discarded.
	Our MPSs are then defined via maps $\psi: \calB \to (\C^n)^{\otimes \ell}$:
	\begin{equation} \label{def:MPSs of SO(n) chain}
		\psi(B) := \sum_{i_1,\dots,i_\ell} \Tr \paran{ B \gamma_{i_\ell} \dots \gamma_{i_1}} \ket{i_1, \dots, i_\ell}, \qquad B\in \calB ,
	\end{equation} where for any site $x$, $i_x = 1,\dots, n$ labels an orthonormal basis of $\C^n$.
    \begin{lemma} \label{lem:SO(n) MPS are ground states}
        The matrix product states are ground states of the Hamiltonian $H_{\ell}\geq 0$ defined by (\ref{def:FF interaction h=1+SWAP-2Q}), i.e.
        \begin{equation}
            \calG_{\ell}:= \{\psi(B): B\in \calB\} \subseteq \ker H_{\ell}.
        \end{equation}
    \end{lemma} 
    \begin{proof}
        Take a site in the chain $x\in [1,\ell-1]$. Let us compute using the Clifford relations (\ref{def:gamma operators (final chap)}): for any $B\in \calB$, 
        \begin{align*}
            h_{x,x+1}\psi(B) &= \sum_{i_1,\dots,i_\ell} \Tr \paran{ B \gamma_{i_\ell} \dots \gamma_{i_{x+1}} \gamma_{i_x} \dots \gamma_{i_1}} \ket{i_1, \dots, i_\ell} \\
            & \qquad + \Tr \paran{ B \gamma_{i_\ell} \dots \gamma_{i_{x}} \gamma_{i_{x+1}} \dots \gamma_{i_1}} \ket{i_1, \dots, i_\ell} \\
            & \quad - 2\delta_{i_{x}, i_{x+1}} \Tr \paran{ B \gamma_{i_\ell} \dots \gamma_{i_{x+1}} \gamma_{i_x} \dots \gamma_{i_1}} \ket{i_1, \dots, i_\ell} \\
            &= 0 . 
        \end{align*} Then, $H_{\ell} \psi(B) = \sum_{x=1}^{\ell-1} h_{x,x+1} \psi(B) = 0$, which is what we wanted to show. 
    \end{proof}
 
	Associated with these MPSs is the completely positive (CP) map $\bbE: \calA_x \otimes \calB \to \calB$, where $\calA_x:= M_n(\C)$ denotes the physical on-site spin algebra given by $n\times n$ complex matrices, given by
	\begin{equation} \label{def:E}
		\bbE_A(B) := \bbE(A\otimes B) = T^*(A\otimes B)T = \frac{1}{n} \sum_{i,j = 1}^n A_{ij} \gamma_i B \gamma_j  , 
	\end{equation} where $\sqrt{n}T := \sum_{i=1}^n \ket{i} \otimes \gamma_i$ is a map such that When $n$ is even, $T:\C^{2^{n/2}}\to \C^n \otimes \C^{2^{n/2}}$, and When $n$ is odd, $T:\C^{2^{(n-1)/2}}\to \C^n \otimes \C^{2^{(n-1)/2}}$.
	Observe that $T^*T=\idty$, i.e. $T$ is an isometry. In the even $n$ case where $\calB = \calC_n$, the map $\bbE_A:\calB\to \calB$ is trivially a well defined map for all $A\in \calA_x$. In the odd $n$ case where $\calB=P_+\calC_n$, the map $\bbE_A$ is well-defined since $P_+$ is in the center of $\calC_n$ and so $\bbE_A(P_+ B P_+) = P_+ \bbE_A(B) P_+$.
 
    We turn our attention to the transfer operator $\bbE_\idty:\calB\to\calB$, whose spectral properties control the purity of the thermodynamic limiting states as per Theorem \ref{thm:Primitivity}.
	$\bbE_\idty$ is readily diagonalized once we choose a basis of $\calB$ expressed in products of our Clifford algebra generators $\{\gamma_i\}$--namely, the basis in (\ref{basis for gamma operators}) given in multi-index notation by $\{\gamma_I : \abs{I}\leq n\}$. 
    When $n$ is even, we have 
	\begin{equation} \begin{split} \label{eqn:diagonalization of E even n}
			\bbE_\idty(\idty) &= \idty ,\\ \bbE_\idty(\gamma_0) &= -\gamma_0  ,\\ 
			\bbE_\idty(\gamma_I) &= \lambda_{I} \gamma_I , \quad 1<\abs{I}<n, 
	\end{split} \end{equation} where 
	$\lambda_I = (-1)^{\abs{I}} \frac{(n-2\,\abs{I})}{n}$. 
	Note that this constant depends only on $\abs{I}$, the number of gamma operators, and crucially that $\abs{\lambda_I}<1$ whenever $\gamma_I$ is not $\idty$ or $\gamma_0$. 
    When $n$ is odd, a basis for $\calB = P_+\calC_n$ can be found by starting with the basis (\ref{basis for gamma operators}), projecting by $P_+(\cdot)P_+$, and removing linearly dependent $\gamma_I$. In any case we can quickly compute
    \begin{equation} \begin{split} \label{eq:diagonalization of E odd n}
			\bbE_\idty(P_+) &= P_+ ,\\
			\bbE_\idty(P_+ \gamma_I P_+) &= \lambda_{I} (P_+ \gamma_I P_+) , \quad 1<\abs{I}<n, 
	\end{split} \end{equation} with the same $\abs{\lambda_I} < 1$ for every basis element except $P_+$.

    Diagonalization of $\bbE_\idty$ in hand, we now pass to the finitely correlated state picture. 
	Any such state $\omega$ can be reconstructed using the CP map $\bbE:\calA_x\otimes \calB \to \calB$, combined with a positive element $e\in \calB$ and a positive linear functional $\rho \in \calB^*$ satisfying compatibility conditions $\bbE(\idty\otimes e) = e$ and $\rho(\bbE(\idty\otimes B)) = \rho(B)$ for all $B\in \calB$. 
	Here, choosing $e:=\idty\in \calB$ and $\rho :=\frac{1}{D}\Tr(\cdot)\in \calB^*$ where $D = \Tr(e)$
	does the trick. When $n$ is even, $D = 2^{n/2}$, and When $n$ is odd, $D = 2^{(n-1)/2}$.
    We can then compute the expectations in $\omega$ of local observables $A_1\otimes \dots \otimes A_\ell \in \calA_{[1,\ell]}$ by
	\begin{equation}\begin{split} \label{def:omega}
			\omega(A_1\otimes \dots \otimes A_\ell) &= \rho(\bbE_{A_1}\circ \dots \circ \bbE_{A_\ell}(e)) \\
			&= \frac{1}{D} \Tr \paran{\bbE_{A_1}\circ \dots \circ \bbE_{A_\ell}(\idty)}.
	\end{split}\end{equation} 
    Now, when $n$ is odd, the diagonalization (\ref{eq:diagonalization of E odd n}) reveals that $\bbE_\idty$ is primitive, meaning $\idty$ is the only eigenvector of eigenvalue $1$ and every other eigenvector has eigenvalue $\abs{\lambda}<1$, and so $\omega$ is a pure state (recall Definition \ref{def:primitive MPS}). 
    When $n$ is even, $\bbE_\idty$ has exactly two eigenvalues $\lambda = \pm 1$ with $\abs{\lambda}=1$, and all other eigenvalues have $\abs{\lambda}<1$. This state $\omega$ is not pure, but it is ergodic: Proposition \ref{prop:ergodic and periodic FCS decompositions} then reveals that the $-1$ eigenvalue indicates that $\omega$ decomposes into two pure 2-periodic states $\omega = \frac{1}{2}(\omega_1+\omega_2)$. 
    We will extract explicit descriptions of these states in the next section.

	% % ---------------------------------------------------
	\subsection{Constructing the states $\omega_\pm$ and the CP maps $\bbF^{(1)}, \bbF^{(2)}$}\label{sec:Constructing the pure states omega+- and the CP maps F1,F2}
	In this section, we restrict to the case where $n$ is even. In this case, our diagonalization of the transfer operator revealed that $\omega$ is the equal weight superposition of two 2-periodic pure states $\omega = \frac{1}{2}\paran{\omega_+ + \omega_-}$. We will recover $\omega_\pm$ explicitly and show that they are the pure $2$-periodic FCS we seek by demonstrating that their transfer operators are primitive: this is the content of Corollary \ref{cor:omega+- are pure}.
    We can recover this decomposition by expressing $\bbE_{(\cdot)}\circ \bbE_{(\cdot)}$ (\ref{def:E}) as a product of two CP maps $\bbF^{(1)}_{(\cdot)}, \bbF^{(2)}_{(\cdot)}$. 
	The ordering of these maps, or equivalently, an appropriate change in boundary conditions, will decide whether a state in question is $\omega_+$ or $\omega_-$.

	We begin by recalling the mutually commuting orthogonal projectors (see also (\ref{def:P+ and P-}))
    \begin{equation}
        P_\pm = \frac{1}{2}(\idty \pm \gamma_0) \in \calC_n . 
    \end{equation} We can use $P_+ + P_- = \idty\in \calC_n$ to decompose the right boundary condition from (\ref{def:omega}) into two pieces, yielding a decomposition $\omega = \frac{1}{2}(\omega_+ + \omega_-)$ where $\omega_\pm$ are defined by
		\begin{equation} \label{def:omega pos and neg, E maps}
				\omega_\pm(A_1\otimes\dots\otimes A_\ell) := \frac{2}{D}\Tr\paran{\bbE_{A_1}\circ \dots \circ \bbE_{A_\ell}(P_\pm)}.
		\end{equation}
			Now, observe that $\gamma_0\in \calC_n^{[ev]}$ and so both $P_\pm\in \calC_n^{[ev]}$. 
			Using the definition of the CP map $\bbE$ (\ref{def:E}), for any $A\in \calA_x$ we have the restriction $\bbE_A : \calC_n^{[ev]}\to \calC_n^{[ev]}$, since it is a linear combination of maps of the form $\gamma_i(\cdot)\gamma_j$.	
			Thus, we lose no information restricting our bond algebra $\calB$ from $\calC_n$ to $\calC_n^{[ev]}$.
			
			As discussed in Section \ref{secapp:the even subalgebra}, the even subalgebra decomposes as an associative algebra $\calC_n^{[ev]} = P_+\calC_n^{[ev]}\oplus P_-\calC_n^{[ev]}$. So, we may the projectors $P_\pm$ to block diagonalize any $B\in \calC_n^{[ev]}$:
			\[
				B = P_+ B P_+ + P_- B P_- , \qquad \text{ for all } B\in \calC_n^{[ev]}.
			\]
			Notice that since $\gamma_i\gamma_0 = -\gamma_0 \gamma_i$ for $i=1,\dots, n$, we get that $\gamma_i P_\pm = P_\mp \gamma_i$, and so the isometry $T$ defined in (\ref{def:E}) enjoys the relations
			\begin{equation} \label{eq:T swaps P+ and P-}
				(\idty\otimes P_\pm )T = T P_\mp .
			\end{equation}
			From here, we see that for any $A\in \calA_x$, $\bbE_A:\calC_n^{[ev]}\to \calC_n^{[ev]}$ must swap between the two algebras $P_+\calC_n^{[ev]}$ and $P_-\calC_n^{[ev]}$, since
	       \begin{equation*}\begin{split}
					\bbE_A (P_\pm B P_\pm) &= T^*(A\otimes P_\pm B P_\pm)T \\
					&= P_\mp T^* (A\otimes B) T P_\mp \\
					&=  P_{\mp} \bbE_A(B) P_{\mp} \, .
			\end{split}\end{equation*}
			In particular, $\bbE_A:P_\pm \calC_n^{[ev]} \to P_\mp \calC_n^{[ev]}$. 
			So, if we introduce an automorphism $\alpha: \calC_n^{[ev]}\to\calC_n^{[ev]}$ that maps $P_\mp \calC_n^{[ev]} \to P_\pm \calC_n^{[ev]}$, the map $\alpha\circ \bbE_A$ will respect the decomposition $\calC_n^{[ev]}=P_+\calC_n^{[ev]}\oplus P_-\calC_n^{[ev]}$.
			Defining $\alpha(B) := \gamma_1 B \gamma_1$ will accomplish this, since $\gamma_1 P_\pm = P_\mp \gamma_1$ (note Remark \ref{rem:choice of alpha is arbitrary}):
			\begin{equation*}\begin{split}
					\alpha \circ \bbE_A(P_\pm B P_\pm) &= \alpha (P_{\mp} \bbE_A(B) P_{\mp}) \\
					&= \gamma_1 (P_{\mp} \bbE_A(B) P_{\mp}) \gamma_1 \\
					\qquad &= P_\pm (\alpha \circ \bbE_A(B) ) P_\pm \; .
			\end{split}\end{equation*}
			In other words, the map $\alpha\circ \bbE_A$ respects the algebra decomposition $\calC_n^{[ev]}=P_+\calC_n^{[ev]}\oplus P_-\calC_n^{[ev]}$: 
			\begin{equation}
				\alpha\circ \bbE_A \Big|_{P_\pm \calC_n^{[ev]}}:P_\pm \calC_n^{[ev]} \to P_\pm \calC_n^{[ev]} . 
			\end{equation}
			We are almost ready to construct our pair of CP maps $\bbF^{(1)},\bbF^{(2)}$. But we first need a key identity allowing us to commute $\alpha$ through $\bbE_A$ at the cost of an automorphism on physical spins.
			\begin{proposition} \label{prop:commuting E with alpha identity}
				Define $\sigma:\calA_x\to \calA_x$ by
				\[
					\sigma(A) = RAR \, , \qquad R = \begin{pmatrix} -1 & & & \\ & 1 & & \\ & & \ddots & \\ & & & 1 \end{pmatrix} \, .
				\]
				Then we have for all $A\in\calA_x$,
				\[
				\alpha \circ \bbE_A = \bbE_{\sigma(A)} \circ \alpha ,
				\] where $\alpha: P_\mp \calC_n^{[ev]}\to P_\pm \calC_n^{[ev]}$ is defined by $\alpha(B) = \gamma_1 B \gamma_1$.
			\end{proposition}
			\begin{proof} Let $B\in P_+ \calC_n^{[ev]}$. Then unraveling the definition of $\bbE$ (\ref{def:E}) and using Clifford relations (\ref{def:gamma operators (final chap)}),
				\begin{equation*}\begin{split}
						\alpha \circ \bbE_A(B) &= \frac{1}{n} \sum_{i,j=1}^n A_{ij} \gamma_1 \gamma_i B \gamma_j \gamma_1 \\
						\qquad &= \frac{1}{n} \sum_{i,j=1}^n A_{ij} (-1)^{1-\delta_{1,i}} (-1)^{1-\delta_{1,j}} \gamma_i \gamma_1 B \gamma_1 \gamma_j \\
						&= \frac{1}{n} \sum_{i,j=1}^n (-1)^{\delta_{1,i}+\delta_{1,j}} A_{ij} \gamma_i \alpha(B) \gamma_j \\
						&= \bbE_{RAR} \circ \alpha (B).
				\end{split}\end{equation*} This calculation goes through identically if $B\in P_-\calC_n^{[ev]}$.
			\end{proof}
        \begin{remark}\label{rem:choice of alpha is arbitrary}
            The choice of $\alpha(\cdot) = \gamma_1 (\cdot) \gamma_1$ is one of convenience: choosing any $\gamma_i$ will work, thanks to $\gamma_i P_\pm = P_\mp \gamma_i$. In geometric algebra, these types of automorphisms are commonly interpreted as corresponding to reflections. In principle, one can study other automorphisms or antiautomorphisms swapping the two (isomorphic) subalgebras $P_+\calC_n$ and $P_-\calC_n$, like those arising as involutions or anti-involutions on $\calC_{n-1} \cong \calC_n^{[ev]}$. An interesting case is the transposition map $t:P_\pm \calC_n^{[ev]} \to P_\mp \calC_n^{[ev]}$ which sends $t(\gamma_1\gamma_2 \dots \gamma_n) = \gamma_n \dots \gamma_2 \gamma_1$. In this case, Proposition \ref{prop:commuting E with alpha identity} takes the amusing form $t \circ \bbE_A = \bbE_{A^T} \circ t$ where $A^T$ is the usual matrix transpose. Note that $t$ is not completely positive, and so this cannot be used as-is to extract a finitely correlated state decomposition, but it does hint towards an interesting charge conjugation symmetry and reflection parity breaking story in Proposition \ref{prop:charge conjugation symmetry breaking} and Proposition \ref{prop:reflection parity symmetry breaking}
        \end{remark}
			
		\subsubsection{Finitely correlated state picture for the states $\omega_\pm$}
		Now, let us put together a few observations to study $\omega_\pm$. For brevity, we demonstrate the case of even chain lengths $\ell$, but odd $\ell$ works similarly.
		We adopt the convention of fixing our finite chains to be of the form $[1,2,\dots, \ell]$, although one should keep in mind that to study the thermodynamic limit on $\Z$, we take both sides of the chain to $\pm \infty$.
		Using our definitions of $\omega_\pm$ (\ref{def:omega pos and neg, E maps}) and $\alpha^2 = \ide$, we have 
			\begin{equation}\begin{split} \label{def:omega pos and neg, F maps}
					\omega_+(A_1\otimes \dots \otimes A_\ell) &= \frac{2}{D} \Tr \paran{\bbE_{A_1}\circ \bbE_{A_2} \circ \bbE_{A_3}\circ \dots \circ \bbE_{A_{\ell-1}} \circ \bbE_{A_\ell}(P_+)} \\
					&= \frac{2}{D} \Tr \paran{\bbE_{A_1}\circ \alpha^2 \circ \bbE_{A_2} \circ \bbE_{A_3} \circ \alpha^2\circ \dots \circ \bbE_{A_{\ell-1}} \circ \alpha^2\circ \bbE_{A_\ell}(P_+)} \\
					&= \frac{2}{D} \Tr \paran{\alpha \circ \bbE_{\sigma(A_1)} \circ \alpha \circ \bbE_{A_2} \circ \alpha \circ \bbE_{\sigma(A_3)} \circ \dots \circ \alpha \circ \bbE_{\sigma(A_{\ell-1})} \circ \alpha\circ \bbE_{A_\ell}(P_+)} \\ 
					&=: \frac{2}{D} \Tr \paran{\bbF_{A_1}^{(1)} \circ \bbF_{A_2}^{(2)} \circ \bbF_{A_3}^{(1)} \circ \dots \circ \bbF_{A_{\ell-1}}^{(1)} \circ \bbF_{A_\ell}^{(2)}(P_+)} , 
     \end{split}\end{equation} where we have used Proposition \ref{prop:commuting E with alpha identity} on the third line. Similarly, since $\alpha(P_-) =\gamma_1 P_- \gamma_1 = P_+$, 
            \begin{align*}
					\omega_-(A_1\otimes \dots \otimes A_\ell) &= \frac{2}{D} \Tr \paran{\bbF_{A_1}^{(1)} \circ \bbF_{A_2}^{(2)} \circ \bbF_{A_3}^{(1)} \circ \dots \circ \bbF_{A_{\ell-1}}^{(1)} \circ \bbF_{A_\ell}^{(2)} \circ\alpha^2 \circ(P_-)} \\
					&=: \frac{2}{D} \Tr \paran{\bbF_{A_1}^{(2)} \circ \bbF_{A_2}^{(1)} \circ \bbF_{A_3}^{(2)} \circ \dots \circ \bbF_{A_{\ell-1}}^{(2)} \circ \bbF_{A_\ell}^{(1)} (P_+)},
			\end{align*} In the process we have defined CP maps which ``factorize'' $\bbE_{(\cdot)}\circ \bbE_{(\cdot)}$,
			\begin{equation}\begin{split} \label{def:F maps}
					\bbF^{(1)}&: \calA_x\otimes P_+\calC_n^{[ev]} \to P_+\calC_n^{[ev]}  , \qquad \bbF_A^{(1)} = \alpha \circ \bbE_{\sigma(A)} \Big|_{P_+\calC_n^{[ev]}} , \\
					\bbF^{(2)}&: \calA_x\otimes P_+\calC_n^{[ev]} \to P_+\calC_n^{[ev]} , \qquad \bbF_A^{(2)} = \alpha \circ \bbE_{A} \Big|_{P_+\calC_n^{[ev]}} ,
			\end{split}\end{equation}
            where $\sigma:\calA_x\to\calA_x$ is the involution defined in Proposition \ref{prop:commuting E with alpha identity}.
			In particular, as algebras we have $P_+\calC_n^{[ev]}\cong M_{2^{n/2-1}}(\C)$, so the bond dimension of these states $2^{n/2 - 1}$ is half the original bond dimension $2^{n/2}$ and they recover all expectations of local observables for the states $\omega_\pm$ defined in (\ref{def:omega pos and neg, E maps}). 
			
			Using the operators $\gamma_1\gamma_i$, which generate the even subalgebra $\calC_n^{[ev]}$ from (\ref{def:even Gamma operators}), we can quickly find the Kraus operators $\{F_i\}$ which encode the CP maps $\bbF^{(1)},\bbF^{(2)}$. 
			Define
			\begin{equation}
			F_i := \begin{cases} 
				\idty & i=1 \\
				\gamma_1\gamma_i & i=2,\dots,n
			\end{cases}
			\end{equation} This choice of operators ensures that the collection $\{F_i\}$ has the property $\frac{1}{n}\sum_{i=1}^n F_i^* F_i = \idty$, which in turn will guarantee that the transfer operators $\bbF_\idty^{(1)} =\bbF_\idty^{(2)}$ are unital CP maps. 
			Now, combining the definitions of $\bbE$ (\ref{def:E}) and $\bbF^{(1)},\bbF^{(2)}$ (\ref{def:F maps}) we see for any $B\in P_+\calC_n^{[ev]}\subseteq \calC_n^{[ev]}$, Proposition \ref{prop:commuting E with alpha identity} yields
			\begin{equation} \begin{split} \label{matrix form of F maps}
					\bbF_A^{(1)}(B) &= \frac{1}{n} \sum_{i,j=1}^{n} A_{ij} F_i^* B F_j \;, \qquad 
					\bbF_A^{(2)}(B)
					= \frac{1}{n} \sum_{i,j=1}^{n} A_{ij} F_i B F_j^* \;.
			\end{split}\end{equation}

			Let us write the associated isometries $T_1,T_2$ for the CP maps $\bbF^{(1)},\bbF^{(2)}: \calA_x \otimes P_+\calC_n^{[ev]} \to  P_+\calC_n^{[ev]}$ in terms of the earlier isometry $T = \frac{1}{\sqrt{n}}\sum_{i=1}^n \ket{i}\otimes \gamma_i$. $P_+\calC_n^{[ev]}$ is isomorphic to the algebra of $2^{n/2-1}\times 2^{n/2-1}$ complex matrices, so we can identify $\C^{2^{n/2 -1}}= \im(P_+)\subseteq \C^{2^{n/2}}$ to define maps $T_k: \C^{2^{n/2 -1}}\to \C^n \otimes \C^{2^{n/2 -1}}$, $k=1,2$, where
			\begin{alignat}{2} \label{def:T1 and T2}
				T_1 &= (\idty\otimes \gamma_1)T\Big|_{\im(P_+)} &&= \frac{1}{\sqrt{n}}\sum_{i=1}^n \ket{i} \otimes F_i \\
				T_2 &= T\gamma_1\Big|_{\im(P_+)} &&= \frac{1}{\sqrt{n}}\sum_{i=1}^n \ket{i} \otimes F_i^* ,
			\end{alignat} which are isometries $T_k^*T_k = \idty$. Note that they indeed map $T_k: \im(P_+)\to \C^n\otimes \im(P_+)$ thanks to the algebra decomposition $\calC_n^{[ev]} = P_+\calC_n^{[ev]}\oplus P_-\calC_n^{[ev]}$. This then gives an equivalent definition of $\bbF^{(k)}:\calA_x\otimes P_+\calC_n^{[ev]}\to P_+\calC_n^{[ev]}$ by writing
			\begin{equation}
				\bbF^{(1)}_A(B) = T_1^*(A\otimes B)T_1, \qquad 
				\bbF^{(2)}_A(B) = T_2^*(A\otimes B)T_2. \end{equation} From this expression it is clear that $\bbF_A^{(1)},\bbF_A^{(2)}$ are adjoints of each other with respect to the Hilbert-Schmidt inner product.
                In the special case of $A=\idty$, the transfer operators are in fact equal since $\bbF_\idty^{(1)} = \bbF_{\sigma(\idty)}^{(1)} = \bbF_\idty^{(2)}$ from (\ref{def:F maps}). So, by checking the primitivity condition of this shared transfer operator $\bbF_\idty$, we may determine whether the two 2-periodic FCSs $\omega_\pm$ are the two pure states given by Proposition \ref{prop:ergodic and periodic FCS decompositions}.

        \begin{proposition} \label{prop:F maps are primitive}
			The shared transfer operator $\bbF_\idty:P_+\calC_n^{[ev]}\to P_+\calC_n^{[ev]}$ of both states $\omega_\pm$ is a primitive CP map. In the case of $n=4$, these maps are rank-1 orthogonal projections onto $\idty\in P_+\calC_4$.
		\end{proposition} 
		\begin{proof}
			By definition (\ref{def:F maps}), we have that $\bbF_\idty = \alpha \circ \bbE_\idty$, restricted to the domain $P_+\calC_n^{[ev]}$. Since $\alpha(\cdot) = \gamma_1(\cdot)\gamma_1$ an automorphism, the spectrum of this operator is the same as of $\bbE_\idty$.
			We computed earlier the diagonalization of $\bbE_\idty$ (\ref{eqn:diagonalization of E even n}). 
			Notice that every eigenvector of $\bbE_\idty$ except $\idty$ and $\gamma_0$ have eigenvalues $\abs{\lambda}<1$. 
			Then, when we restrict by projecting $P_+(\cdot)P_+: \calC_n \to P_+\calC_n^{[ev]}$, we have 
			\[
				P_+(\idty)P_+ = P_+ , \qquad P_+(\gamma_0)P_+ = P_+P_- = 0 .
			\] This means in particular that $\bbF_\idty = \alpha\circ \bbE_\idty\Big|_{P_+\calC_n^{[ev]}}$ only has one eigenvalue with $\abs{\lambda}=1$, namely $P_+$ with $\lambda=1$. 
			Thus this eigenvalue is simple. 
			
			To see the case of $n=4$, observe that the diagonalization on $\calC_4^{[ev]}$ explicitly reads 
			\begin{alignat*}{2}
				\bbE_\idty(\idty) &= \idty \\
				\bbE_\idty(\gamma_i\gamma_j) &= 0 \; ,  &&1\leq i < j \leq 4\\
				\bbE_\idty(\gamma_0) &= -\gamma_0 .
			\end{alignat*} So once we restrict via $P_+(\cdot)P_+$, the $\gamma_0$ dies and $\bbF_\idty(P_+) = \gamma_1 P_- \gamma_1 = P_+$, i.e. $\bbF_\idty$ is the orthogonal projection onto the identity $P_+\in P_+\calC_n^{[ev]}$.
		\end{proof}
        \begin{corollary} \label{cor:omega+- are pure}
            The ground states $\omega_\pm$ are pure 2-periodic FCS.
        \end{corollary} While Proposition \ref{prop:ergodic and periodic FCS decompositions} assures us that $\omega$ decomposes into two distinct pure states, one may worry that the states we have constructed here are actually translation invariant and equal, $\omega_+ = \omega_-$. Such worries will be put to rest by Theorem \ref{thm:w+ and w- are distinct dimerized states for l large enough}: there, we will find an observable $A$ (actually, a family of observables) such that $\omega_+(A) \neq \omega_-(A)$.

        We also note a Corollary of independent interest when $n=4$.
		\begin{corollary}
			For $n=4$, the states $\omega_\pm$ exhibit 1-dependence, meaning any local observables separated by at least one site have zero correlation: i.e. for any $j\in [1,\ell]$,
			\[
			\omega_{\pm}(A_1\otimes \dots \otimes A_j \otimes \idty \otimes A_{j+2} \otimes \dots \otimes A_\ell) = 	\omega_{\pm}(A_1\otimes \dots \otimes A_j) \,  \omega_\pm(A_{j+2} \otimes \dots \otimes A_\ell) .
			\]
		\end{corollary}

        Finally, we reproduce a known result from~\cite{tu2008class}.
		\begin{corollary} \label{cor: correlation length diverges as n increases}
			For even $n>4$, the correlation length $\xi$ of $\omega_\pm$ is given by the eigenvalue $\lambda_2$ of eigenvectors $\gamma_i\gamma_j + \star (\gamma_i\gamma_j)$, $1\leq i<j\leq n$ from (\ref{eqn:diagonalization of E even n}):
			\[
				\xi = -\frac{1}{\log \lambda_2} = \frac{1}{\log \paran{\frac{n}{n-4}}}.
			\] This agrees with the value computed in~\cite{tu2008class}.
		\end{corollary}

        Note that the correlation length $\xi$ is an increasing function of $n$ with $\lim_{n\to\infty} \xi(n) = \infty$. Glancing at the phase diagram for $O(n)$-invariant spin chains with nearest-neighbor interactions presented in Figure \ref{fig:SO(n) phase diagram}, this may not be too surprising: in the large $n$ limit, this frustration-free model converges to the Reshitikhin point, an exactly solvable gapless model with power-law correlations \cite{reshetikin1983method}.

	% % ---------------------------------------------------
	
 \section{Symmetry Breaking from \texorpdfstring{$O(n)$}{O(n)} to \texorpdfstring{$SO(n)$}{SO(n)}} \label{sec:symmetry breaking from O(n) to SO(n)}
    We now turn to further studying the fascinating properties of the even $n$ ground states $\omega_\pm$. It is obvious from their definition (\ref{def:omega pos and neg, F maps}) that we can map between $\omega_+$ and $\omega_-$ by applying a single site translation $\tau$ to the left (or the right), $\omega_+\circ \tau = \omega_-$. This type of symmetry breaking from a translation-invariant Hamiltonian to 2-periodic ground states was first observed in~\cite{tu2008class}. But there is another symmetry breaking picture that was, to our knowledge, not previously known: as a surprising consequence of Proposition \ref{prop:commuting E with alpha identity}, we may also map between $\omega_+$ and $\omega_-$ by applying the on-site symmetry $\sigma$. 
	More concretely, for any on-site observables $A_i\in \calA_x$,
	\begin{equation}\begin{split}
			\omega_+(A_1\otimes \dots \otimes A_{\ell}) &= \omega_-(\sigma(A_1) \otimes \dots \otimes \sigma(A_{\ell})) \; .
	\end{split}\end{equation}
	Further, we will show in the next section that these states are invariant under the group $SO(n)$ and we will compute an explicit representation of this group.
	This in turn means that any
	local symmetry $\sigma_M(\cdot) = M(\cdot) M^{-1}$ with $M\in O(n)$ such that $\det{M} = -1$ will map between $\omega_+$ and $\omega_-$ by
	\begin{equation}\begin{split} \label{determinant -1 maps between omega+ and omega-}
			\omega_+( M A_1 M^{-1} \otimes \dots \otimes M A_\ell M^{-1} ) &= \omega_-( (R M) A_1 (RM)^{-1} \otimes \dots \otimes (RM) A_\ell (RM)^{-1} ) \\
			&= \omega_-( A_1 \otimes \dots \otimes A_\ell)
	\end{split}\end{equation} 
	where the second equality holds since $RM\in SO(n)$. 
	So, the $O(n)$ symmetry of our initial Hamiltonian (\ref{def:model Hamiltonian}) is broken into $SO(n)$ symmetry, and we can map between these two dimerized ground states $\omega_\pm$ by applying an on-site symmetry $\sigma_M$ which conjugates each on-site observable by an orthogonal matrix $M$ with $\det(M)=-1$. Recall that any reasonable measure of entanglement is invariant under local unitaries (we stated this for bipartite entanglement entropy in Equation (\ref{eq:entanglement entropy invariant under local unitaries})). This gives us the following theorem.

    \begin{theorem} \label{thm:omega_pm have identical entanglement}
        When $n$ is even, the ground states $\omega_\pm$ have identical entanglement structure.
    \end{theorem}
    \begin{proof}
        By (\ref{determinant -1 maps between omega+ and omega-}), $\omega_+$ and $\omega_-$ are related by an on-site symmetry, which cannot change entanglement.
    \end{proof}
	This is in stark contrast to the usual picture of dimerization. For instance, in Section \ref{sec:Warm-up: the Majumdar-Ghosh Chain Ground States}, we saw that the Majumdar-Ghosh Chain ground states $\wt{\omega}_\pm$ are pure vector states given by alternating products of singlet bonds between neighboring qubits $[i,i+1] = \ket{01}-\ket{10}$, i.e. $\wt{\omega}_+ = \dots \otimes [1,2]\otimes [3,4]\otimes [5,6] \otimes \dots$ while $\wt{\omega}_- = \dots \otimes [2,3]\otimes [4,5]\otimes [6,7] \otimes \dots$. In this case, one can distinguish the two states by measuring the entanglement entropy across a cut between sites $1$ and $2$. Following this theme, it was proven in~\cite{bjornberg2021dimerization} that in an open region surrounding the south pole point of Figure \ref{fig:SO(n) phase diagram}, one can similarly distinguish the pair of dimerized ground states by checking the entanglement entropy across sites $1$ and $2$. Theorem \ref{thm:omega_pm have identical entanglement} asserts that no entanglement measurement will distinguish $\omega_+$ from $\omega_-$.

    We conclude by saying a bit about the $O(n)$ symmetry breaking story.
    At the south pole point of the phase diagram in Figure~\ref{fig:SO(n) phase diagram}, it was noted in~\cite{bjornberg2021dimerization} that the $O(n)$ symmetry is unbroken and the pair of ground states $\wt{\omega}_+,\wt{\omega}_-$ are both separately invariant under the full orthogonal group $O(n)$. In other words, the vector space of ground states $\overline{\calS}_0\supseteq \C \wt{\omega}_+\oplus \C\wt{\omega}_-$ carries two 1D trivial representations of $O(n)$.\footnote{It is strongly believed (but not proven) that $\overline{\calS}_0$ is two-dimensional, hence the $\supseteq$ instead of equality.} But at the $SO(n)$ AKLT point, the $O(n)$ symmetry is broken to $SO(n)$, since the pair of ground states are exchanged by an operator of determinant $-1$. This means that $\C(\omega_+ + \omega_-)$ is the 1D trivial representation of $O(n)$, while $\C(\omega_+ - \omega_-)$ is the 1D sign representation of $O(n)$, which is not equivalent to the trivial representation. So, the vector space of ground states $\overline{\calS}_1 = \C(\omega_+ + \omega_-)\oplus\C(\omega_+ - \omega_-)$ carries a different representation of $O(n)$, and so by Theorem~\ref{thm:bachmann gapped}, we arrive at the following result.
    
    \begin{theorem}
        When $n$ is even, the south pole point and the $SO(n)$ AKLT point of the phase diagram in Figure~\ref{fig:SO(n) phase diagram} occupy distinct $O(n)$-symmetry protected topological phases.
    \end{theorem}
    %-------------------------------------------
    \subsection{Invariance of ground states under the spin representations $\Pi$ of $SO(n)$}
    Recall for even\footnote{The story is essentially unchanged for odd $n$, except that we replace $\C^{2^{n/2}}$ with $\C^{2^{(n-1)/2}}$.} $n$ the spin representations $\Pi:Spin(n)\to \calU(\C^{2^{n/2}})$ from Section \ref{sec:the spin representations Pi of SO(n)}, which pass to a projective representation $\Pi:SO(n)\to \calU(\C^{2^{n/2}})$. We are identifying $\calC_n\cong M_{2^{n/2}}(\C)$, whence $\calU(\C^{2^{n/2}})\subseteq \exp(\calC_n)$.
    In particular, the spin representations enjoyed relation (\ref{eq:representation of SO(n) on clifford algebra}), which we recall below for convenience:
    \begin{equation} \label{eq:representation of SO(n) on clifford algebra (SO(n) chapter)}
			\Pi(w) \gamma_i \Pi(w)^{-1} = \sum_{j} w_{ji} \gamma_j, \qquad w\in SO(n).
	\end{equation}

    This then means that the isometry $T = \frac{1}{\sqrt{n}}\sum_{i=1}^n \ket{i} \otimes \gamma_i$ defined in (\ref{def:E}) intertwines the defining representation $w$ of $SO(n)$ on physical spins $V = \C^n$ with the (projective) representation $\Pi$:
	\begin{equation} \label{eq:T an intertwiner}
			(w\otimes \Pi(w))T = T \,\Pi(w), \qquad w\in SO(n).
	\end{equation}
	\begin{proof} Equation (\ref{eq:representation of SO(n) on clifford algebra (SO(n) chapter)}) immediately implies that
    \begin{equation}
            (\idty \otimes \Pi(w)) \paran{\sum_{i=1}^n \ket{i} \otimes \gamma_i} \Pi^{-1}(w) = (w\otimes \idty) \sum_{i} \ket{i}\otimes \gamma_i ,
    \end{equation} which is equivalent to the identity we desire.
	\end{proof}
 
    Notice that in the case of $n=3$, this is precisely the intertwining relation (\ref{def:MPS AKLT (isometry)}) which defined the MPS tensor for the AKLT chain--here, the defining representation of $SO(3)$ is the spin-1 representation $V_1$ of $SU(2)$, and the spin representation of $SO(3)$ is the spin-1/2 representation $V_{1/2}$. 
    Similarly to this case, for odd $n$, this intertwining relation will essentially determine the symmetry of the resultant MPS and guarantee its invariance under $SO(n)$ (we will actually have full invariance under $O(n)$, since we will later show that this is the unique MPS ground state of an $O(n)$-invariant parent Hamiltonian). In the case of even $n$, a similar relation holds for the isometries $T_1,T_2$ given by (\ref{def:T1 and T2}), which in turn will tell us that the states $\omega_\pm$ are $SO(n)$ invariant.
	We emphasize that these states are \textit{not} invariant under $O(n)$, which will be apparent momentarily. 
    This can be straightforwardly seen by judicious insertion of $\gamma_1^2 = \idty$'s into (\ref{eq:T an intertwiner}), combined with the restriction from $\C^{2^{n/2}}$ to $\im(P_+) \cong \C^{2^{n/2 - 1}}$ we used to define $T_1,T_2$. Recall that $\Pi$ is a reducible representation and splits into irreps $\Pi = \Pi_+\oplus \Pi_-$ which correspond to $\C^{2^{n/2}} = \im(P_+)\oplus \im(P_-)$. After suppressing the explicit $w$ in our notation for readability $\Pi := \Pi(w)$, we have 
        \[
            (w\otimes \Pi) (\idty\otimes \gamma_1) (\idty\otimes \gamma_1) T = (\idty \otimes \gamma_1 )( \idty \otimes \gamma_1) T \Pi,
        \] which, after multiplying on the left by $(\idty \otimes \gamma_1)$ and restricting to the domain $\im(P_+)$ of $T_1$, yields $(w\otimes \alpha(\Pi_+))T_1 = T_1 \Pi_+$.
        A similar computation reveals that $(w\otimes \Pi_+) T_2 = T_2 \, \alpha(\Pi_+)$.
        So, we have the following relations connecting the representation $\Pi = \Pi_+\oplus \Pi_-$ to the isometries $T_1,T_2$:
        \begin{equation}\begin{split} \label{eq:isometries T1 and T2 kinda commute with Pi}
            (w\otimes \alpha(\Pi_+)) T_1 &= T_1 \Pi_+ \\
            (w\otimes \Pi_+) T_2 &= T_2 \, \alpha(\Pi_+) .
        \end{split}\end{equation}
        We will later see in Corollary~\ref{cor:alpha swaps spin reps} that as representations $\alpha(\Pi_+)\cong \Pi_-$.

        These relations allow us to prove that both states $\omega_\pm$ enjoy $SO(n)$ symmetry.
        \begin{theorem}\label{thm:Invariance of ground states under SO(n)}
		When $n$ is even, the states $\omega_\pm$ defined in (\ref{def:omega pos and neg, F maps}) are invariant under local $SO(n)$ symmetry, meaning if $w\in SO(n)$, then 
			\[
				\omega_{\pm}(wA_1w^{-1}\otimes \dots \otimes wA_\ell w^{-1}) = \omega_{\pm}(A_1\otimes \dots \otimes A_\ell).
			\]
        When $n$ is odd, the state $\omega$ is invariant under $O(n)$.
		\end{theorem}\begin{proof}
        The odd $n$ case follows from being the unique finitely correlated ground state of a frustration-free Hamiltonian $H$ with $O(n)$ symmetry. The even $n$ case is the interesting one. 
		Let's apply the relations (\ref{eq:isometries T1 and T2 kinda commute with Pi}) to a short chain expressed using $\bbF^{(k)}_A(B) = T_k^*(A\otimes B)T_k$, where $k=1,2$. The proof for longer chains is a straightforward extension of this, modulo the inconsequential distinction between even and odd length chains. 
        Note that by unitarity, $\Pi(w)^* \Pi(w) = \idty$, and recall from Section \ref{sec:the spin representations Pi of SO(n)} that $\Pi(w)\in \calC_n^{[ev]}$ for all $w\in SO(n)$. For readability, we abuse notation to write $\Pi:= \Pi_+(w)$.
		\begin{align*}
			\omega_{+}(w^{-1}A_1w\otimes w^{-1}A_2w) &= \frac{2}{D} \Tr \; \paran{\bbF_{w^{-1}A_1w}^{(1)} \circ \bbF_{w^{-1}A_2w}^{(2)} (P_+)} \\
			 &= \frac{2}{D} \Tr \; T_1^* \paran{w^{-1}A_1w \otimes T_2^*(w^{-1} A w \otimes \Pi^* \Pi P_+ \Pi^* \Pi)T_2}T_1 \\
			&= \frac{2}{D} \Tr \; T_1^* \paran{w^{-1}A_1w \otimes \paran{ \alpha(\Pi)^*( T_2^*( A \otimes P_+ )T_2 } \alpha(\Pi)}T_1 \\
			&= \frac{2}{D} \Tr \; \Pi^* \paran{T_1^* \paran{A_1 \otimes T_2^*( A \otimes P_+ )T_2 }}T_1 \Pi \\	
			&= \frac{2}{D} \Tr \; \Pi^* \paran{\bbF_{A_1}^{(1)} \circ \bbF_{A_2}^{(2)} (P_+)}  \Pi \\
			&= \omega_+(A_1\otimes A_2), 
		\end{align*} 
		where third line follows from Equation (\ref{eq:isometries T1 and T2 kinda commute with Pi}) and from $\Pi P_+ \Pi^* = P_+ \Pi \Pi^* = P_+$, since $P_+$ is in the center of $\calC_n^{[ev]}$. Likewise, since again $\alpha(\Pi) P_+ \alpha(\Pi)^* = P_+$, we can perform the same computation for $\omega_-$ and see 
		\begin{align*}
			\omega_{-}(w^{-1}A_1w\otimes w^{-1}A_2w) &= \frac{2}{D} \Tr \; \alpha(\Pi)^* \paran{\bbF_{A_1}^{(2)} \circ \bbF_{A_2}^{(1)} (P_+)}  \alpha(\Pi) \\
			\qquad &= \omega_-(A_1\otimes A_2).
		\end{align*} 
		\end{proof}
        We have still not demonstrated that the states $\omega_\pm$ are indeed distinct. Indeed, in the next section we prove this, but strangely, even though they are 2-periodic, they are only distinguishable by observables supported on $\geq n/2$ sites.

	% % ---------------------------------------------------
	
 \section{Two Distinct Ground States on Long Chains} \label{sec:Two Distinct Ground States on Long Chains}
    In this section we finally prove in Theorem \ref{thm:w+ and w- are distinct dimerized states for l large enough} that the ground states $\omega_\pm$ are distinct. This is accomplished by using MPSs to recover reduced density matrices $\rho_\ell^\pm$ on chains of length $\ell$. Along the way, we discover yet another peculiar feature of these ground states, accompanying their unique $O(n)$-to-$SO(n)$ symmetry breaking from Section \ref{sec:symmetry breaking from O(n) to SO(n)}: Corollary \ref{cor:rho +- are identical for small l} reveals that while the support spaces of $\rho_\ell^{\pm}$ are orthogonal for chains of length $\ell \geq n$, the two reduced density matrices are identical for small length chains $\ell < n/2$. 
	This again contrasts with the usual picture of dimerization. The dimerized ground states of the Majumdar-Ghosh model in Section \ref{sec:Warm-up: the Majumdar-Ghosh Chain Ground States} may be distinguished by computing the expectation of the two-body operator $P_{i,i+1}$ given by the orthogonal projection onto the singlet $\ket{01}-\ket{10}$. 
    In an open region surrounding the south pole in Figure \ref{fig:SO(n) phase diagram}, it was proven in \cite{aizenman1994geometric} that dimerization can be detected with an observable supported on two neighboring sites.
	But here, no observable supported on small $\ell<n/2$ chains can detect dimerization, even though the two states are 2-periodic and differ only by a one site translation. The distinction may only be detected on $\ell=n/2$ sites.

    \subsection{MPSs corresponding to $\omega_\pm$ and a suggestive example} \label{sec:MPSs corresponding to omega_pm and a suggestive example}
	Recall the first definition of our MPSs (\ref{def:MPSs of SO(n) chain}). When $n$ is odd, $\calB = P_+\calC_n$, and we will later prove the maps $\psi:\calB\to (\C^n)^{\otimes \ell}$ are injective when $\ell\geq n$. This case requires no adjustment.
    When $n$ is even, the initial bond algebra was $\calB = \calC_n$, but we had to perform a blocking procedure on even $\ell$ chains to produce pure 2-periodic states. We then restricted to the subalgebra $\calC_n^{[ev]}$ to take positive boundary conditions for the blocked FCSs. The necessity of this can similarly be seen at the MPS level. 
    Lemma \ref{lem:traces of products of gamma operators} tells us that $\Tr(\gamma_I) =0$ whenever $\gamma_I$ is not a scalar multiple of $\idty$. As a consequence, we have that whenever $\abs{I}\text{mod } 2 \neq \ell \text{ mod } 2$, $\psi(\gamma_I) = 0$. So, the MPSs with the naive bond algebra $\calC_n$ are never truly injective; this however is easily resolved, simply by restricting to the vector spaces of even or odd elements $\calC_n^{[ev]},\calC_n^{[odd]}$ in $\calC_n$, according to the parity of the chain length $\ell$. For the following discussions, we take the case where $\ell$ is even, and so we restrict the domain of our MPSs to obtain $\psi: \calC_n^{[ev]}\to (\C^n)^{\otimes \ell}$. With this restriction, we will see that $\psi$ is injective for $\ell\geq n$. The corresponding results for odd $\ell$ are straightforward to extract and essentially the same, only requiring us to treat the domain of our MPSs as $\calC_n^{[odd]}\cong P_+\calC_n P_- \oplus P_- \calC_n P_+$.

    Let $n$ even. The MPSs $\psi_+$ and $\psi_-$ associated to $\omega_+$ and $\omega_-$ can be expressed in several ways. One is to define them as a restriction of $\psi$ to the subalgebras $P_+\calC_n^{[ev]} = P_+\calC_n P_+$ and $P_-\calC_n^{[ev]} = P_-\calC_n P_-$, respectively.\footnote{For odd $\ell$, the corresponding restriction is to the off-diagonal blocks $\calC_n^{[odd]} = P_+\calC_n P_- \oplus P_-\calC_n P_+$.}
    Equivalently, we may use that the automorphism $\alpha$ swaps matrix blocks\footnote{For odd $\ell$, we have that $\alpha(P_+\calC_n P_-) = P_-\calC_n P_+$.}
    \begin{equation}
        \alpha(P_+\calC_n P_+) = P_-\calC_n P_- . 
    \end{equation} Then, we may start with the pure states $\omega_\pm$ from Definition (\ref{def:omega pos and neg, F maps}) and use the CP maps $\bbF^{(1)},\bbF^{(2)}$ with their associated isometries $T_1,T_2$ defined by (\ref{def:T1 and T2}) to write $\psi_\pm: P_+\calC_n^{[ev]} \to (\C^n)^{\otimes \ell}$ in the following way, where for $\psi_-$ we think of $P_+\calC_n^{[ev]} = \alpha(P_-\calC_n^{[ev]})$: 
	\begin{equation}\begin{split} \label{def of matrix product states}
	\psi_+(B) &= \sum_{i_1,\dots,i_\ell}^{n} \Tr(B F_{i_\ell}^*F_{i_{\ell-1}}\dots F_{i_2}^*F_{i_1})\ket{i_1, i_2, \dots, i_{\ell-1},i_{\ell}}, \qquad B\in P_+\calC_n^{[ev]} \\
     &= \sum_{i_1,\dots,i_\ell}^{n} \Tr(B \gamma_{i_\ell} \gamma_{i_{\ell-1}} \dots \gamma_{i_2} \gamma_{i_1})\ket{i_1, i_2, \dots, i_{\ell-1},i_{\ell}} \\
	\psi_-(B) &= \sum_{i_1,\dots,i_\ell}^{n} \Tr(B F_{i_\ell}F_{i_{\ell-1}}^*\dots F_{i_2}F_{i_1}^*)\ket{i_1, i_2, \dots, i_{\ell-1},i_{\ell}}, \qquad B\in P_+\calC_n^{[ev]}\\
    &= \sum_{i_1,\dots,i_\ell}^{n} \Tr(\alpha(B)\gamma_{i_\ell} \gamma_{i_{\ell-1}} \dots \gamma_{i_2} \gamma_{i_1})\ket{i_1, i_2, \dots, i_{\ell-1},i_{\ell}}, 
	\end{split}\end{equation} where we have twice used that $F_i^*F_j = \gamma_i\gamma_j$ and $F_i F_j^* = \alpha(\gamma_i\gamma_j)$, along with cyclicity of trace. In any case, it is computationally convenient to note that if we use the above expressions and allow $\psi_\pm$ to have domain $\calC_n^{[ev]}$,\footnote{For odd $\ell$, we extend to the domain $\calC_n^{[odd]}$.} we obtain the following useful expression: 
    \begin{equation} \label{eqn:alpha swaps psi_+ with psi_-}
        \psi_-(B) = \psi_+(\alpha(B)) \quad \text{for all } B\in \calC_n^{[ev]}
    \end{equation}

    In what follows, we will make repeated use of the basis of $P_+\calC_n$ given by $\gamma_I + \star\gamma_I, \, \abs{I}\leq n/2$ from Section \ref{sec:Clifford algebras}.
    \begin{example}
    Let $n=4$ and chain length $\ell=2$. Let $D=\Tr \idty$. Choose $B = \gamma_I+\star \gamma_I$ with $\gamma_I := \gamma_1\gamma_2$ and compute using Lemma \ref{lem:traces of products of gamma operators}:
	\begin{align*}
		\frac{1}{D}\psi_+(\gamma_I + \star\gamma_I) &= \frac{1}{D}\sum_{i_1,i_2}^{n} \Tr((\gamma_1\gamma_2-\gamma_3\gamma_4) \gamma_{i_2}\gamma_{i_1} )\ket{i_1, i_2} \\
			&= \ket{12} - \ket{21} - \ket{34} + \ket{43} .
	\end{align*} Continuing this example, we can use Equation (\ref{eqn:alpha swaps psi_+ with psi_-}) to write
	\begin{align*}
		\frac{1}{D}\psi_-(\gamma_I + \star\gamma_I) &= \frac{1}{D}\sum_{i_1,i_2}^{n} \Tr(\alpha(\gamma_1\gamma_2-\gamma_3\gamma_4) \gamma_{i_2} \gamma_{i_1} )\ket{i_1, i_2} \\
        &= \frac{1}{D}\sum_{i_1,i_2}^{n} \Tr(-\gamma_1\gamma_2-\gamma_3\gamma_4) \gamma_{i_2} \gamma_{i_1} )\ket{i_1, i_2} \\
			&= - \ket{12} + \ket{21} - \ket{34} + \ket{43} .
	\end{align*} Here, it is clear that these two states are orthogonal. A key point to note is that if we had instead chosen $B= \idty + \star \idty = \idty + \gamma_0 = 2P_+$, then the states $\psi_+$ and $\psi_-$ are not orthogonal on $\ell=2$ sites: since $\psi_+(\gamma_0) = 0$ by Lemma \ref{lem:traces of products of gamma operators} and since $\alpha(\idty)=\idty$, we have
    \[
        \psi_+(\idty + \gamma_0) = \psi_+(\idty) + \psi_+(\gamma_0) = \psi_+(\idty) = \psi_- (\idty) = \psi_- (\idty + \gamma_0). 
    \] Indeed, the proof of lack of distinguishibility for small chains follows from a similar computation. If $\abs{I}<n/2$, then $\abs{\star I} > n/2$ and so by Lemma \ref{lem:traces of products of gamma operators}, the contribution of $\star \gamma_I$ is zero whenever $\ell < n/2$, making it impossible to distinguish $\ket{\psi_+}$ from $\ket{\psi_-}$. This is the content of Corollary \ref{cor:rho +- are identical for small l}.
    \end{example}
    
    Before we can prove the full Theorem \ref{thm:w+ and w- are distinct dimerized states for l large enough}, we need to better understand the $SO(n)$ representations on the bond algebra, which will tightly constrain the spectrum of the reduced density matrices $\rho_\ell^\pm$ and allow us to later prove the parent property Theorem \ref{thm:parent property (SO(n) chains)}.
    
  %~~~~~~~~~~~~~~~~~~~~~~~~~~~~~~~~~~~~~~~~~~~~~~~~~~~~~~  
    \subsection{$SO(n)$ representations on MPSs} \label{SO(n) representations on MPS}
    Let $V=\C^n$ denote the defining representation of $SO(n)$, as before. Our goal now is to understand the ground state space as an invariant subspace of the tensor power of the defining representation:
    \begin{equation}
        \calG_\ell := \{\psi(B): B\in \calB\} \subseteq V^{\otimes \ell}.
    \end{equation}
    Recall the intertwining relation (\ref{eq:T an intertwiner}).
    Working naively on the initial bond algebras $\calB = P_+\calC_n$ for odd $n$ and $\calB = \calC_n$ for even $n$, we have for all $B\in \calB$ and $w\in SO(n)$
	\begin{equation}\label{eqn:physical rep lifts to bond algebra rep}
		w^{\otimes \ell}\psi(B) = \psi\paran{\Pi(w)(B)\Pi(w)^{-1}}, 
	\end{equation} where we recall that $\Pi(\cdot)\Pi^{-1}$ is a true representation of $SO(n)$ on $\calC_n$ since the phases cancel. Refer to Section \ref{sec:The Adjoint of the Spin Representation Pi (-) Pi} for more information on this representation, which we will use heavily from here. Equation (\ref{eqn:physical rep lifts to bond algebra rep}) exactly means that the maps $\psi$ intertwine the representations $w^{\otimes \ell}$ acting on $\calG_\ell$ and $\Pi(\cdot)\Pi^{-1}$ acting on $\calC_n$.
    We will later prove an injectivity statement for the MPSs in Theorem \ref{thm:w+ and w- are distinct dimerized states for l large enough}: When $n$ is odd, the MPSs are injective for $\ell>n/2$, and When $n$ is even, the MPSs are injective for $\ell>n$ after the aforementioned restriction to the invariant subspace $\calC_n^{[ev]}$ for even $\ell$ (and $\calC_n^{[odd]}$ for odd $\ell$). Once the MPSs are injective, we have that $\calB \cong \calG_\ell$ as representations.
    Lemma \ref{lem:rep of Pi( )Pi^-1 on bond algebras} addresses exactly the decomposition of this representation into irreps. Thus we get Corollaries \ref{cor:reps of odd ground state spaces} and \ref{cor:reps of even ground state spaces}. The odd $n$ case is a bit simpler.
        
    \begin{corollary} \label{cor:reps of odd ground state spaces}
        Let $\calB= P_+\calC_n$ and $V=\C^n$ be the defining representation of $SO(n)$. If the maps $\psi:\calB \to (\C^n)^{\otimes \ell}$ are injective, then Lemma \ref{lem:rep of Pi( )Pi^-1 on bond algebras} tells us that as representations of $SO(n)$, the ground state space $\calG_\ell$ is given by
        \begin{align*}
            \calG_\ell := \{\psi(B) : B\in \calB\} &\cong \Exterior^0 V \oplus \Exterior^2 V \oplus \dots \oplus \Exterior^{n-1} V \\
            &\cong \Exterior^1 V \oplus \Exterior^3 V \oplus \dots \oplus \Exterior^{n} V,
        \end{align*} We may pick bases for these irreps by
        \begin{equation}
            \Exterior^k V = \text{span} \{ \psi(\gamma_I + \star \gamma_I): \abs{I} = k \} . 
        \end{equation}
        \end{corollary}

    Onto the even $n$ case. Lemma \ref{lem:rep of Pi( )Pi^-1 on bond algebras} tells us that the irrep decomposition of $\Pi(\cdot)\Pi^{-1}$ acting on $\calC_n = \calC_n^{[ev]}\oplus \calC_n^{[odd]}$ is  
        \begin{align*}
            \calC_n^{[ev]} &\cong \Exterior^0 V \oplus \Exterior^2 V \oplus \dots \oplus \Exterior^{n-2}V \oplus \Exterior^{n} V \\
            \calC_n^{[odd]} &\cong \Exterior^1 V \oplus \Exterior^3 V \oplus \dots \oplus \Exterior^{n-3}V \oplus \Exterior^{n-1} V ,
        \end{align*} noting that the middle exterior power splits as $\Exterior^{n/2}V\cong U_+\oplus U_-$ where $U_+\subseteq P_+\calC_n^{[ev]}$ and $U_-\subseteq P_-\calC_n^{[ev]}$ are distinct irreps of the same dimension. Again, per the blocking procedure, MPSs of odd length $\ell$ may be thought of as having domain $\calC_n^{[odd]}$, and MPSs of even length $\ell$ may be thought of as having domain $\calC_n^{[ev]}$. 

        \begin{lemma} \label{lem:alpha swaps P_+ and P_-} 
		Let $n$ even and let $\gamma_I\in \calC_n$. 
		Then the automorphism $\alpha(\cdot) = \gamma_1(\cdot)\gamma_1$ sends
			\[
				\alpha(\gamma_I + \star \gamma_I) = (-1)^{s_I}(\gamma_I - \star \gamma_I),
			\] where $s_I\in \N$ depends on $I$.
		\end{lemma}
	\begin{proof}
    We know that for any fixed product $\gamma_I$, we have $\gamma_1 \gamma_I \gamma_1 = (-1)^{s_I} \gamma_I$ where $s_I$ tracks the sign. Observe that
        \begin{align*}
            \alpha(\gamma_I + \star \gamma_I) = \gamma_1 \gamma_I \gamma_1 + \gamma_1 \star \gamma_I \gamma_1 
            = \gamma_1 \gamma_I \gamma_1 + \gamma_1 \gamma_0 \gamma_I \gamma_1 
            &= \gamma_1 \gamma_I \gamma_1 - \gamma_0 \gamma_1 \gamma_I \gamma_1 \\
            &= \gamma_1 \gamma_I \gamma_1 - \star(\gamma_1 \gamma_I \gamma_1) \\
            &= (-1)^{s_I}(\gamma_I - \star \gamma_I),
        \end{align*} where in the second line we have used that $\star \gamma_I = \gamma_0 \gamma_I$ by (\ref{eqn:hodge dual is gamma0 multiplication}) and in the third line we have used that $n$ is even and so $\gamma_0\gamma_1=-\gamma_1\gamma_0$.
		\end{proof}
        Recall that the basis elements $\gamma_I+\star\gamma_I$ are in the image of $P_+$, while the elements $\gamma_I-\star \gamma_I$ are in the image of $P_-$. Thus using the explicit bases provided by Lemma \ref{lem:rep of Pi( )Pi^-1 on bond algebras}, the above Lemma \ref{lem:alpha swaps P_+ and P_-} reveals that $\alpha$ is a automorphism mapping the $SO(n)$ representations in the following way:
        \begin{equation}
            \alpha(\Exterior^k V) = \Exterior^{n-k} V \cong \Exterior^k V, \quad \alpha(U_\pm) = U_\mp ,
        \end{equation} where we have recalled from Proposition \ref{prop:ext k and ext n-k are isomorphic} that as representations of $SO(n)$, $\Exterior^k V \cong \Exterior^{n-k} V$. In particular, since $\psi_+\circ \alpha = \psi_-$, we obtain the irrep decomposition corresponding to the states $\psi_\pm$ for sufficiently long ($\ell\geq n$) chains. Of particular importance are the irreps $U_\pm$, so we opt to only state the parities of $\ell$ for which these irreps are visible.

        \begin{corollary} \label{cor:reps of even ground state spaces}
        Let $n\bmod{4} = 0$. Let $\ell$ be even so that $\calB:= P_+\calC_n^{[ev]}$. If $\psi_\pm: \calB \to (\C^n)^{\otimes \ell}$ are injective, then Lemma \ref{lem:rep of Pi( )Pi^-1 on bond algebras} tells us that as representations of $SO(n)$, the ground state spaces $\calG_\ell^+ \oplus \calG_\ell^- = \calG_\ell$ are given by 
            \begin{align*}
                \calG_\ell^+ := \{\psi_+(B): B\in \calB\} &\cong \Exterior^0 V \oplus \Exterior^2 V \oplus \dots \oplus \Exterior^{n/2 - 2} V \oplus U_+ \\
                \calG_\ell^- := \{\psi_-(B): B\in \calB\} &\cong \Exterior^0 V \oplus \Exterior^2 V \oplus \dots \oplus \Exterior^{n/2 - 2} V \oplus U_- .             
            \end{align*}
        Likewise, let $n\bmod{4} = 2$, and let $\ell$ be odd so $\calB := P_+\calC_n^{[odd]} = P_+\calC_n P_-$. If $\psi_\pm:\calB\to (\C^n)^{\otimes \ell}$ are injective, then 
            \begin{align*}
                \calG_\ell^+ := \{\psi_+(B): B\in \calB\} &\cong \Exterior^1 V \oplus \Exterior^3 V \oplus \dots \oplus \Exterior^{n/2 - 2} \oplus U_+ \\
                \calG_\ell^- := \{\psi_-(B): B\in \calB\} &\cong \Exterior^1 V \oplus \Exterior^3 V \oplus \dots \oplus \Exterior^{n/2 - 2} \oplus U_- .          
            \end{align*}

        We may pick bases for these irreps by
        \begin{align*}
            \calG_\ell^+:& \quad \Exterior^k V = \text{span} \{ \psi(\gamma_I + \star \gamma_I): \abs{I} = k \} , \quad U_+ = \text{span} \{ \psi(\gamma_I + \star \gamma_I): \abs{I} = n/2 \} \\
            \calG_\ell^-:& \quad \Exterior^k V = \text{span} \{ \psi(\gamma_I - \star \gamma_I): \abs{I} = k \} , \quad U_- = \text{span} \{ \psi(\gamma_I - \star \gamma_I): \abs{I} = n/2 \}
        \end{align*}
        \end{corollary} 
        It is worth observing that the only irreps distinguishing $\calG_\ell^+$ from $\calG_\ell^-$ are the $U_\pm$ irreps, which are contained in $\Exterior^{n/2}V$ and so are only supported in the MPS spaces when $\ell\geq n/2$. This is actually sufficient to demonstrate that $\omega_+ \neq \omega_-$, since $\psi_\pm(\gamma_I +\star \gamma_I)\neq 0$ when $\ell\geq n/2$ and Schur orthogonality of the distinct irreps $U_+$ and $U_-$. But with a little more work, we will obtain a stronger result that the representation content of $\calG_\ell$ can only suggest:\footnote{It may well be that two MPSs are identical as representations, but not be the same MPSs.} that $\calG_\ell^+$ and $\calG_\ell^-$ are completely identical for $\ell<n/2$.
        
        Before proceeding, we record one more corollary. It will not be used to prove anything, but it provides another perspective into the automorphism $\alpha$ which swaps $\omega_+$ with $\omega_-$.
        \begin{corollary} \label{cor:alpha swaps spin reps}
            The spin representations $\Pi\cong \Pi_+ \oplus \Pi_-$ of $Spin(n)$ are given by the action of invertible elements in $\calC_n^{[ev]}$ on $\C^{n/2}\cong \im(P_+)\oplus \im(P_-)$. In particular, we can write the representatives as $\Pi_\pm(w) = \gamma_I \pm \star \gamma_I$ with $\abs{I}$ even. Lemma \ref{lem:alpha swaps P_+ and P_-} then means that $\alpha$ swaps the two half-spin representations, i.e. 
            \[
                \alpha(\Pi_\pm(w)) = (-1)^{s_w}\Pi_\mp(w) \qquad  \text{for all } w\in Spin(n) ,
            \] where $s_w$ is an integer tracking the sign as before.
        \end{corollary}

    \subsection{Proving Distinctness}
    We can now proceed to the main theorem. Again, we state the version of the theorem for even $\ell$ to match the expressions of $\omega_\pm$ as FCSs, but it holds equally well for odd $\ell$ by restricting the domain of $\psi$ from $\calC_n$ to $\calC_n^{[odd]}$ (and so $\psi_\pm$ have domain $P_+\calC_n P_-$, where $\psi_- = \psi_+\circ \alpha$).
    
	\begin{theorem} \label{thm:w+ and w- are distinct dimerized states for l large enough}
			Let $\ell \geq n$. 
			Take the basis $\{\gamma_I + \star \gamma_I : \, \abs{I}\leq n/2\}$ of $\calB = P_+\calC_n^{[ev]}$ defined in (\ref{eqn:basis gamma_I + star gamma_I}), and define normalized matrix product vectors $\varphi_\pm: P_+\calC_n^{[ev]}\to (\C^n)^{\otimes\ell}$ using (\ref{def of matrix product states})
		\begin{equation*}\begin{split}
			\varphi_\pm^{(I)} &:= c_\pm(I) \psi_\pm(\gamma_I+\star\gamma_I),
		\end{split}\end{equation*} where $c_\pm(I)$ is a normalization constant ensuring $\norm{\varphi_\pm^{(I)}}=1$. 
	
		Let $\rho_\ell^\pm$ denote the reduced density matrices of $\omega_\pm$ restricted to the interval $[1,\dots,\ell]$. 
		Then there are eigenvalues $\mu_{\,\abs{I}}(\ell)\in (0,1]$, depending only on the length of the chain $\ell$ and the size of $\abs{I}$, such that
		\begin{equation*}
				\rho_\ell^\pm = \sum_{\abs{I}\leq n/2} \mu_{\,\abs{I}} \ket{\varphi_\pm^{(I)}}\bra{\varphi_\pm^{(I)}} .
		\end{equation*}
		In particular, the set of eigenvectors $\{\varphi_\pm^{(I)}\}$ is a mutually orthogonal set, so the support spaces of $\rho_\ell^+$ and $\rho_\ell^-$ are orthogonal and thus $\rho_\ell^+\neq \rho_\ell^-$. 
  
    The entanglement structure of these two reduced density matrices is identical since $\rho_\ell^+ = R^{\otimes \ell}\rho_\ell^- R^{\otimes \ell}$, where $R$ is given in Prop. \ref{prop:commuting E with alpha identity}.
		
		In the large chain limit $\ell\to\infty$, the eigenvalues converge to a single value $\mu_{\,\abs{I}}(\ell) \to 1/2^{n-2}$.
		\end{theorem}
	
		\begin{proof}
			The eigenvectors of our reduced density matrices are all recoverable as matrix product states.

            \vspace{5mm} 
			\noindent
   \textbf{Orthogonality of eigenvectors}
    
    We begin by noting that multi-index notation can handle products of $F_i$ and its adjoint, e.g. 
        \[
			F_{i_\ell}^*F_{i_{\ell-1}} \dots F_{i_2}^* F_{i_1} = \pm \gamma_K,
		\] where $\abs{K}\leq \min(\ell, n)$ since some $F_{i_k}$ may be $\idty$. Combining this observation with Lemma \ref{lem:traces of products of gamma operators} and its corollary, we may readily compute matrix coefficients of $\varphi_\pm$. In particular, using the Lemma and orthogonality of $\ket{i_k}$ it is straightforward to verify that if $I\neq J$ where $\abs{I},\abs{J}\leq n/2$, then 
	   \[
		  \inprod{\varphi(\gamma_I + \star\gamma_I), \varphi(\gamma_J + \star \gamma_J)} = 0,
	   \] and likewise for $\varphi_-^{(I)}$. Direct computation reveals that $\varphi_\pm^{(I)}$ is nonzero whenever $\ell \geq \abs{I}$, so since $\abs{I}\leq n/2$, choosing $\ell \geq n/2$ is sufficient to show that $\varphi^{(I)}_+$ and $\varphi^{(I)}_-$ are separately injective linear maps. 
    
        The interesting case then is orthogonality of $\varphi_+^{(I)}$ and $\varphi_-^{(I)}$. Recall observation (\ref{eqn:alpha swaps psi_+ with psi_-}), which states that $\varphi_-(B) = \varphi_+\circ \alpha(B)$ for all $B\in \calC_n^{[ev]}$, and Lemma \ref{lem:alpha swaps P_+ and P_-} to see that 
        \begin{equation} \begin{split}\label{relation of + and - matrix product states via alpha}
			\varphi_-(\gamma_I + \star \gamma_I) &= (-1)^{s_I} \varphi_-(\alpha(\gamma_I - \star \gamma_I) \\
				&= (-1)^{s_I} \varphi_+(\gamma_I - \star \gamma_I). 
		\end{split} \end{equation} We will demonstrate they are orthogonal on chains of length $\ell=n$, which then implies orthogonality for chains of longer length $\ell\geq n$.
        Using relation (\ref{relation of + and - matrix product states via alpha}) and that $\inprod{\varphi(\gamma_I),\varphi(\star\gamma_I)}=0$ since the multi-indices $I\neq \star I$, we may compute:
        \begin{align*}
				\langle\varphi_+ (\gamma_I + & \star\gamma_I),  \varphi_-(\gamma_I + \star \gamma_I)\rangle  \\
				&= (-1)^{s_I} \inprod{\varphi_+(\gamma_I + \star \gamma_I), \varphi_+(\gamma_I - \star \gamma_I)} \\
				&= (-1)^{s_I} \paran{\inprod{\varphi_+(\gamma_I), \varphi_+(\gamma_I)} - \inprod{\varphi_+(\star\gamma_I), \varphi_+(\star\gamma_I)} } \\
				&= (-1)^{s_I} \paran{\sum_{K} \abs{\Tr(\gamma_I\cdot \gamma_K)}^2 - \sum_{K'} \abs{\Tr((\star\gamma_I)\cdot \gamma_{K'})}^2 } \\
				&= (-1)^{s_I} \paran{\sum_{\gamma_K = \pm \gamma_I} \abs{\Tr(\gamma_I \cdot \gamma_I)}^2 - \sum_{\gamma_{K'} = \pm \star \gamma_I} \abs{\Tr(\star \gamma_I \cdot \star \gamma_I)}^2 } \\
			\end{align*} where in the final line we have used Lemma \ref{lem:traces of products of gamma operators} to see that many terms of this sum vanish. We may use the observation $\abs{\Tr(\gamma_I\cdot \gamma_I)} = \abs{\Tr(\star\gamma_I\cdot \star\gamma_I)} = \Tr(\idty) =: d$ to continue:
            \begin{align*}
                \langle\varphi_+ (\gamma_I +  \star\gamma_I),  \varphi_-(\gamma_I + \star \gamma_I)\rangle &= (-1)^{s_I} \paran{\sum_{\gamma_K = \pm \gamma_I} d^2 - \sum_{\gamma_{K'} = \pm \star \gamma_I} d^2 } \\
				&= (-1)^{s_I} \paran{\binom{n}{\abs{I}} d^2 - \binom{n}{n-\abs{I}} d^2} \\
				&= 0, 
            \end{align*} where the second line follows from the crucial assumption that $\ell \geq n$ (in fact, it suffices to assume that $\ell\geq n-\abs{I}$). If we instead had $\ell< n$, it will be that any $\gamma_I$ with $\ell < \abs{I}\leq n$ will not be cancelled in the second line, leaving a nonzero overlap.
			
			To see that $\rho_+ = R^{\otimes \ell}\rho_-  R^{\otimes \ell}$, we need only combine Proposition (\ref{prop:commuting E with alpha identity}) with the eigenvector relationship $\ket{\varphi_+(\gamma_I +\star \gamma_I)} = \ket{\varphi_-(\alpha(\gamma_I + \star \gamma_I))}$ to write
			\begin{align*}
				\ket{\varphi_+^{(I)}}\bra{\varphi_+^{(I)}} &= \ket{\varphi_-(\alpha(\gamma_I +\star \gamma_I))}\bra{\varphi_-(\alpha(\gamma_I +\star \gamma_I))} \\
				&= R^{\otimes \ell} \ket{\varphi_-^{(I)}}\bra{\varphi_-^{(I)}} R^{\otimes \ell}.
			\end{align*}

            \vspace{5mm} 
			\noindent \textbf{Eigenvalues $\mu$ for fixed $\ell$ depend only on $\abs{I}$}

            We now demonstrate that the eigenvalues $\mu_{\, \abs{I}}$ depend only on $\abs{I}$ for a fixed length. 
			Thanks to the relation $R^{\otimes \ell}\rho_+ R^{\otimes \ell} = \rho_-$, we can restrict our attention to $\rho_+$.
			We proved earlier in Theorem \ref{thm:Invariance of ground states under SO(n)} that the state $\omega_+$ is invariant under local $SO(n)$ symmetry, and so this means for any $w\in SO(n)$, 
			\[
				w^{\otimes \ell}\rho_+ (w^{-1})^{\otimes \ell} = \rho_+.
			\] Now, by the simultaneous block diagonalization Lemma \ref{lem:sim block diag}, any irreducible invariant subspace of the representation $V^{\otimes \ell}$ must be contained in single eigenspace of $\rho_+$. In particular, Corollary \ref{cor:reps of even ground state spaces} reveals that the vector space $\text{span}\{\varphi_+^{(I)}: \abs{I} = k\}\cong \Exterior^k V$ (or $\cong U_+$ in the case of $k=n/2$) is an irrep of $SO(n)$. Thus each vector $\varphi_+^{(I)}$ with $\abs{I}=k$ is an eigenvector of $\rho_+$ of the same eigenvalue $\mu_{\abs{I}}$. 

            \vspace{5mm} 
			\noindent \textbf{Convergence of $\mu_{\,\abs{I}}(\ell) \to \mu$}
			
			We earlier diagonalized the transfer operators $\bbF_\idty^{(1)},\bbF_\idty^{(2)}$ in the proof of Proposition \ref{prop:F maps are primitive}.
			Noting that $\bbF_\idty^{(1)},\bbF_\idty^{(2)}$ are adjoints of each other with respect to the Hilbert-Schmidt inner product and are primitive CP maps with eigenvector $\idty$ corresponding to eigenvalue $\lambda = 1$, it is immediate that $\idty$ is an eigenvector of the adjoint transfer operators $\bbF_\idty^{(2)},\bbF_\idty^{(1)}$.
			Using purity of the states $\omega_\pm$, we can then apply Proposition 5.9 in \cite{fannes1992finitely} to argue that every eigenvalue $\mu_{\,\abs{I}}(\ell)$ must converge to the same value $\mu$. 
			By injectivity of these matrix product states as maps $\varphi_\pm(\cdot): P_+\calC_n^{[ev]}\to (\C^n)^{\otimes \ell}$ and the density matrix normalization $\Tr(\rho_\pm) = 1$, we can compute $\mu$ by computing the dimension of $P_+\calC_n^{[ev]}$: thus, $\mu = \frac{1}{2^{n-2}}$.
		\end{proof}

		% ------------------------------------------------------------------------
		% ------------------------------------------------------------------------
		\begin{corollary}\label{cor:rho +- are identical for small l}
		Let $\ell<n/2$. Then the reduced density matrices $\rho_\ell^\pm$ are identical:
		\[
			\rho_\ell^+ = \rho_\ell^- .
		\]
		\end{corollary}
		\begin{proof}
		We again compute coefficients in the matrix product states, as in the above proof. 
		We can again split coefficients like
		\[
			\Tr((\gamma_I+\star\gamma_I)\gamma_K) = \Tr(\gamma_I \gamma_K) + \Tr((\star\gamma_I) \gamma_K).
		\] But in this case, since $\abs{K}\leq \ell < n/2$ and since $\abs{I}\leq n/2$, $I$ and $K$ can not be the same index up to sign, so Lemma (\ref{lem:traces of products of gamma operators}) implies that $\Tr((\star\gamma_I) \gamma_K) = 0$. 
		This means that $\varphi_+(\star\gamma_I) = \varphi_-(\star\gamma_I) = 0$, so
		\begin{align*}
			\varphi_+^{(I)} &= \varphi_+(\gamma_I+\star\gamma_I) \\
			&= \varphi_+(\gamma_I) +0 \\
			&= (-1)^{s_I} \varphi_-(\gamma_I) \\
			&= (-1)^{s_I} \varphi_-(\gamma_I + \star\gamma_I) .
		\end{align*} where in the third line we have used relation (\ref{relation of + and - matrix product states via alpha}).
		So indeed, $\varphi_+^{(I)} = \pm \varphi_-^{(I)}$ and thus $\ket{\varphi_+^{(I)}}\bra{\varphi_+^{(I)}} = \ket{\varphi_-^{(I)}}\bra{\varphi_-^{(I)}}$, proving $\rho_\ell^+=\rho_\ell^-$.
		\end{proof}
        As a trivial consequence, for short chains $\ell<n/2$, we may combine $\rho_\ell^+ = \rho_\ell^-$ with $R^{\otimes \ell}\rho_\ell^+ R^{\otimes \ell} = \rho_\ell^-$ and $SO(n)$-invariance to see that $\rho_\ell^+$ is invariant under the entire $O(n)$ group. 
  
		\begin{remark}
		This result is sharp in the sense that when $\ell = n/2$, $\rho_\ell^+ \neq \rho_\ell^-$ and so dimerization is detectable. The key is that when $\ell=n/2$, the MPSs $\varphi_\pm^{(I)}$ are nonzero for $\abs{I} = n/2$, and the orthogonality result $\inprod{\varphi_+^{(I)}, \varphi_-^{(I)}} = 0$ goes through thanks to nonzero contributions from $\Tr((\star \gamma_I)\gamma_K)$ terms since $\abs{\star I} = n/2$. 
	    So as we raise $\ell$, eigenvectors common to $\rho_\ell^+$ and $\rho_\ell^-$ become distinct until the full support spaces of these matrices are orthogonal at $\ell = n$. 	
		\end{remark}
        \begin{remark}
            The orthogonality statement for these MPSs trivially implies that the injectivity length, i.e. the first chain length $\ell$ such that the map $\psi:\calC_n^{[ev]}\to (\C^n)^{\otimes \ell}$ are injective, is $\ell = n$. A similar computation reveals that for odd $n$, the maps $\psi:P_+\calC_n \to (\C^n)^{\otimes \ell}$ are injective when $\ell > n/2$, as a basis for $P_+\calC_n$ is given by $\gamma_I + \star \gamma_I$ with $\abs{I}<n/2$. 
        \end{remark}

    %--------------------------------------------

   %---------------------------------------------
    
\section{The \texorpdfstring{$SO(n)$}{SO(n)} Haldane Phase is a Nontrivial SPT Phase}\label{sec:two different SPT phases}
    Pick a cut of the lattice $\Z = \Z_L\cup \Z_R$, where $\Z_L = (-\infty,0]\cap \Z$ and $\Z_R = [1,\infty)\cap \Z$. Theorem~\ref{thm:split state index} guarantees that we may assign a well-defined index $h_\omega\in H^2(SO(n),U(1))$ to $\omega$ when $n$ is odd, and likewise a pair of indices $h_{\omega_\pm}$ for $\omega_\pm$ when $n$ is even.
    As discussed in Section~\ref{sec:excess spin operator}, we may use the excess spin operator construction to compute this index.\footnote{We could use the calculation in Section~\ref{sec:the H2 index for MPS} for odd $n$, but we need the result for 2-periodic MPSs as well.} In the odd $n$ case, $\omega$ is pure finitely correlated and in the even $n$ case, $\omega_\pm$ are pure 2-periodic finitely correlated--for brevity, we write $\omega$. The state $\omega$ has the GNS representation $(\calH_L\otimes \calH_R,\pi_L\otimes \pi_R,\Omega_\omega)$, and so Theorem~\ref{thm:bachmann excess spin FCS} guarantees that there is a strongly continuous representation $U^R:Spin(n)\to \calU(\calH_R)$. This then descends to a strongly continuous projective representation $U^R:SO(n)\to P\calU(\calH_R)$ to which we may associate a second group cohomology index $H^2(SO(n),U(1))$. Recall that $H^2(SO(n),U(1)) = \{1,\sigma\}$ from Theorem~\ref{thm:second group cohomology of semi-simple}.

    When $n$ is odd, $\omega$ is translation-invariant and in particular a pure finitely correlated state. Since the isometry $T$ defining $\omega$ enjoys the intertwining relation (\ref{eq:T an intertwiner}), we may then immediately apply Theorem~\ref{thm:rep content excess spin MPS} to see that the representation $U^R:SO(n)\to P\calU(\calH_R)$ consists only of (infinitely many copies of) the spin representation $\Pi:SO(n)\to P\calU(\C^{2^{(n-1)/2}})$ from Section \ref{sec:the spin representations Pi of SO(n)}, which is nontrivial projective by Remark~\ref{rem:spin reps are projective} and so corresponds to a nontrivial index $h_\omega = \sigma\in H^2(SO(n),U(1))$. 

    When $n$ is even, the story is the same, with the small subtlety that we must track the precise location of the cut. The index of either $\omega_+$ or $\omega_-$ may depend on the cut location, but of course the collection $\{h_{\omega_+},h_{\omega_-}\}$ does not. We may mimic the proof of Theorem~\ref{thm:rep content excess spin MPS}. Pick first $\omega:=\omega_+$. Recall the relations (\ref{eq:isometries T1 and T2 kinda commute with Pi}), which writing $U_g := g$ for all $g\in Spin(n)$ allows us to show
    \begin{equation}
        \bbF^{(1)}_{U_g}(\alpha(\Pi_+(g))) = \Pi_+(g), \quad \text{and}\quad  \bbF^{(2)}_{U_g}(\Pi_+(g)) = \alpha(\Pi_+(g)) . 
    \end{equation} In particular, $\Pi_+(g)$ is an eigenvector of eigenvalue 1 of the primitive map $\bbF^{(1)}_{U_g}\circ \bbF^{(2)}_{U_g}$ for all $g$ in a neighborhood of the identity $\calN\subseteq G$. Let $A,B\in \calA_L$, and for ease of notation let $\bbF_A$ be shorthand for the appropriately ordered products of $\bbF^{(1)}$ and $\bbF^{(2)}$. we may compute similarly to before:
    \begin{equation}\begin{split} \label{eq:matrix elements gns so(n) haldane}
        \inprod{\pi_L(A)\Omega_\omega, U_g^R \pi_L(B) \Omega_\omega} &= \inprod{\Omega_\omega, \pi_L(A^*B) U_g^R \Omega_\omega}\\
        &= \frac{2}{D}\Tr\paran{\rho\circ \bbF_{A^* B} \circ \bbF^{(1)}_{U_g} \circ \bbF^{(2)}_{U_g} \circ \bbF^{(1)}_{U_g} \circ \dots \circ \bbF^{(1)}_{U_g} \circ \bbF^{(2)}_{U_g}(\Pi_+(g)) }\\
        &= \frac{2}{D} \paran{ \rho \, \bbF_{A^* B} (\Pi_+(g) }  .
    \end{split}\end{equation}
    We may again apply the characterization of primitivity\footnote{If one desires a more careful discussion of primitivity for $p$-periodic maps, one may appreciate Proposition 3.3 in~\cite{fannes1992finitely}.} Theorem~\ref{thm:Primitivity} to see that
    \begin{equation}
        \text{span}\{\rho \circ \bbF^{(1)}_{A_{-\ell}} \circ \bbF^{(2)}_{A_{-\ell+1}} \dots \circ \bbF^{(2)}_{A_0}: \ell\in \N \text{ is odd}, A_i\in M_n(\C)\} = P_+\calC_n . 
    \end{equation}
    
    In particular, Equation (\ref{eq:matrix elements gns so(n) haldane}) gives a one-to-one correspondence between the matrix elements of the representatives $U_g^R$ and the matrix elements of $\Pi_+(g)$. Thus, $\Pi_+$ determines the full set of representation content of $U^R$. Since $U^R:G\to \calU(\calH_R)$ is infinite dimensional, it must contain the finite dimensional $\Pi_+:Spin(n)\to \calU(\C^{2^{n/2-1}})$ with infinite degeneracy.

    To adjust this proof for $\omega_-$, one simply observes that $\alpha(\Pi_+(g))$ is an eigenvector of eigenvalue 1 of the primitive map $\bbF^{(2)}_{U_g}\circ\bbF^{(1)}_{U_g}$ for all $g$ in a neighborhood of the identity $\calN\subseteq G$. Equation (\ref{eq:matrix elements gns so(n) haldane}) becomes

    \begin{equation}\begin{split}
    \inprod{\pi_L(A)\Omega_\omega, U_g^R \pi_L(B) \Omega_\omega}
    &= \frac{2}{D}\Tr\paran{\rho\circ \bbF_{A^* B} \circ \bbF^{(2)}_{U_g} \circ \bbF^{(1)}_{U_g} \circ \bbF^{(2)}_{U_g} \circ \dots \circ \bbF^{(2)}_{U_g} \circ \bbF^{(1)}_{U_g}(\alpha(\Pi_+(g))) }\\
        &= \frac{2}{D} \paran{ \rho \, \bbF_{A^* B} (\alpha(\Pi_+(g)) }  .
    \end{split}\end{equation} whence we similarly obtain that $U^R$ contains only the finite dimensional $\alpha(\Pi_+)$ with infinite multiplicity. Then, we wrap up by observing that the representations $\alpha(\Pi_+)\cong \Pi_-$ by Corollary~\ref{cor:alpha swaps spin reps}.

    In any case, the half-spin representations $\Pi_+$ and $\Pi_-$ both descend to nontrivial projective representations of $SO(n)$ as per Remark~\ref{rem:spin reps are projective}, and so have associated nontrivial cohomology indices $\sigma\in H^2(SO(n),U(1))$. Let us wrap up the calculations of this section in a single theorem.
    \begin{theorem} \label{thm:SOn haldane chains have nontrivial index}
        Let $G=SO(n)$ and recall the set of second cohomology classes $h_\calS$ from Theorem~\ref{thm:SPT invariants of ground state space}. 
        If $n$ is odd, $h_\calS = \{\sigma\}$ consists of a single nontrivial index. If $n$ is even, $h_\calS = \{\sigma, \sigma\}$ consists of a pair of nontrivial second cohomology classes.
    \end{theorem}
    Thanks to Corollary~\ref{cor:south pole ground state is trivial SPT}, we finally have the central result of this thesis.
    \begin{corollary}\label{cor:south pole and MPS point are different SPT phases}
        The south pole point and the $SO(n)$ AKLT point of the phase diagram in Figure~\ref{fig:SO(n) phase diagram} occupy different $SO(n)$-symmetry protected topological phases.
    \end{corollary}
    
    %--------------------------------------------
    
    \section{CPT symmetry and \texorpdfstring{$n\bmod{4}$}{n mod4} Behavior} \label{sec:CPT symmetry}
    We have now proven that $\omega_\pm$ are distinct ground states arising from the spontaneous symmetry breaking from $O(n)$ to $SO(n)$, and so we can now study how they react under other natural symmetries enjoyed by our Hamiltonian $H$. The Charge-Parity-Time Reversal (CPT) theorem in quantum field theory provides three symmetry operations of great interest~\cite{streater2000pct}. The CPT theorem suggests that the combined symmetries of charge conjugation, spatial reflection parity, and time reversal should be together unbroken, in that if one symmetry is broken, another symmetry must also be broken.\footnote{Of course, there is no \textit{a priori} reason to believe that the ground states of a spin system need obey the conditions of the CPT theorem.}
    We find the result of the CPT theorem to hold with an interesting $n\bmod{4}$ behavior. Just as in the AKLT chain~\cite{pollmann2010entanglement}, Proposition \ref{prop:time reversal invariance} verifies that for any even $n$, $\omega_\pm$ are invariant under time-reversal symmetry (this is immediate for odd $n$ by uniqueness of the ground state in the thermodynamic limit). Propositions \ref{prop:charge conjugation symmetry breaking} and \ref{prop:reflection parity symmetry breaking} together show that when $n\bmod{4} = 0$, the ground states $\omega_+$ and $\omega_-$ are invariant under charge conjugation and reflection parity symmetry, while when $n\bmod{4}=2$, both symmetries are broken by the pair of ground states $\omega_+$ and $\omega_-$.

    \subsection{Charge conjugation symmetry}
    We follow the mathematician's convention of denoting the complex conjugate of $z\in \C$ by $\overline{z}$. 
    \begin{proposition} \label{prop:charge conjugation symmetry breaking}
    Take observables $A_1\otimes \dots \otimes A_\ell \in \calA_{[1,\ell]}$.
    \begin{itemize}
        \item When $n\bmod{4} = 0$, charge conjugation symmetry is unbroken, i.e. 
        \begin{align*}
            \omega_+(\overline{A_1\otimes \dots \otimes A_\ell}) &= \overline{\omega_+(A_1\otimes \dots \otimes A_\ell)} \\ 
            \omega_-(\overline{A_1\otimes \dots \otimes A_\ell}) &= \overline{\omega_-(A_1\otimes \dots \otimes A_\ell)}.
        \end{align*}
        \item When $n\bmod{4} = 2$, charge conjugation symmetry is broken, i.e. \begin{align*}
            \omega_+(\overline{A_1\otimes \dots \otimes A_\ell}) &= \overline{\omega_-(A_1\otimes \dots \otimes A_\ell)} \\ 
            \omega_-(\overline{A_1\otimes \dots \otimes A_\ell}) &= \overline{\omega_+(A_1\otimes \dots \otimes A_\ell)}.
        \end{align*}
    \end{itemize}
    Equivalently, in terms of their reduced density matrices from Theorem \ref{thm:w+ and w- are distinct dimerized states for l large enough}
    \[
        \overline{\rho_\ell^\pm} = \begin{cases}
            \rho_\ell^\pm & n\bmod{4} = 0 \\
            \rho_\ell^\mp & n\bmod{4} = 2.
        \end{cases}
    \]
    \end{proposition}
    \begin{proof}
    We begin by recalling Remark \ref{rem:choice of alpha is arbitrary}, which says that if $t:\calC_n\to \calC_n$ is defined on products by $t(\gamma_1\gamma_2\dots \gamma_k) = \gamma_k \dots \gamma_2\gamma_1$, then
    \[
        t\circ \bbE_A = \bbE_{A^T} \circ t , \qquad \text{for all } A\in \calA_x,
    \] where $A^T$ is the usual matrix transpose. Then, observe that $t\circ \alpha = \alpha \circ t$, since for any $B\in \calC_n$,
    \[
        t\circ \alpha(B) = t(\gamma_1 B \gamma_1) = \gamma_1 t(B) \gamma_1 = \alpha \circ t(B),
    \] which then quickly means that for both $\bbF$ maps, 
    \[
        t\circ \bbF^{(1)}_A = \bbF^{(1)}_{A^T} \circ t, \text{ and } t\circ \bbF^{(2)}_A = \bbF^{(2)}_{A^T} \circ t.
    \]
    For convenience, work on even $\ell$, as the odd case is similar. Then, noting that $t^2 = \ide$ and letting $A_1\otimes \dots \otimes A_\ell \in \calA_{[1,\ell]}$, we have from Definition \ref{def:omega pos and neg, F maps} that
    \begin{align*}
    \omega_\pm(A_1\otimes \dots \otimes A_\ell) &= \frac{2}{D}\Tr( \bbF_{A_1}^{(1)}\circ \bbF_{A_2}^{(2)} \circ \dots \circ  \bbF_{A_{\ell-1}}^{(1)}\circ \bbF_{A_{\ell}}^{(2)}(P_\pm) ) \\
    &= \frac{2}{D}\Tr( \bbF_{A_1}^{(1)}\circ \bbF_{A_2}^{(2)} \circ \dots \circ  \bbF_{A_{\ell-1}}^{(1)}\circ \bbF_{A_{\ell}}^{(2)}\circ t^2 \circ (P_\pm) ) \\
    &= \frac{2}{D}\Tr( t \circ \bbF_{A_1^T}^{(1)}\circ \bbF_{A_2^T}^{(2)} \circ \dots \circ  \bbF_{A_{\ell-1}^T}^{(1)}\circ \bbF_{A_{\ell}^T}^{(2)}\circ t \circ (P_\pm) ) \\
    &= \frac{2}{D}\Tr( \bbF_{A_1^T}^{(1)}\circ \bbF_{A_2^T}^{(2)} \circ \dots \circ  \bbF_{A_{\ell-1}^T}^{(1)}\circ \bbF_{A_{\ell}^T}^{(2)}\circ t \circ (P_\pm) ),
    \end{align*} where the final equality holds by checking that $\Tr \, t(B) = \Tr B$ on the basis $\{\gamma_I:\abs{I}\leq n\}$ of $\calC_n$, since $t(\idty)= \idty$, and $t(\gamma_I)=\pm \gamma_I$. Now, note that when $n\bmod{4} =0$, $t(\gamma_0) = \gamma_0$ and thus $t(P_\pm) = P_\pm$, while when $n\bmod{4} =2$, $t(\gamma_0) = -\gamma_0$ and so $t(P_\pm) = P_\mp$. We then have
    \[
        \omega_\pm(A_1\otimes \dots \otimes A_\ell) = \begin{cases}
        \omega_\pm(A_1^T\otimes \dots \otimes A_\ell^T) 
        & n\bmod{4} = 0\\
        \omega_\mp(A_1^T\otimes \dots \otimes A_\ell^T) & n\bmod{4} = 2 , 
        \end{cases}
    \] which is equivalent to the claim since $\omega(\overline{A}^T) = \omega(A^*) = \overline{\omega(A)}$ for any $A\in \calA_{[1,\ell]}$. To see the statement for reduced density matrices, we write the case for $n\bmod{4}=0$ and $\omega_+$ and the others are analogous:
    \[
        \Tr \rho_\ell^+ A = \omega_+(A) = \omega_+(A^T) = \Tr\rho_\ell^+ A^T = \Tr A (\rho_\ell^+)^T = \Tr \overline{\rho_\ell^+} A.
    \]
    \end{proof}
    One could similarly prove this proposition by working on the reduced density matrices $\rho_\ell^\pm$ directly. The key is that for any simple product $\gamma_J$, the vector $\ket{\psi(\gamma_J)}$ from Definition \ref{def:MPSs of SO(n) chain} has real coefficients in the basis $\{\ket{i_1,\dots, i_\ell} \}$ of $(\C^n)^{\otimes\ell}$. Then, we can take an eigenvector $\ket{\varphi_+(\gamma_I + \star \gamma_I)}$ of $\rho_\ell^\pm$ and compute the complex conjugate of the orthogonal projector onto this eigenvector. Now, $\star \gamma_I = \gamma_0\gamma_I$ from Section \ref{secapp:hodge duals}. Then since $\gamma_0 = \gamma_1\dots \gamma_n$ when $n\bmod{4} = 0$ while $\gamma_0 = i\gamma_1\dots \gamma_n$ when $n\bmod{4} =2$, we can quickly compute 
    \[
        \overline{\ket{\varphi_+(\gamma_I + \star \gamma_I)}\bra{\varphi_+(\gamma_I + \star \gamma_I)}} = \begin{cases}
            \ket{\varphi_+(\gamma_I + \star \gamma_I)}\bra{\varphi_+(\gamma_I + \star \gamma_I)} & n\bmod{4} = 0 \\
            \ket{\varphi_+(\gamma_I - \star \gamma_I)}\bra{\varphi_+(\gamma_I - \star \gamma_I)} & n\bmod{4} = 2 ,
        \end{cases}
    \] which is equivalent to our result.

    This is compatible with the description of the ground state spaces $\calG_\ell^\pm$ as irreps of $\so(n)$ from Corollary \ref{cor:reps of even ground state spaces}.
    Here, by charge conjugation symmetry we mean symmetry under the map which sends a unitary representation $(\Pi,U)$ to its dual representation $(\overline{\Pi}, \overline{U})$. By unitarity, the dual representation is exactly the same as the conjugate representation. As a well known fact in Lie representation theory (see Theorem 3.E in~\cite{samelson1990notes}), one can express conjugation of finite dimensional irreps of a simple complex Lie algebra $\g$ as a special automorphism of its corresponding Dynkin diagram. In type $D_{n/2}$ ($\so(n)$ with $n$ even), this corresponds to a flip swapping the two dangling nodes. It turns out that when $n \bmod{4} = 0$, every irrep $(\pi,U)$ is self-dual. In particular, the two invariant subspaces $U_+,U_-$ of the representation spaces $\calG_\ell^+,\calG_\ell^-$ are self-dual, so $\overline{U_+} = U_+$ and $\overline{U_-} = U_-$. When $n\bmod{4} = 2$, not all irreps of $\so(n)$ are self-dual: in fact, $\overline{U_+} = U_-$.
    This means that as $\so(n)$ representations, 
    \[
        \overline{\calG_\ell^\pm} = \begin{cases}
            \calG_\ell^\pm & n \bmod{4} = 0 \\
            \calG_\ell^\mp & n \bmod{4} = 2
        \end{cases}
    \] While this is not enough to fully prove Proposition \ref{prop:charge conjugation symmetry breaking}, as distinct matrix product ground state spaces may be isomorphic as representations, it does strongly suggest this result.

    \subsubsection{$SU(4)$ quantum antiferromagnets with exact $C$-breaking ground states}
    As an aside, we briefly describe a connection to the work~\cite{affleck1991SU2n} and the closely related work~\cite{gozel2019novel}. In both cases the authors construct a pair of ground states which break $C$-symmetry for Hamiltonians which are $SU(2k)$-invariant. The case of $SU(4)$ coincides with the situation here, by virtue of the Lie algebra isomorphism $\su(4)\cong \so(6)$. Here, we think of each site $x\in \Z$ as supporting the antisymmetric square of the defining representation of $SU(4)$, the 6-dimensional self-conjugate representation $\Exterior^2 \C^4$. The Littlewood-Richardson rule informs us that we may diagrammatically use Young tableaux to decompose the tensor product of these two $SU(4)$ irreps~\cite{georgi1982lie}:
    \[
   \ydiagram{1,1} \otimes  \ydiagram{1,1} \cong \ydiagram{1,1,1,1} \oplus \ydiagram{2,1,1}\oplus \ydiagram{2,2}
    \]
    Realizing that $\Exterior^2 \C^4$ is isomorphic to a single copy of the defining representation $V=\C^6$ of $\so(6)$, the above equality is in fact the same as the two-site decomposition (\ref{eq:irrep decomposition of two site O(n) rep}):
    \[
    V\otimes V \cong \C \ket{\xi} \oplus \Exterior^2 V \oplus M_2.
    \]

    \subsection{Reflection Parity Symmetry}
    We now prove the promised statement for reflection parity.
    \begin{proposition}\label{prop:reflection parity symmetry breaking}
    Take observables $A_1\otimes \dots \otimes A_\ell \in \calA_{[1,\ell]}$. We will think of the axis of reflection as lying at $\ell/2$.
    \begin{itemize}
        \item When $n\bmod{4} = 0$, reflection parity symmetry is unbroken, i.e. 
        \begin{align*}
            \omega_+(A_1\otimes A_2\otimes  \dots \otimes A_{\ell-1} \otimes A_\ell) &= \omega_+(A_\ell\otimes  A_{\ell-1} \otimes \dots \otimes A_2 \otimes A_1) \\ 
            \omega_-(A_1\otimes A_2\otimes  \dots \otimes A_{\ell-1} \otimes A_\ell) &= \omega_-(A_\ell\otimes  A_{\ell-1} \otimes \dots \otimes A_2 \otimes A_1).
        \end{align*}
        \item When $n\bmod{4} = 2$, reflection parity symmetry is broken, i.e. 
        \begin{align*}
            \omega_+(A_1\otimes A_2\otimes  \dots \otimes A_{\ell-1} \otimes A_\ell) &= \omega_-(A_\ell\otimes  A_{\ell-1} \otimes \dots \otimes A_2 \otimes A_1)
        \end{align*}

        Equivalently, in terms of their reduced density matrices from Theorem \ref{thm:w+ and w- are distinct dimerized states for l large enough}, if $R_{\ell/2}$ is the linear map implementing reflection parity symmetry $R_{\ell/2} \ket{i_1,i_2,\dots, i_{\ell-1}, i_\ell} = \ket{i_\ell,i_{\ell-1},\dots, i_2, i_1}$, then
        \[
            R_{\ell/2} \, \rho_\ell^{\pm} \, R_{\ell/2} = \begin{cases}
                \rho_\ell^\pm & n\bmod{4} = 0 \\
                \rho_\ell^\mp & n\bmod{4} = 2 
            \end{cases}
        \]
    \end{itemize}
    \end{proposition}
    \begin{proof}
        As usual, we phrase the statement for $\ell$ is even, noting that the result for odd $\ell$ requires more careful sign treatment than the other results in this paper. Recall the map $t(\gamma_1\gamma_2 \dots \gamma_k) = \gamma_k \dots \gamma_2\gamma_1$ from the proof of Proposition \ref{prop:charge conjugation symmetry breaking}. Working on eigenvectors $\ket{\varphi_+(\gamma_I +\star \gamma_I)}$, $\abs{I}\leq n/2$ of the density matrix $\rho_\ell^+$, we see that
        \begin{align*}
            R_{\ell/2} \ket{\varphi_+(\gamma_I+\star\gamma_I)} &= \sum_{i_1,\dots, i_\ell} \Tr\paran{(\gamma_I + \star\gamma_I)\gamma_{i_\ell}\gamma_{i_{\ell-1}} \dots \gamma_{i_2} \gamma_{i_1} } R_{\ell/2} \ket{i_1, i_2, \dots , i_{\ell-1},i_\ell } \\
            &= \sum_{i_1,\dots, i_\ell} \Tr\paran{(\gamma_I + \star\gamma_I) \gamma_{i_1}\gamma_{i_{2}} \dots \gamma_{i_{\ell-1}} \gamma_{i_\ell} } \ket{i_1, i_2, \dots , i_{\ell-1},i_{\ell} } \\
            (\Tr \, t(A) = \Tr A) \qquad &= \sum_{i_1,\dots, i_\ell} \Tr\paran{t(\gamma_{i_\ell}\gamma_{i_{\ell-1}} \dots \gamma_{i_2} \gamma_{i_1})t(\gamma_I + \star\gamma_I) } \ket{i_1, i_2, \dots , i_{\ell-1},i_{\ell} } \\
            &= (-1)^{c}\sum_{i_1,\dots, i_\ell} \Tr\paran{t(\gamma_I + \star\gamma_I) \gamma_{i_\ell}\gamma_{i_{\ell-1}} \dots \gamma_{i_2} \gamma_{i_1} } \ket{i_1, i_2, \dots , i_{\ell-1},i_{\ell} }
        \end{align*}
        where in the final line we have used that $t(\gamma_{i_\ell}\gamma_{i_{\ell-1}} \dots \gamma_{i_2} \gamma_{i_1}) = (-1)^c (\gamma_{i_\ell}\gamma_{i_{\ell-1}} \dots \gamma_{i_2} \gamma_{i_1})$ for some integer $c$, just by using Clifford relations to reorganize. Now, we have that since $\gamma_0\in Z(\calC_n^{[ev]})$ and since $t(\gamma_0)=\gamma_0$ for $n\bmod{4}=0$ and $t(\gamma_0)=-\gamma_0$ for $n\bmod{4}=2$,
        \[
            t(\gamma_I + \star \gamma_I) = t(\gamma_I) + t(\gamma_I)t(\gamma_0) = \begin{cases}
            (-1)^c(\gamma_I + \gamma_0\gamma_I) & n\bmod{4}=0 \\
            (-1)^c(\gamma_I - \gamma_0\gamma_I) & n\bmod{4}=2 .
            \end{cases}
        \] Putting this all together, we have
        \[
            R_{\ell/2}\, \ket{\varphi_+(\gamma_I+\star \gamma_I)}\bra{\varphi_+(\gamma_I+\star \gamma_I)} R_{\ell/2} = \begin{cases}
            \ket{\varphi_+(\gamma_I+\star \gamma_I)}\bra{\varphi_+(\gamma_I+\star \gamma_I)}  & n\bmod{4}=0 \\
            \ket{\varphi_-(\gamma_I+\star \gamma_I)}\bra{\varphi_-(\gamma_I+\star \gamma_I)}  & n\bmod{4}=2 
            \end{cases}
        \]and thus our result.
    \end{proof}

    \subsection{Time Reversal Symmetry}
    Finally, we address the case of time reversal. To define this symmetry, we follow e.g. Appendix of~\cite{ogata2019lieb}.We can treat our local Hilbert space $\C^{n}$ as a spin-$s$ irrep of $SU(2)$, writing $n=2s+1$. It will be convenient to relabel our basis $\ket{1},\dots ,\ket{n}$ as $\ket{-s} ,\ket{-s+1}, \dots , \ket{s-1}, \ket{s}$ so that we have representatives $S_1, S_2, S_3\in M_{2s+1}(\C)$ of $SU(2)$ satisfying 
    \[
        S_3\ket{\mu} = \mu \ket{\mu}, \qquad (S_1\pm iS_2)\ket{\mu} = \sqrt{s(s+1) - \mu(\mu\pm 1)} \ket{\mu\pm 1} .
    \] Recall that an antilinear unitary map on a complex vector space $\theta:W\to W$ is a map which has for all $w,u\in W$ and $z\in\C$,
    \[
        \theta(w+u) = \theta(w)+\theta(u), \quad \theta(zw) = \overline{z} w, \quad \theta^* \theta = \idty .
    \] Define $\theta$ to be the antilinear unitary on $\C^{2s+1}$ via
    \begin{equation}\label{def:time reversal theta}
        \theta\ket{\mu} = (-1)^{s-\mu}\ket{-\mu}.
    \end{equation} One verifies that it has the desired property of ``time-reversing spin'', i.e.
    \[
        \theta^* S_j \theta = -S_j ,\qquad j=1,2,3. 
    \] This then induces an antilinear unital $*$-automorphism $\Xi$ on local observables 
    \begin{equation}\label{def:time reversal symmetry}
        \Xi(A) = (\theta^*)^{\otimes \ell} A (\theta)^{\otimes \ell}, \qquad A\in \calA_{[1,\ell]}.
    \end{equation} which is a bounded operator and thus extends uniquely to the algebra of quasilocal observables $\calA$. Here, thanks to antilinearity, the right condition for invariance of a state $\omega$ is that
    \[
        \omega(\Xi(A) ) = \omega(A^*), \quad A\in \calA_{[1,\ell]}.
    \]
    We are now in a position to state the final piece of CPT symmetry. Just like the AKLT chain, the ground states $\omega_\pm$ are unbroken under time reversal symmetry.

    \begin{proposition} \label{prop:time reversal invariance} 
        Take observables $A_1\otimes \dots \otimes A_\ell\in\calA_{[1,\ell]}$. Then for all even $n$, time reversal symmetry $\Xi$ (\ref{def:time reversal symmetry}) is unbroken, i.e.
        \begin{align*}
            \omega_+(\Xi(A_1\otimes \dots \otimes A_\ell)) &= \omega_+(A_1^*\otimes \dots \otimes A_\ell^*)\\
            \omega_-(\Xi(A_1\otimes \dots \otimes A_\ell)) &= \omega_-(A_1^*\otimes \dots \otimes A_\ell^*)
        \end{align*}
    \end{proposition}\begin{proof}
    
        Split any local observable in $\calA_{[1,\ell]}$ into $A+iB$ where $A,B$ real matrices in the spin basis $\{\ket{\mu}\}$ of $\C^{2s+1}$. Then by antilinearity,
        \[
            \Xi(A+iB) = \Xi(A)-i\Xi(B),
        \] and $\Xi$ is linear when restricted to real matrices. It is obvious that $\theta$ in Definition (\ref{def:time reversal theta}) has $\theta\in O(n)$ when we restrict to its real part, as its columns form an orthonormal basis. We will now inductively show that $\theta = \theta(n) \in SO(n)$,  where we have made explicit the dimension dependence $n=2s+1$. When $n=1$, we have $\theta(1) = \begin{pmatrix} 1\end{pmatrix}$ which trivially has $\det(\theta(1))=1$. When $n>1$, one can verify by hand the recursive equation
        \[
            \theta(n+1) = \begin{pmatrix}
                 & 1 \\
                 -\theta(n) & 
            \end{pmatrix},
        \] where it is then quick to see that $\det(\theta(n+1)) = \det(\theta(n))$ by cofactor expansion. In particular, we proved that both ground states $\omega_\pm$ are invariant under on-site $SO(n)$ transformations in Theorem \ref{thm:Invariance of ground states under SO(n)}, so 
        \[
        \omega_+(\Xi(A+iB)) = \omega_+(\Xi(A))-i\omega_+(\Xi(B)) = \omega_+(A) - i \omega_+(B) = \overline{\omega_+(A+iB)} = \omega_+((A+iB)^*),
        \] and likewise for $\omega_-$, which is what we wanted to show.
    \end{proof}

	%~~~~~~~~~~~~~~~~~~~~~~~~~~~~~~~~~~~~~~~~~~~~~
	\section{The Parent Property}\label{sec:parent property (SO(n) chains)}
    We will now prove for all $n\geq 2$ that the frustration-free Hamiltonian $H$ defined by the interaction $h=\idty + \SWAP - 2Q$ from (\ref{def:FF interaction h=1+SWAP-2Q}) is in fact a parent Hamiltonian for the family of MPSs $\psi_\ell$, extending the well known $n=3$ result of the AKLT chain Theorem \ref{thm:AKLT parent Hamiltonian}. As a reminder, $\SWAP(u\otimes v) = v\otimes u$ and $Q=\ket{\xi}\bra{\xi}$ where $\ket{\xi}= \frac{1}{\sqrt{n}}\sum_{i=1}^n\ket{ii}$. We will now make the chain length $\ell$ explicit in our notation, and in the even $n$ case we will not distinguish the states $\psi_+$ and $\psi_-$. To reiterate, our MPSs are maps $\psi_\ell:\calB \to (\C^n)^{\otimes \ell}$ defined by (\ref{def:MPSs of SO(n) chain}), where for odd $n$ we have $\calB = P_+\calC_n$ and for even $n$ we restrict to the subspaces $\calB = \calC_n^{[ev]}$ when $\ell$ is even and $\calB = \calC_n^{[odd]}$ when $\ell$ is odd (recall the discussion regarding this parity adjustment above Corollary \ref{cor:reps of even ground state spaces}). Theorem \ref{thm:w+ and w- are distinct dimerized states for l large enough} assures us that these MPSs are injective when $\ell\geq n$. The MPS spaces $\calG_\ell := \{\psi_\ell(B): B\in\calB\}$ are exactly as before.
    
    Recall Definition \ref{def:parent hamiltonian} of a parent Hamiltonian: given a family of MPSs $\psi_\ell$, a parent Hamiltonian $H_\ell = \sum_{x} h_x$ is defined by a frustration-free interaction supported on finitely many sites such that $\ker H_\ell = \calG_\ell$ for all $\ell$.
    In the special case where $h\geq 0$, frustration-freeness means that the kernel of $H$ is nonempty and enjoys the following intersection property:
    \begin{equation} \label{eq:intersection property (SO(n) chains)}
        \ker H_\ell = \ker \paran{\sum_x h_x} = \bigcap_{x} \ker{h_x} \neq \{0\} . 
    \end{equation}
    It was a quick computation to show in Lemma~\ref{lem:SO(n) MPS are ground states} that $\calG_{\ell}\subseteq \ker H_\ell$. But to see that this containment is saturated is less obvious. Our approach will heavily leverage the structure of these ground states as $SO(n)$ representations. We recall from Corollary \ref{cor:reps of odd ground state spaces} and Corollary \ref{cor:reps of even ground state spaces} that under the action of the tensor power of the defining representation $w^{\otimes \ell}$ of $SO(n)$, the MPS spaces $\calG_\ell$ admit the following decomposition into $SO(n)$ representations:
    \begin{equation}\label{eqn:decomp of irreps for G_l}
        \calG_\ell \cong \begin{cases} 
        \Exterior^0 V \oplus \Exterior^2 V \oplus \dots \oplus \Exterior^{n-2} V \oplus \Exterior^n V & \text{ if }\ell \text{ even } \\
        \Exterior^1 V \oplus \Exterior^3 V \oplus \dots \oplus \Exterior^{n-3} V \oplus \Exterior^{n-1} V & \text{ if }\ell \text{ odd }   ,
        \end{cases}
    \end{equation} where every representation here is irreducible, except $\Exterior^{n/2}V$ which splits into two distinct irreps of equal dimension $U_+\oplus U_-$. From here, a Pieri-type Lemma (\ref{lem:pieri rule for irreps of exterior x defining}) will allow us to understand the irreps arising on $\ell+1$ sites, and the intersection property (\ref{eq:intersection property (SO(n) chains)}) alongside liberal use of Schur's Lemma will allow us to prove that all representations except those arising from $\calG_{\ell+1}$ have nonzero energy. As a warm up, in Section \ref{sec:Preview the parent property for the AKLT chain} we revisit the proof of the parent property for the AKLT chain with a different approach, working more suggestively on the exterior powers of the defining representation of $V = \C^3$ of $\so(3)$. Since $\su(2)\cong \so(3)$, this are the same representations we encountered before, but this approach will more naturally lead to the general proof. We then demonstrate the proof for $n=4$ in Section~\ref{sec:the n=4 case proof of parent property} from which it is not terribly difficult to see the general $n$ case in Section~\ref{sec:parent property (SO(n) chains)}.

    \subsection{The Pieri-type lemma}
    The following lemma, upon which the entire proof centrally relies, is a variant of Pieri's formula, which in turn is a special case of a generalized version of the celebrated Littlewood-Richardson Rule. General formulae for decomposing tensor products of representations is a field unto itself\footnote{Let us mention a few alternative approaches. One could use the theory of crystal bases as in~\cite{nakashima1993crystal}, tableaux combinatorics as in~\cite{okada16pieri}, branching rules as in~\cite{king1975branching}, or Newell-Littlewood numbers as in~\cite{gao2021newell}.}, but it turns out this particular Lemma is much simpler to extract using elementary techniques~\cite{fulton1991representation}.\footnote{I would like to thank Eugene Gorsky for directing me to these results.} 
    
    Pieri's formula describes the decomposition of the tensor product of two $\sl(n,\C)$-irreps: the defining representation $V=\C^n$, and the \textit{Weyl module} (or \textit{Schur functor}) $\mathbb{S}_\lambda(V)$ for the partition $\lambda$ (Ch. 6 in~\cite{fulton1991representation}). Pieri's formula can be expressed diagrammatically using Young tableaux and ``adding and subtracting boxes'', as in~\cite{georgi1982lie}: for us, when $0<k<n$, the representation $\Exterior^k V$ is $\mathbb{S}_{(1^k)}(V)$, where the partition $\lambda = (1^k)$ corresponds to the Young diagram with a single column and $k$ rows. Pieri's formula then tells us that as $\sl(n,\C)$ representations,
    \begin{equation} \label{eq:sln pieri formula}
        \mathbb{S}_{(1^k)}(V) \otimes V \cong \mathbb{S}_{(1^{k+1})}(V) \oplus \mathbb{S}_{\mu}(V), \qquad \text{where }\mu = (2, \underbrace{1,1,\dots ,1}_{k-1 \text{ times }}).
    \end{equation} This decomposition may be similarly seen by considering the natural $\so(n)$ equivariant wedge map and considering its kernel and image. In our context, when chain length $\ell\geq k+1$, we may think of $\Exterior^k V\subseteq \calG_{\ell}$ as being $\text{span}\{ \psi_\ell(\gamma_{i_1}\dots \gamma_{i_k}): 1\leq  i_1< i_2 < \dots < i_k \leq n \}$, noting that $\Exterior^1 V \cong V$ here is directly implemented by $\psi_1(\gamma_i) = \ket{i}$. Consider the $\so(n)$-equivariant ``wedge map'' $\phi_\wedge :\Exterior^k V \otimes V \to \Exterior^{k+1} V$ given by
    \begin{equation}
        \phi_{\wedge} \paran{\psi_{\ell}(\gamma_{i_1}\dots \gamma_{i_k}) \otimes \psi_1(\gamma_j) }= \begin{cases}
            0 & \text{ if } j\in (i_1,\dots,i_k) \\
            \psi_{\ell+1}(\gamma_{i_1}\dots\gamma_{i_k}\gamma_j ) &\text{ if } j\not\in (i_1,\dots,i_k).
        \end{cases}
    \end{equation} The decomposition Equation (\ref{eq:sln pieri formula}) is exactly given by $\ker \phi_\wedge = \mathbb{S}_{\mu}(V)$ and $\im \phi_\wedge = \mathbb{S}_{(1^{k+1})}(V)$. Physically, we may think of states in the image $\im \phi_\wedge$ as those formed by antisymmetrically ``bonding'' a state $\psi_\ell(\gamma_K)$ with a particle on the $\ell+1$ site--indeed, they are in $\ker (\idty + \SWAP_{\ell,\ell+1})$.\footnote{Note that one cannot ``antisymmetrize'' forever, since $\Exterior^k V = 0$ if $k>n$. In other words, one cannot build purely antisymmetric ground states for chains of arbitrary length: there are no zero energy states for the interaction $\idty + \SWAP$ in the thermodynamic limit. To find interesting zero energy states, the $Q$ term really is crucial.}

    Treating $\so(n)$ as a subalgebra of $\sl(n)$, restricting to a representation of $\so(n)$ respects this direct sum decomposition. The first factor $\mathbb{S}_{(1^{k+1})}(V)=\Exterior^{k+1} V$ is irreducible when $k+1\neq n/2$ and otherwise splits into two different irreps $U_+\oplus U_- = \Exterior^{n/2} V$ of the same dimension. The second factor $\mathbb{S}_{\mu}(V)$ however is reducible as a representation of $\so(n)$, splitting into 
    \[
        \mathbb{S}_\mu(V) \cong \Exterior^{k-1} V \oplus M_k ,
    \] where these are irreducible except when $k\in [n/2-1,n/2+1]$. These can be found explicitly by a procedure known as Weyl's construction for orthogonal groups, which we will now mimic for our particular context. See Chapter 19.5 of~\cite{fulton1991representation} for more details and the general story.
    
    Let us use the maximally entangled state $\ket{\xi}$ to generate a family of contractions with an orthogonal bilinear form: let $1\leq p\leq k$ and write
    \begin{align*}
        c_{p}: \Exterior^k V\otimes V &\to \Exterior^{k-1}V \\
        \psi_\ell(\gamma_{i_1} \dots \gamma_{i_k})\otimes \psi_1(\gamma_m) &\mapsto \inprod{\xi \vert i_p,m} \psi_{\ell+1}(\gamma_{i_1} \dots \hat{\gamma_{i_p}} \dots \gamma_{i_k}),
    \end{align*} where $\hat{\gamma_{i_p}}$ denotes the omission of $\gamma_{i_p}$.\footnote{Note that since $i_1<i_2<\dots<i_k$, this is a more concise way of writing the usual contraction maps from Ch. 19 of~\cite{fulton1991representation}.} These contraction maps are indeed $\so(n)$-equivariant, since $R\otimes R\ket{\xi} = \ket{\xi}$ for any $R\in SO(n)$. Then, the representation $M_k\subseteq \ker \phi_\wedge$ is exactly
    \[
        M_k = \bigcap_{1\leq p \leq k} \ker (c_p) ,
    \] and the orthogonal complement of $M_k$ in $\ker \phi_\wedge$ can be seen to be isomorphic to $\Exterior^{k-1} V$. Physically, we may think of $\Exterior^{k-1} V$ as the family of states given by forming a maximally entangled bond between $\psi_\ell(\gamma_K)$ and a single particle on site $\ell+1$--these are in the kernel $\ker\idty-Q_{\ell,\ell+1}$--and we may think of $M_k$ as the family of states obtained by symmetric bonds orthogonal to the maximally entangled state $\ket{\xi}$--these are in the orthogonal complement $(\ker \idty -Q_{\ell,\ell+1})^\perp$. For example, in the case of $n=2$ and $\ell=1$, $M_1$ is spanned by the basis elements $\ket{i_1 m} + \ket{m i_1}, i_1 < m$ which are in $\ker c_1$, since $\inprod{\xi \vert i_1, m} = 0$. Then, $\Exterior^0 V = \C \ket{\xi}$, the maximally entangled bond across $\ell,\ell+1$, is exactly the orthogonal complement of $M_2$ in $\ker\phi_\wedge$.

    It can be shown be shown that these representations are irreducible for $\so(n)$ in all cases except when $k = n/2$ (in this case, one needs only a mild adjustment stemming from the split $\Exterior^{n/2} V\cong U_+\oplus U_-$). This is straightforward by the Weyl character formula, which allows us to compute the dimension of an irrep in terms of its highest weight:~\cite{fulton1991representation} in the language from Section~\ref{sec:highest weight reps}, the irreps $M_k$ correspond to the highest weight $2L_1 + L_2 + \dots + L_k$ when $k<n/2$ and the irreps $\Exterior^k V$ correspond to the highest weight $L_1+L_2+\dots + L_k$ when $k<n/2$, and when $k>n/2$ we have that $M_k\cong M_{n-k}$ and $\Exterior^k V\cong \Exterior^{n-k} V$ from Lemma~\ref{prop:ext k and ext n-k are isomorphic}\footnote{We only stated this result for $\Exterior^k V\cong \Exterior^{n-k}V$, but the result is the same for $M_k\cong M_{n-k}$. In all cases it is essential that these are representations of $\so(n)$, as this is not true for other Lie algebras, including $\sl(n)$.}
    
    \begin{lemma} (A Pieri-type rule)\cite{fulton1991representation, gao2021newell} \label{lem:pieri rule for irreps of exterior x defining}
    
    Let $V=\C^n$ be the defining representation of $\so(n)$. Then for all $0<k<n$, we have
    \[
        \Exterior^k V \otimes V \cong \Exterior^{k-1}V \oplus \Exterior^{k+1} V \oplus M_k . 
    \] 
    \begin{itemize}
        \item When $n$ is odd (Lie type $B_m$ where $n=2m+1$), all three representations on the right are irreducible. 
        \item When $n$ is even (Lie type $D_m$ where $n=2m$), all three representations are irreducible when $1<k<n/2-1$ or $n/2+1<k<n$. 
        \begin{itemize}
        \item If either $k + 1 = n/2$ or $k - 1 = n/2$, then $\Exterior^{n/2} V\cong U_+\oplus U_-$ with $U_+,U_-$ irreps of equal dimension as in Lemma \ref{lem:rep of Pi( )Pi^-1 on bond algebras} and $M_k$ is irreducible.
        \item If $k=n/2$, then $\Exterior^{k\pm 1} V$ is irreducible and $M_{n/2}\cong M_{n/2,+} \oplus M_{n/2,-}$ with $M_{n/2,\pm}$ irreps.
        \end{itemize}
    \end{itemize}

    When $k=0$ or $k=n$, then $\Exterior^0 V \cong \Exterior^n V \cong \textbf{1}$, the trivial representation, and so $\Exterior^0 V\otimes V\cong V\cong \Exterior^1 V$ and $\Exterior^n V \otimes V \cong V \cong \Exterior^{n-1} V$, which is irreducible.
    \end{lemma}

    \subsection{Preview: the parent property for the AKLT chain, $n=3$} \label{sec:Preview the parent property for the AKLT chain}
    When $n=3$, we have $\su(2)\cong \so(3)$. The ground state space for a single copy of the interaction $h_{1,2}$ is exactly given by (\ref{eq:two site ground states MPS}), which we may rewrite as
    \[
        \calG_2 \cong \Exterior^0 V \oplus \Exterior^2 V \cong \ker h_{1,2} . 
    \] To contact our earlier decomposition for the AKLT chain, $\Exterior^0 V$ is the $s=0$ irrep of $\su(2)$, and $\Exterior^2 V$ is the $s=1$ irrep of $\su(2)$. At this point, since $n=3$ is odd and $\ell=2 > n/2$, the MPS $\psi_2$ is injective. The intersection property (\ref{eq:intersection property (SO(n) chains)}) means that 
    \[
        \ker H_3 = \ker (h_{1,2} + h_{2,3}) = \ker h_{1,2} \cap \ker h_{2,3} \cong \paran{ (\Exterior^0 V \oplus \Exterior^2 V) \otimes V } \cap \ker h_{2,3}.
    \] 
    Clebsch-Gordan decomposition Proposition~\ref{prop:clebsch-gordan} may be viewed as a special case of the Pieri-type Lemma~\ref{lem:pieri rule for irreps of exterior x defining}. We see that $(\Exterior^0 V \oplus \Exterior^2 V) \otimes V \cong 2 \Exterior^1 V \oplus \Exterior^3 V \oplus M_2 $ where $M_2$ corresponds to a spin-2 representation. Note that Proposition \ref{prop:ext k and ext n-k are isomorphic} gives that $\Exterior^k V\cong \Exterior^{n-k} V$, here $\Exterior^0 V \cong \Exterior^3 V$ and $V\cong \Exterior^1 V \cong \Exterior^2 V$. Since $H$ commutes with the action of $\so(n)$, $H$ restricted to any irrep must have fixed eigenvalue by Schur's Lemma. One then computes directly on highest weight vectors in these representations to see that one spin-1 representation $\Exterior^1 V$ and the spin-2 representation $M_2$ have nonzero energy, so $\ker H_3 \subseteq \Exterior^1 V \oplus \Exterior^3 V$. But since $\calG_3 \cong \Exterior^1 V \oplus \Exterior^3 V$ and $\calG_3\subseteq \ker H_3$, this containment is an equality.

    %--------------------------------------------------
    \subsection{The proof of the parent property for $n=4$} \label{sec:the n=4 case proof of parent property}
    Let us trace the $n=4$ case of the proof of the parent property Theorem \ref{thm:parent property (SO(n) chains)}, leaving the details to the general proof. Here, $V=\C^4$ is the defining representation of $\so(4)$. It is not hard to show that $\calG_\ell = \ker h_{\ell}$ for $\ell< 4$, so let us begin at $\ell=4$, the injectivity length. Then the decomposition of the MPS space $\calG_\ell$ into $\so(4)$-representations (\ref{eqn:decomp of irreps for G_l}) reads as 
    \begin{align*}
        \calG_\ell \cong \Exterior^0 V \oplus \Exterior^2 V \oplus \Exterior^4 V\cong \Exterior^0 V \oplus U_+\oplus U_- \oplus \Exterior^4 V,
    \end{align*} In particular, using Section~\ref{secapp:comment for proof of parent property}, the corresponding highest weight vectors in the above decomposition are $\psi_4(\idty)$, $ \psi_4(f_1f_2)$, $ \psi_4(f_1\wt{f_2})$, and $ \psi_4(\star \idty)$. Then, using the intersection property, 
      \begin{equation}\begin{split} 
        \calG_{5} &\cong \paran{\Exterior^1 V \oplus \Exterior^3 V \oplus \Exterior^5 V} \\
        &\subseteq \ker H_{5} \\
        &= (\calG_{4}\otimes V)\cap \ker h_{4,5} \\
        &\cong \paran{(\Exterior^0 V\oplus \Exterior^2 V \oplus \Exterior^4 V)\otimes V }\cap \ker h_{4,5} \\
        &\cong \paran{2\paran{\Exterior^1 V \oplus \Exterior^3 V \oplus \Exterior^{5} V} \oplus M_2 }\cap \ker h_{4,5}.
    \end{split}\end{equation} where we have used the Pieri-type lemma~\ref{lem:pieri rule for irreps of exterior x defining} in the final line. 
    We need to show that this containment is saturated, by demonstrating that $M_2$ and a single copy of each odd $k$, $\Exterior^k V$ are orthogonal to $\ker h_{4,5}$ by irreducibility.

    The irrep decomposition of $M_2$ is $M_2\cong M_{2,+}\oplus M_{2,-}$ and we may find the highest weight vectors to be
    \begin{align*}
        v_{2,+} := \psi_4(f_1 f_2)\otimes \psi_1(f_1) \\
        v_{2,-} := \psi_4(f_1 \wt{f_2})\otimes \psi_1(f_1),
    \end{align*} which have respective weights $2L_1+L_2$ and $2L_1-L_2$. Applying the interaction term, we have $h_{4,5}\, v_{2,\pm}\neq 0$, and by irreducibility, $M_{2,\pm}$ are both orthogonal to $\ker{h_{4,5}}$.

    For $\Exterior^1 V$ and $\Exterior^3 V$, we take
    \begin{align*}
        w_1 &:= \psi_4(\idty)\otimes \psi_1(f_1) \\
        w_3 &:= \psi_4(\gamma_0)\otimes \psi_1(f_1).
    \end{align*} Grading considerations immediately give that $w_1$ and $w_3$ occupy different irreps from each other and from $M_{2,+},M_{2,-}$. One then computes that $h_{4,5} w_1\neq 0$ and $h_{4,5} w_3\neq 0$ and so $\ker h_{4,5}$ contains at most one copy of each $\Exterior^1 V$ and $\Exterior^3 V$.

    Now, let us move from $\ell=5$ to $\ell=6$. We again have by using the intersection property and the Pieri-type Lemma~\ref{lem:pieri rule for irreps of exterior x defining} that 
    \begin{equation}\begin{split} 
        \calG_{6} &\cong \paran{\Exterior^0 V \oplus \Exterior^2 V \oplus \Exterior^4 V} \\
        &\subseteq \ker H_{6} \\
        &= (\calG_{5}\otimes V)\cap \ker h_{5,6} \\
        &\cong \paran{(\Exterior^1 V\oplus \Exterior^3 V ) \otimes V }\cap \ker h_{5,6} \\
        &\cong \paran{\Exterior^0 V \oplus 2\paran{\Exterior^2 V} \oplus \Exterior^4 V \oplus (M_1\oplus M_3) }\cap \ker h_{5,6}.
    \end{split}\end{equation}
    For $M_1$ and $M_3$, we find the highest weight vectors
    \begin{align*}
        v_1 &:= \psi_5(\gamma_1)\otimes \psi_1(\gamma_1) \\
        v_3 &:= \psi_5(\star \gamma_3)\otimes \psi_1(\gamma_1) = -\psi_5(\gamma_1\gamma_2\gamma_4)\otimes \psi_1(\gamma_1),
    \end{align*} which both have highest weights $2L_1$. Applying the interaction term, we have $h_{5,6}v_{1}\neq 0$ and $h_{5,6}v_{3}\neq 0$, and by irreducibility, $M_{1}$ and $M_3$ are both orthogonal to $\ker{h_{5,6}}$.

    For $\Exterior^2 V \cong U_+\oplus U_-$, we start by defining
    \begin{align*}
        \wt{w_1} &:= \psi_5(f_1)\otimes \psi_1(f_2) \\
        \wt{w_3} &:= \psi_5(\star f_1)\otimes \psi_1(\wt{f}_2).
    \end{align*} Let us focus on $\wt{w_1}$. In this case, we need to perform Gram-Schmidt orthogonalization it has nontrivial overlap with $M_1$. By irreducibility $M_1$ is cyclically generated by the action of $[\pi,\cdot]$ on its highest weight vector $v_1$, so we see the only overlap occurs when $\pi(X) = \frac{1}{4} \wt{f_1}f_2$. In this case, this element of $M_1$ is
    \[
         \frac{1}{4}\paran{\psi_5([\wt{f_1}f_2,f_1])\otimes \psi_1(f_1) + \psi_5(f_1)\otimes \psi_1([\wt{f_1}f_2,f_1]) } = -(\psi_5(\gamma_2)\otimes \psi_1(\gamma_1) + \psi_5(\gamma_1)\otimes \psi_1(\gamma_2))
    \]and so Gram-Schmidt orthogonalizing $\wt{w}_1\in \Exterior^1 V\otimes V$ yields
    \begin{align*}
        w_1 &=  \psi_5(\gamma_2) \otimes \psi_1(\gamma_1) - \psi_5(\gamma_1) \otimes \psi_1(\gamma_2),
    \end{align*} which must be orthogonal to $M_1$, and by grading, to $M_3$. One then computes to see that $h_{5,6} w_1 \neq 0$. By a similar procedure, one acquires $w_3$ and then computes to see that $h_{5,6} w_3 \neq 0$. Grading/weight considerations force these to occupy different irreducibles in $\Exterior^2 V \subseteq \Exterior^1 V \otimes V$, and so $U_+$ and $U_-$ must both be orthogonal to $\ker h_{5,6}$.

    \subsection{The proof of the parent property for all $n$}
    \begin{theorem}\label{thm:parent property (SO(n) chains)}
        $H$ is a parent Hamiltonian for the MPS ground states, i.e. $\ker H_\ell = \calG_\ell$ for all chain lengths $\ell$.
    \end{theorem}
    \begin{corollary} \label{cor:parent prop means uniqueness}
        When $n$ is odd, $H$ has a unique gapped ground state $\omega$. When $n$ is even, $H$ has exactly two pure ground states $\omega_\pm$.
    \end{corollary}
    \begin{proof}
    We will demonstrate the proof for even $n$: as one may suspect from looking at Lemma \ref{lem:pieri rule for irreps of exterior x defining}, the case for odd $n$ is essentially an easier version of this proof, as it avoids the clumsy reducibility problems when $k=n/2$ and $k=n/2\pm 1$. The proof will proceed by induction on number of sites $\ell$. We will start at the injectivity length $\ell = n$ and discuss the straightforward adjustments for the shorter chains immediately after the proof. 

    Assume towards induction that $\ell\geq n$ and $\ker H_{\ell} = \calG_\ell$, defined in Eqn. (\ref{def:MPSs of SO(n) chain}).
    Now, by the intersection property for frustration-free interactions (\ref{eq:intersection property (SO(n) chains)}),
    \[
        \ker H_{\ell+1} = (\calG_\ell \otimes \C^n) \cap \ker h_{\ell,\ell+1} .
    \]
    Then, combining this with the decomposition of $\calG_{\ell+1}$ into $SO(n)$ irreps (\ref{eqn:decomp of irreps for G_l}), we have
    \begin{equation}\begin{split} \label{eqn:Parent property big containment}
        \calG_{\ell+1} &\cong \paran{\Exterior^1 V \oplus \Exterior^3 V \oplus \dots \oplus \Exterior^{n-3} V \oplus \Exterior^{n-1} V} \\
        &\subseteq \ker H_{\ell+1} \\
        &= (\calG_{\ell}\otimes V)\cap \ker h_{\ell,\ell+1} \\
        &\cong \paran{(\Exterior^0 V\oplus \Exterior^2 V \oplus \dots \oplus \Exterior^{n-2} V \oplus \Exterior^n V)\otimes V }\cap \ker h_{\ell,\ell+1} \\
        &\cong \paran{2\paran{\Exterior^1 V \oplus \Exterior^3 V \oplus \dots \oplus \Exterior^{n-3} V \oplus \Exterior^{n-1} V} \oplus \bigoplus_{\text{even } 0< k <n} M_k}\cap \ker h_{\ell,\ell+1},
    \end{split}\end{equation}
    where the final isomorphism is via the Pieri-type Lemma (\ref{lem:pieri rule for irreps of exterior x defining}).
    Our goal going forward is to show that the containment $\calG_{\ell+1} \subseteq \ker H_{\ell+1}$ is saturated by showing that the kernel $\ker h_{\ell,\ell+1}$ contains no $M_k$ representations and at most 1 copy of each $\Exterior^{k+1} V$ representation, which will then yield that $\calG_{\ell+1} = \ker H_{\ell+1}$ by the above. Irreducibility of each irrep $W$ is key to this argument, as each irrep must respect the eigendecomposition of $H_{\ell+1}$ since $H_{\ell+1}$ is $\so(n)$-invariant and so is either in $\ker h_{\ell,\ell+1}$ or orthogonal to it.
    
    \vspace{5mm} 
    \noindent \textbf{$\ker{h_{\ell,\ell+1}}$ contains no $M_k$ representations.}
    
    We start by assuming that $0<k<n/2$. Notice that each $M_k$ appears as $M_k \subseteq \Exterior^k V\otimes V$: in fact, since 
    \[
        2L_1 + L_2 +\dots + L_k = (L_1+ L_2 + \dots + L_k) + (L_1),
    \] its highest weight is the sum of the highest weights of $\Exterior^k V$ and $V$, and by a standard result in representation theory, its highest weight space is 1-dimensional and spanned by a single vector $v_k$ which is given by taking the tensor product of the highest weight vectors of $\Exterior^k V$ and $V$, respectively. Section \ref{secapp:comment for proof of parent property} details what these highest weight vectors are, and we are left with
    \begin{equation}\label{def:Mk highest weight vecs}
        v_k := \psi_\ell(f_1 f_2 \dots f_k) \otimes \psi_1(f_1) .
    \end{equation} Apply the interaction term $h_{\ell,\ell+1} = \idty + \SWAP - 2\ket{\xi}\bra{\xi}$ across the $\ell,\ell+1$ bond:
    \begin{align*}
        h_{\ell,\ell+1} \, v_k &= \sum_{i_1,\dots,i_\ell, i_{\ell+1}} \Tr(f_1 \gamma_{i_{\ell+1}}) \Tr\paran{ (f_1\dots f_k) \gamma_{i_{\ell}}\dots \gamma_{i_1}} \ket{i_1,\dots,i_\ell, i_{\ell+1}} \\
        & \qquad \qquad \;\; + \Tr(f_1 \gamma_{i_{\ell}}) \Tr\paran{ (f_1\dots f_k) \gamma_{i_{\ell+1}}\gamma_{i_{\ell-1}}\dots \gamma_{i_1}} \ket{i_1,\dots,i_\ell, i_{\ell+1}} \\
        & \quad  - \frac{2}{n}\delta_{i_{\ell+1},i_\ell}\sum_j \Tr(f_1 \gamma_{i_{\ell+1}}) \Tr\paran{ (f_1\dots f_k) \gamma_{i_{\ell}}\dots \gamma_{i_1}} \ket{i_1,\dots,j,j} .
    \end{align*} Now, examine the coefficients when $i_{\ell+1} = 1$ and $i_{\ell}= 3$. The third term corresponding to $\ket{\xi}\bra{\xi}$ is zero, and we have
    \begin{equation}\label{eqn:parent prop vk nonzero terms}
        \Tr(f_1 \gamma_1) \Tr\paran{(f_1 \dots f_k) \gamma_{3} \gamma_{i_{\ell-1}} \dots \gamma_{i_1}} + \Tr(f_1 \gamma_{3}) \Tr\paran{(f_1 \dots f_k) \gamma_{1} \gamma_{i_{\ell-1}} \dots \gamma_{i_1}}.
    \end{equation} Since $f_1 = \gamma_1 + i\gamma_2$, by Lemma \ref{lem:traces of products of gamma operators} we have $\Tr(f_1\gamma_3) = 0$ and so the second term is zero. Now, observe that
    \begin{equation} \label{eqn:parent prop f1f2...fk}
    f_1 f_2 \dots f_k = (\gamma_1\gamma_3 \gamma_5 \dots \gamma_{2k-1}) + (\text{other terms}) ,
    \end{equation} where the other terms $\Gamma$ have $\Tr (\gamma_1\gamma_3 \gamma_5 \dots \gamma_{2k-1} \Gamma) = 0$. So, picking the remaining $i_1, \dots, i_{\ell-1}$ to be such that $(\gamma_3 \gamma_{i_{\ell-1}} \dots \gamma_{i_1}) = \gamma_{2k-1} \dots \gamma_5 \gamma_3 \gamma_1$, which we can always do once $\ell\geq k$, we have that the first term of Equation (\ref{eqn:parent prop vk nonzero terms}) is equal to $\Tr(\idty)\Tr(\idty) = D^2$. Thus, this coefficient is nonzero, and so $h_{\ell,\ell+1}v_k \neq 0$. In particular, since there is only one copy of $M_k$ in the decomposition (\ref{eqn:Parent property big containment}), and since $M_k$ irreducible, Schur orthogonality relations assure us that it is orthogonal to $\ker h_{\ell,\ell+1}$ and thus to $\ker H_{\ell+1}$.

    When $k>n/2$, one can perform a similar calculation
    \[
        v_k := \psi_\ell(\star(f_1 \dots f_{n-k}))\otimes \psi_1(f_1),
    \] where after looking at Equation (\ref{eqn:parent prop f1f2...fk}) we see that
    \[
        \star(f_1f_2 \dots f_{n-k}) = \star(\gamma_1\gamma_3 \gamma_5 \dots \gamma_{2(n-k)-1}) + (\text{other terms)},
    \] where the other terms $\Gamma$ have $\Tr(\star(\gamma_1\gamma_3 \gamma_5 \dots \gamma_{2(n-k)-1})  \Gamma) = 0$. So, we examine the coefficients of $h_{\ell,\ell+1} v_k$ when $i_{\ell+1} = 1$ and $i_{\ell} = 2$. Then again the term corresponding to $\ket{\xi}\bra{\xi}$ vanishes, and similarly to Equation (\ref{eqn:parent prop vk nonzero terms}), we have
    
    \[
    \Tr(f_1 \gamma_1) \Tr\paran{\star(f_1 \dots f_{n-k}) \gamma_{3} \gamma_{i_{\ell-1}} \dots \gamma_{i_1}} + \Tr(f_1 \gamma_{3}) \Tr\paran{\star(f_1 \dots f_{n-k}) \gamma_{1} \gamma_{i_{\ell-1}} \dots \gamma_{i_1}}.
    \] The second term again dies since $\Tr(f_1 \gamma_3) = 0$, and picking $i_1,\dots, i_{\ell-1}$ such that 
    \[(\gamma_3 \gamma_{i_{\ell-1}} \dots \gamma_{i_1}) = \paran{\star(\gamma_1\gamma_3 \gamma_5\dots \gamma_{2(n-k)-1})}^{-1},
    \]which can be done as long as $\ell \geq (n-k)$, the first term becomes $D^2$, and so $h_{\ell,\ell+1} v_{k} \neq 0$, showing that $M_k$ is orthogonal to $\ker H_{\ell+1}$.

    When $k=n/2$, there is an inconsequential subtlety: here, the $M_{n/2}$ splits into two irreps $M_{n/2} \cong M_{n/2,+} \oplus M_{n/2,-}$. These arise as $M_{n/2,+} \subseteq U_+\otimes V$ and $M_{n/2,-} \subseteq U_-\otimes V$, and so just as before, they have highest weights $2L_1 + L_2 +\dots + L_{n/2-1} + L_{n/2}$ and $2L_1 + L_2 +\dots + L_{n/2-1} - L_{n/2}$ with corresponding highest weight vectors
    \begin{align*}
        v_{n/2,+} &:= \psi_\ell(f_1 \dots f_{n/2-1} f_{n/2} )\otimes \psi_1(f_1) \\
        v_{n/2,-} &:= \psi_\ell(f_1 \dots f_{n/2-1} \wt{f}_{n/2} )\otimes \psi_1(f_1).
    \end{align*} The above computation still works and $h_{\ell,\ell+1} v_{n/2,\pm} \neq 0$. Using irreducibility, we again have that $M_{n/2,\pm}$ is orthogonal to $\ker{H_{\ell+1}}$.

    %---------------------------------------------
    \vspace{5mm}
    \noindent \textbf{$\ker{h_{\ell,\ell+1}}$ contains at most one copy of each exterior power $\Exterior^{(\cdot)} V$.}
    
    We again first demonstrate the case where $ k+1 <n/2$.
    In this case, we will show that for even $k$, one of the $\Exterior^{k+1} V$ must be orthogonal to $\ker h_{\ell,\ell+1}$.
    
    Due to the degeneracy of each $\Exterior^{k+1} V$ term for even $k$, there are many possible choices for vectors of highest weight, and indeed, it is difficult to find vectors which are orthogonal to $\calG_{\ell+1}$. However, all we really need to do though is find one vector $w_k$ in each pair $\Exterior^{k+1} V\oplus \Exterior^{k+1} V$ with $w_k\not\in \ker h_{\ell,\ell+1}$, which by will guarantee at least one of the irreps will be orthogonal to $\ker(h_{\ell,\ell+1})$. One way to accomplish this is to just find $w_k$'s that are orthogonal to every $M_{k'}$ subspace with nonzero energy.
    By irreducibility each $M_{k'}$ subspace is cyclically generated by the action of $[\pi,\cdot]$ on its highest weight vector $v_{k'}$ (\ref{def:Mk highest weight vecs}), i.e.
    \begin{equation} \label{eqn:Mk cyclic representation}
        M_{k'} =
        \{\psi_\ell([\pi(X),f_1\dots f_{k'}])\otimes \psi_1(f_1) + \psi_\ell(f_1\dots f_{k'})\otimes \psi_1([\pi(X),f_1]): X\in \so(n)\}.
    \end{equation} Just as before, the image of this Lie algebra representation is determined by generators $\gamma_J$, or equivalently by even products $\pi(X) = f_J\wt{f}_{\bar{J}}$ where $\abs{J}+\abs{\bar{J}}$ is even. Crucially, $[\pi,(\cdot)]$ cannot change the grading of a product of gamma operators and so every element in $M_{k'}$ is a linear combination of $\psi_{\ell}(f_I)\otimes \psi_{1}(f_i)$ with $\abs{I}=k'$. With this in mind, define $\wt{w_k} \in \Exterior^k V \otimes V$ by
    \[
        \wt{w_k} := \psi_\ell(f_1 \dots f_k)\otimes \psi_1(f_{k+1}) . 
    \] By looking at the grading this is orthogonal to $M_{k'}$ with $k\neq k'$. Indeed, this means that when $k=0$, $\wt{w_k} = \psi_\ell(\idty) \otimes \psi_1(f_1)$ is immediately orthogonal to all $M_{k'}$. But if $k=k'$, we need to perform Gram-Schmidt orthogonalization.
    We first observe the orthgonality relations, which quickly follow from recalling that $\psi_1(\gamma_i) = \ket{i}$ forms an orthonormal basis of $\C^n$: when $i\neq j$,
    \begin{equation} \label{eqn:parent prop orthogonality of fi basis}
    \inprod{\psi_1(f_i),\psi_1(f_j)} = 0, \quad \inprod{\psi_1(f_i),\psi_1(\wt{f}_j)} = 0, \text{ and } \inprod{\psi_1(f_i),\psi_1(\wt{f}_i)} = 0.
    \end{equation}
    It is also useful to note the anticommutation relations satisfied by the $f_i,\wt{f}_i$: again letting $i\neq j$,
    \begin{equation} \label{eqn:parent prop anticommutation relations f}
    \{f_i,f_j\} =0, \quad \{f_i,\wt{f}_j\} =0,\quad \{f_i,f_i\} =0, \text{ and } \{f_i,\wt{f}_i\} = 4\idty .
    \end{equation}
    Then, using the form of vectors $M_k$ furnished by Equation (\ref{eqn:Mk cyclic representation}), one sees that due to the orthogonality relations (\ref{eqn:parent prop orthogonality of fi basis}) $\wt{w_k}$ only has nonzero overlap with $M_{k}$ when $\pi(X)=\frac{1}{4} \wt{f}_1 f_{k+1}$ since $\frac{1}{4}[\wt{f}_1 f_{k+1},f_1] = f_{k+1}$ and so Gram-Schmidt leads us to define
    \begin{align*}
        w_k :&= \wt{w_k} - \frac{c}{4}(\psi_\ell([\wt{f}_1 f_{k+1}, f_1 \dots f_k])\otimes \psi_1(f_1) + \psi_\ell(f_1 \dots f_k)\otimes \psi_1([\wt{f}_1f_{k+1}, f_1])) \\
        &= c_1 \psi_\ell(f_1\dots f_k)\otimes \psi_1(f_{k+1}) + c_2 \psi_\ell(f_2\dots f_k f_{k+1})\otimes \psi_1(f_1), 
    \end{align*} where the equality follows by liberal use of anticommutation relations (\ref{eqn:parent prop anticommutation relations f}) to see
    \begin{align*}
    \frac{1}{4}[\wt{f}_1f_{k+1}, f_1f_2 \dots f_k] &= \frac{1}{4} (\wt{f}_1f_{k+1}f_1f_2 \dots f_k - f_1f_2 \dots f_k\wt{f}_1f_{k+1}) \\
    &= \frac{1}{4} (\wt{f}_1f_{k+1}f_1 - f_1\wt{f}_1f_{k+1}) f_2 \dots f_k \\
    &= \frac{1}{4} \{\wt{f}_1,f_1\} f_{k+1} f_2 \dots f_k \\
    &= (-1)^{k-1}f_2 \dots f_k f_{k+1}.
    \end{align*} In particular, it is clear by inspection that picking $c_1 = 1$ and $c_2 = (-1)^{k}$ will yield a nonzero $w_k$ orthogonal to $M_k$:
    \begin{equation}
        w_k = \psi_\ell(f_1\dots f_k)\otimes \psi_1(f_{k+1}) + (-1)^{k}\psi_\ell(f_2\dots f_k f_{k+1})\otimes \psi_1(f_1).
    \end{equation} So, $w_k$ is orthogonal to every $M_{k'}$ representation. Let us now compute the bond $\ell,\ell+1$ interaction acting on $w_k$:
    \begin{align*}
        h_{\ell,\ell+1} \, w_k &= \sum_{i_1,\dots,i_{\ell+1}} \Big(\Tr(f_{k+1} \gamma_{i_{\ell+1}})  \Tr\paran{ (f_1 \dots f_k) \gamma_{i_{\ell}}\dots \gamma_{i_1}}  \\
        & \qquad\qquad\qquad\qquad + (-1)^k\Tr(f_{1} \gamma_{i_{\ell+1}}) \Tr\paran{ (f_2 \dots f_{k+1})  \gamma_{i_{\ell}}\dots \gamma_{i_1}} \Big) \ket{i_1,\dots,i_\ell, i_{\ell+1}} \\
        & \qquad  +\Big( \Tr(f_{k+1} \gamma_{i_{\ell}}) \Tr\paran{ (f_1 \dots f_k) \gamma_{i_{\ell+1}}\gamma_{i_{\ell-1}}\dots \gamma_{i_1}} \\
        & \qquad\qquad\qquad\qquad + (-1)^k\Tr(f_{1} \gamma_{i_{\ell}}) \Tr\paran{ (f_2 \dots f_{k+1}) \gamma_{i_{\ell+1}}\gamma_{i_{\ell-1}} \dots \gamma_{i_1}} \Big) \ket{i_1,\dots,i_{\ell},i_{\ell+1}} \\
        & \quad - \frac{2}{n}\delta_{i_{\ell+1},i_\ell}\sum_j \Big( \Tr(f_{k+1} \gamma_{i_{\ell+1}}) \Tr\paran{ (f_1 \dots f_k) \gamma_{i_{\ell}}\dots \gamma_{i_1}} \\
        & \qquad\qquad\qquad\qquad + (-1)^k\Tr(f_{1} \gamma_{i_{\ell+1}}) \Tr\paran{ (f_2 \dots f_{k+1}) \gamma_{i_{\ell}}\dots \gamma_{i_1}} \Big)\ket{i_1,\dots,j,j} .
    \end{align*} We proceed similarly to before. Examine the coefficients with $i_{\ell+1} = 1$ and $i_{\ell} = 2$. The third term corresponding to $\ket{\xi}\bra{\xi}$ is again zero, and we have 
    \begin{equation} \begin{split}
    &\Tr(f_{k+1} \gamma_{1}) \Tr\paran{ (f_1 \dots f_k) \gamma_{2} \gamma_{i_{\ell-1}}\dots \gamma_{i_1}} + (-1)^k\Tr(f_{1} \gamma_{1}) \Tr\paran{ (f_2 \dots f_{k+1}) \gamma_{2}\gamma_{i_{\ell-1}}\dots \gamma_{i_1}} \\
    & + \Tr(f_{k+1} \gamma_{2}) \Tr\paran{ (f_1 \dots f_k) \gamma_{1}\gamma_{i_{\ell-1}}\dots \gamma_{i_1}} + (-1)^k\Tr(f_{1} \gamma_{2}) \Tr\paran{ (f_2 \dots f_{k+1}) \gamma_{1}\gamma_{i_{\ell-1}} \dots \gamma_{i_1}} 
    \end{split}\end{equation} Since $k>0$, $\Tr(f_{k+1}\gamma_1) = \Tr(f_{k+1}\gamma_2) = 0$, and so this becomes 
    \begin{equation}\begin{split} \label{parent prop coefficient for bad exterior power}
    &(-1)^k\Tr(f_{1} \gamma_{1}) \Tr\paran{ (f_2 \dots f_{k+1}) \gamma_{2}\gamma_{i_{\ell-1}}\dots \gamma_{i_1}} + (-1)^k\Tr(f_{1} \gamma_{2}) \Tr\paran{ (f_2 \dots f_{k+1}) \gamma_{1}\gamma_{i_{\ell-1}} \dots \gamma_{i_1}} \\
    &= (-1)^k \paran{ D \, \Tr\paran{ (f_2 \dots f_{k+1}) \gamma_{2}\gamma_{i_{\ell-1}}\dots \gamma_{i_1}} + iD \, \Tr\paran{ (f_2 \dots f_{k+1}) \gamma_{1}\gamma_{i_{\ell-1}} \dots \gamma_{i_1}}}
    \end{split} \end{equation}   
    Similarly to (\ref{eqn:parent prop f1f2...fk}), we have 
    \[
    f_2 \dots f_{k+1} = \gamma_3 \gamma_5 \dots \gamma_{2k+1} + (\text{other terms}),
    \] where the other terms $\Gamma$ have $\Tr(\gamma_3\gamma_5 \dots \gamma_{2k+1}\Gamma) = 0$. So, if we pick $i_1, \dots, i_{\ell-1}$ such that $\gamma_{i_{\ell-1}} \dots \gamma_{i_1} = \gamma_2(\gamma_3 \gamma_5 \dots \gamma_{2k+1})$, the coefficient (\ref{parent prop coefficient for bad exterior power}) becomes
    \[
        (-1)^k \paran{D^2 + 0} \neq 0,
    \] which means that $h_{\ell,\ell+1} w_k \neq 0$. 
    
    To wrap up the $k+1<n/2$, we realize that each $w_k\in \Exterior^k V\otimes V$ has different grading and since $[\pi,\cdot]$ respects grading, they must all occupy different irreducibles, and so $\ker h_{\ell,\ell+1}$ contains at most one copy of each $\Exterior^{k+1}V$ for even $k+1 < n/2$.

    To handle the case where $k-1>n/2$ is similar, and in this case we show that $\ker h_{\ell,\ell+1}$ contains at most one copy of each $\Exterior^{k-1} V$. One starts first with $\wt{w}_k \in \Exterior^k V \otimes V$ given by
    \[
        \wt{w}_k := \psi_{\ell}(\star(f_1 \dots f_{n-k}))\otimes \psi_{1}(f_{(n-k)+1}).
    \] One again observes that $\wt{w}_k$ is orthogonal to any $M_k'$ with $k\neq k'$, and performs Gram-Schmidt in the case where $k=k'$. The Gram-Schmidt procedure is essentially the same as the previous case, since any $\so(n)$ representative $\pi(X)$ is in $\calC_n^{[ev]}$ and thus commutes with $\gamma_0$, so by (\ref{eqn:hodge dual is gamma0 multiplication}) we have
    \[
        [\pi(X), \star(f_1 \dots f_{n-k})] = \star [\pi(X), f_1 \dots f_{n-k}].
    \] We eventually arrive at the orthogonalized $w_k$ with $h_{\ell,\ell+1} w_k \neq 0$, showing $\ker H_{\ell,\ell+1}$ contains at most one copy of $\Exterior^{k-1}V$.

    Finally, we must address the cases where $k+ 1 = n/2$ and $k-1 = n/2$. In these cases, $\Exterior^{n/2}V\cong U_+\oplus U_-$ and all $M_{k'}$ are irreducible. One then picks $\wt{w}_{+} \in \Exterior^{n/2-1} V \otimes V$ and $\wt{w}_{-} \in \Exterior^{n/2+1} V \otimes V$ to be 
    \begin{align*}
        \wt{w}_{+} & := \psi_{\ell}(f_1 \dots f_{n/2-1})\otimes  \psi_{1}(f_{n/2}) \\
        \wt{w}_{-} & := \psi_{\ell}(\star(f_1 \dots f_{n/2-1}))\otimes  \psi_{1}(\wt{f}_{n/2}),
    \end{align*} which, after orthogonalizing, we notice that $w_{+}$ and $w_{-}$ have respective weights $L_1 + \dots + L_{n/2 -1} + L_{n/2}$ and $L_1 + \dots + L_{n/2-1} - L_{n/2}$. The first weight belongs exclusively to the weight lattice of $U_+$ and the second exclusively to $U_-$. In both cases, $h_{\ell,\ell+1}w_{+}\neq 0 $ and $h_{\ell,\ell+1}w_{-}\neq 0 $, proving $\ker H_{\ell+1}$ contains at most one copy of $\Exterior^{n/2}V$.

    \vspace{5mm}
    \noindent\textbf{From $\ell+1$ to $\ell+2$}
    
    To complete the induction, the parity dependence on $\ell$ means we must explain how to move from $\ell+1$ to $\ell+2$. Let us mimick the series of equations \ref{eqn:Parent property big containment}. Using the intersection property and the Pieri-type Lemma \ref{lem:pieri rule for irreps of exterior x defining}, we have 
    \begin{equation}\begin{split}
        \calG_{\ell+2} &\cong \paran{\Exterior^0 V \oplus \Exterior^2 V \oplus \dots \oplus \Exterior^{n-2} V \oplus \Exterior^{n} V} \\
        &\subseteq \ker H_{\ell+2} \\
        &= (\calG_{\ell+1}\otimes V)\cap \ker h_{\ell+1,\ell+2} \\
        &\cong \paran{(\Exterior^1 V\oplus \Exterior^3 V \oplus \dots \oplus \Exterior^{n-3} V \oplus \Exterior^{n-1} V)\otimes V }\cap \ker h_{\ell+1,\ell+2} \\
        &\cong \paran{\Exterior^0 V \oplus 2\paran{\Exterior^2 V \oplus \dots \oplus \Exterior^{n-2} V} \oplus \Exterior^{n} V \oplus \bigoplus_{\text{odd } 0< k <n} M_k}\cap \ker h_{\ell+1,\ell+2}.
    \end{split}\end{equation}
    But here, except the interesting fact that $\Exterior^0 V$ and $\Exterior^n V$ are multiplicity free and appear in the ground state space, the proof is very much the same: show that $\ker h_{\ell+1,\ell+2}$ contains no $M_k$ representations and at most one copy of each $\Exterior^{k+1} V$ for odd $k<n/2-1$, at most one copy of each $\Exterior^{k-1} V$ for odd $k>n/2+1$, and at most one copy of $\Exterior^{n/2}V$, if it appears. Changing each mention of ``even $k$'' in the previous sections to ``odd $k$'', the arguments go through unscathed, and we arrive at 
    \[
        \calG_{\ell+2} = \ker H_{\ell+2}.
    \] By induction on chain length $\ell$, we have that
    \[
        \calG_{\ell} = \ker H_{\ell},
    \] and so $H_\ell$ is a parent Hamiltonian for the MPSs $\psi_\ell$. 
    \end{proof}
    To adjust this proof to handle cases where $\ell<n$, one can inductively build up the ground state space $\calG_\ell$ and similarly kill off $M_k$ irreps and single copies of $\Exterior^{k'}V$ irreps where $k'\leq n$ by showing they have nonzero energy across the bond $h_{\ell,\ell+1}$. The key point to accomplishing this in each case was the existence of inverses $\gamma_{I_{\ell}}:= \gamma_3 \gamma_{i_{\ell-1}} \dots \gamma_{i_1}$ for particular strings of gamma operators, as in Equation (\ref{eqn:parent prop f1f2...fk}): there, we needed the inverse of $\gamma_{2k-1} \dots \gamma_5\gamma_3\gamma_1$, which can be found explicitly as long as we have  $k \leq \abs{I_{\ell}}=\ell$. The same step was required after taking the Hodge dual $\star$, where we could find the requisite inverse as long as $n-k\leq \abs{I_\ell} = \ell$. In either case, we do not need the full injectivity length, just a sufficiently large chain $\ell$. Assuming that we are beyond the injectivity length, $\ell \geq n$, cleans up the argument and allows us to not treat each irrep separately, but the proof still goes through with only aesthetic modifications to handle this case. 
    
    This provides an interesting counterexample to a natural conjecture. One may suspect that MPSs with increasing correlation lengths or injectivity lengths may require parent Hamiltonians with increasingly large interaction lengths: indeed, as we saw in Section~\ref{sec:The Parent Property}, the standard construction of a parent Hamiltonian for a family of MPSs grants interaction terms $h$ whose support scales with the injectivity length of the MPS. But combining this result with Corollary \ref{cor: correlation length diverges as n increases}, we realize that as we increase local dimension $n$, the $\psi_\ell$ form a family of MPSs with arbitrarily large correlation length which all have parent Hamiltonians given by a nearest-neighbor interaction, i.e. that each term $h_{x,x+1}$ acts at most on two adjacent sites.
    \begin{corollary} \label{cor:parent prop counterexample cor length}
    There exists a family of MPSs $\psi_\ell$ with arbitrarily large correlation length and injectivity length whose parent Hamiltonians are given by nearest-neighbor interactions.
    \end{corollary}

\printbibliography
% \appendix
\end{document}